\documentclass[botnum, fleqn, final]{unmeethesis}

\usepackage{amsthm,amsmath, amssymb, color,graphicx,subfigure,geometry}
\usepackage{hyperref}
\usepackage{verbatim,algorithmic}
\usepackage[section,boxed]{algorithm}
\usepackage{boxedminipage}
\usepackage{epigraph}
\usepackage{enumerate}
\usepackage{threeparttable}
\usepackage{bbding}
\usepackage{array}

\newtheorem{lemma}{Lemma}[chapter]

\newtheorem{corollary}{Corollary}
\newtheorem{theorem}{Theorem}[chapter]
\newtheorem{observation}{Observation}[chapter]

\renewcommand{\epigraphsize}{\footnotesize}
\renewcommand{\epigraphrule}{0pt}
\newcommand{\epitext}{\textsf }	% The font for the text.
\newcommand{\epicharacter}{\textrm}	% the character uttering the text.
\newcommand{\episource}{\textit}	% The book or the source.
\newcommand{\epiauthor}{\textsc} % The author, if it is a book, poem etc.	

\DeclareGraphicsExtensions{.pdf}

\newcommand{\DASH}{\mathrm{DASH}}
\newcommand{\SDASH}{\mathrm{SDASH}}
\newcommand{\rem}{\mathrm{rem}}
\newcommand{\G}{\mathrm{G}}
\newcommand{\W}{\mathrm{W}}
\newcommand{\T}{\mathrm{T}}
\newcommand{\RT}{\mathrm{RT}}
\newcommand{\RTree}{\mathrm{Reconstruction\ Tree}}

\newcommand{\FT}{\mathrm{FT}}
\newcommand{\SubRT}{\mathrm{SubRT}}

\newcommand{\RTparent}{\mathrm{RTparent}}

\newcommand{\ID}{\mathrm{ID}}

\newcommand{\will}{\mathrm{will}}
\newcommand{\heir}{\mathrm{heir}}
\newcommand{\helper}{\mathrm{helper}}
\newcommand{\parent}{\mathrm{parent}}
\newcommand{\hparent}{\mathrm{hparent}}

\newcommand{\children}{\mathrm{children}}
\newcommand{\hchildren}{\mathrm{hchildren}}
\newcommand{\isreadyheir}{\mathrm{isreadyheir}}
\newcommand{\Degree}{\mathrm{degree}}
\newcommand{\diam}{\mathrm{diam}}
\newcommand{\depth}{\mathrm{depth}}
\newcommand{\delanc}{\mathrm{danc}}
\newcommand{\h}{\mathrm{h}}
\newcommand{\ishelper}{\mathrm{ishelper}}
\newcommand{\nexthparent}{\mathrm{nexthparent}}
\newcommand{\nexthchildren}{\mathrm{nexthchildren}}
\newcommand{\nextparent}{\mathrm{nextparent}}
\newcommand{\bypass}{\mathrm{bypass}}

\newcommand{\Empty}{\emph{EMPTY}}

\newcommand{\Endpoint}{\mathrm{endpoint}}
\newcommand{\BT}{\mathrm{BT}}
\newcommand{\PrRoots}{\mathrm{PrRoots}}

\newcommand{\representative}{\mathrm{Representative}}
\newcommand{\edge}{\mathrm{edge}}

\newcommand{\descendantcount}{\mathrm{desccount}}
\newcommand{\height}{\mathrm{height}}
\newcommand{\False}{\mathrm{FALSE}}
\newcommand{\True}{\mathrm{TRUE}}

\newcommand{\haft}{\mathrm{haft}}
\newcommand{\FG}{\mathrm{FG}}
\newcommand{\FTree}{\mathrm{ForgivingTree}}
\newcommand{\FGraph}{\mathrm{ForgivingGraph}}

\newcommand{\Strip}{\mathrm{Strip}}
\newcommand{\Merge}{\mathrm{Merge}}

\newcommand{\numchildren}{\mathrm{numchildren}}

\newcommand{\size}{\mathrm{size}}

\newcommand{\RTfragment}{\mathrm{RTfragment}}
\newcommand{\Nset}{\mathrm{Nset}}
\newcommand{\hleftchild}{\mathrm{hleftchild}}
\newcommand{\hrightchild}{\mathrm{hrightchild}}

\renewcommand{\algorithmiccomment}[1]{// #1}

\begin{document}

%\thispagestyle{empty}
%
%\begin{figure}[t!]
%\centering
%\includegraphics{images/Dissertation010}
%\end{figure}
%\clearpage

\frontmatter

\title{Algorithms for Self-Healing Networks}

\author{Amitabh Trehan}

\degreesubject{Ph.D., Computer Science}

\degree{Doctor of Philosophy \\ Computer Science}

\documenttype{Dissertation}

%\previousdegrees{ B.Sc., Biology, Punjab University, Chandigarh, India, 1994\\ 
%		M.C.A., Indira Gandhi National Open University, India, 2000  %Leaving the newline out. formats correct, wathch out in print.
%                 M.Tech., Indian Institute of Technology, Delhi, India, 2002}
                 
                 \previousdegrees{ B.Sc., Biology, Punjab University, 1994\\ 
	\quad M.C.A., Indira Gandhi National Open University, 2000 \qquad  %Leaving the newline out. formats correct, wathch out in print.
                 M.Tech., Indian Institute of Technology Delhi,  2002}

\date{May, 2010}

\maketitle

\makecopyright

\begin{dedication}
   To the sun, the moon\\
   and the intrepid spirit,\\
   To my family \\
    who made this journey possible.\\[3ex]
    
    The question walks\\
    the length of pages,\\
    rain drops on roof.
%     -- \emph{7 Haikus}
     % by Amitabh Trehan
  % ``A bird in hand is worth two in the bush''
     %    -- Anonymous
\end{dedication}

\begin{acknowledgments}
   \vspace{1.1in}
   If this were an Oscar awards ceremony, my list of thank yous would have had the music director going crazy trying to hound me off the stage. There are many many to thank for the journey responsible for  this document.   My foremost gratitude goes towards my advisor, Professor Jared Saia, for his constant enthusiastic guidance. He has patiently ironed out a multitude of rough edges that I, as a scientist, have presented, and has taught the virtues of discipline and mathematical rigour. My academic collaborator and committee member Professor Thomas Hayes has been a source of constant inspiration. I am thankful to my dissertation committee (Professors Saia, Hayes,  Cris Moore and Tanya-Beger Wolf)  who have provided me with much insight and guidance.  I am thankful to all my close friends, who have been with me through good times and bad, especially  Navin Rustagi and Vaibhav Madhok (their  lively discussions have lit up many evenings!). I am thankful to the US educational system, for its support of quality graduate education and research. I owe a debt of gratitude to all my teachers and friends in India, and to the Art of Living foundation and it's founder Sri Sri Ravi Shankar, for Sudershan Kriya, the meditation and the satsangs . Finally, I have to thank my biggest inspiration: my mother, and my family: my late father, my step-father, my brothers, my sister-in-laws, my nieces and my nephew, without whose support and love I would never have been able to pursue the path around the world and in my academic world that I have.
   
%   I would like to thank, as well.\footnote{To my brother and sister, who
%   are really cool.}
\end{acknowledgments}

\maketitleabstract

\begin{abstract}
  Many modern networks are \emph{reconfigurable}, in the sense that the topology of the network can be changed by the
nodes in the network.  For example, peer-to-peer, wireless and ad-hoc networks are reconfigurable.  More generally, many
social networks, such as a company's organizational chart; infrastructure networks, such as an airline's transportation
network; and biological networks, such as the human brain, are also reconfigurable.
Modern reconfigurable networks have a complexity unprecedented in the history of engineering, resembling more a dynamic
and evolving living animal rather than a structure of steel designed from a blueprint. Unfortunately, our mathematical
and algorithmic tools have not yet developed enough to handle this complexity and fully exploit the flexibility of
these networks. \\

We believe that it is no longer possible to build  networks that are scalable and never have node failures. Instead, these networks should be able to admit small, and maybe, periodic failures
and still recover like skin heals from a cut.
This process, where the network can recover itself by maintaining  key invariants in response to attack
by a powerful adversary is what we call \emph{self-healing}. \\

Here, we present several fast and provably good distributed algorithms for self-healing in reconfigurable dynamic networks. Each of these algorithms have different properties, a different set of gaurantees and limitations. We also discuss future directions and theoretical questions we would like to answer.
%in the final dissertation that this document is proposed to lead to.
\clearpage %(required for 1-page abstract)
\end{abstract}

\tableofcontents
\listoffigures
\listoftables

%\begin{glossary}{Longest  string}
%   \item[$a_{lm}$]
%      Taylor series coefficients, where $l,m = \{0..2\}$
%   \item[$A_{\bf{p}}$]
%      Complex-valued scalar denoting the amplitude and phase.
%   \item[$A^T$]
%      Transpose of some relativity matrix.
%\end{glossary}

\mainmatter

\chapter{Introduction}
\label{chapter: Intro}

\begin{epigraphs}
\qitem{ \epitext{Begin at the beginning and go on till you come to the end: then stop.}}{\epicharacter{The king of hearts}\\ \episource{Alice in Wonderland}}\\ 
%\qitem{\epitext{`Take some more tea,' the March Hare said to Alice, very earnestly.
%`I've had nothing yet,' Alice replied in an offended tone, `so I can't take more.'
%`You mean you can't take less,' said the Hatter: `it's very easy to take more than nothing.'}}{\episource{Alice in Wonderland}}
%\qitem{Speak English! I don't know the meaning of half those long words, and I don't believe you do either!}{\emph{Eaglet}\\ \texttt{Alice in Wonderland}} 
\end{epigraphs}

Networks in the modern age have grown by leaps and bounds, both in size and complexity. The size of some networks spans nations and even the globe. Networks provide a multitude of services using a wide variety of protocols and  components to the extent that they have now begun to resemble self-governed living entities. The Internet is the obvious example but there are others too like cellular phone networks. There are networks which have always been around but which only now have been scrutinized  by tools of computer science, such as the social networks. Most networks are dynamic since nodes can enter the network or be removed by choice, failure or attack. We are also fortunate that we live in a time where we can observe and 
inßuence the evolution of a dynamic network like the Internet. Due to the scale and nature of design of modern networks, it may simply not be practical to build robustness into the individual nodes or into the structure of the initial network itself. 

Many important networks are also \emph{reconfigurable} in the sense that they can change their topology.  Often, individual nodes can initiate new connections or drop existing connections.
  For example, peer-to-peer, wireless and ad-hoc networks are reconfigurable.  Looking beyond computer networks, many social networks, such as a company's organizational chart, or friendship networks on social networking sites are reconfigurable.  Infrastructure networks, such as an airline's transportation
network are reconfigurable. Many biological networks, including the human brain, which shows such capacity  for learning and  adaptability, are also reconfigurable. From an engineering aspect,  modern reconfigurable networks have a complexity unprecedented in  history.  We are approaching scales of billions of components.
Such systems are less akin to a traditional engineering enterprise built from a blueprint 
such as a bridge, and more akin to a dynamic and evolving living organism in terms of
complexity.  A bridge must be designed so that key components never
fail, since there is no way for the bridge to automatically recover
from system failure.  In contrast, a living organism can not be
designed so that no component ever fails: there are simply too many
components.  For example, skin can be cut and still heal. Designing
skin that can heal is much more practical than designing skin that is
completely impervious to attack.  Unfortunately, current algorithms
ensure robustness in computer networks through hardening individual
components or, at best, adding lots of redundant components.  Such an
approach is increasingly unscalable.

Our mathematical and algorithmic tools have not yet developed enough to handle the complexity and fully exploit the flexibility of modern networks. As an example, on August 15, 2007 the Skype network crashed for about $48$ hours, disrupting service to approximately $200$ million users~\cite{fisher,malik, moore, ray, stone}.  Skype attributed this outage to failures in their ``self-healing mechanisms''~\cite{garvey}.  We believe that this outage is indicative of the much broader problems outlined earlier. 

In the following chapters, we will propose some algorithms for self-healing. Informally, we define self-healing to be maintenance of certain properties within desirable bounds by the nodes in a network suffering from failures or under attack. As the name implies, self-healing has to be initiated and executed by the nodes themselves. As such, the algorithms we have proposed here are fully distributed. Equivalenty we can say  that a self-healing system, when starting from a correct state, can only be temporarily out of a  correct state i.e. it recovers to a correct state, in presence of  attacks. Self-healing is one of the so called `Self-*' properties which systems such as autonomic systems may be required to have. Section~\ref{sec: Intro-Self-*} has a brief discussion on these  properties.

One approach towards self-healing is to  add additional capacity or rerouting in anticipation of failures. There has been plenty of work which has followed this approach. However, there are obvious limitations including wastage of resources and limitations on additional capacity. 
In this Dissertation, we have adopted a \emph{responsive} approach.  Our approach is responsive in the sense that it responds to an attack (or component failure) by changing the topology of the network.  This  approach works irrespective of the initial state of the network, and is thus orthogonal and complementary to traditional  non-responsive techniques.  

Informally, the model we adopt in this work is as follows. We assume that the network is initially a connected graph over $n$ nodes.  An adversary repeatedly attacks the network. This adversary knows the network topology and our algorithm, and it has the ability to delete arbitrary nodes from the  network or insert a new node in the system which it can connect to any subset of the nodes currently in the system.   However, we assume the adversary is constrained in that in any time step it can only delete or insert a single node. Following that, the self-healing algorithm has a short time to reconfigure and heal the network by adding edges between remaining nodes before the next act of the adversary. Our model  captures what can happen when a worm or software error propagates through the population of nodes. This model is described in more detail Section~\ref{sec: Intro-self-healingModel}.

\section{Naive self-healing}
Even in a very simple setting, we need to be smart about reconfiguring. Suppose we are trying to maintain a property such as connectivity of the network but our algorithm is not very sophisticated. Then, it may be very easy for the adversary to force the algorithm to cause high degree increase (which may lead to overload and eventual network breakdown) or increase in distances between nodes (which may lead to poor communication).  Figure~\ref{fig: naiveheal} shows a naive algorithm attempting to heal the network by using only a small number of edges at each timestep. However, node $v$ in the figure ends up increasing its degree by 3 over a course of 3 deletions. Thus, a naive algorithm could yield a degree increase as high as $\theta(n)$.

  \begin{figure}[h!]
\centering
\subfigure[First deletion]{
\includegraphics[scale=0.5]{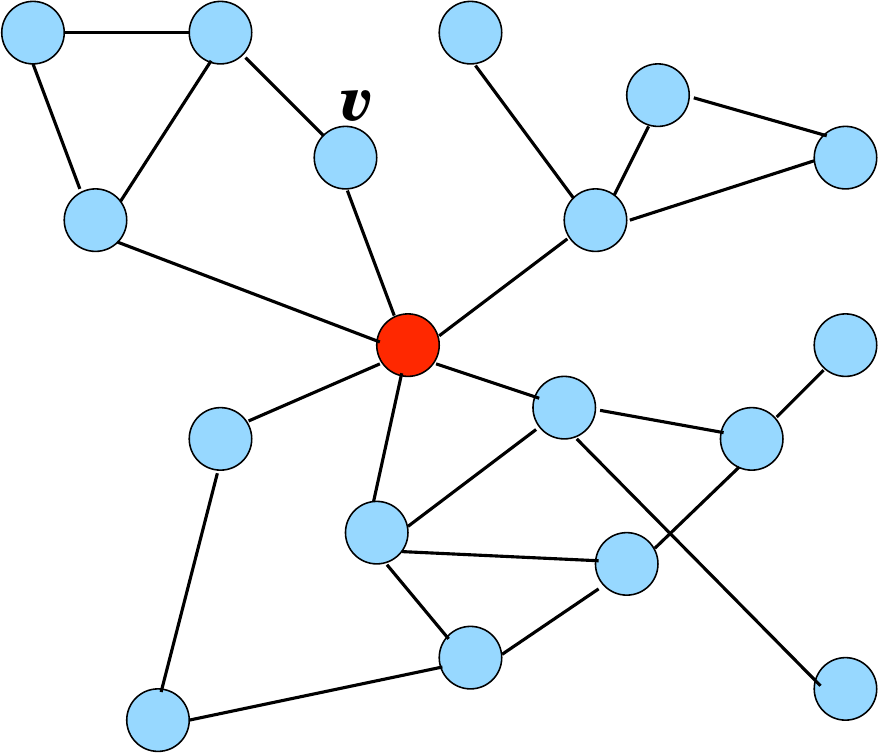}
 } 
\subfigure[Neighbors detect deletion]{
\includegraphics[scale=0.5]{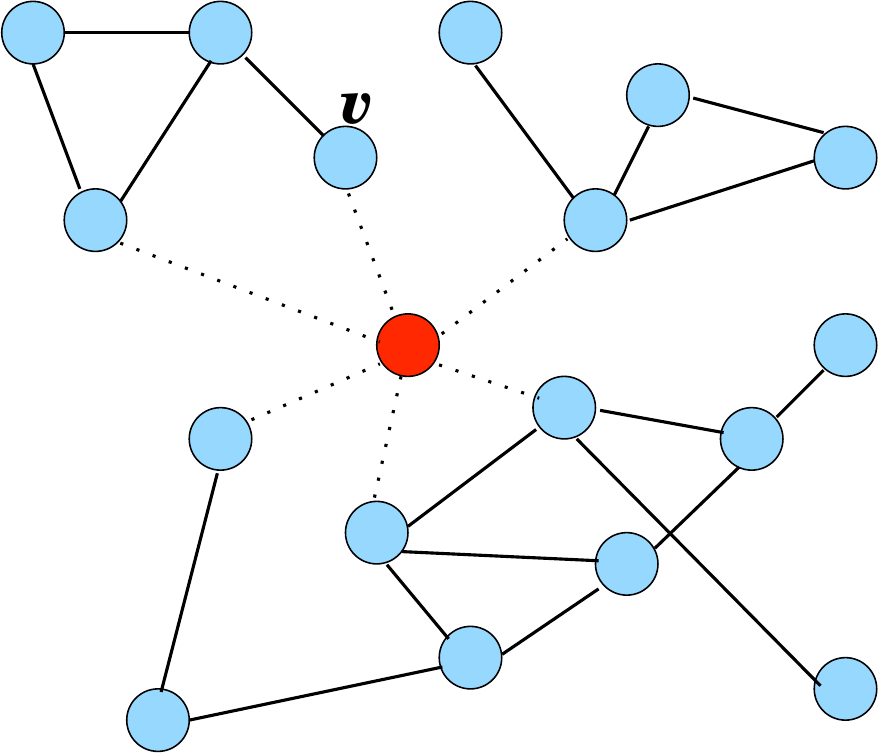}
%\label{fig:subfig2}
 }
\subfigure[Reconnection: $v$ increases degree]{
\includegraphics[scale=0.5]{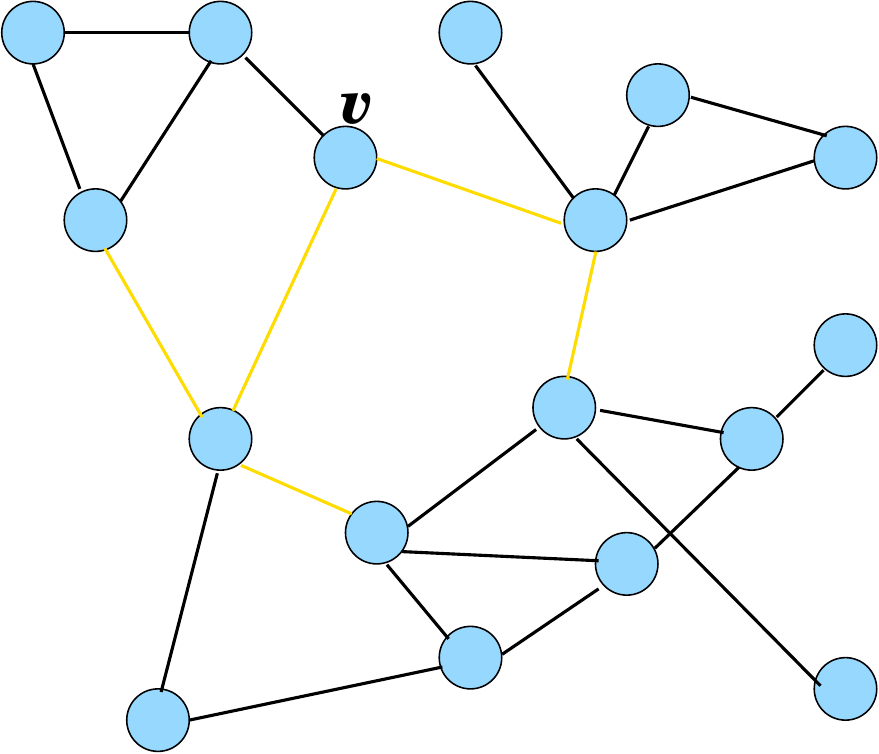}
%\label{fig:subfig3} 
}\\
\subfigure[Second deletion]{
\includegraphics[scale=0.5]{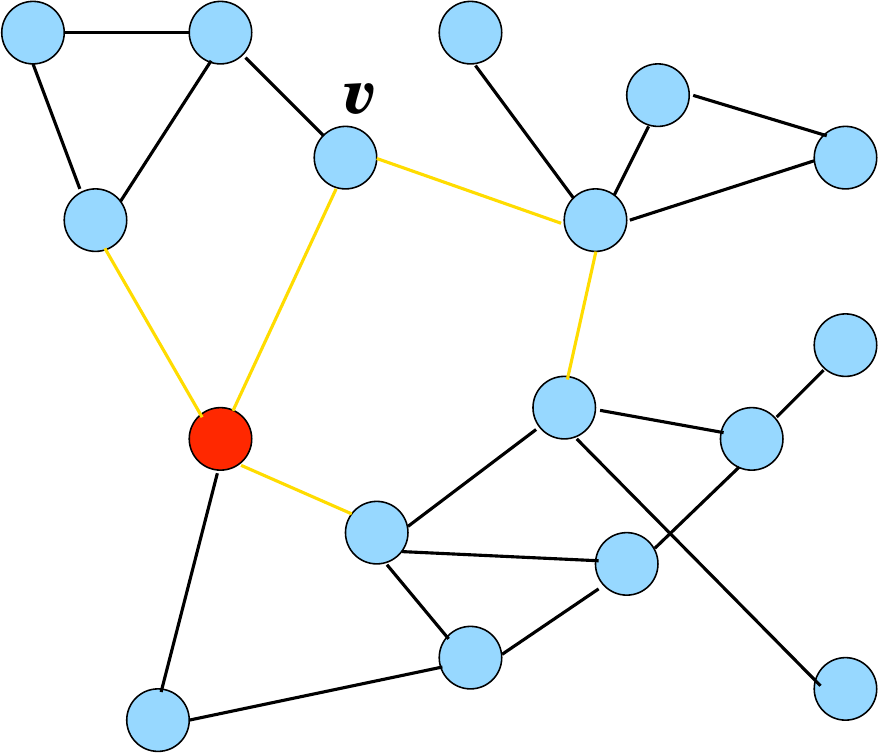}
 } 
\subfigure[Detection]{
\includegraphics[scale=0.5]{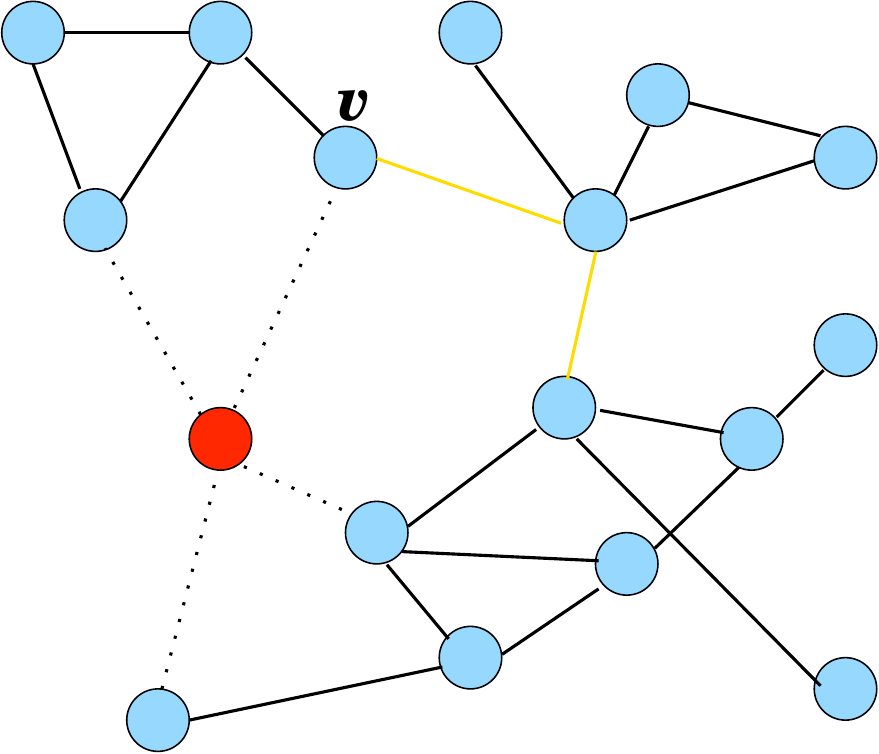}
 } 
 \subfigure[$v$'s degree increases by 2]{
\includegraphics[scale=0.5]{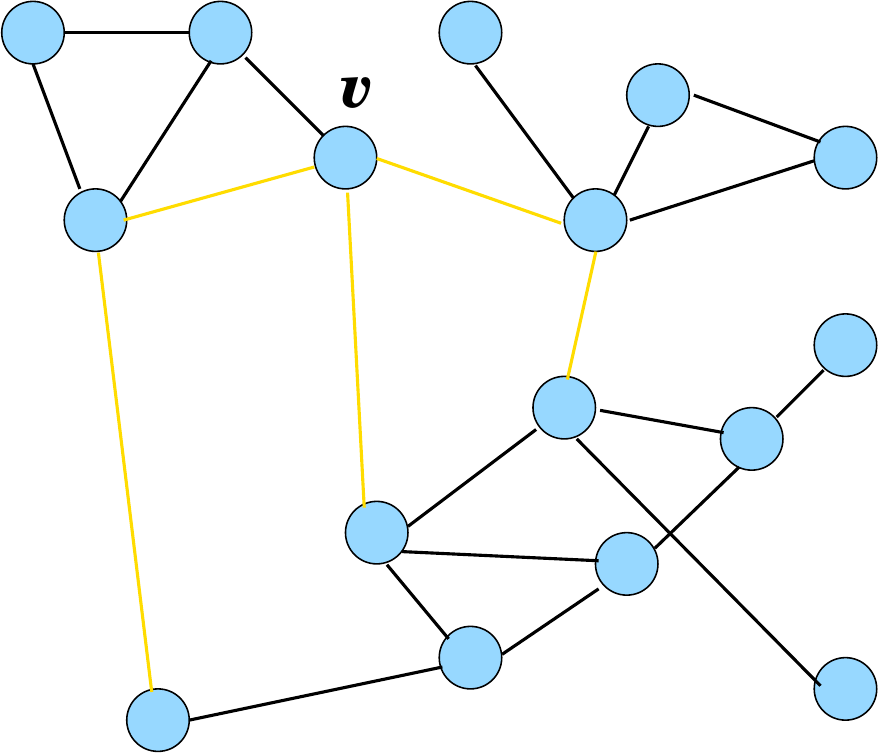}
 } \\
\subfigure[Third deletion]{
\includegraphics[scale=0.5]{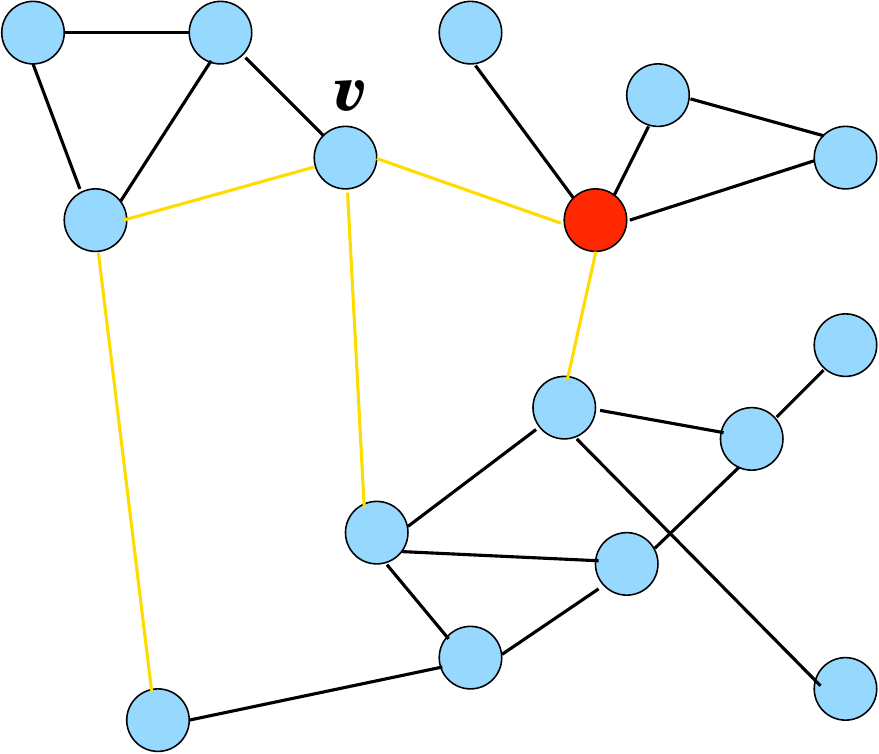}
%\label{fig:subfig2}
 }
 \subfigure[Detection]{
\includegraphics[scale=0.5]{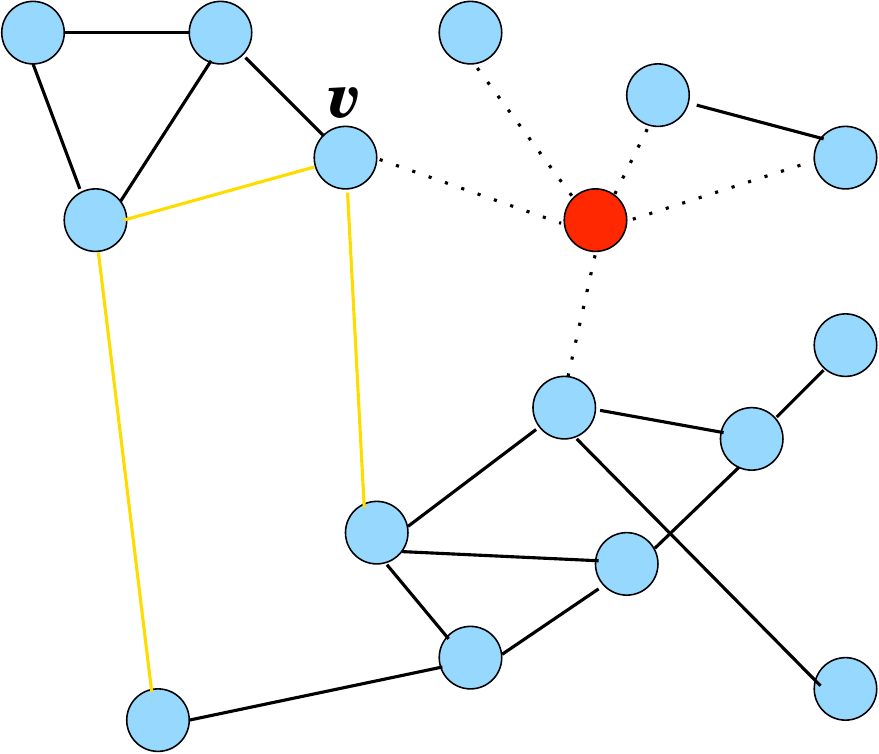}
%\label{fig:subfig2}
 }
\subfigure[$v$'s degree increases by 3]{
\includegraphics[scale=0.5]{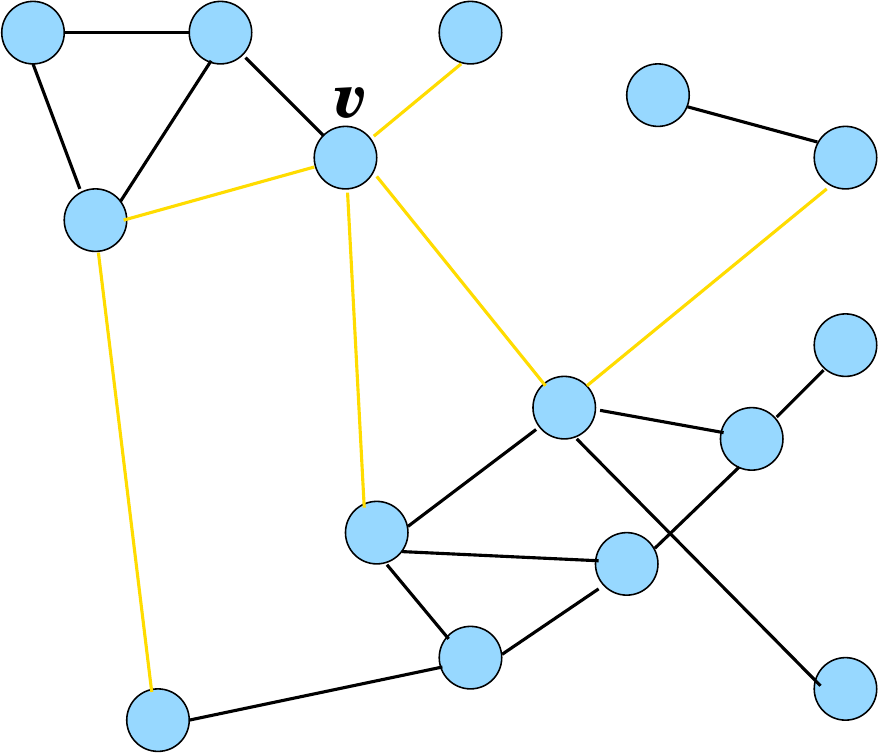}
%\label{fig:subfig3} 
}
\label{fig: naiveheal}
\caption{A sequence of 3 deletions and healings using a naive algorithm. A node marked red is deleted by the adversary. The neighbors of the deleted node reconnect  (golden edges) to maintain connectivity. Notice node $v$ increases its degree by 3.}
\end{figure}

 \section{Model of self-healing}
\label{sec: Intro-self-healingModel}

%\epigraph{\epitext{Speak English! I don't know the meaning of half those long words, and I don't believe you do either!}}{\epicharacter{Eaglet.} \episource{Alice in Wonderland}} 

Our general model of self-healing is shown in Figure~\ref{algo: model-general}. The specific models used in our algorithms are special cases of this model, differing mainly in the way the success metrics of the graph properties are presented.
% Also, the model used in  $\FTree$ algorithm is a simplification in the sense that $\FTree$ does not handle node insertions. 
This model is very similar to the model described in Figure~\ref{algo: model-2}.  Let $G = G_0$ be an arbitrary graph on $n$ nodes, which represent processors in a distributed network.  In each step, the adversary either deletes or adds a node.  After each deletion, the algorithm gets to add some new edges to the graph, as well as deleting old ones.  At each insertion, the processors follow a protocol to update their information.
The algorithm's goal is to maintain the chosen graph properties within the desired bounds. At the same time, the algorithm wants to minimize the resources spent on this task.  Initially, each processor only knows its neighbors in $G_0$, and is unaware of the structure of the rest of $G_0$. After each deletion or insertion, only the neighbors of the deleted or inserted vertex are informed that the deletion or insertion has occured. After this, processors are allowed to communicate by sending a limited number of messages to their direct  neighbors.  We assume that these messages are always sent and received successfully.  The processors may also request new edges be added to the graph. The only synchronicity assumption we make is that no other  vertex is deleted or inserted until the end of this round of computation and communication has concluded. To make this assumption more reasonable, the per-node communication cost should be very small in $n$ (e.g. at most logarithmic).

We also allow a certain amount of pre-processing to be done before the first attack occurs.  This may, for instance,
be used by the processors to gather some topological information about $G_0$, or perhaps to 
coordinate a strategy.  Another success metric is the amount of computation and communication needed during this
preprocessing round.  For our success metrics, we compare the graphs at time $T$: the actual graph $G_T$ to the graph $G'_{T}$ which is the graph with only the original nodes (those at $G_0$) and insertions without regard to deletions and healing. This is the graph which would have been present if the adversary was not doing any deletions and (thus) no self-healing algorithm was active. This is the natural graph for comparing results.  
%Notice if there were no insertions happening in our model, we could have compared $G_T$ to $G_0$ but since insertions are happening, $G_T$ may not even have the same nodes as $G_0$ rendering a node-based comparison impossible.
Figure~\ref{fig: Intro-ComparisonGraphs} shows an example of $G'_T$ and a corresponding $G_T$. The figure also shows, in  $G'_T$,  the nodes and edges inserted  and deleted,  and in $G_T$, the edges inserted by the healing algorithm, as the network evolved over time.

%\begin{algorithm}[h!]
\begin{figure}[t]
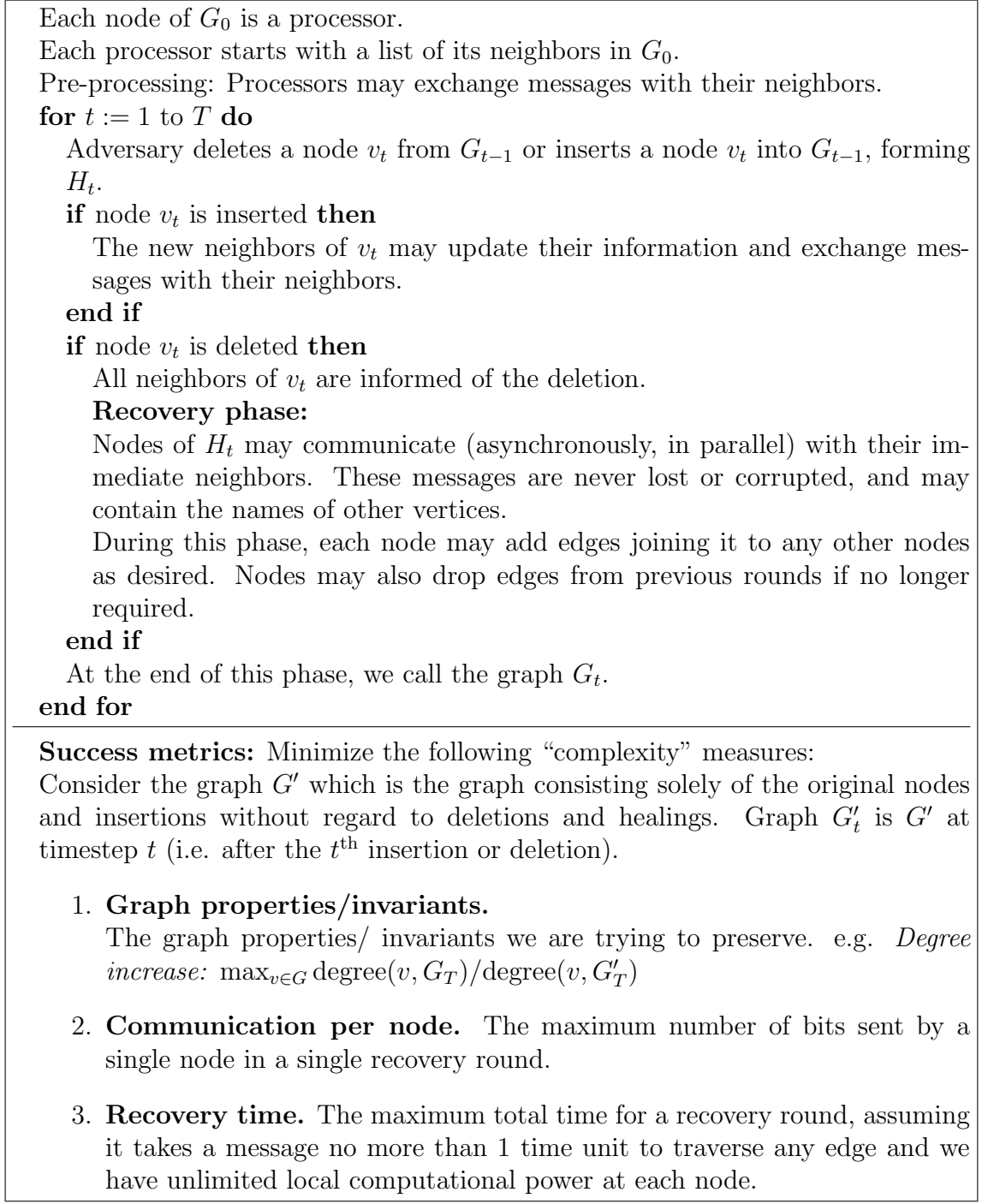

\caption{The general distributed Node Insert, Delete and Network Repair Model.}
\label{algo: model-general}
\begin{boxedminipage}{\textwidth}
\begin{algorithmic}
\STATE Each node of $G_0$ is a processor.  
\STATE Each processor starts with a list of its neighbors in $G_0$.
\STATE Pre-processing: Processors may exchange messages with their neighbors.
%send messages to and from their neighbors.
\FOR {$t := 1$ to $T$}
\STATE Adversary deletes a node $v_t$ from $G_{t-1}$ or inserts a node $v_t$ into $G_{t-1}$, forming $H_t$.
\IF{node $v_t$ is inserted} 
\STATE The new neighbors of $v_t$ may update their information and exchange messages with their neighbors.
\ENDIF
\IF{node $v_t$ is deleted} 
\STATE All neighbors of $v_t$ are informed of the deletion.
\STATE {\bf Recovery phase:}
\STATE Nodes of $H_t$ may communicate (asynchronously, in parallel) 
with their immediate neighbors.  These messages are never lost or
corrupted, and may contain the names of other vertices.
\STATE During this phase, each node may add edges
joining it to any other nodes as desired. 
Nodes may also drop edges from previous rounds if no longer required.
\ENDIF
\STATE At the end of this phase, we call the graph $G_t$.
\ENDFOR
\vspace{5pt}
\hrule
\STATE
\STATE {\bf Success metrics:} Minimize the following ``complexity'' measures:\\
Consider the graph  $G'$ which is the graph consisting solely of the original nodes and insertions without regard to
deletions and healings. Graph $G'_{t}$ is $G'$ at timestep $t$ (i.e. after the $t^{\mathrm{th}}$ insertion or deletion).
%Graph $G'_{t}$ is $G'$ at timestep $t$ which is equivalent to $G'_{t'}$ where the $t' \le t$ is
%the timestep at which the latest insertion on or before $t$ occured. 
 \begin{enumerate}
\item{\bf Graph properties/invariants.}\\
The graph properties/ invariants we are trying to preserve. e.g.  \emph{Degree increase:} 
  $\max_{v \in G} \Degree(v,G_T) / \Degree(v,G'_T)$
%\item {\bf Network stretch.} $\max \left( (x,y) \in G_{t}, G_{t'}; t, t' <T, \distance_{t'}(x,y) / \distance_{t}(x,y)
%\right)$
%\item {\bf Network stretch.} $\max_{x, y \in G_{T}} \frac{dist(x,y,G_{T})}{dist(x,y,G'_{T})}$, where, for a graph $G$ and nodes $x$ and $y$ in $G$, $dist(x,y,G)$ is the
%length of the shortest path between $x$ and $y$ in $G$.
%For any pair of nodes $x$ and $y$, $ \distance(x,y,G_{T}) / \distance(x,y,G'_t)$
\item{\bf Communication per node.} The maximum number of bits sent by a single node in a single recovery round.
% \tom{Want to modify this or omit?}
\item{\bf Recovery time.} The maximum total time for a recovery round,
assuming it takes a message no more than $1$ time unit to traverse any edge and we have unlimited local computational power at each node.
\end{enumerate}
\end{algorithmic}
\end{boxedminipage}
\end{figure}
%\end{algorithm}

\clearpage

\begin{figure}[h!]
\centering
\subfigure[$G'_T$: Nodes in red (dark gray in grayscale) deleted, and nodes in green (patterned) inserted, by the adversary.]{ \label{sfig: GraphOrigInserts}
 \makebox[0.4\textwidth][c]{\includegraphics[scale=0.6]{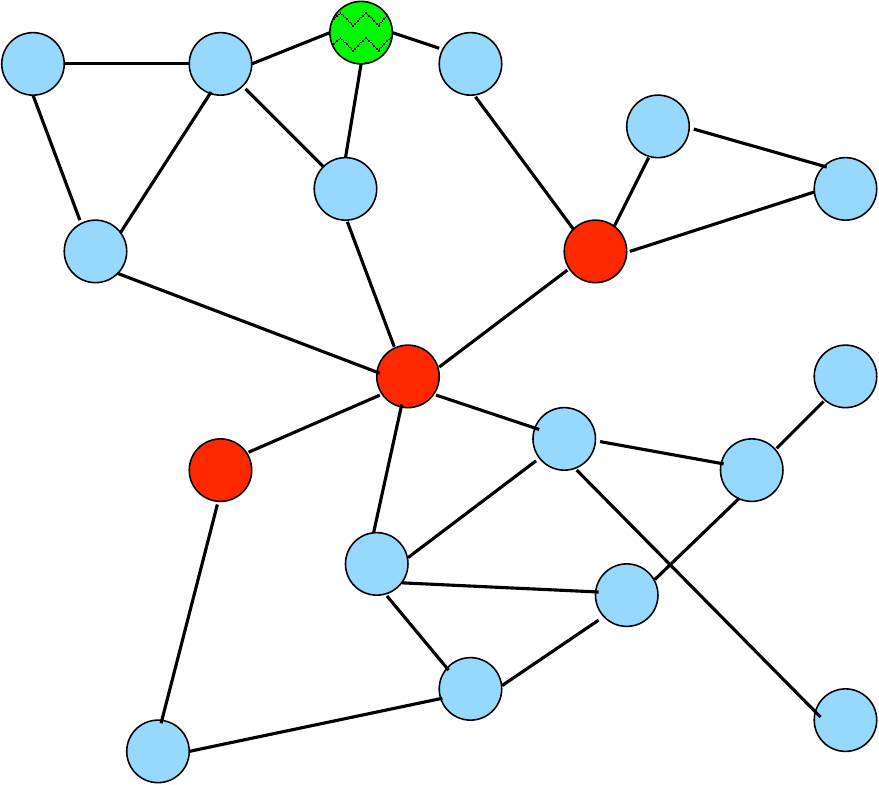}} }
%\hspace{30pt}
\hspace{0.1 \textwidth}
\subfigure[$G_T$: The actual graph. Edges added by the healing algorithm shown in gold (light shaded in grayscale) color.]
%and Every haft is a union of complete binary trees. In our notation, $T_a$ is a complete binary
%tree and $|T_{a}|$ is the number of leaf nodes in $T_{a}$. ]
%{ \label{sfig: haft-as-join} \includegraphics[scale=0.9]{images/FG/Haft} }
{ \label{sfig: GraphHealed}
\makebox[0.4\textwidth][c]{ \includegraphics[scale=0.6]{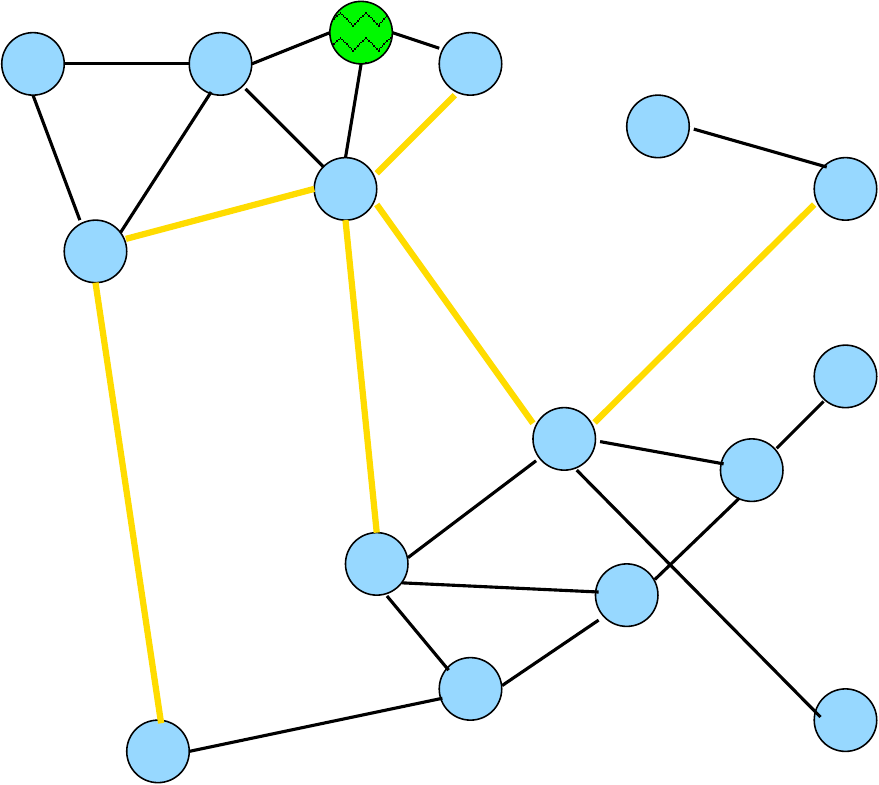} }}
\caption{ \emph{Graphs at time T}. $G'_T$: The graph of initial nodes and insertions over time, $G_T$: The actual healed graph.}
\label{fig: Intro-ComparisonGraphs}
\end{figure}

\section{Healing by Reconstruction Trees}
\label{sec: Intro-healingbyRT}

\begin{figure}[h!]
\centering
\includegraphics[scale=0.65]{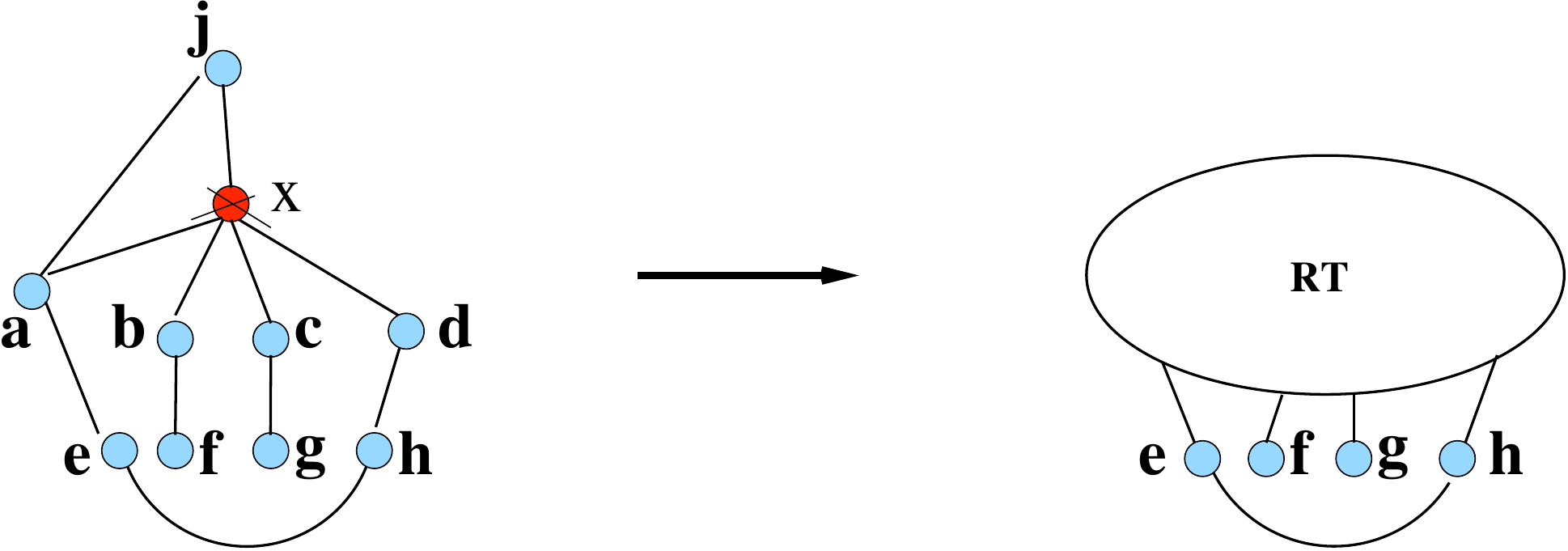}
\caption{Deleted node $x$ (in red, crossed) replaced by a Reconstruction Tree, which is a structure formed by its neighbors ($a, b, c, d, j$). }
 \label{fig: RTheal}
\end{figure}

Our algorithms ($\DASH, \FTree, \FGraph$) use the same basic principle: when a node is deleted, replace it by a tree based structure formed from its neighbors, as shown in Figure~\ref{fig: RTheal}. This structure we call the $\RTree$ ($\RT$), and thus, we can also call these algorithms $\RTree$ healing algorithms. It turns out that trees are a natural choice for the graph properties we have tried to maintain. A balanced tree is a structure which has low distance between nodes (at most $2 \log_2 n$ for a balanced binary tree) while each node has a small degree  (at most 3 for a binary tree). At the same time, coming up with the suitable $\RT$s and maintaining them over the run of the algorithm is quite a significant challenge.

\section{Our Results}
\label{sec: Intro-algorithms}
  In our algorithms, we have focused on some fundamentally important properties: maintaining connectivity, ensuring low degree increase for all nodes, and simultaneously, in later algorithms, ensuring low increase of diameter (or a stronger property, the stretch) of the network. Figure~\ref{fig: healtimeline} (repeated as Figure~\ref{fig: DASHhealtimeline}) shows a series of snapshots from a simulation of our algorithm called $\DASH$ (Chapter~\ref{chapter: DASH}). Notice that the network stays connected, and no individual node gets a large number of extra edges during healing.  
 
 \begin{figure}[h!]
\centering
\subfigure[single deletion]{
\includegraphics[scale=0.21]{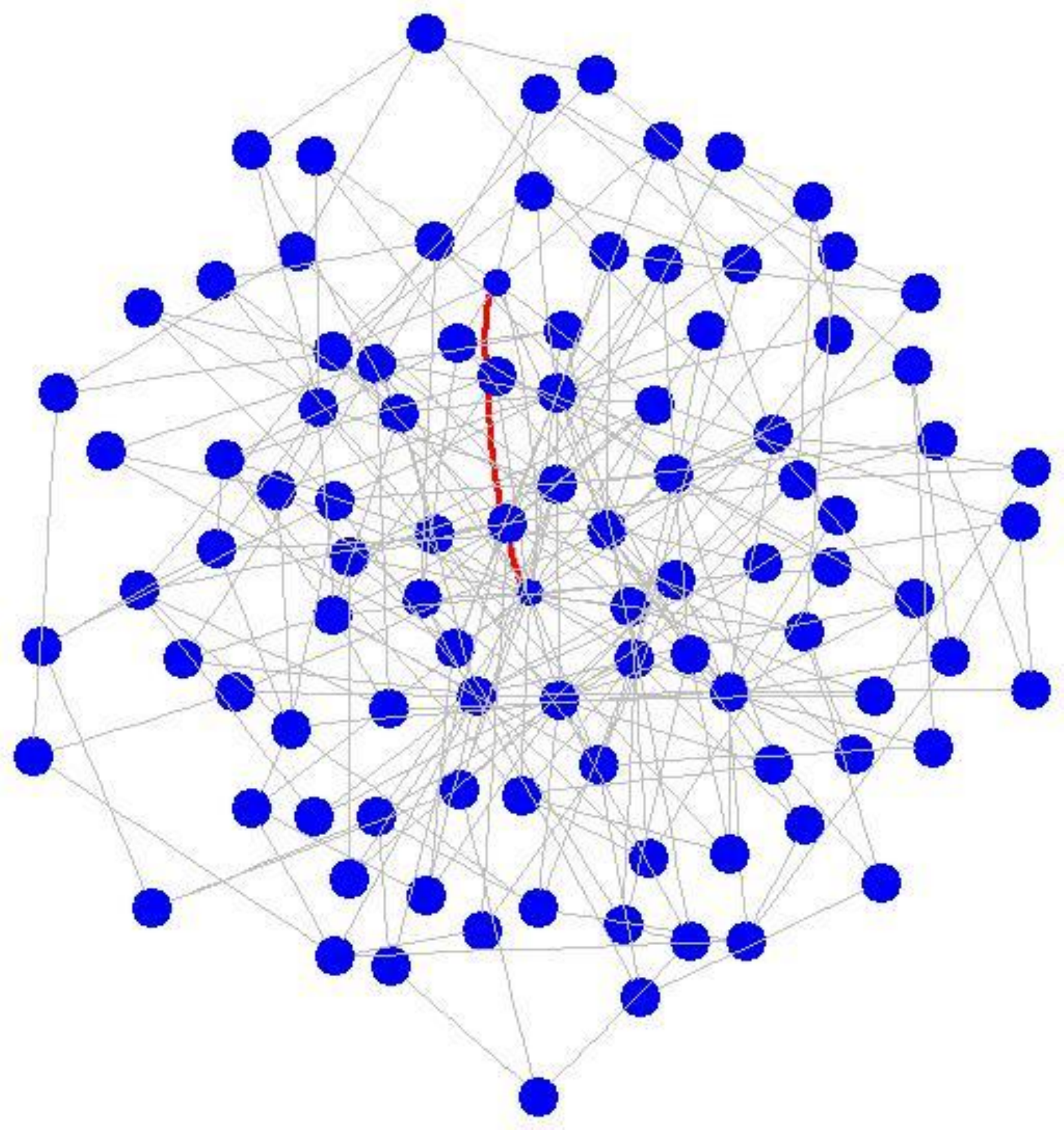}
 } 
\subfigure[10 deletions]{
\includegraphics[scale=0.22]{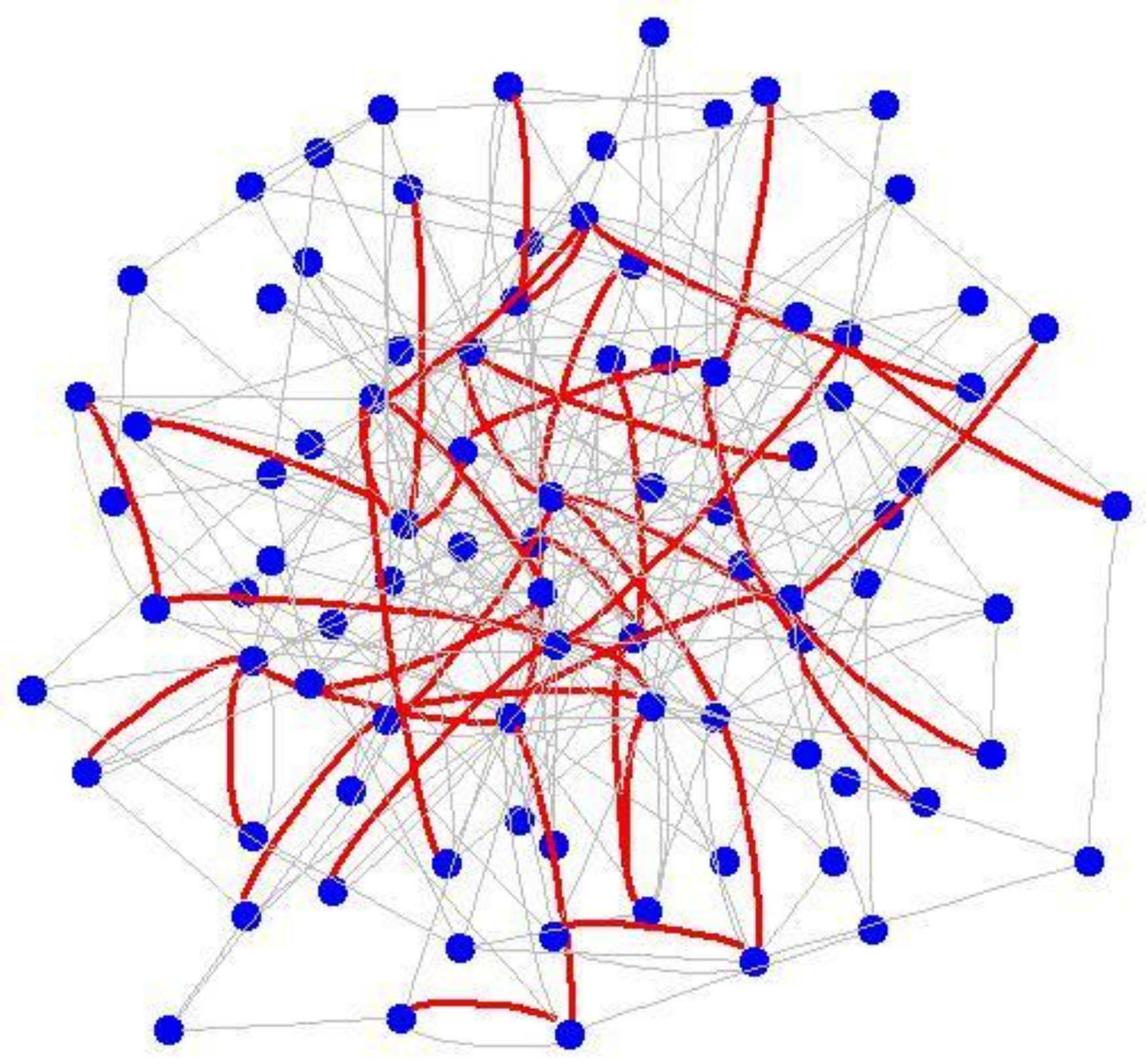}
%\label{fig:subfig2}
 }
\subfigure[30 deletions]{
\includegraphics[scale=0.21]{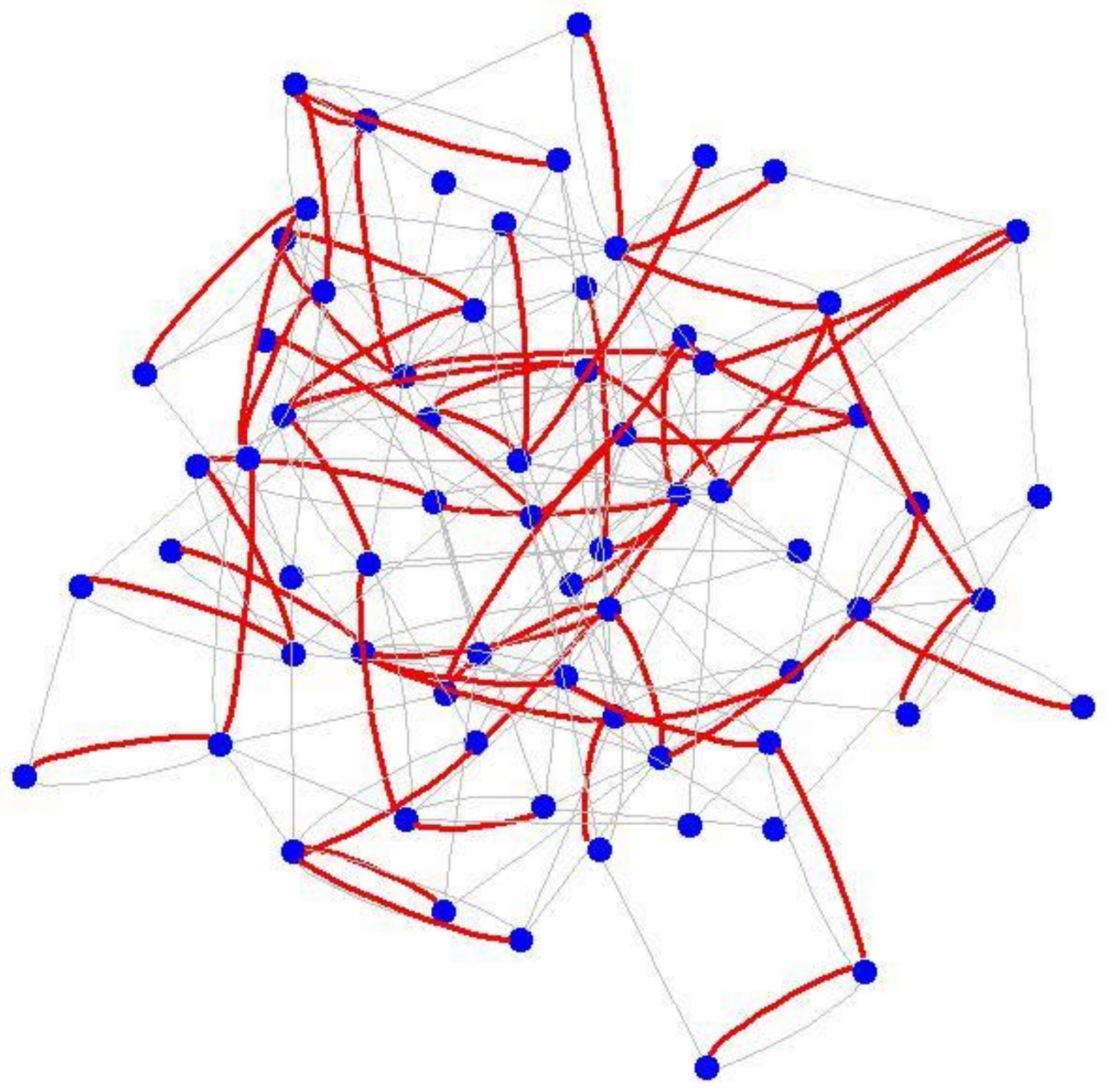}
%\label{fig:subfig3} 
}\\
\subfigure[40 deletions]{
\includegraphics[scale=0.18]{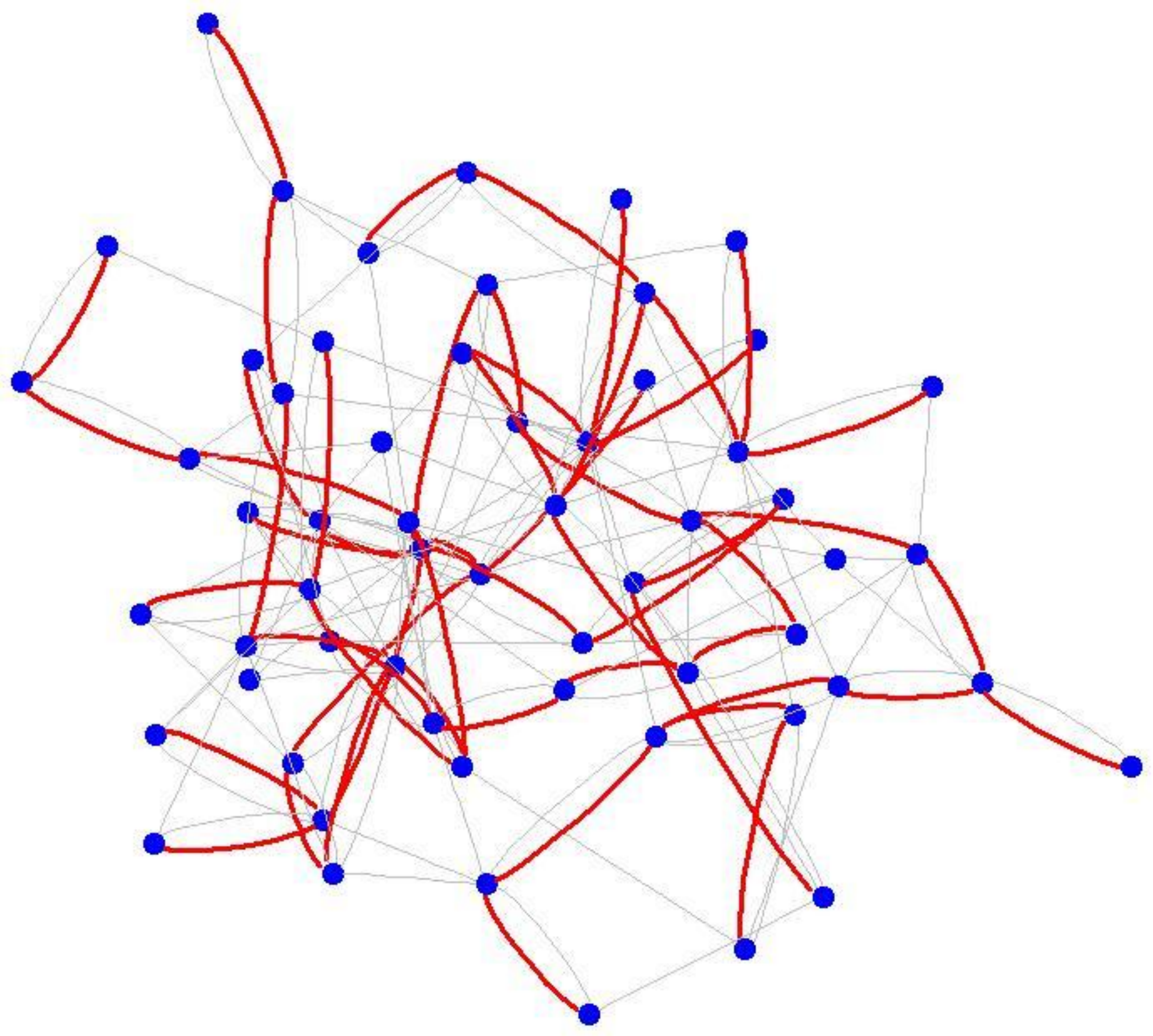}
 } 
\subfigure[50 deletions]{
\includegraphics[scale=0.17]{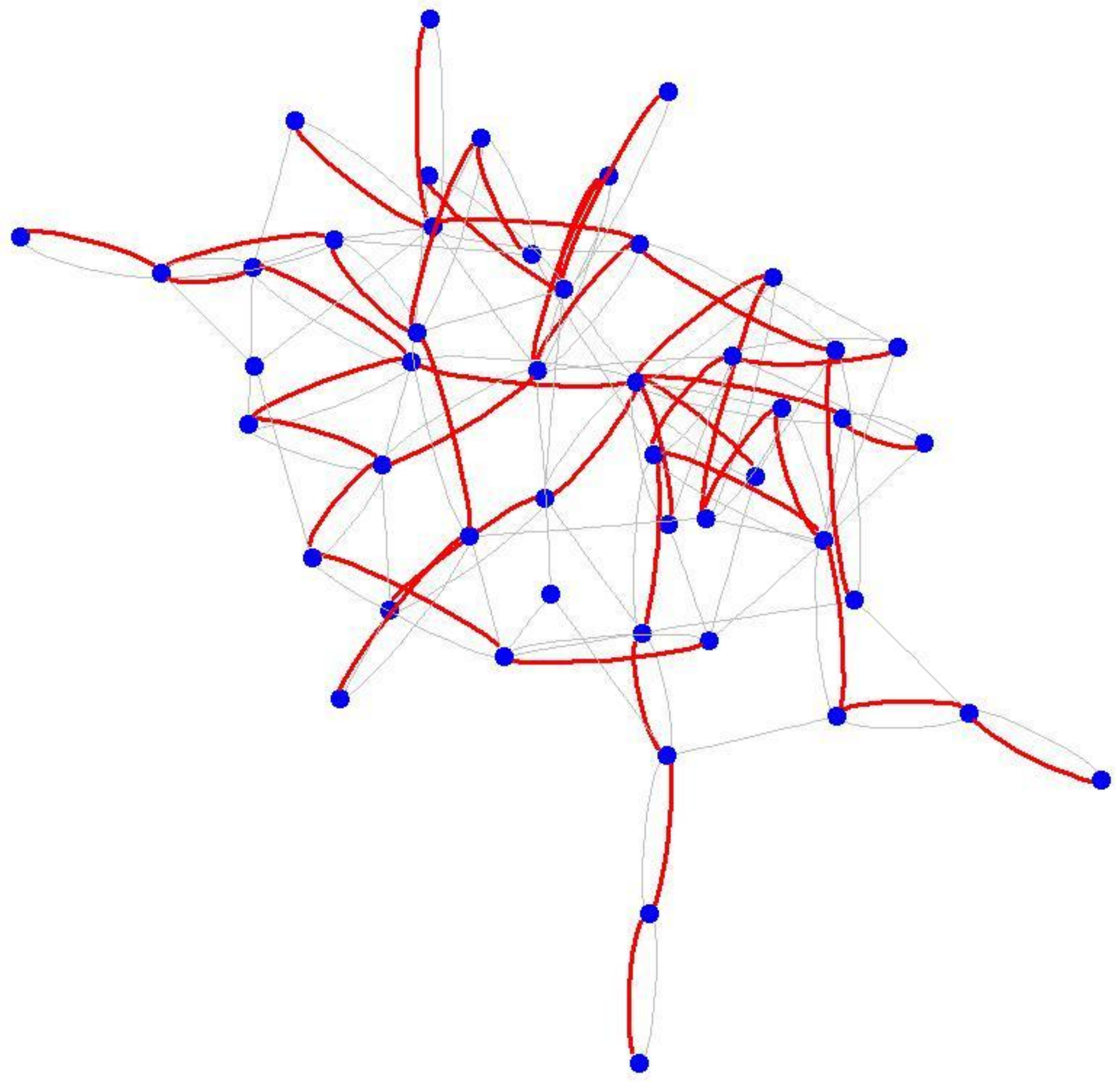}
 } 
 \subfigure[60 deletions]{
\includegraphics[scale=0.18]{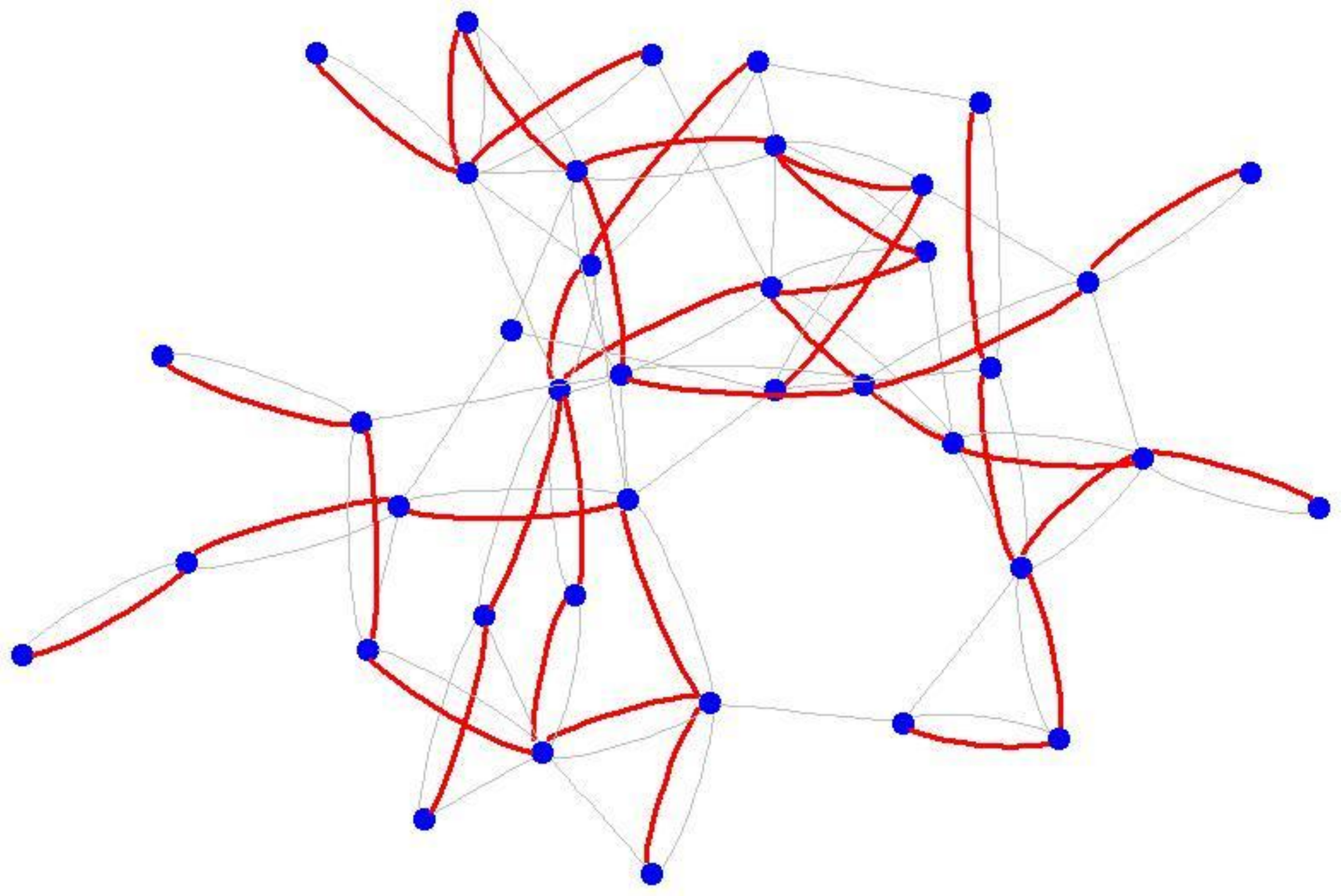}
 } \\
\subfigure[70 deletions]{
\includegraphics[scale=0.2]{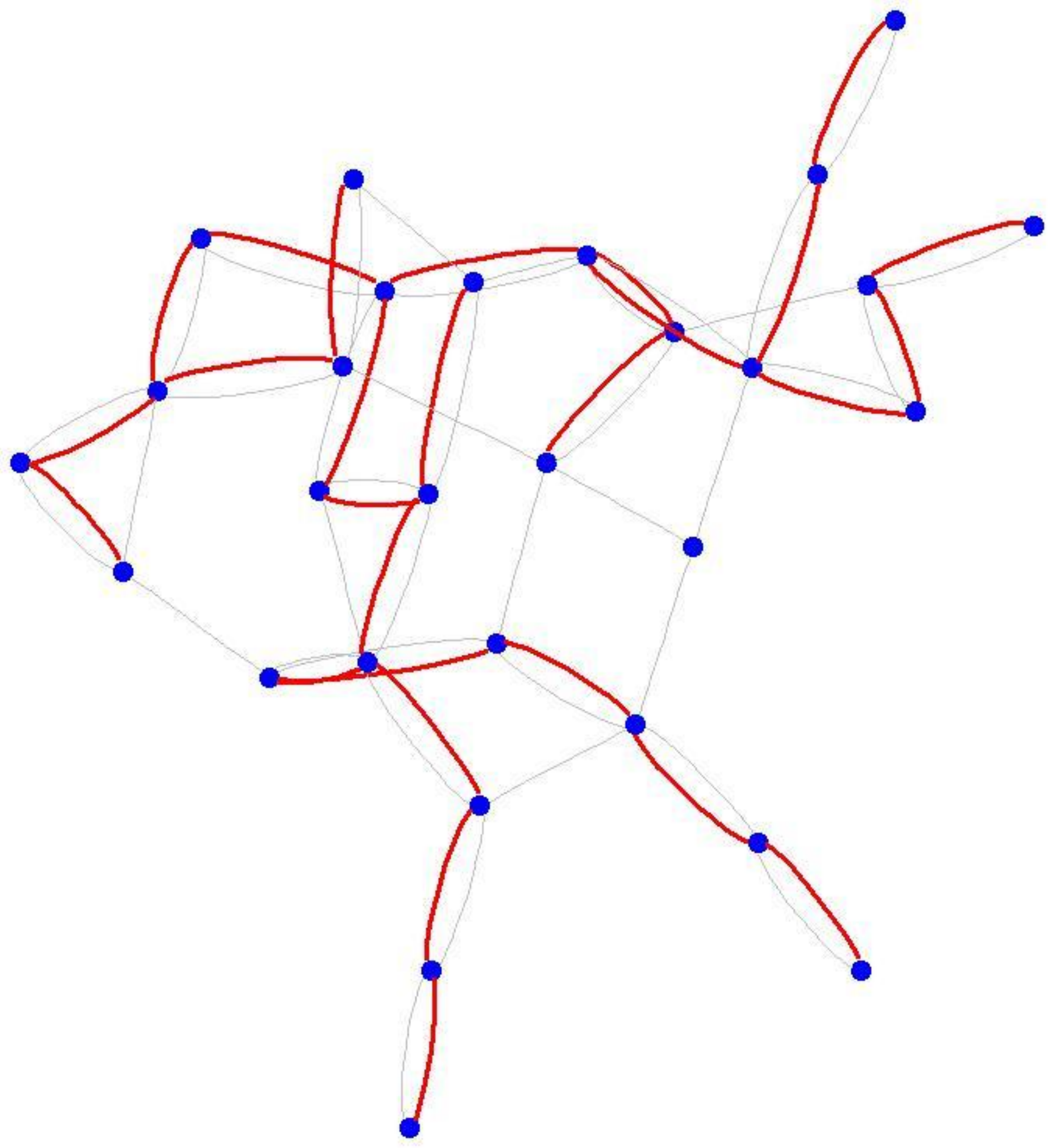}
%\label{fig:subfig2}
 }
 \subfigure[80 deletions]{
\includegraphics[scale=0.19]{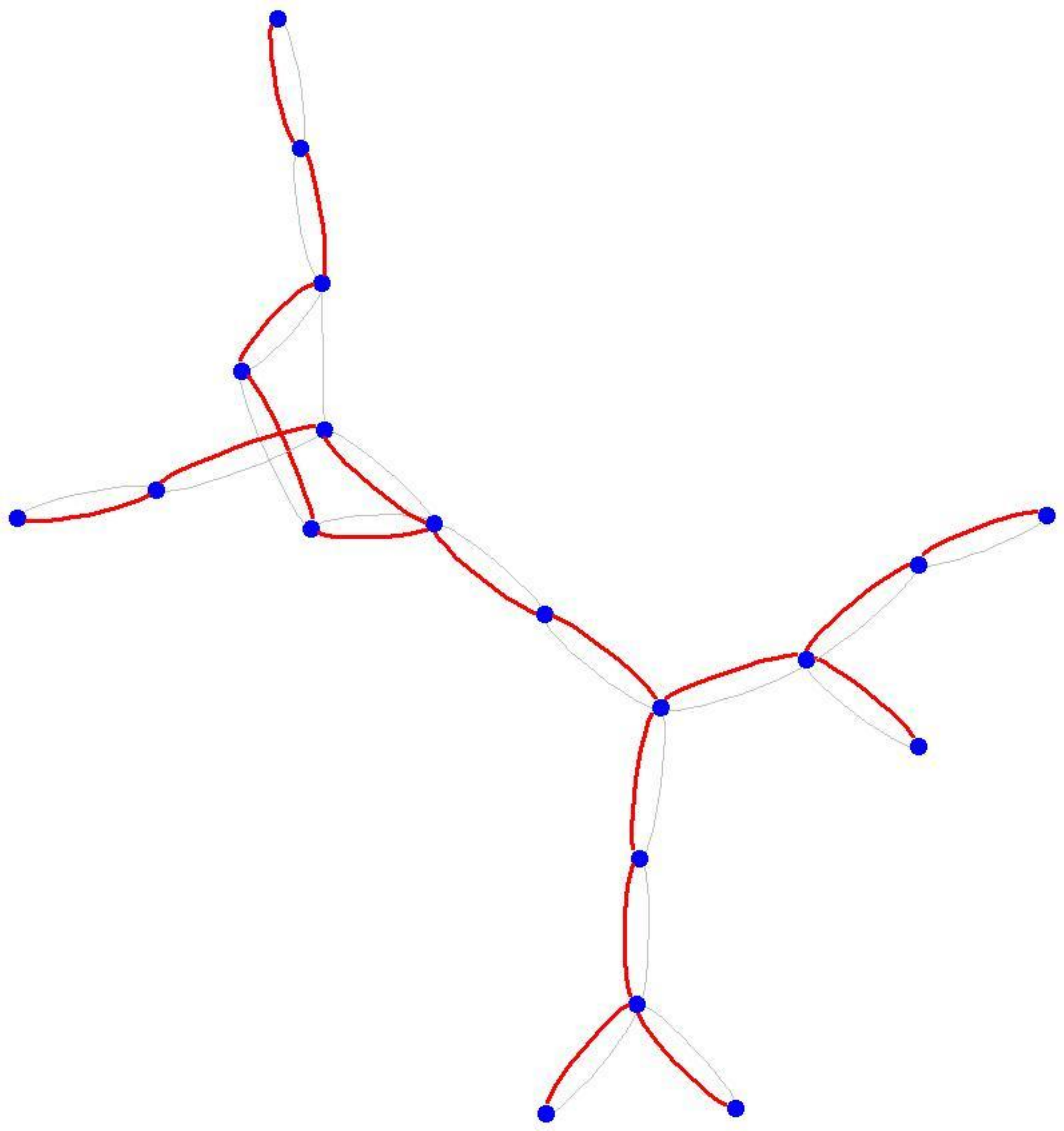}
%\label{fig:subfig2}
 }
\subfigure[90 deletions]{
\includegraphics[scale=0.2]{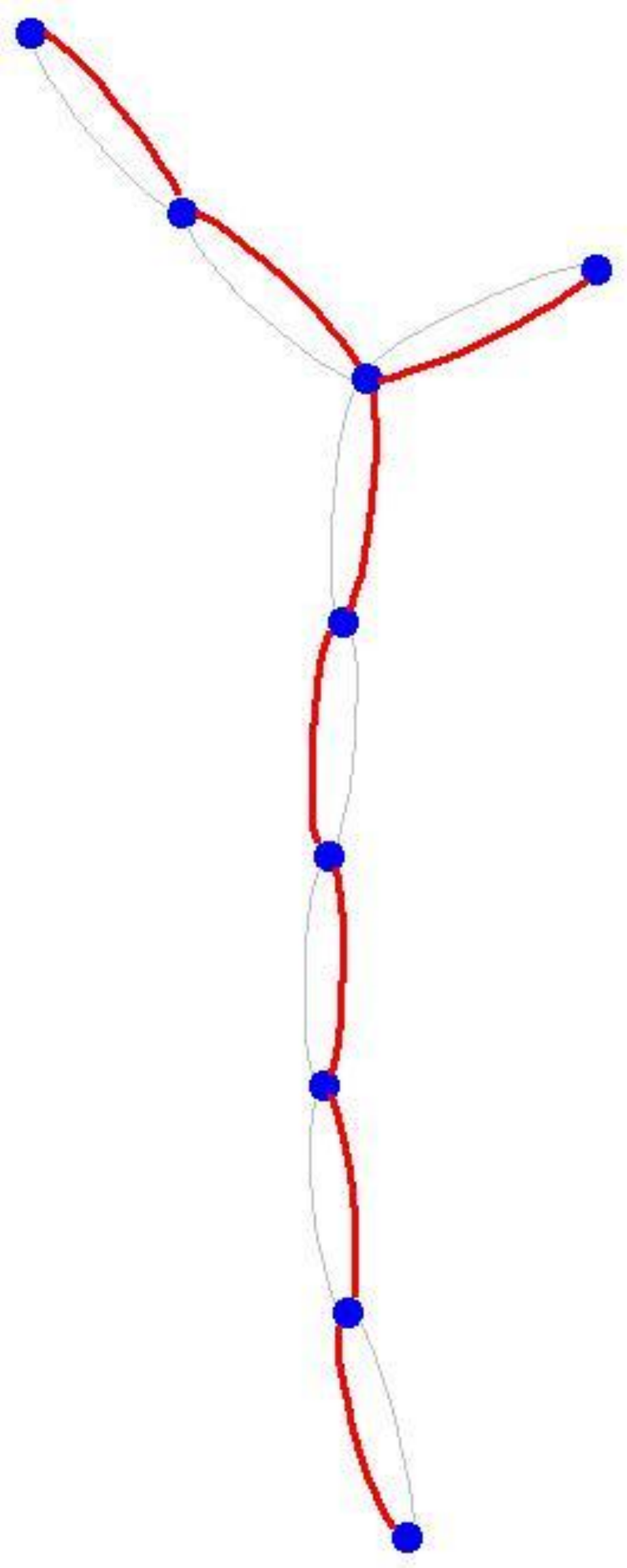}
%\label{fig:subfig3} 
}
\label{fig: healtimeline}
\caption{A timeline of deletions and self healing in a network with 100 nodes. The gray edges are the original edges and
the red edges are the new edges added by our self-healing algorithm.}
\end{figure}

%  Here we have focussed  on
% the simplest and most fundamental invariants: maintaining connectivity, ensuring the diameter of the network (or as a stronger invariant, the stretch of the network) and the degrees of all nodes do not increase by much, in the face of attacks in the form of either a deletion or an insertion.  

We have developed three different distributed self-healing algorithms, whose results are optimal (i.e. with a matching lower bound) for their particular objectives. All of them fulfill the objectives of maintaining connectivity in the network in face of adverserial attacks, and low degree increase for individual nodes. These algorithms were presented at reputed conferences and have been well received by the academic community.
 These algorithms are:

 \begin{itemize}
   
 \item \noindent {\sc DASH:} \emph{Degree Assisted Self Healing}:
% ,  presented at  \emph{IEEE  International Parallel \& Distributed Processing Symposium(IPDPS)} 2008, 
 $\DASH$ guarantees network connectivity and degree increase of at most $2 \log n$, where $n$ is the number of nodes initially in the network.   $\DASH$ is  locality-aware i.e. only the immediate neighbors of a deleted node are involved in reconstruction. Also, the healing algorithm always adds in less edges than the adversary has removed from the system.
   Empirical results show that   $\DASH$ performs well in practice on power-law networks. This is joint work with Jared Saia. An earlier version~\cite{SaiaTrehanIPDPS08} was presented at the conference \emph{IEEE International Parallel \& Distributed Processing Symposium 2008}.

\item \noindent{\sc ForgivingTree}:
% Presented at
%\emph{Principles of Distributed Computing (PODC 2008)}, 
This algorithm efficiently maintains a special spanning tree which guarantees at worst a constant additive degree increase and diameter increase of only a $\log \Delta$ factor, where $\Delta$ is the
maximum degree of a node in the original network, by a system of inheritance and wills. 
%$\FTree$ is  locality-aware i.e. only the neighbors of a deleted node are involved in reconstruction. 
This is  work  jointly done with Tom Hayes, Navin Rustagi and Jared Saia.
 An earlier version~\cite{HayesPODC08} was   presented  at the conference \emph{ACM Principles of Distributed Computing 2008}.

\item \noindent{\sc ForgivingGraph}:
%Presented at PODC2009,
 This algorithm efficiently maintains a general graph of the network, handling both deletions and insertions, while guaranteeing at worst a constant multiplicative degree increase and the simultaneously challenging property of a low ($\log n$) factor stretch (maximum distance increase between any two nodes).  Also, we introduce a  novel mergable data structure called half-full trees(\emph{haft}) having a one-to-one correspondence with binary numbers, with the merge corresponding to binary addition. This is joint work with Tom Hayes and Jared Saia.
  An earlier version~\cite{HayesPODC09} was presented  at the conference \emph{ACM Principles of Distributed Computing 2009}.

\end{itemize}

\begin{table}[h!]
\centering
%\begin{tabular}{|l|p{1.8in}|}
%\begin{tabular}{|p{0.1\textwidth}|}
\begin{threeparttable}[b]
%\newcolumntype{V}{>{\centering\arraybackslash} m{.1\linewidth} }
%\begin{tabular}{|p{0.1\textwidth}| V |m{0.1\textwidth}|c|c|c|}
\begin{tabular}{|l|c|c|p{0.08\textwidth}|p{0.12\textwidth}|p{0.11\textwidth}|p{0.08\textwidth}|}
\hline
 & \multicolumn{2}{|c|}{Adversarial Attack} &  \multicolumn{4}{|c|}{Property bounded}\\
 \hline
 & Deletion & Insertion & Connec\-tivity & Degree \mbox{(orig: d)\tnote{\textasteriskcentered}} & Diameter \mbox{(orig: D)\tnote{\textasteriskcentered}} & Stretch\\
 \hline
 DASH & \checkmark & \checkmark& \checkmark &$d + 2\log n$ & \textemdash & \textemdash \\
 \hline
 Forgiving Tree & \checkmark&$\times $& \checkmark & $d + 3$ & $D\log \Delta$  & \textemdash\\
 \hline
  Forgiving Graph & \checkmark& \checkmark & \checkmark & $3d$ & $D\log n$ & $\log n$ \\
  \hline
 \end{tabular}
 
 \begin{tablenotes} 
 \item [\textasteriskcentered] `orig:'  the original value of the property in the graph (i.e. the value in the graph $\G'$ in our model)\\
\end{tablenotes} 
\end{threeparttable}

\begin{threeparttable}[b]
%\newcolumntype{V}{>{\centering\arraybackslash} m{.1\linewidth} }
%\begin{tabular}{|p{0.1\textwidth}| V |m{0.1\textwidth}|c|c|c|}
\begin{tabular}{|l|p{0.14\textwidth}|p{0.23\textwidth}|p{0.1\textwidth}|p{0.08\textwidth}|p{0.08\textwidth}|}
\hline
 & \multicolumn{3}{|c|}{Costs} &  \multicolumn{2}{|c|}{}\\
 \hline
 & Repair time & \# Msgs  per deletion& Msg size & match lower bound\tnote{$\ddagger$} & locality (hops)\tnote{$\sharp$} \\
 \hline
 DASH & $O(\log n)$ \tnote{$\dagger$}& $O(\delta\log n + \log^2 n)$ \tnote{$\dagger$}& $O(\log n)$ & \checkmark  & 1  \\
 \hline
 Forgiving Tree &$O(1)$& $O(\delta)$ &$O(\log n)$ & \checkmark & 2 \\
 \hline
  Forgiving Graph & $O(\log \delta \log n)$ &$O(\delta \log n)$ & $O(\log^2 n)$& \checkmark & $\log n$  \\
  \hline
 \end{tabular}
 
 \begin{tablenotes} 
 \item[$\dagger$] \emph{with high probability},  and amortized over $O(n)$ deletions.
 \item[$\ddagger$] The lower bounds differ according to the properties being bounded.
 \item[$\sharp$] Number of hops from the deleted node to nodes involved in repair.
\end{tablenotes} 
\end{threeparttable}

\caption{Comparison of our self-healing Algorithms. $d$ is the degree of an individual node,  $\Delta$ is the maximum degree of a node in the graph, and $\delta$ is the degree of the deleted node.}
\label{tab: AlgoCompare}
\end{table}

Table~\ref{tab: AlgoCompare} gives a comparison of these self-healing algorithms with regards to various criteria including methods of adverserial attack, properties maintained, and costs of the algorithm.
Many important open questions remain and there are  many promising directions towards which our work can be extended. Some of these are discussed in the last chapter (Chapter~\ref{chapter: FD}).

 \section{Related Work}  

There have been numerous papers that
discuss strategies for adding additional capacity or rerouting in
anticipation of failures \cite{ AwerbuchAdapt92, doverspike94capacity,
frisanco97capacity, iraschko98capacity, murakami97comparative,
caenegem97capacity, xiong99restore}.  Results that
are responsive in some sense include the following.
  M\'{e}dard, Finn, Barry, and Gallager
\cite{medard99redundant} propose constructing redundant trees to make
backup routes possible when an edge or node is deleted.  Anderson,
Balakrishnan, Kaashoek, and Morris \cite{anderson01RON} modify some
existing nodes to be RON (Resilient Overlay Network) nodes to detect
failures and reroute accordingly. Some networks have enough redundancy
built in so that separate parts of the network can function on their
own in case of an attack~\cite{goel04resilient}.  In all these past results, the network topology is fixed. In contrast, our algorithms add or deletes edges as node failures occur.  Moreover, our
algorithms do not dictate routing paths or specifically require
redundant components to be placed in the network initially. 

%which proposed a simple line algorithm for self-healing to maintain
%network connectivity.

There has also been recent research in the physics community on
preventing cascading failures.  In the model used for these results,
each vertex in the network starts with a fixed capacity. When a vertex
is deleted, some of its ``load'' (typically defined as the number of
shortest paths that go through the vertex) is diverted to the
remaining vertices.  The remaining vertices, in turn, can fail if the
extra load exceeds their capacities. Motter, Lai, Holme, and Kim have
shown empirically that even a single node deletion can cause a
constant fraction of the nodes to fail in a power-law network due to
cascading failures\cite{holme-2002-65, motter-2002-66}. Motter and Lai
propose a strategy for addressing this problem by intentional removal
of certain nodes in the network after a failure begins
\cite{motter-2004-93}.  Hayashi and Miyazaki propose another strategy,
called emergent rewirings, that adds edges to the network after a
failure begins to prevent the failure from
cascading\cite{hayashi2005}.  Both of these approaches are
shown to work well empirically on many networks.  However, unfortunately, they
perform very poorly under adversarial attack.

A responsive approach was followed by the authors in \cite{IchingThesis,BomanSAS06},  which proposed a simple line algorithm for self-healing to maintain network connectivity. This algorithm has obvious drawbacks with regard to properties such as diameter maintenance but has served as a useful starting point for our research.

 \subsection{Self-healing and Self-* properties}
 \label{sec: Intro-Self-*}

  The importance of self-healing in systems is worth mentioning.  As an example, self-healing is one of the main components of IBM's autonomic systems initiative~\cite{IBMAutonomicManifesto,IBMAutonomicVision}. Autonomic computing itself is one of the building blocks of pervasive computing, an anticipated future computing model in which tiny - even invisible - computers will be all around us, communicating through increasingly interconnected networks~\cite{Whatis?AutonomicComputing}.  Self-healing forms one of the eight crucial elements in IBM's autonomic computing vision. Self-healing is one of the self-* properties that a system can possess, where the `*' in self-* is a wildcard character that can take on many different forms. IBM's vision often refers to an autonomic computing system  as a  self-managing system that has  the so-called self-CHOP properties: \emph{self-configuring, self-healing, self-optimizing,} and \emph{self-protecting}. Often, \emph{self-management} is a generic term which implies the system has at least one of the other self-* properties i.e. it has some desired autonomic behavior~\cite{Berns09DissectingSelf-*}.

   In the distributed systems world, perhaps the most well-known self-* property is \emph{self-stabilization}~\cite{Djikstra74SelfStabilizing, DolevBookSelfStabilization, Dolev09EmpireofColoniesSelf-stabilizing, GerardTelDistributedAlgosBook}. Self-stabilization was introduced by Djikstra in 1974~\cite{Djikstra74SelfStabilizing}. A self-stabilizing system is a system  which, starting from an arbitrary state and being affected by adversarial transient failures, can, in finite time, recover to a correct state. Often, self-stabilization does not take code corruption (byzantine behavior) or fail-stop failures (node crashes) into account. A self-healing system, when starting from a correct state, can only be temporarily out of a  correct state i.e. it recovers to a correct state, in presence of some adversarial attacks including node removal.  Other self-* properties, often broadly defined, include  \emph{self-scaling, self-repairing} (similar to self-healing), \emph{self-adjusting} (similar to self-managing), \emph{self-aware/self-monitoring, self-immune, self-containing}~\cite{Berns09DissectingSelf-*}.

%Chapter \ref{chapter: DASH}, \ref{chapter:
%FT} and Section \ref{chapter: FG} present three fully distributed self-healing algorithms, Chapter \ref{chapter: DASH} presents \emph{DASH}, an algorithm that manages to do self-healing efficiently while allowing  only small degree increase ($O(\log n)$) for any node in the network. 
% The algorithms in Chapter \ref{chapter: FT} and Chapter \ref{chapter: FG} 
%% are less constrained in the sense that they sometimes allow a few nodes to connect to nodes that were not immediate neighbors of a deleted node. 
% allow us to limit the increase in diameter and stretch of the network respectively during
%self-healing.  introduces the \emph{ForgivingTree}, a self-healing algorithm that achieves
%self-healing while allowing  a  degree increase of at most  3 for any node and keeping the diameter increase of the
%network to  within 2$\log n$ of the original diameter. and is presently submitted for publication. Section
%\ref{sec: FG} discusses the \emph{ForgivingGraph}, a generalization of the ForgivingTree we are working on, which we
%believe will help us solve the harder problem of 
% bounding the stretch of the network (the maximum multiplicative increase of distance between any two pair of
%nodes). This  work  is being pursued with Tom Hayes(TTI-C), Navin Rustagi(UNM) and Jared Saia(UNM). Both DASH
%and ForgivingGraph are designed to work  in a dynamic network, where both insertions and deletions are taking place. 

\section{Structure of the document}
\label {sec: Intro-DocStructure}
The next three chapters are self-contained presentations of the three algorithms with an occasional reference to the Introduction. Chapter \ref{chapter: DASH} presents $\DASH$, chapter \ref{chapter: FT} describes $\FTree$,  chapter \ref{chapter: FG} presents $\FGraph$. Chapter \ref{chapter: FD}  sketches some open problems and possible directions. For chapter~\ref{chapter: DASH} of this dissertation, we gratefully acknowledge the help of Iching Boman,  Dr. Deepak Kapur and his class \emph {Introduction to Proofs, Logic and Term-rewriting}, and the UNM Computer Science Theory Seminar.  

% A short appendix at the end presents some known material for completion.
%  Section \ref{sec: Sensor} discusses ideas related to self-healing in Sensor networks. This work is being Pursued with Shripad Thite (CalTech). Section \ref{sec: Social} discusses ideas related to self-healing in Social networks. The ideas presented here  arise from discussions with Willemien Kets (SFI).

%We end this introduction with a high level but more formal description of our general model and techniques.

%\subsection*{Our results}

% Chapter 1: DASH 

\chapter{DASH}
\label{chapter: DASH}

\begin{epigraphs}
%\qitem {\epitext{I was built for the long run, not for the short dash, I guess.}}%
%{\epiauthor{William Shatner}}
\qitem{\epitext{But he said what mattered most of all was the dash between those years}}%
{\episource{The Dash Poem}\\ \epiauthor{Linda Ellis} }
\end{epigraphs}

This chapter presents the first of our self-healing algorithms called $\DASH$ (short for \emph{Degree Assisted Self-Healing}, which first appeared at \emph{IEEE International Parallel \& Distributed Processing Symposium 2008}~\cite{SaiaTrehanIPDPS08} To recap, we consider the problem of self-healing in networks that are \emph{reconfigurable} in the sense that they can change their topology during an attack.  Our goal is to maintain connectivity in these networks, even in the presence of repeated adversarial node deletion, by carefully adding edges after each attack.  We present a new algorithm, $\DASH$ which provably ensures that: 1) the network stays connected even if an adversary deletes up to all nodes in the network; and 2)  no node ever increases its degree by more than $2 \log n$, where $n$ is the number of nodes initially in the network.  $\DASH$ is fully distributed; adds new edges only among neighbors of deleted nodes; and has average latency and bandwidth costs that are at most logarithmic in $n$.   $\DASH$ has these properties irrespective of the topology of the initial network, and is thus orthogonal and complementary to traditional topology-based approaches to defending against attack. The detailed model used in $\DASH$ and its relation to the general model we described in  Section~\ref{sec: Intro-self-healingModel} is given in Section~\ref{sec: DASH-Intro}. 

We also prove lower-bounds showing that $\DASH$ is asymptotically optimal in terms of minimizing maximum degree increase over multiple attacks.  Finally, we present empirical results on power-law graphs that show that $\DASH$ performs well in practice, and that it significantly outperforms naive algorithms in reducing maximum degree increase.

\section{Introduction}
\label{sec: DASH-Intro}

%On August 15, 2007 the Skype network crashed for about $48$ hours, disrupting service to approximately $200$ million users~\cite{fisher, malik, moore, ray, stone}.  Skype attributed this outage to failures in their ``self-healing mechanisms''~\cite{garvey}. 

Earlier in Chapter~\ref{chapter: Intro}, we have made a case for better ``self-healing mechanisms''  and of the need for using \emph{responsive} approaches for maintaining robust networks. There are many desirable invariants to maintain in the face of an attack.  Here we focus only on the simplest and most fundamental invariants: maintaining network connectivity and ensuring low node degree increase.

\medskip
\noindent
{\bf Our Model:} 
\label{DASH-Model}
We now describe our model of attack and network response.  We assume that the network is initially a connected graph over $n$ nodes.  We assume that every node knows not only its neighbors in the network but also the neighbors of its neighbors i.e. neighbor-of-neighbor (NoN) information.  In particular, for all nodes $x$,$y$ and $z$
such that $x$ is a neighbor of $y$ and $y$ is a neighbor of $z$, $x$ knows $z$.  There are many ways that such
information can be efficiently maintained, see e.g.~\cite{MNW,  Naor04knowthy}.

We assume that there is an adversary that is attacking the network.  This adversary knows the network topology and our algorithm, and it has the ability to delete carefully selected nodes from the network.  However, we assume the adversary is constrained in that in any time step it can only delete a small number of nodes from the network\footnote{Throughout this chapter, for ease of exposition, we will assume that the adversary deletes only one node from the network before the algorithm responds.  However, our main algorithm, $\DASH$, can easily handle the situation where any number of nodes are removed, so long as the neighbor-of-neighbor graph remains connected.}.  We further assume that after the adversary deletes some node $x$ from the network, that the neighbors of $x$ become aware of this deletion and that they have a small amount of time to react.

%\pagebreak  % To fix widows and orphans

When a node $x$ is deleted, we allow the neighbors of $x$ to react to this deletion by adding some set of edges amongst themselves.  We assume that these edges can only be between nodes which were previously neighbors of $x$.  This is to ensure that, as much as possible, edges are added which respect locality information in the underlying network.  We assume that there is very limited time to react to deletion of $x$ before the adversary deletes another node.  Thus, the algorithm for deciding which edges to add between the neighbors of $x$ must be fast and localized.

This model can be seen as a special case of our general model (Section~\ref{sec: Intro-self-healingModel}). We do not explicitly discuss node insertions in our further treatment but assume we begin with a connected graph of $n$ vertices.  $\DASH$ can easily handle insertions in a natural way, and thus, as long as the number of insertions are $O(n)$, our bounds hold. Also, for the same reason, for our bounds, we need only compare our graph properties in the present graph at timestep $t$ ($G_t$), to the initial graph $G_0$ which has $n$ vertices (notice $n$ is the maximum number of nodes the network will have in this model). 

%\medskip
%\pagebreak
\noindent {\bf Our Results:}  We introduce an algorithm for self-healing of reconfigurable networks, called  $\DASH$
(an acronym for \emph{Degree Assisted Self-Healing}). $\DASH$ is \emph{locality-aware} in that it uses only the
neighbors of the deleted node for reconnection.  We prove that $\DASH$ maintains connectivity in the network, and
that it increases the degree of any node by no more than $O(log n)$.  During reconnection of nodes, our algorithm uses
only local information, therefore, it is scalable and can be implemented in a completely distributed manner. Algorithm
 $\DASH$ is described as Algorithm \ref{algo: dash} in Section~\ref{sec: dash}.  The main characteristics of $\DASH$ are summarized in the following theorem that is proved in Section~\ref{sec: dash}.

\begin{theorem} 
$\DASH$ guarantees the following properties even if up to all the nodes in the network are deleted:
\begin{itemize}
\item The degree of any vertex is increased by at most $2 \log n $.
\item The number of messages any node of initial degree $d$ sends out and receives is no more than $2 (d + 2 \log n) \ln
n $
\emph{with high
probability}\footnote{Throughout this text, we use the phrase with high probability (w.h.p) to mean with probability at
least $1-1/n^{C}$ for any fixed constant $C$.} over all node deletions.
%\item The number of messages any node sends out and receives is no more than $2 \ln n $ \emph{with high
%probability}\footnote{Throughout this paper, we use the phrase with high probability (w.h.p) to mean with probability at
%least $1-1/n^{C}$ for any fixed constant $C$.} over all node deletions.
\item The latency to reconnect is $O(1)$ after attack; and the amortized latency to update the state of the network over $\theta(n)$ deletions 
 is $O(\log n)$ with high probability. 
\end{itemize}
\end{theorem}

\noindent
%We also prove (in Section~\ref{sec:lower}) the following lower bound that shows that  \emph{DASH} is asymptotically optimal.
We also prove (in Section~\ref{sec: lower}) the following lower bound that shows that  $\DASH$ is
asymptotically optimal.

\begin{theorem}
 Consider any locality-aware algorithm that increases the degree of any node after an attack by at most a
fixed constant.  Then there exists a graph and a strategy of deletions on that graph that will force the algorithm to
increase the degree of some node  by at least $\log n$.
\end{theorem}

We also present empirical results (in Section~\ref{sec: empirical}) showing that $\DASH$ performs well in practice and that
it significantly outperforms naive algorithms in terms of reducing the maximum degree increase. Finally (in Section~\ref{sec: empirical}) we describe $\SDASH$, a heuristic based on $\DASH$ that we show empirically both keeps node degrees small and also keeps shortest paths between nodes short.
%\medskip
%\noindent

In this chapter, we build on earlier work done in \cite{IchingThesis,BomanSAS06}, which proposed a simple line algorithm for self-healing to maintain network connectivity.

\medskip
\noindent
{\bf Table of Contents:}  
The rest of this chapter is organized as follows.  Section~\ref{sec: dash} describes the algorithm $\DASH$, and its
theoretical properties.  Section~\ref{sec: lower} gives a lower bound on locality-aware algorithms.  Section~\ref{sec: empirical} gives empirical results for $\DASH$, and several other simple
algorithms on random power-law networks. It also describes and gives results for $\SDASH$. We conclude and give areas for future work in Section~\ref{sec: conclusions}. 
%Finally, some of our more technical proofs are included in the Appendix to this paper. This includes the proof of our
%lower bounds are in Section \textbf{B}. Section \textbf{C} of the appendix gives some more empirical results and
%heuristics. 

%\textbf{The Appendix has been submitted to the Program Chair}.

\section{$\DASH$: An Algorithm for Self-Healing}
\label{sec: dash}
\setcounter{theorem}{0}

In this Section, we describe $\DASH$ and prove certain properties about it. In brief, when a deletion occurs, $\DASH$ asks the neighbors of the deleted node to reconnect themselves into a certain kind of complete binary tree. Then messages are propagated so that the nodes can keep track of which connected component they belong to. 

Let the actual network at a particular time step be $G(V,E)$. Let $E_h$ be the edges (i.e. \emph{healing edges}), that have been added by the algorithm up to that time step (note $E_h\subseteq E$). 
 Let $G_h=(V, E_h)$. We show that $G_h$ is a forest in Lemma \ref{lemma: forest}.

\subsection{$\DASH$: Degree Assisted Self-Healing}
\label{subsec: dash}

 As the acronym suggests, $\DASH$ employs information of previous degree increase to control further degree increase for a node. When a deletion occurs, we assume the neighbors of the deleted node are able to detect the deletion. Then they employ $\DASH$ to heal. To maintain connectivity, $\DASH$ connects the neighbors of a deleted node as a binary tree. The tree is structured so that the vertices which have incurred the maximum degree increase previously get to be leaves and thus not increase their degree in this round. Notice that at least half the vertices in a binary tree are leaves. The nodes maintain information about the virtual network and their connected component in this network. The algorithm tries to use only a single node from each component during reconnection and thus adds only a low number of new edges during healing.

 To describe $\DASH$ we give some definitions. Let $N(v,G)$ be the neighbors of vertex $v$ in the graph $G$ representing
the real network.  Let $N(v,G_h)$ be the neighbors of vertex $v$ in graph $G_h$ consisting of the edges added by the
healing algorithm. Let $\delta(v)$ be the degree increase of the vertex $v$ compared to its initial degree. Note that
this is not the same as the degree of  $v$ in $G_h$. 

When a node $v$ is deleted, 
partition  on the basis of their $ID$ all the neighbors of $v$ in $G$ (not having the same $ID$ as $v$). Let $UN(v,G)$
(\emph{Unique Neighbors})
be  the  set having one representative from each of the partitions. If there is  more than one node as a  possible
representative from a partition, we include the one with the lowest initial $ID$.
 %Let $UN(v,G)$ (\emph{Unique Neighbors}) be the neighbors of $v$ in $G$ with $ID$ not equal to $ID$ of $v$ such that no
%two members of $UN(v,G)$ have the same $ID$. If two or more neighbors of $v$ have the same $ID$ we include the one with
%the lowest initial $ID$.

 Note that $UN(v,G) \cap N(v,G_h) = \phi$ and $UN(v,G) \cup N(v,G_h) \subseteq N(v,G)$ . The $ID$
of a node allows us to keep track of which connected component in $G_h$ it belongs to.  The lowest $ID$ of any node in
that component is broadcast and  all the nodes in the component take on this $ID$. 

\begin{algorithm}[h!]
\caption{\textbf{DASH:} Degree-Based Self-Healing}
\label{algo: dash}
\begin{algorithmic}[1]
\STATE \emph{Init:} for given network $G(V,E)$, Initialize each vertex with a random number $ID$ between [0,1] selected
uniformly at random. 
\WHILE {true}
\STATE \emph{If a vertex $v$ is deleted, do}
\STATE Nodes in $UN(v,G) \cup N(v,G_h)$ are reconnected into a \emph{complete binary tree}. To connect the tree, go left
to right, top down, mapping nodes to the \emph{complete binary tree} in increasing order of $\delta$ value.

%that is they reconnect into a binary tree such that the Breadth First traversal gives a list sorted in ascending order
%on the Degree of the nodes. The following procedure 
\STATE Let $MINID$ be the minimum $ID$ of any node in $UN(v,G) \cup N(v,G_h)$.
 Propagate $MINID$ to all the nodes in the tree of $UN(v,G) \cup N(v,G_h)$ in $G_h$. All these nodes now set their $ID$ to
$MINID$.
%\ENDFOR
\ENDWHILE
\end{algorithmic}
\end{algorithm}

Our main results about $\DASH$ are stated in Theorem~\ref{theorem: DASH}.

\begin{theorem} 
$\DASH$ is a distributed algorithm with the following properties:
\label{theorem: DASH}
\begin{itemize}
\item The degree of any vertex is increased by at most $2 \log n $.
\item The latency to reconnect is $O(1)$.
\item The number of messages any node of degree $d$ sends out and receives is no more than $(2d + 2 \log n) \ln n $ \emph{with high probability} over all node deletions.
%\item The number of messages any node sends out and receives is no more than $2 \ln n $ \emph{with high probability} over all node deletions.
%\item The amortized latency for  $ID$ propagation is $O(log n)$ \emph{with high probability}.
\item The amortized latency for  $ID$ propagation is $O(log n)$ \emph{with high probability} over all node deletions. 
%\item The algorithm is completely distributed.
\end{itemize}
\end{theorem}

\subsection{Towards the proof of Theorem~\ref{theorem: DASH}}

For analysis, we use the following definitions:  
\def \max{\mathop{\rm max}\limits}%
\begin{itemize}
   \item Let $T(x,y)$ be the tree in $G_h-y$ that contains $x$.     
\item Each vertex $v$ will have a weight, $w(v)$. The weight of a vertex will start at 1 and may increase during the
algorithm. If $v$ is deleted, $w(v)$ is added to an arbitrarily chosen neighbor in $G_h$.
   %  \item For vertices $v$ and $x$, let 
    %     $W(v,x) = \sum\limits_{v' \in T(v,x)} w(v')$
     \item Let $W(S) = \sum\limits_{v \in V} w(v)$, for a graph $S(V,E)$ i.e. the sum of the weights of all vertices in $S$.
   \item For vertex $v$, let $\rem(v)$ =
 \[
%\rem(v) = \sum\limits_{u \in N(v,G_h)} W(T(u,v)) - 
\sum_{u \in N(v,G_h)}\!\!\!\!\!\! W(T(u,v))\, - 
  \max_{u \in N(v,G_h)}\!\!\!\!\!\!( W(T(u,v)))\, + w(v).
\]
             We will show that as the degree of a vertex increases in our algorithm, so will the  $\rem$ value of that vertex. 
           Intuitively $\rem(v)$ is large when removing $v$ from its tree in $G_h$ gives rise to many connected components with large weight.
           %$rem$ and degree increase and thereby prove our upper bound on degree increase.

\end{itemize}

\begin{lemma}
\label{lemma: forest} The edges added by the algorithm, $E_h$, form a forest.
\end{lemma}
\begin{proof}We prove this by induction on the number of nodes deleted.\\

\noindent\emph{Base Case:} Initially, $G_h$ is a forest because $E_h$ is empty.\\

\noindent We note that $E_h$ and $G_h$ change only when a deletion occurs. Consider the $i^{th}$ deletion and let $v$ be
the node deleted.\\
 Let $v$ belong to tree $T_v$ in $G_h$ just prior to the deletion of $v$.
 % Let $I_v$ be the $ID$ of each node in $T_v$. 
Now, for all $ x, y  \in N(v,G_h)$\ x and y are not connected in $E_h$ since that would have implied the existence of a
cycle through $v$ contradicting the Inductive Hypothesis. Note also that for all $z \in UN(v,G), z \notin T_v$. Since we
select only 1 node from each tree $T_i$ in which $v$ had a neighbor, no pair of nodes in $UN(v,G) \cup N(v,G_h)$ are
connected in $G_h$. We reconnect all the nodes in $UN(v,G) \cup N(v,G_h)$ in a Binary Tree and propagate the minimum ID.
Since we are adding edges between nodes which previously were in separate connected components in $G_h$, no cycles are
introduced. Hence, $G_h$ remains a forest. 

\end{proof}

\pagebreak % fix widow, orphan. Move lemma statement.

\begin{lemma}
\label{lemma: nondecreasing} For any vertex $v$, $rem(v)$ is non-decreasing over any vertex deletion where $v$ has not
been deleted.
\end{lemma}
\begin{proof}

By Lemma \ref{lemma: forest}, every vertex $v$ in $G_h$ belongs to some tree, which we will call $T_v$. For every $T_v$
in $G_h$,  $W(T_v)$ is the sum of the weights 
of all vertices in $T_v$.

By definition,
$\rem(v)$ =
 \[
%\rem(v) = \sum\limits_{u \in N(v,G_h)} W(T(u,v)) - 
\sum_{u \in N(v,G_h)}\!\!\!\!\!\! W(T(u,v))\, - 
  \max_{u \in N(v,G_h)}\!\!\!\!\!\!( W(T(u,v)))\, + w(v).
\]

%$ rem(v) = \sum\limits_{u \in N(v,G_h)} W(T(u,v)) - 
%            \max_{u \in N(v,G_h)} \left( W(T(u,v)) \right) + w(v) $

Therefore,\\
 $rem(v) =  W(T_v) - \max_{u \in N(v,G_h)} W(T(u,v))$

Observe first that $W(T_v)$ cannot decrease even when there is a deletion in $T_v$ because the deleted vertex's weight
is not ``lost", but added to some member of $T_v$. 
 
%Initially, $v$ does not have any subtrees, so $rem(v)=1$. When one edge is added to $v$, there is only one subtree,
%therefore, $rem(v)$ is still one. 
%When two edges have been added to $v$, $rem(v)$ increases. 
 
%For all subsequent steps,
 Since $W(T_v)$ cannot decrease, $rem(v)$ can only 
decrease if the maximum subtree weight increases more than $W(T_v)$. 
Since the maximum subtree is a subset of the tree, $T_v$, 
any increases or decreases in the maximum subtree is also counted in $W(T_v)$. Thus,  $rem(v)$ cannot decrease.

\end{proof}

\pagebreak % fix widow, orphan. Move lemma statement.

\begin{lemma}
\label{lemma: WgreatRem}
 For any node $v$, for all nodes $q \in N(v,G_h)$ , $W(T(v,q)) \ge rem(v)$.
 \end{lemma}

\begin{proof}

\begin{figure}[h!]
\centering
\includegraphics[scale=0.7]{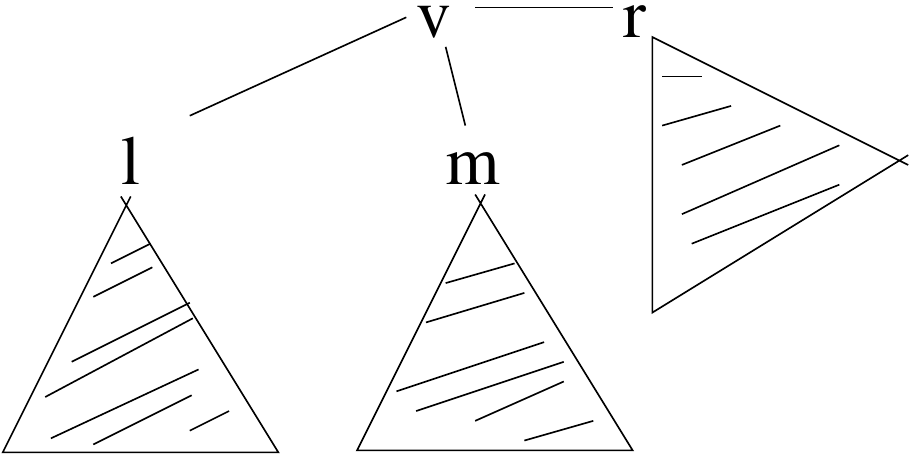}
\caption{  $W(T(v,m)) \ge rem(v)$.}
\label{fig: WgreatRem}
\end{figure}

\noindent For all nodes $q$,
\begin{eqnarray*}
   W(T(v,q)) & = &\sum\limits_{u \in N(v,G_h)\atop{u \neq q}} W(T(u,v)) + w(v)\\
    & \ge & \sum\limits_{u \in N(v,G_h)} W(T(u,v))\\
     & & - \max_{u \in N(v,G_h)} W(T(u,v)) + w(v)\\ 
    & = & rem(v)
 \end{eqnarray*}
   
%This is greater than or equal to $rem(v)$.\\

For example, in figure \ref{fig: WgreatRem}, $W(T(V,M)) = W(T(L,V)) + W(T(R,V)) + w(v) \ge rem(v)$.
\end{proof}

\pagebreak % fix widow, orphan. Move lemma statement.

\begin{lemma}
%\end{lemma}

\label{lemma: deglimit}
  For any node v, $rem(v) \ge 2^{\delta(v)/2}$, where $\delta(v)$, as defined earlier, is the degree increase of the
vertex $v$ in $G$.
\end{lemma}

\begin{proof} Let $t$ be the number of rounds of healing where a round is a single adversarial deletion followed by
self-healing by $\DASH$. We prove this lemma by induction on $t$.\\
 Let $\G_{ht}$, $\rem_t(v)$ and $\delta_{t}(v)$ be $\G_h$, $\rem(v)$ and $\delta(v)$ respectively at time $t$.\\

\emph{Base Case:} t = 0:
  In this case, all nodes $v$  have  $\delta(v) = 0$; $\rem(v) = 1$.  Thus, $\rem(v) \ge 2^{0}$.\\

 \emph{Inductive Step:}
 Consider the network at round $t$. We assume by the inductive hypothesis that for all nodes $v$ in $\G_h$,
$\rem_{t-1}(v) \ge 2^{\delta_{t-1}(v)/2} $. Our goal is to  show that $\rem_{t}(v) \ge 2^{\delta_{t}(v)/2} $.

 Suppose node $x$ was deleted at round $t$. According to our algorithm, some or all of the neighbors of $x$ will be
reconnected as a binary tree. Let us call this  tree $\RT$ (short for \emph{Reconstruction Tree}).  Let $T(x,y)$ be
the tree in $G_{h(t-1)}-y$ that contains $x$, and $T'(x,y)$ be the tree in $G_{ht}-y$ that contains $x$.

 Consider a surviving vertex $v$. If $v$ is not a part of $\RT$, then by a simple application of lemma \ref{lemma:
nondecreasing}, our induction holds. If $v$ is a part of $\RT$, there are 3 possibilities:

\begin{enumerate}
\item \emph{$v$ is a leaf node in $\RT$}

 The degree of $v$ did not change. Thus, $\delta_t(v) = \delta_{t-1}(v)$. By Lemma  \ref{lemma: nondecreasing},
$\rem_{t}(v) \ge \rem_{t-1}(v)$. Thus, using  the induction hypothesis, $\rem_{t}(v) \ge 2^{\delta_{t}(v)/2} $.

\item \emph{$v$ is the root of $\RT$}

\begin{figure}[h!]
\centering
\includegraphics[scale=0.55]{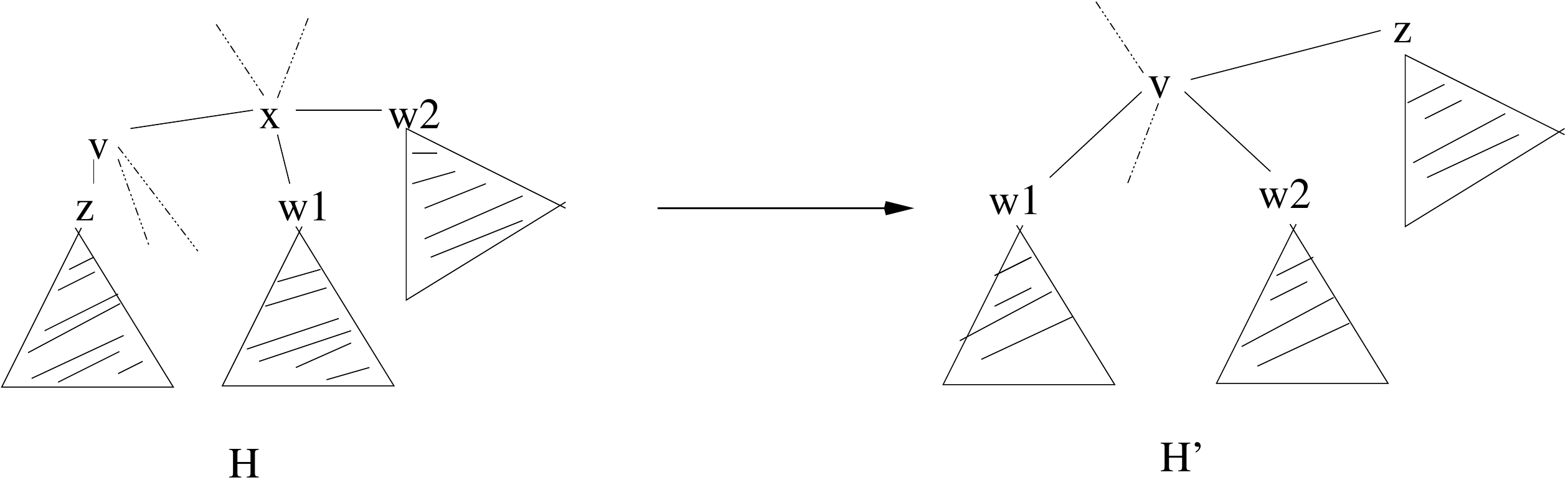}
\caption{node $v$ is the root, with 2 children}
\label{fig: vroottree}
\end{figure}

 If $v$ has only one child in $\RT$, then this is the same as the previous case with the parent and child role reversed
and the induction holds. Let us consider the case when $v$ has two children in $\RT$. Now, $\delta_t(v)$ has increased
by 1. Let $z$ be the neighbor of $v$ such that $W(T(z,v))$ is the largest among all neighbors of $v$ except $x$. Note
that $W(T'(z,v)) = W(T(z,v))$, since this subtree was not involved in the reconstruction.
Consider the possibly empty subtree of $v$ rooted at $z$.
 % such that the subtree headed by $v$ is the heaviest (has the highest $W(.)$ ) amongst all neighbors of $v$ not
%including $x$. 
  Let the two children of $v$ in $\RT$ be $w_1$ and $w_2$, as illustrated in figure \ref{fig: vroottree}. By our
algorithm, we
know that $\delta_{t-1}(w_1) \ge \delta_{t-1}(v)$ and $\delta_{t-1}(w_2) \ge \delta_{t-1}(v)$. Thus, using the inductive
hypothesis and lemma
\ref{lemma: WgreatRem}, we have that $\W(\T(w_1,x)) \ge \rem_{t-1}(w_{1}) \ge 2^{\delta_{t-1}(w_{1})/2}$ and $\W(\T(w_2,x))
\ge \rem_{t-1}(w_{2}) \ge 2^{\delta_{t-1}(w_{2})/2}$. 
 By lemma \ref{lemma: nondecreasing}, this implies that in $G_{ht}$,
  \begin{eqnarray*}
    W(\T'(w_{1},v)) & \ge  2^{\delta_{t-1}(w_{1})/2} & \ge 2^{\delta_{t-1}(v)/2}\\
    W(\T'(w_{2},v)) & \ge  2^{\delta_{t-1}(w_{2})/2} & \ge 2^{\delta_{t-1}(v)/2}
   \end{eqnarray*}
   
  Assume without loss of generality that   $\W(\T'(w_1,v)) \le \W(\T'(w_2,v))$. There are two cases:

\begin{enumerate}
\item $\W(\T(z,v)) < \W(\T'(w_1,v))$

  \qquad In this case $\rem_{t-1}(v)$ did not include $\W(\T(x,v))$.
 But  $\rem_t(v)$ will include $\W(\T'(w_1,v))$
 % and possibly $w(x)$ if the weight of $x$ was added  to $v$. 
 Hence,\\
\begin{eqnarray*}
 \rem_t(v) & \ge & \rem_{t-1}(v) + \W(\T'(w_1,v))\\
  & \ge & 2^{\delta_{t-1}(v)/2} + 2^{\delta_{t-1}(v)/2}\\
  & = & 2^{(\delta_{t-1}(v)+2)/2} \\
  & = & 2^{(\delta_{t}(v)+1)/2}
  \end{eqnarray*}

\item $\W(\T(z,v)) \ge \W(\T'(w_1,v))$

 \qquad In this case  $\rem_t(v)$ will include $\W(\T'(w_1,v))$ and the smaller of $\W(\T'(w_2,v))$ and  $\W(\T'(z,v))$.
  Note that by Lemmas \ref{lemma: WgreatRem} and \ref{lemma: nondecreasing}, the inductive hypothesis, and the fact
that $\delta_{t-1}(w_{1})
\ge
\delta_{t-1}(v) $,  $\W(T'(w_1,v)) \ge \rem_{t}(w_{1}) \ge \rem_{t}(w_{1}) \ge 2^{\delta_{t-1}(w_{1})/2} \ge 
2^{\delta_{t-1}(v)/2}$. \\ 
Also, since by assumption $\W(T'(w_2,v)) \ge \W(T'(w_1,v)) $, we know that $\W(T'(w_2,v)) \ge
2^{\delta_{t-1}(v)/2} $. \\ 
 Further, since $\W(T'(z,v)) =  \W(T(z,v)) \ge \W(T'(w_1,v))$ we know that $\W(T'(z,v)) \ge
2^{\delta_{t-1}(v)/2}$. 
  
%Since $\W(T'(w_1,v)) \ge 2^{\delta_{t-1}(v)/2}$,  and $\W(T'(z,v)) \ge 2^{\delta_{t-1}(v)/2}$, 
Hence,
 \begin{eqnarray*}
 \rem_t(v) & \ge & 2^{\delta_{t-1}(v)/2} + 2^{\delta_{t-1}(v)/2}\\
  & = & 2^{(\delta_{t-1}(v)+2)/2} \\
  & = & 2^{(\delta_{t}(v)+1)/2}
  \end{eqnarray*}

\end{enumerate}

\item \emph{$v$ is an internal node in $T'$}
\label{case: deg1-intnode}
\begin{figure}[h!]
\centering
\includegraphics[scale=0.55]{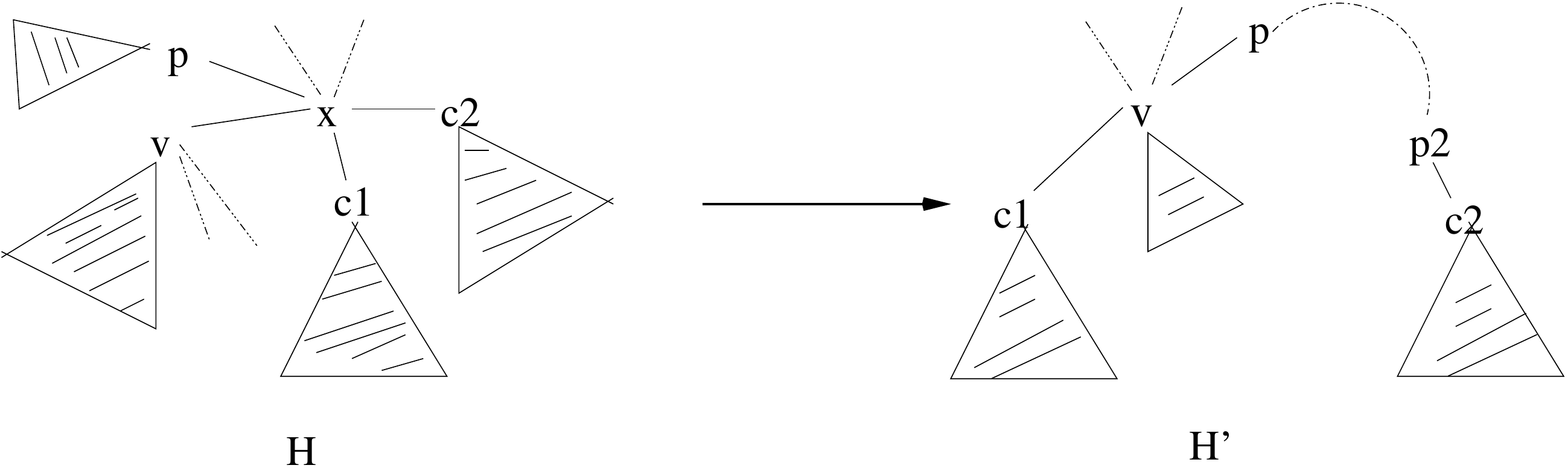}
\caption{Internal node $v$ with 1 child}
\label{fig: vinternalnode}
\end{figure}

\begin{figure}[h!]
\centering
\includegraphics[scale=0.6]{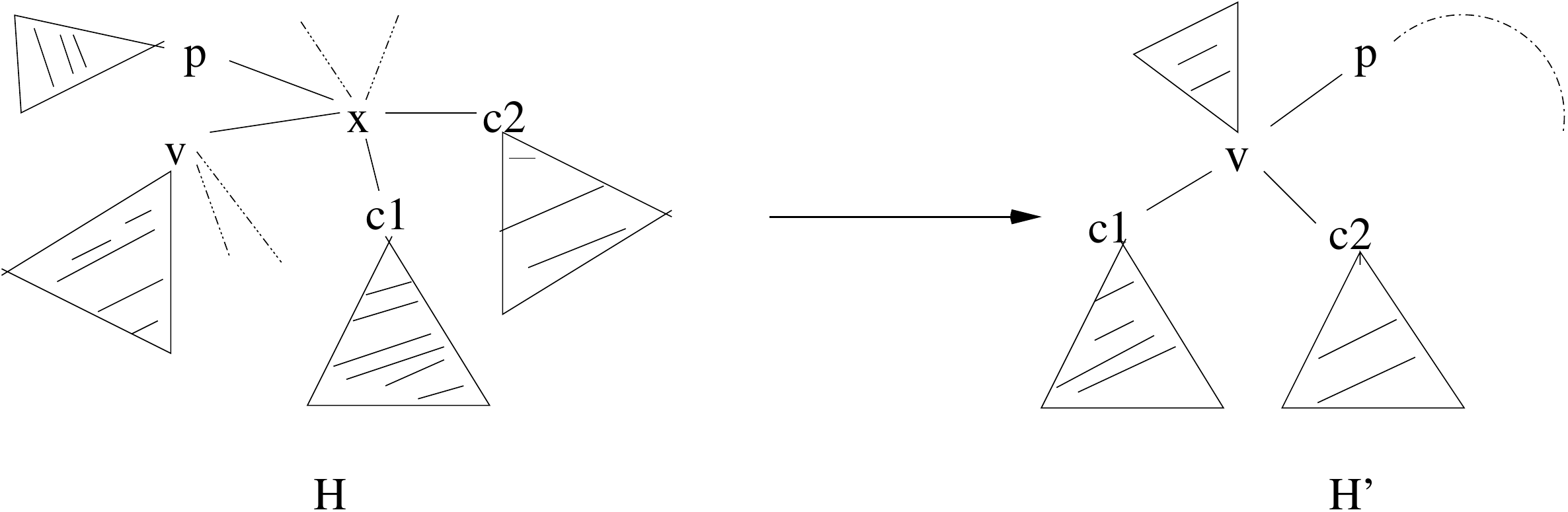}
\caption{Internal node $v$ with 2 children}
\label{fig: vdegplustwo}
\end{figure}

 For node $v$ to become an internal node, the deleted
neighbor $x$ must have at least three other neighbors.  Three neighbors of $x$  are shown as $C1$, $C2$ and $P$ in the
figures \ref{fig: vinternalnode} and \ref{fig: vdegplustwo}. 
Also, now $v$'s degree can increase by 1, as illustrated in figure \ref{fig: vinternalnode}, or by 2, as illustrated in figure  \ref{fig: vdegplustwo}.  Let us consider these cases separately:

\begin{enumerate}
 \item  $\delta_t(v) = \delta_{t-1}(v) + 1$

This can only happen when $v$ has a parent and a single child in $\RT$ as in figure \ref{fig: vinternalnode}. Let $P$ be
the parent of $v$ and
$C1$ the child of $v$.  $C1$ has to be a leaf node since the tree is complete and $v$ has only one child.  Observe that
there
exists at least one leaf node besides $C_{1}$ in the tree, accessible to $v$ only via $P$. Let this node be $C2$ and let
$P2$ be its parent. Note that $P2$ and $P$ may even be the same node. In our algorithm, any leaf node in $\RT$ has a
$\delta$ value no less than the $\delta$ value of any internal node. Thus,
 \begin{eqnarray*}
 \delta_{t-1}(C1) & \ge &\delta_{t-1}(v); \textrm{ and}\\
 \delta_{t-1}(C2) & \ge & \delta_{t-1}(v)
 \end{eqnarray*}
 %These inequalities and the  Inductive Hypothesis, 
 %\begin{eqnarray*}
 %\rem_{t-1}(C1) & \ge & 2^{(\delta_{t-1}(v))/2}\ \textrm{; and}\\
 %\rem_{t-1}(C2) & \ge & 2^{(\delta_{t-1}(v))/2}
 %\end{eqnarray*}

These inequalities, Lemmas \ref{lemma: nondecreasing} and \ref{lemma: WgreatRem}, and the  Inductive Hypothesis, imply
that
 \begin{eqnarray*}
\W(\T'(C1,v)) & \ge &  \rem_{t}(C1) \\
 &\ge & \rem_{t-1}(C1) \\
 &\ge & 2^{\delta_{t-1}(v)/2};\\
\W(\T'(C2,P2)) & \ge &  \rem_{t}(C2) \\
 &\ge & \rem_{t-1}(C2) \\
  &\ge & 2^{\delta_{t-1}(v)/2};\\
\W(\T(v, x)) & \ge & \rem_{t}(v)\\
  & \ge & \rem_{t-1}(v) \\
  &\ge &  2^{\delta_{t-1}(v)/2}.
 \end{eqnarray*}

 Since $\rem_t(v)$ can exclude at most one of $W(T'(C1,v))$, $W(T'(C2,P2))$ and $W(T(v, x))$, 
\begin{eqnarray*}
\rem_t(v) & \ge &  2^{\delta_{t-1}(v)/2} +  2^{\delta_{t-1}(v)/2} \\
  & = &  2^{(\delta_{t}(v) + 1)/2}
\end{eqnarray*} 

\item   $\delta_t(v) = \delta_{t-1}(v) + 2$
 
In this case $v$ has two children in $\RT$, $C1$ and $C2$, as illustrated in figure \ref{fig: vdegplustwo}. The analysis
is similar to the case above. The value  $\rem_t(v)$ can exclude at most one of $W(T'(C1,v))$, $W(T'(C2,v))$ and
$W(T(v,x))$ and we can show that all three of these values are at least  $2^{\delta_{t-1}(v)/2}$.
Thus, $ \rem_t(v)\ge 2^{(\delta_{t}(v))/2} $.
\end{enumerate}

\end{enumerate}

 Hence, the induction holds.
 
\end{proof}

%\begin{corr}
%\label{corr: doubleRem}
 %The value $rem(v)$ doubles each time $v's$ degree increases due to a deletion. 
%\end{corr}
%\begin{proof} When $v's$ degree goes from $j-1$ to $j$ or to $j+1$,  $rem(v)$ goes from $2^{\frac{j-2}{2}}$ to at least
%$2^{\frac{j}{2}}$, thus $rem(v)$ at least doubles.
 %\end{proof}
 
 \pagebreak % fix widow, orphan. Move lemma statement.

\begin{lemma} For all vertices $v$, $rem(v)$ is always no more than n. \end{lemma}
\label{lemma: remNorless}
\begin{proof} No vertex is counted twice in a $rem$ value since the subtrees of a vertex are disjoint. Since the number of vertices in the subtrees cannot be more than the number of vertices remaining, the $rem$ value is always no more than the sum of the weights of all undeleted vertices in $G_h$.

Define $W^*$ to be the sum of weights of all undeleted vertices in $G_h$.  After initialization, $W^* = n$, since there are $n$ vertices.  At each step of the algorithm, $W^* = n$ , since the weight of the deleted vertex is added to one of the remaining vertices.  Thus, for node $v$,  $rem(v) \le n$. 

\end{proof}

%\pagebreak
\begin{lemma}
\label{l:degree} 
$\DASH$ increases the degree of any vertex by at most $O(\log n)$. \end{lemma}
\begin{proof} 

Every vertex $v$ starts with $\rem(v)=w(v)=1$. We know that $\rem(v)  \ge  2^{\delta(v)/2}$  \textrm{by Lemma \ref{lemma:
deglimit}}.  since $\rem(v)$ is at most n,  $2^{\delta(v)/2}  \le n$ . Taking $\log$ of both sides, $\delta(v)/2  \le 
\log n $.  Solving for $\delta(v)$ gives $ \delta(v)  \le  2 \log n  $. 

\end{proof}

\begin{lemma} 
\label{l:latency}
The latency to reconnect the network in $\DASH$ is $O(1)$.
\end{lemma}
\begin{proof}
During the reconnection process, $\DASH$ requires communication only between nodes one hop away, thus, the latency is just $O(1)$. 
\end{proof}

\pagebreak % fix widow, orphan. Move lemma statement.

\begin{lemma} 
\label{l:messages}
The number of messages any node of initial degree $d$ sends out and receives is no more than $2(d + 2 \log n)\ln n $
\emph{with high probability} over all node deletions.
%The number of messages any node sends out and receives is no more than $2 \ln n$ with high probability over all node deletions.
\end{lemma}
\begin{proof}
 In $\DASH$, after the reconnections have been made, messages are sent out by nodes when the minimum $ID$ has to be
propagated. With similarity to the \emph{record breaking problem} \cite{Glick-RecordBreaking78}(Section~\ref{subsec: recordbreaking}), it is easily shown that \emph{w.h.p.}, a node
has its
$ID$ reduced no more than 2 $\ln n$ times, where the record is the node's $ID$. These are the only messages the node needs to 
 transmit or receive. Each time its $ID$ changes, the node sends this message to all its neighbors, Thus, it sends or receives
$O((d + \log n) \ln n)$ messages, since the final degree of the node is at most $d + 2 \log n$.

%\ref{app: recordbreaking}.
\end{proof}

\begin{lemma}
\label{l:latencyid}
 The amortized latency for  $ID$ propagation is $O(\log n)$ \emph{with high probability} over all node deletions.
\end{lemma}

\begin{proof}
Again, with similarity to the record breaking problem,  a node  sends messages to its neighbors
(neighbors, by definition, are a single hop away) only $O(\log n)$ times with high probability. Thus, messages are
transmitted $O(n \log n)$ times over all the nodes. Over $O(n)$ deletions, this implies that the amortized latency for
messages (involving $ID$ propagation) is only $O(\log n)$ .
 %The latency (number of hops)%that are required for ID propagation after a node deletion is no more than the number of
%messages sent after the node deletion.  Thus the proof follows immediately from lemma~\ref{l:messages}.
\end{proof}
 
 \subsection{The Record Breaking Problem}
\label{subsec: recordbreaking} Here we recap the well known record breaking problem. 
Given a sequence of deleted vertices, $v_1, v_2, ..., v_n$, we define 
$id(v_j)$ ($j \le n$) to be a record value if $id(v_j) < id(v_i)$ for all 
$1 \ge i < j$. \\  

Let $X_1, X_2, ..., X_n$ be indicator random variable:
\[
%X_j = \begin{cases}
%X_j =  1 & \text{if $id(v_j)$ is a record value} \\
X_j = 
\begin{array}{cc} 1  & \textrm{if $id(v_j)$ is a record } \\
0  & \textrm{otherwise}
\end{array}
%  0 & \text{otherwise}
 %     \end{cases}
\]

The probability that $v_j$ is a record  is $P_j=\frac{(j-1)!}{j!} = 
\frac{1}{j}$. Therefore: $E[X_j] = 1/j$.

Let $X = \sum\limits_{j=1}^{n} X_j$. \\

By linearity of expectation: \\
$E[X] = \sum\limits_{j=1}^{n} E[X_j] = \sum\limits_{j=1}^{n} 1/j = \theta ( ln(n) )$ 

The variance for $X_j$ is $Var(X_j) = E[X_j^2] - E[X_j]^2$. We calculate $E[X_j^2]$ from the second derivative of the
moment generating 
function for $X_j$.

\begin{eqnarray*}
M''(t) & = & E[X_j^2 e^{tX_j}] \\
      & =  & \sum_j X_j^2 e^{tX_j} P_j \\
      & = & (1) (e^{t(1)})(1/j) + 0 \\
      & = & e^t/j
\end{eqnarray*}

\begin{eqnarray*}
   Var(X_j) & = & E[X_j^2] - (E[X_j])^2 \\
& = & M''(0) - (1/j)^2 \\
& = & 1/j - 1/j^2 \\
& = & (j-1)/j^2
\end{eqnarray*}

\subsection{Proof of Theorem~\ref{theorem: DASH}}

The proof of Theorem~\ref{theorem: DASH} now follows immediately from Lemmas
\ref{l:degree}, \ref{l:latency}, \ref{l:messages} and \ref{l:latencyid}.
  
%\begin{theorem} The algorithm is completely distributed.
%\end{theorem}

%\begin{proof}
%\emph{DASH} is completely distributed since it requires no global communication and can be implemented by nodes locally
%using local communication.
%\end{proof}
%\pagebreak % fix widow, orphan. Move lemma statement.

\section{Lower bounds on Locality-aware algorithms}

\label{sec: lower}

%We define a \emph{Locality aware} aware algorithm to be one which adds edges only among the neighbors of the node
%deleted during the current timestep and uses only information available with these nodes.

To begin with, we give an insight as to why a healing strategy might need to keep track of connected components.

\subsection{Necessity of Component tracking for healing strategies}
\label{subsec: comptrack}

\begin{lemma}
\label{lemma: d+2} For  a tree, deletion of  a node of degree $d$ increases the sum total of degrees of its neighbors
by $d - 2$ for a locality-aware acyclic healing strategy.
\end{lemma}
\begin{proof}

A \emph{locality-aware acyclic healing strategy} will reconnect the neighbors of a deleted node without creating any cycles. If there were no cycles in the original graph involving the neighbors and not involving the deleted node, then such a strategy can only reconnect these neighbors as a tree to maintain their connectivity.

 A node of degree $d$ has $d$ neighbors. Since it was part of a tree, this node and its neighbors also constitute a tree. Let us call this the \emph{immediate subtree}. The \emph{immediate subtree} had $d$ edges and a total of $2d$ degrees. 
 These $d$ neighbors are now reconnected as a tree with $d-1$ edges and $2(d-1)$ degrees.
 Each of these neighbors lost a single degree due to the deletion of their edge to the deleted node.
 Thus, the total degrees gained on reconstruction are $2(d-1) - d = d - 2$. 
 
\end{proof}

 %Here, we give some intuition as to why an efficient healing algorithm has to keep track of connected component information in $G_h$.
 It is reasonable to assume that an efficient healing algorithm adds close to  the minimum possible edges at each step to maintain
connectivity of the neighbors of the deleted node. In $G_h$, if a deleted node $v$ had two neighbors which had an
alternate path between themselves not involving $v$, then the algorithm may need to use only one of them for
reconnection to other nodes. By extension, if there were many neighbors which had alternate connections between them, the algorithm
may need to use only one of these nodes. This is equivalent to stating that the algorithm  may need to use only one node from a  connected component. Knowing that certain nodes are in the same component would allow the algorithm to do this.
%We know that within a \emph{connected component} all the nodes are connected to each
%other. Thus, if a deleted node had more than one neighbour which were part of the same component in $G_h$, it is
%sufficient for an algorithm to use just one of these in the reconnection process.  Such a strategy would be efficient. 
 $G_h$ is comprised only of edges added by the healing algorithm, and is always a forest. If the adversary mainly deletes nodes with degree greater than 2 and the algorithm does not use the component information, the sum total of degrees of the neighbors of the deleted nodes will increase by $ (d -2)$ i.e. at least 1, at each step. After many ($O(n)$) deletions,  only a few nodes will be left, and these will have $O(n)$ degree increase.

\subsection{A lower bound on healing by Degree-bounded locality-aware healing algorithms}
\label{subsec: lognlowerbound}

 We prove a result regarding the lower bounds for degree-bounded locality-aware algorithms in Theorem \ref{thm: lognbound}. We also show a lower bound which shows that any locality-aware healing algorithm (not necessarily degree-bounded) will increase node degree by at least $\log_2\log_3n$ in  Section~ \ref{subsec: loglognlower}.
 %However,  a stronger bound  ($\Omega(\log_3 n)$) exists in~\cite{IchingThesis, BomanSAS06} but this theorem is included  because the proof construction maybe of some interest.

  Our lower bound occurs on graphs that are originally trees. To state the proof, we need to prove some other lemmas. 
  
First, we define the following operation that the adversary can perform on trees, where we assume self-healing is applied after every deletion:
\begin{description}
%\item[Prune (r,s)]: For a node $r$ and its subtree headed by node $s$, the $Prune$ operation on $s$ leads to deletion
%of
%all the nodes in that subtree including $s$ without any change in degree of $r$ despite self-healing. This operation
%can
%be accomplished by repeatedly deleting leaf nodes in the subtree till all the nodes including $s$ are deleted.
\item[Prune (r,s)]: For a node $r$ and its subtree headed by node $s$, the $Prune$ operation on $s$ leads to deletion of
all the nodes in that subtree including $s$. This operation can be accomplished by repeatedly deleting leaf nodes in the
subtree till all the nodes including $s$ are deleted. 

\begin{figure}[h!]
\centering
%includegraphics[scale=0.8]{images/DASH/prune.pdf}
\includegraphics{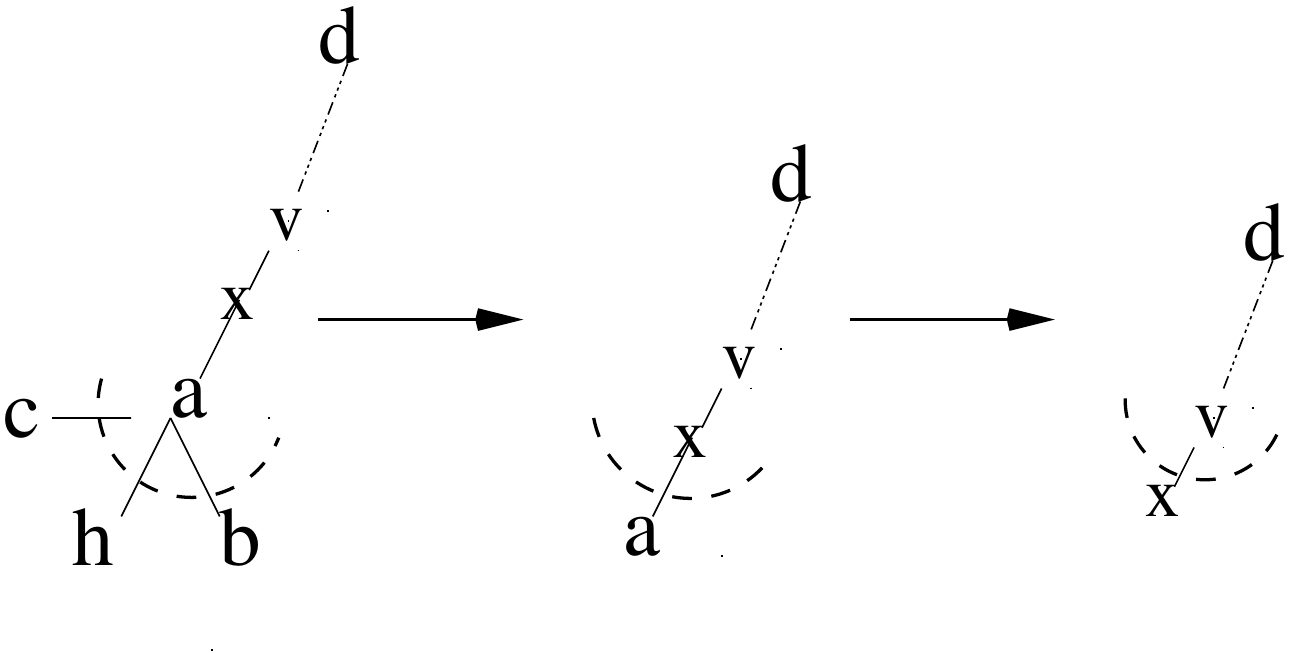}
\caption{Steps in Prune(v,x). Leaf nodes are deleted at each step. }
\label{fig: 3nbrs}
\end{figure}
%\item[Graft (r,s)]: Given a node $r$ and another node $s$ in a subtree of $r$, the Graft operation makes $r$ and $s$
%neighbors without changing the degree increase of either of them. This can be accomplished as follows: Take a node $x$
%on the path between $r$ and $s$. \emph{Prune} all subtrees of $x$ except those containing $r$ and $s$, then delete $x$.
%Repeat this process for all nodes on the path between $r$ and $s$.
\end{description}

\begin{lemma}
\label{lemma: degup} Deletion of a node with degree at least 3 increases the degree of at least one node by degree 1, no matter how the healing occurs.
\end{lemma}
\begin{proof}

\begin{figure}[h!]
\centering
\includegraphics[scale=0.5]{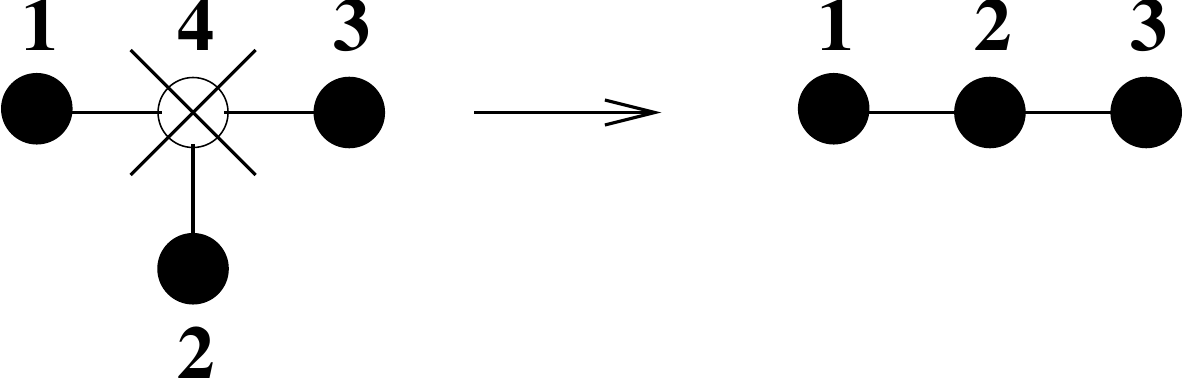}
\caption{An internal node  in  a 3-node line reconnection suffers a degree increase.}
\label{fig: 3nbrs}
\end{figure}

Any reconnection of more than two nodes has a 3-node line (as in figure \ref{fig: 3nbrs}) as a subgraph. Here the
internal node has a degree increase of $1$. Thus, at least one node increases it's degree by at least $1$.

\end{proof}

%\pagebreak

For further discussion, we define the following:
\begin{description}
\item[Degree-bounded / M-degree-bounded :] A healing algorithm is \emph{degree-bounded} or \emph{M-degree-bounded} if
any node can increase its degree by at most $M$ in a single round of deletion and healing.
\end{description}

%\subsection{A lower bound on healing by constant degree increase locality-aware healing algorithms}

\begin{lemma}
%\label{lemma: M+3up1} For a tree,  deletion of a node $v$ with degree $M+3$, leads to degree increase for at least two
%neighbors of $v$, when a M-degree-bounded locality aware healing algorithm is used. 
\label{lemma: M+3up1} Consider a M-degree-bounded locality-aware healing algorithm used on a tree. In such a
situation, deletion of a node $v$ with degree at least M+3 leads to degree increase for at least two neighbors of $v$.
\end{lemma}
\begin{proof}
 Node $v$ has $M+3$ neighbors. By Lemma \ref{lemma: d+2}, the sum total of degree increase of neighbors is $M+1$, when
the
graph is a tree. Since one node can get a maximum degree increase of $M$, at least one node has to incur the rest of the
degree increase. Thus, at least two nodes have to increase their degrees.

\end{proof}

\begin{figure}[h!]
\centering
\includegraphics{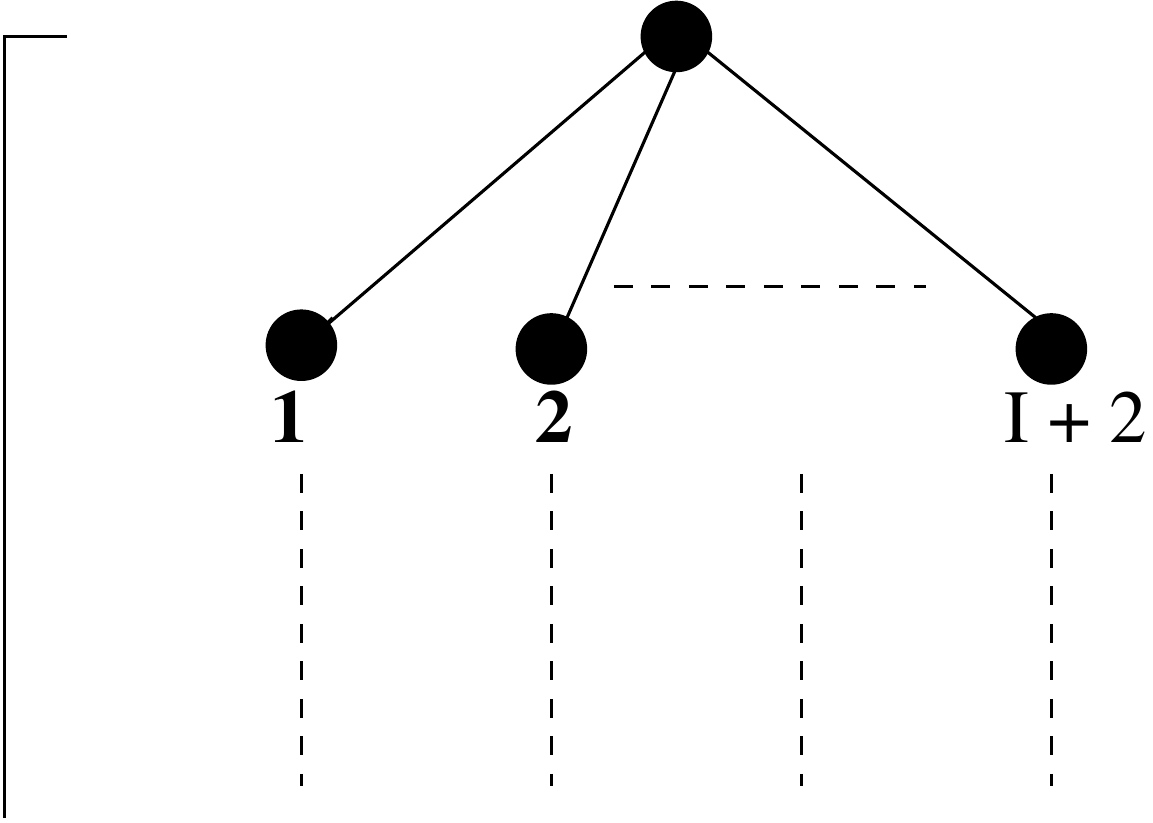}
\caption{M+2 -ary Tree}
\label{fig: Mplus2tree}
\end{figure}

\begin{algorithm}[h!]
\caption{\textsc{LevelAttack}: level-by-level attack on a (M+2)-ary tree}
\label{algo: levelattack}
\begin{algorithmic}[1]
\STMT Consider an (M+2)-ary tree $T$ of depth $D$ with levels numbered $0$ to $D$, the root being at level
$0$.
\STMT $i \leftarrow  D-1$
\WHILE {$i \ge 0$}
 \FOR{each node $v$ at level $i$} \label{levelalgo: forlevel}
 \STATE \label{levelalgo: prune} if $v$ has $c > M+2$ children  remove the excess $c - (M+2)$ nodes by deleting those
with least degree increases and their subtrees by using the \emph{Prune} operation, so that $v$ now has $M+2$ children.
\STATE delete $v$.
\ENDFOR
\STMT $i \leftarrow i-1$
\ENDWHILE
\end{algorithmic}
\end{algorithm}

\pagebreak

Here, we introduce a new attack strategy:\\
\begin{description}
\item{\textsc{LevelAttack}:}
This strategy is described in Algorithm \ref{algo: levelattack}. In brief, the adversary deletes nodes one level at a
time beginning one level above the leaves of a $M+2$-ary complete tree going up to the root. 
The reasoning behind the strategy is the following: If the adversary deletes a node of degree $M+3$ in a tree, this
ensures that a degree increase of at least 1 is passed to its children. What the adversary must do is to  ensure that 
$log n$ of these degree increases are credited to the same node.
\end{description}

\begin{lemma}
%\label{lemma: levelupD} For a $(M+2)-ary$ tree $T$, and the healing algorithm $H_{M}$ and \textsc{LevelAttack}, the deletion of anode at level $L_{i}$, $0 < i < D$ gives a node with degree increase at least $D-i$, and this was a leaf node in the original tree i.e. at level $L_{D}.
\label{lemma: levelupD} Assume a $(M+2)-ary$ tree $T$, a degree-bounded locality-aware healing algorithm and the
\textsc{LevelAttack} adversarial strategy. Then, when \textsc{LevelAttack} deleted a node at level $i$,
$0 < i < D$ some leaf node of the original tree  increases its degree by at least $ D - i$.
\end{lemma}

\begin{proof} The proof is by induction.

\emph{Base case:} In the \textsc{LevelAttack} strategy, the nodes at level $D-1$ are deleted first. Thus, a deletion of a node at
$D-1$ is our base case. A node at level $D-1$ has $M+3$ neighbors. By lemma \ref{lemma:
M+3up1}, there is at least one leaf node that increases its degree by 1 or more. Thus, the base case
holds.

\emph{Inductive step:} Assume the hypothesis holds for nodes at level $i+1$. We now show that it holds for nodes at level $i$.
 Consider a node, say X at level $i \ge 0$ .  It had $M+2$ children at level $i+1$. By the inductive hypothesis, each of
these deletions led to at least one node with degree  $D-(i+1)$. Moreover, $X$ is not among  these $M+2$ nodes.
Moreover, all of  these are now neighbors of $X$, since $X$ itself was involved in each of these deletions.
 The \emph{Prune} algorithm in step \ref{levelalgo: prune} retains only these $M+2$ as children of $X$. Each of these children has degree increase $D-(i+1)$
and was originally a leaf node of $T$. The adversary now deletes $X$.  By lemma \ref{lemma: M+3up1}, at least one of these
children incurs a degree increase.
%3) Adversary prunes the subtrees of X so that it has only these I+2 high degree increase children. Notice X now has
%degree I+3 again. 

%4) delete X, and nodes at its level.

\end{proof}

\begin{theorem}
\label{thm: lognbound}
% There exists a graph $G$ such that for any locality aware algorithm on $G$ that can increase the degree of a node by a maximum constant,  there exists an adversary strategy that forces some node to increase it's degree by $\log n$, where $n$ is the number of nodes in $G$.
% There exists a graph $G$ such that any locality-aware degree-bounded algorithm can be forced to increase the degree
%of some node  by at least $\log n$.
Consider any locality-aware algorithm that increases the degree of any node after an attack by at most a
fixed constant.  Then there exists a graph and a strategy of deletions on that graph that will force the algorithm to
increase the degree of some node  by at least $\log n$.
\end{theorem}

\begin{proof}
 
 It is sufficient to give a graph and an attack strategy such that any degree-bounded  locality-aware healing algorithm
will have to increase a particular node's degree by $\log n$. 
%Let $H_{M}$ be the healing algorithm and
Let $M$ be the constant degree increase that is the maximum that the healing algorithm can impose on any one node in
the graph. Then, for a graph which is a full (M+2)-ary tree ( Figure \ref{fig: Mplus2tree}),
 the adversary uses  \textsc{LevelAttack}.  

 Consider a (M+2)-ary tree $T$ of depth $D$ with levels numbered
$0$ to $D$. By lemma \ref{lemma: levelupD},  after the last deletion in the
adversary strategy, which is the deletion of the  root of $T$ i.e. the node at level $0$ there is at least one node
left which has a degree increase of $D$. Since $D$ is $O(log n)$, this adversary strategy achieves a degree increase of
at least $O(log n)$.

\end{proof}

\subsection{A general lower bound on healing by locality-aware algorithms}
\label{subsec: loglognlower}

%For the  discussion that follows, consider the following structure:  Let $T$ be a top-level 4-level complete subtree,
%as illustrated in figure \ref{sfig: t4l}. Top-level subtree implies that the root node has no parent but each of the
%leaf node may themselves have other subtrees hanging off them. There are four levels labelled from $L0$ to $L3$. Let
%$\delta(v)$ be the increase in degree experienced by node $v$.

For the  discussion that follows, consider the following structure:  Let $T$ be a top-level 3-level complete subtree, as
illustrated in figure \ref{sfig: t4l}. Top-level subtree implies that the root node has no parent but each of the leaf
node may themselves have other subtrees hanging off them. There are three levels labeled from $0$ to $2$. Let
$\delta(v)$ be the increase in degree experienced by node $v$.

We also define another operation called Graft, which uses the previously defined operation Prune. 
%Both operations are described as follows:

\begin{description}
%\item[Prune (r,s)]: For a node $r$ and its subtree headed by node $s$, the $Prune$ operation on $s$ leads to deletion of
%all the nodes in that subtree including $s$. This operation can be accomplished by repeatedly deleting leaf nodes in the
%subtree till all the nodes including $s$ are deleted. 

\item[Graft (r,s)]: Given a node $r$ and another node $s$ in a subtree of $r$, the Graft operation makes $r$ and $s$
neighbors without changing the degree increase of either of them. This can be accomplished as follows: Take a node $x$
on the path between $r$ and $s$. \emph{Prune} all subtrees of $x$ except those containing $r$ and $s$, then delete $x$.
Repeat this process for all nodes on the path between $r$ and $s$.
\end{description}

\begin{algorithm}[h!]
\caption{ \textsc{(Root Node):} Increase degree by 2 for  a 3-level ternary subtree}
\label{algo: strategyup3}
\begin{algorithmic}[1]
\STATE If, at any point, any node has its degree increased by 2, stop.
\STATE  Delete all nodes at level $1$. \label{algostep: delL1}
\STATE  \FOR{Root Node $r$ (Level $0$)}  \label{algostep: L0prune}
 \WHILE{there is a neighbor $v'$ where $\delta(v') = 0$ }
 \STMT delete  $v'$
 \ENDWHILE
 \ENDFOR
 \STATE  \label{algostep: delroot} delete the Root Node (level $0$).
\end{algorithmic}
\end{algorithm}

%\begin{algorithm}[h!]
%\caption{ Increase degree by 3 for  a 4-level ternary subtree}
%\label{algo: strategyup3}
%\begin{algorithmic}[1]
%\STMT If, at any point, any node has its degree increased by 3, stop.
%\STATE \label{algo: sup3L2} Delete all nodes at level $L2$. 
%\STATE  Delete  all nodes $v$ at level $L1$ which have $\delta(v) <  2$, 
%\STATE \emph{round3a:} For the root node $r$ (Level $L0$$) and all its neighbors $v'$ where $\delta(v')$ = 0, till
%there are no such neighbours.
%\STATE \emph{round3b:} delete the root node (level $L0$).
%\end{algorithmic}
%\end{algorithm}

%\begin{figure}[h!]
%\centering
%\subfigure[3-Level complete ternary subtree $T$]{\label{sfig: t4l} \includegraphics[scale=0.5]{DASH/ternary3l.jpg}}\\
%\subfigure[ $T$ after round1]{\label{sfig: t4lr1} \includegraphics[scale=0.5]{DASH/ternary3l1ra.jpg}}
%\subfigure[ $T$ after strategic deletion]{\label{sfig: t4lr2} \includegraphics[scale=0.5]{DASH/ternary3l1rb.jpg}}
%\subfigure[ $T$ after round 2]{\label{sfig: t4lr3a} \includegraphics[scale=0.5]{DASH/ternary3l2r.jpg}}
%%\subfigure[ $T$ after round3]{\label{sfig: t4lr3b} \includegraphics[scale=0.5]{ternary4l3rb.jpg}}
%\caption{\emph{Strategy-1}}
%\end{figure} 

\begin{figure}[h!]
\centering
\subfigure[3-Level complete ternary subtree $T$]{\label{sfig: t4l} \includegraphics[scale=0.75]{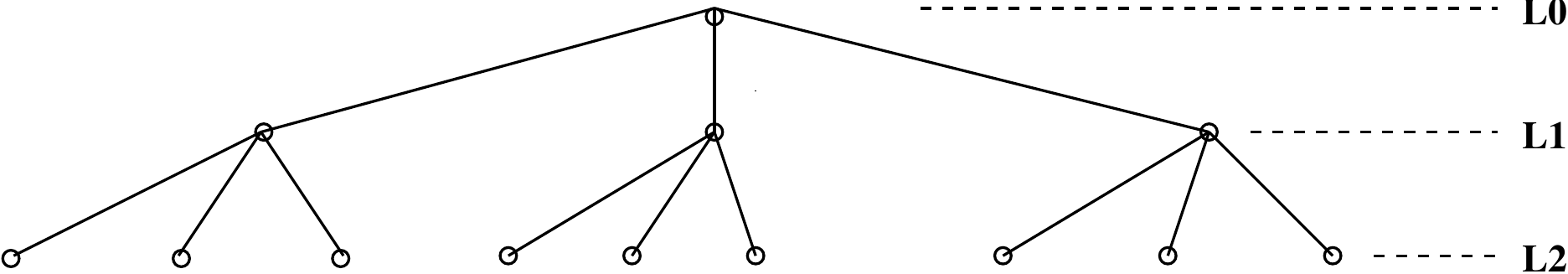}}\\
\subfigure[ $T$ after round1]{\label{sfig: t4lr1} \includegraphics[scale=0.8]{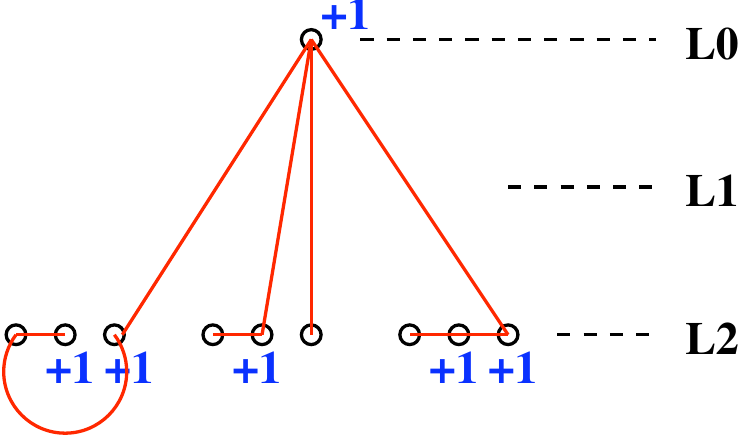}}\hspace*{35pt}
\subfigure[ $T$ after strategic deletion]{\label{sfig: t4lr2} \includegraphics[scale=0.8]{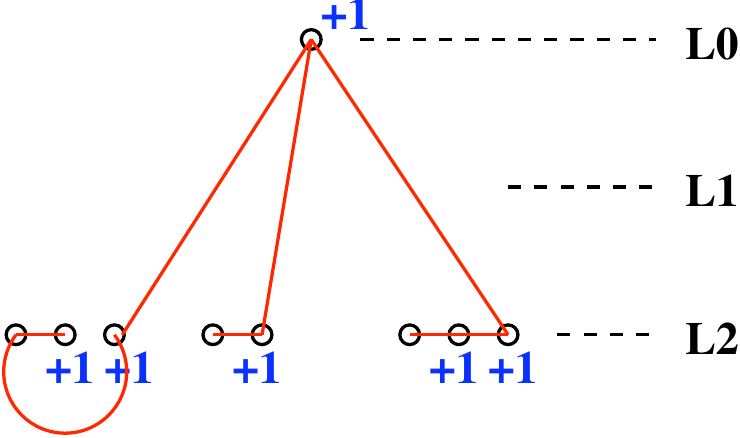}}\\
%\subfigure[ $T$ after round1]{\label{sfig: t4lr1} \includegraphics{images/DASH/ternary3l1ra}}\hspace*{5pt}
%\subfigure[ $T$ after strategic deletion]{\label{sfig: t4lr2} \includegraphics{images/DASH/ternary3l1rb}}\\
%\subfigure[ $T$ after round 2]{\label{sfig: t4lr3a} \includegraphics[scale=0.8]{images/DASH/ternary3l2r}}
\subfigure[ $T$ after round 2]{\label{sfig: t4lr3a} \includegraphics{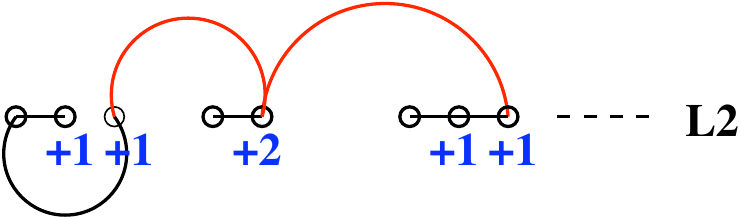}}
%\subfigure[ $T$ after round3]{\label{sfig: t4lr3b} \includegraphics[scale=0.5]{ternary4l3rb.jpg}}
\caption{\emph{Strategy-1}}
\end{figure}

\begin{lemma}
\label{lemma: 3level3tree} For a top-level 3-level complete ternary subtree, for any locality aware algorithm, the
adversary strategy Algorithm \ref{algo: strategyup3} forces some node to increase it's degree by 2.
\end{lemma}
\begin{proof}
  Let $T$ be the top-level 3-level ternary subtree
depicted in figure \ref{sfig: t4l}. There are 3 levels labeled $0$ to $2$. The adversary strategy Algorithm \ref{algo: strategyup3} consists of 3
possible rounds of deletions. As expected, a locality-aware self-healing algorithm does self-healing after every node
deletion. In the following, the steps refer to the steps of the algorithm \ref{algo: strategyup3}.

\begin{itemize}
\item \underline{Round 1; step \ref{algostep: delL1}:}
  The adversary deletes all nodes at level $1$. By lemma \ref{lemma: degup}, at least one neighbor of the deleted node
would get a degree increase of 1. Moreover, this node will now be a neighbor of the parent of the deleted node, at
level $0$. This is shown in figure \ref{sfig: t4lr1}.

\item \underline{Round 2; step\ref{algostep: L0prune}:}
 If no node got degree increase of 2 in the previous round, round 2 and 3 are initiated. In this round, the adversary
will delete all neighbors $v'$ of the root node (level $0$), where $\delta(v')$ = 0, if any. Now each neighbor of
$0$ has degree increase of 1.
 
\item \underline{Round 3; step \ref{algostep: delroot}:}
  After the previous round, the root node will have  3 neighbors, each with degree increase of 1. The adversary now
deletes the root node at level $0$. On Self-healing, one of the nodes will get a degree increase of 2.

\end{itemize}

\end{proof}
\begin{algorithm}[h!]
\caption{\textsc{DegreeUp(V,i):} Recursive Procedure to get a  node of degree increase $i+1$}
\label{algo: degreeup}
\begin{algorithmic}[1]
\STATE \emph{Init:} 
 Let $\delta(v)$ be the increase in degree experienced by node $v$. Let \emph{virgin} subtrees be subtrees none of whose
nodes have been involved in a deletion/self-healing process yet. 
\IF{$i$ = 0}
\STMT$V'$ = \emph{Strategy-1}($V$)
\ENDIF
\FOR{Each of the 3 virgin subtrees of $V$}
\STMT $V' \leftarrow$ root of virgin subtree.
\WHILE {$\delta(V') < i$}
\STMT $V' \leftarrow$ \textsc{DegreeUp} ($V', \delta(V')$)
\STMT \emph{Graft}($V,V'$)
\ENDWHILE
\ENDFOR
\STMT \emph{Prune} subtrees of $V$ not involved above.
\STMT delete $V$.
\STMT \emph{return} node (ex-neighbor of $V$) with highest degree increase.
\end{algorithmic}
\end{algorithm}

\begin{algorithm}[h!]
\caption{ Increase degree of a node by \emph{log log n}, for a $3.2^{a}$-level ternary tree, where $ a
%\caption{ \textsc{Main:} Increase degree of a node by \emph{log log n}, for a $3.2^{a}$-level ternary tree, where $ a
\ge 0$ }
\label{algo: main}
\begin{algorithmic}[1]
\STATE \emph{Init:} Let $G$ be a complete ternary tree of $L$ levels, where $L$ is $3.2^{a}$, where $ a \ge 0$, and $n$
is the total number of nodes . Let $\delta(v)$ be the increase in degree experienced by node $v$. Let \emph{virgin}
subtrees be subtrees none of whose nodes have been involved in a deletion/self-healing process yet. 
\STMT $i \leftarrow 0$
\STMT $V = $ \emph{Strategy-1 (Root of $G$)}
 \WHILE{$\delta(V) < log_{2}log_{3}n$ }
 \STMT $V \leftarrow$ \textsc{DegreeUp ($V, \delta(V)$)}
 \ENDWHILE
\end{algorithmic}
\end{algorithm}

%\begin{algorithm}[h!]
%\caption{\textsc{Strategy-2:} Degree increase for a $2^a, a > 1$-level complete ternary tree}
%\label{algo: strategy2}
%\begin{algorithmic}[1]
%\STATE \emph{Init:} Let $G$ be a complete ternary tree of L levels, where L is a power of 2 greater than 2. Let
%$\delta(v)$ be the increase in degree experienced by node $v$. Let \emph{virgin} subtrees be subtrees none of whose
%nodes have been involved in a deletion/self-healing process yet. 
%\STATE  Apply  \emph{Strategy \ref{algo: strategyup3}}  to the top-level 3-level subtree of $G$ to obtain a node $z$
%with $\delta(z) = 2$.
%\STATE If $z$ has virgin subtrees, then, apply \emph{Strategy \ref{algo: strategyup3}} to each of these subtrees. (By
%construction, $z$ has 3 virgin subtrees). Apply the operation \emph{Graft} so that $z$ is attached to one high degree
%increase (increase of 2) node  in each of these subtrees. 
%\STATE Apply the \emph{Prune} operation to the other subtrees of $z$. 
%\STATE delete $z$ to obtain a node $y$ with degree increase of 3.
%\REPEAT
%\STATE  for a node $y$ with $\delta(y) = i$, apply \textsc{Strategy2}, if possible, using the top $2^{i-1}$-level
%subtree
%of each of its children to obtain a new node $y$ with degree increase of $i+1$. 
%\UNTIL{$y$ does not have at least 3 \emph{virgin} subtrees}
%\end{algorithmic}
%\end{algorithm}
 
 \pagebreak
  
\begin{theorem}
\label{thm: loglognbound}
 There exists a graph $G$ such that for any locality aware algorithm on $G$ there exists an adversary strategy that
forces some node to increase it's degree by $\log \log n$, where $n$ is the number of nodes in $G$.
\end{theorem}

\begin{proof}
%It is sufficient to give example of a graph and an attack strategy such that any healing algorithm will have to
%increase a particular node's degree by $log_3n$. Such a graph is a complete 3-ary (ternary) tree. In a ternary tree,
%let
%$x$ be a node at level $y$ that is connected to three ternary subtrees at level $y+1$. In brief, The adversary strategy
%is to selectively delete nodes in the three subtrees such that the  three children of $x$ all have their degree
%increased. Now, on deletion of $x$, one of its children shall have its degree incremented by at least one more.
%Applying
%this strategy recursively, we can get a degree increase of 1 for each level of the tree, thus, getting a degree
%increase
%of $log_3 n$. The detailed proof follows.
 
 It is sufficient to give example of a graph and an attack strategy such that any healing algorithm will have to
increase a particular node's degree by $\log \log n$. Such a graph $G$ is complete ternary tree with $L$ levels where L
is $3.2^{a}$, where $ a \ge 0$ .

The adversary strategies are described in Algorithms \ref{algo: strategyup3} and \ref{algo: main}.
The intuition behind the adversary strategy is that the strategy has to force any locality-aware self-healing strategy
to have a degree increase. In particular, the adversary strategy wants to avoid the possibility of a node
\emph{surrogating} all the other neighbors of its deleted neighbor. Notice that if a node had four neighbors, three of
which had a degree increase of 2, and the fourth has no degree increase, this fourth neighbor could simply connect to
the three i.e.
\emph{surrogate} them
and incur a degree increase of only 2. Moreover, the resulting geometry makes it difficult to construct a strategy. 
The way around this in Algorithm \ref{algo: main} is that  for a node which is about to be deleted, have a parent with
a  degree increase higher or equal to that of it's three children.  This forces some neighbor to register the required
degree increase on self-healing. Algorithm \ref{algo: strategyup3} gives a method to get a degree increase of 2 for a
node in a 3-level ternary tree. Algorithm \ref{algo: main} uses this as a recursive subroutine and the idea of a
high-degree parent to obtain a certain node with degree increase of at least $O (\log \log n)$ . 

 Consider the following cases for $G$:

 \begin{enumerate}
\item \textbf{$L = 3$}\\
The adversary applies Algorithm \ref{algo: strategyup3}.
% By lemma \ref{lemma: degup}, one of the nodes increases
%its degree by 1 i.e. by $\log \log n$.
\item \textbf{$L > 3$}\\
 The adversary applies Algorithm \ref{algo: main}. To begin with, Algorithm \ref{algo: main} calls
Algorithm \ref{algo: strategyup3} to obtain a top-level node $x$ with degree increase of 2. To get a degree increase of
3, 
 Algorithm \ref{algo: main} calls the algorithm \textsc{DegreeUp} for each of its three children to get  a child of
degree increase 2 in each of these subtrees, using 3 more levels. Using the \emph{graft} operation, these nodes are
attached to $x$. The
\emph{prune} operation removes any other subtrees of  $x$. Now $x$ has exactly three neighbors of degree increase 2
each, and deletion of $x$ leads to a node of degree increase 3. To get a degree increase of 4, the strategy uses the
same strategy described above recursively for three \emph{virgin} children of this node, using 6 more levels. This
will give it three children with degree increase 3 and now we can obtain a node with degree increase 4. Notice, each 
subsequent degree increase involves exponentially larger number of levels in $G$.\\
 Thus,
\begin{eqnarray*}
\textrm{degree increase} & =  & O(\log_{2} (\textrm{number of levels in}) G)\\
 & = & O (\log_{2} \log_{3} n)
\end{eqnarray*}

 \end{enumerate}

\end{proof}

\section{Experiments}

%\begin{epigraphs}
%\qitem {\epitext{No amount of experimentation can ever prove me right; a single experiment can prove me wrong.}}%
%{\epiauthor{Albert Einstein}}
%%\qitem{\epitext{It doesn't matter how beautiful your theory is, it doesn't matter how smart you are. If it doesn't agree with experiment, it's wrong.}}%
%%{ \epiauthor{Richard Feynman} }
%\end{epigraphs}

\label{sec: empirical}
 We carried out a number of experiments to ascertain the performance of various healing algorithms. We used a number of
attack strategies to measure how different healing strategies performed with regard to degree increase and stretch,
where stretch is the maximum ratio of distance increase in the healed network compared to the original network, over all
pairs of nodes. Our empirical results on stretch and a heuristic for maintaining low stretch
are described in Section \ref{app: stretch}. 
  
\subsection{Methodology}
\label{subsec: methodology}
 Most of our experiments were
conducted  on random graphs. These graphs were generated by the \emph{Preferential Attachment} model proposed by Barabasi \cite{barabasi-1999-286, barabasi-2003-50}.  The experimental approach was the following:
\begin{itemize}
\item For each graph size, for a particular deletion and healing strategy, repeat for 30 random instances of the graph:
\begin{itemize}
 \item  Repeat while there are nodes in the graph:
  \begin{itemize}
   \item delete a single node according to the deletion strategy.
   \item repair according to the self-healing strategy.
   \item measure the statistics (e.g. maximum change of degree for any node) for the graph. 
  \end{itemize}
  \end{itemize}
  \item average the statistics for each graph size.
 \end{itemize}
 
\subsection{Attack Strategies}
\label{subsec: attack}

The aim of the adversary is to collapse the network by trying to overload a node beyond it's maximum capacity. There are
many possible attack strategies. One strategy is to delete the node with the maximum degree. We call this the $Max Node strategy$.
 It would seem that a strategy that leads to additional burden on an already high burden node would be a good strategy.
For the adversary,  one good adversarial strategy is to continuously attack/delete a randomly chosen neighbor of the
highest degree node in the network. We call this the $Neighbor of Max Strategy (NMS)$.
 This would also seem plausible as in a real network or the kind of networks we are looking at, it would be reasonable
that the hubs or the high degree nodes would be more well protected and resilient to attack while their less significant
neighbors should be easy to take down.

 \subsection{Healing strategies}
\label{subsec: healing}
 We attempted various locality-aware healing strategies, some of which are the following:
 \begin{itemize}
 \item \emph{Graph heal}: On each deletion, we reconnect the neighbors of the deleted node in a binary tree regardless
of whether we introduced any cycles in the graph formed by the new edges introduced for healing. This seems to be  a
naive algorithm since the nodes use more edges than what are required for maintaining connectivity. 
 \item \emph{Binary tree heal}: On each deletion, we reconnect the neighbors of the deleted node in a binary tree being
careful not to introduce any cycles in the graph formed by the new edges introduced for healing. This is done using
random IDs which can then be used to identify which tree a particular node belongs to. This is an improvement on the
previous algorithm but still naive since it does not take into consideration the previous degree increase suffered by
nodes during healing.
 \item \emph{$\DASH$ (Degree Assisted binary tree heal)}:
 DASH is  smarter than the previous algorithms as borne out by the results of the experiments.  The DASH algorithm has been earlier described in Section \ref{subsec: dash} and stated as Algorithm \ref{algo: dash}.
\item \emph{$\SDASH$ (Surrogate Degree Assisted binary tree heal)}: (described in Section \ref{subsec: SDaSH})
 A heuristic based on $\DASH$ that tries to both keep node degrees and path lengths small.
 \end{itemize}
 
\subsection{Connectivity}
\label{subsec: connectivity}

Figure~\ref{fig: DASHhealtimeline} shows a series of snapshots from a simulation of  $\DASH$ showing that the network stays connected, and no individual node seems to be getting a large number of healing edges during healing.

 \begin{figure}[h!]
\centering
\subfigure[single deletion]{
\includegraphics[scale=0.21]{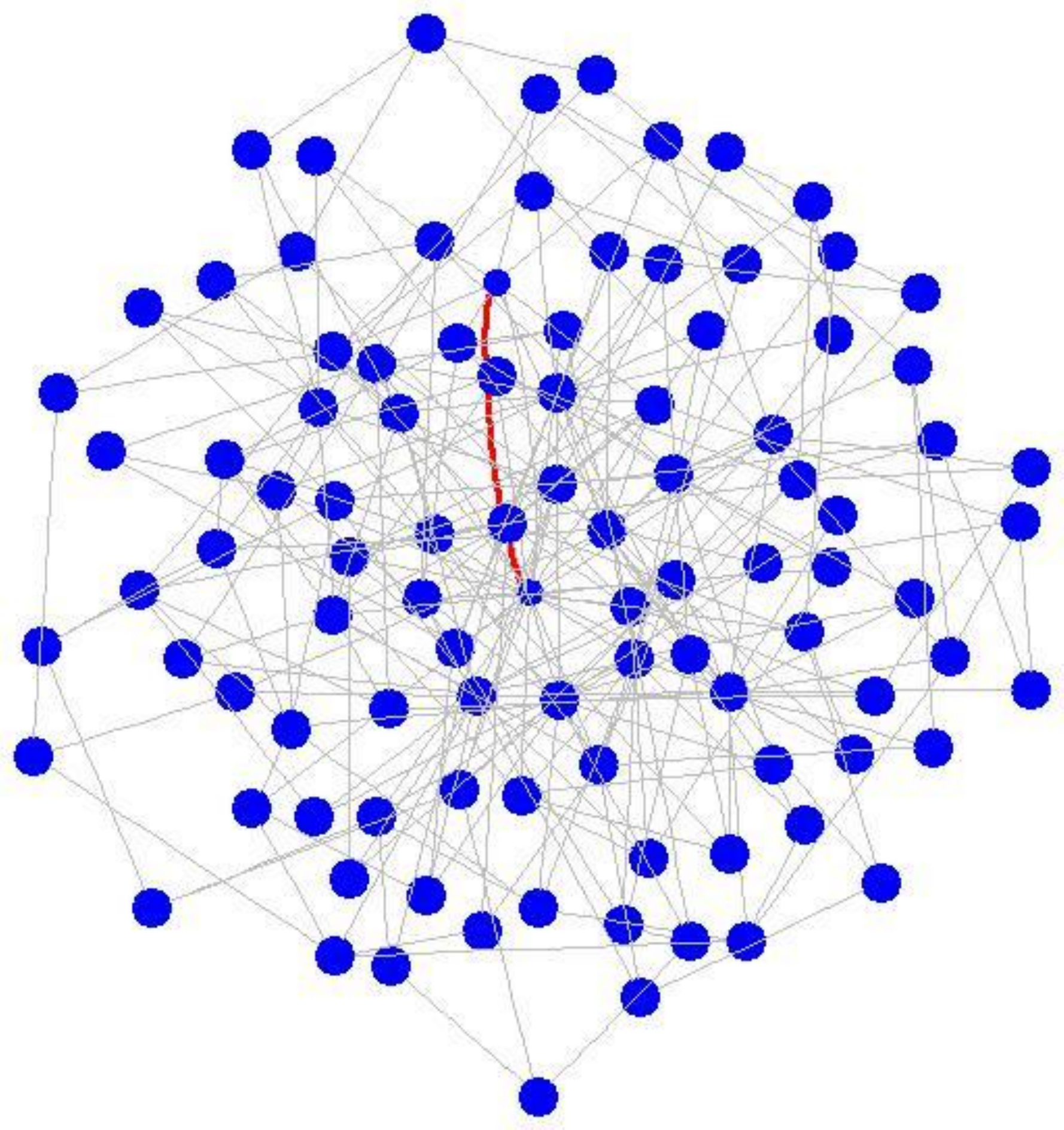}
 } 
\subfigure[10 deletions]{
\includegraphics[scale=0.22]{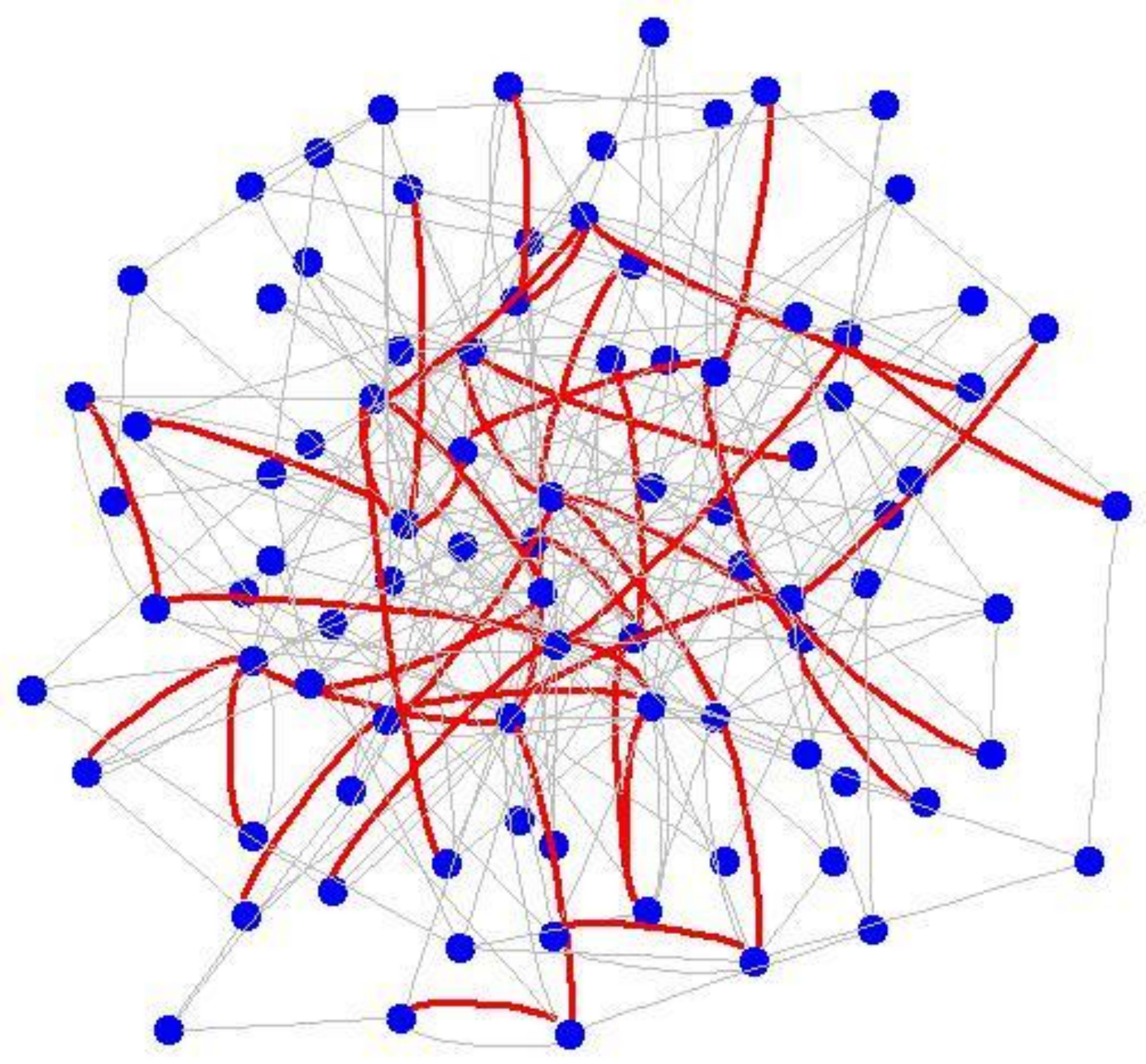}
%\label{fig:subfig2}
 }
\subfigure[30 deletions]{
\includegraphics[scale=0.21]{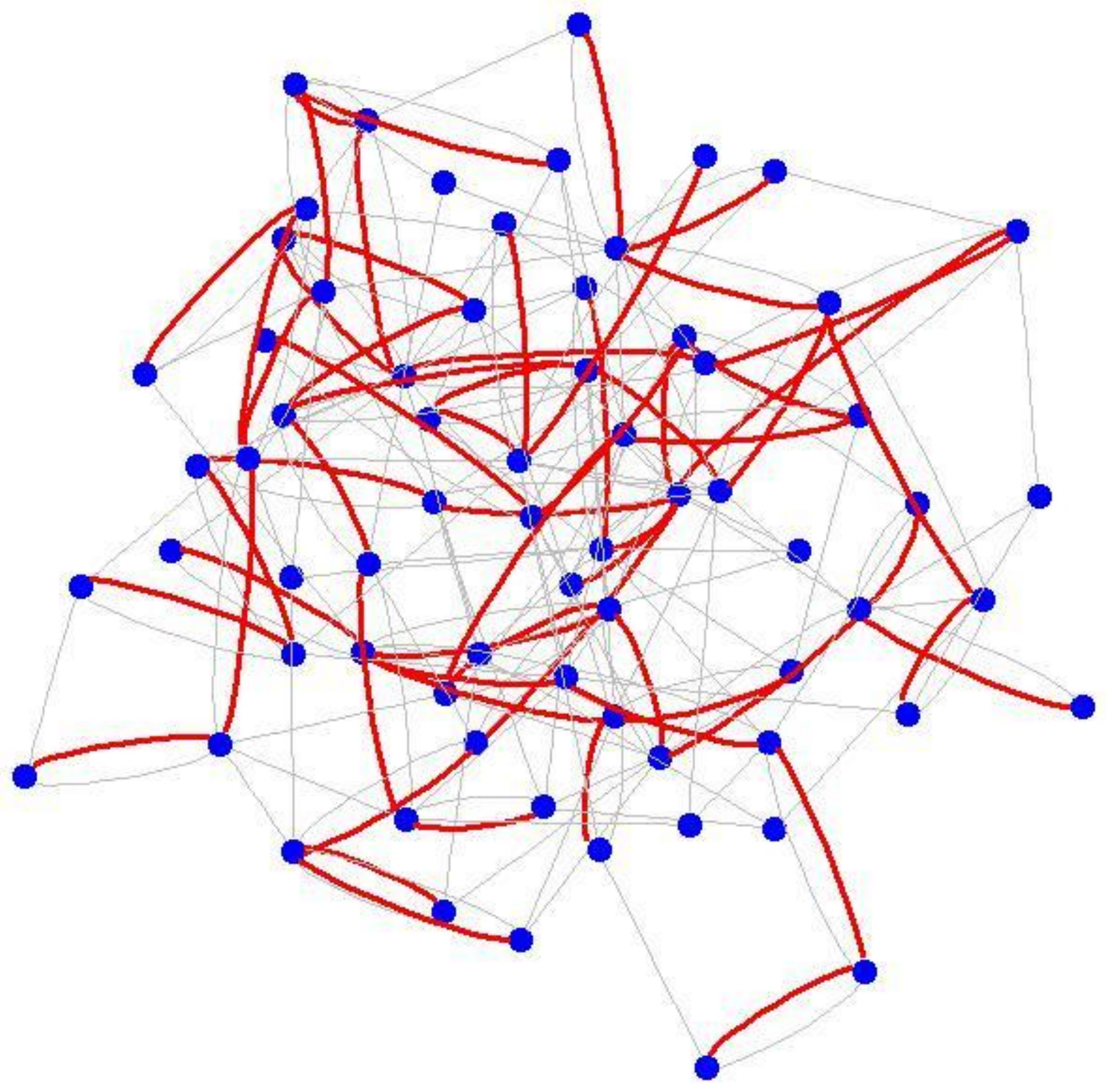}
%\label{fig:subfig3} 
}\\
\subfigure[40 deletions]{
\includegraphics[scale=0.18]{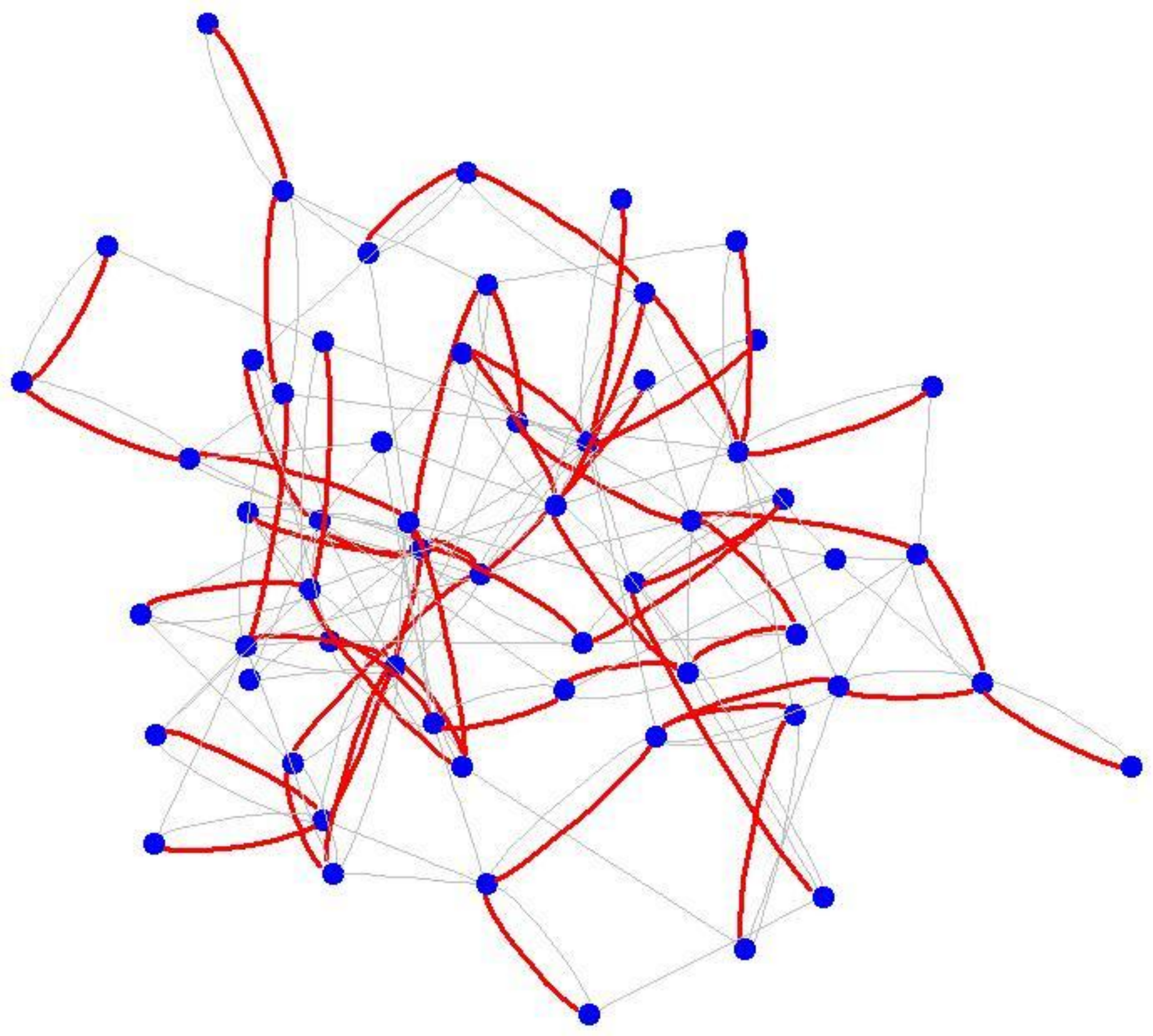}
 } 
\subfigure[50 deletions]{
\includegraphics[scale=0.17]{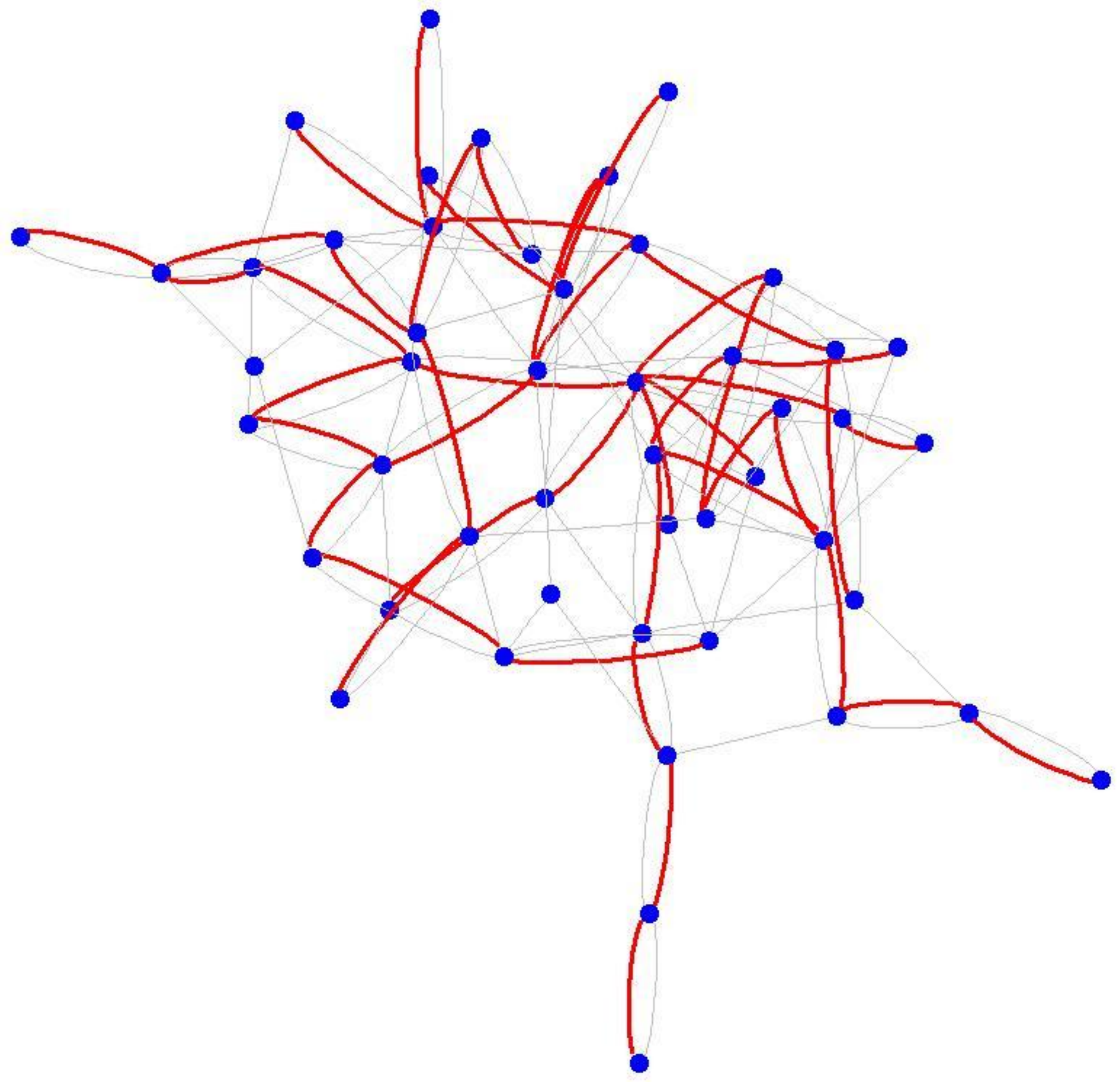}
 } 
 \subfigure[60 deletions]{
\includegraphics[scale=0.18]{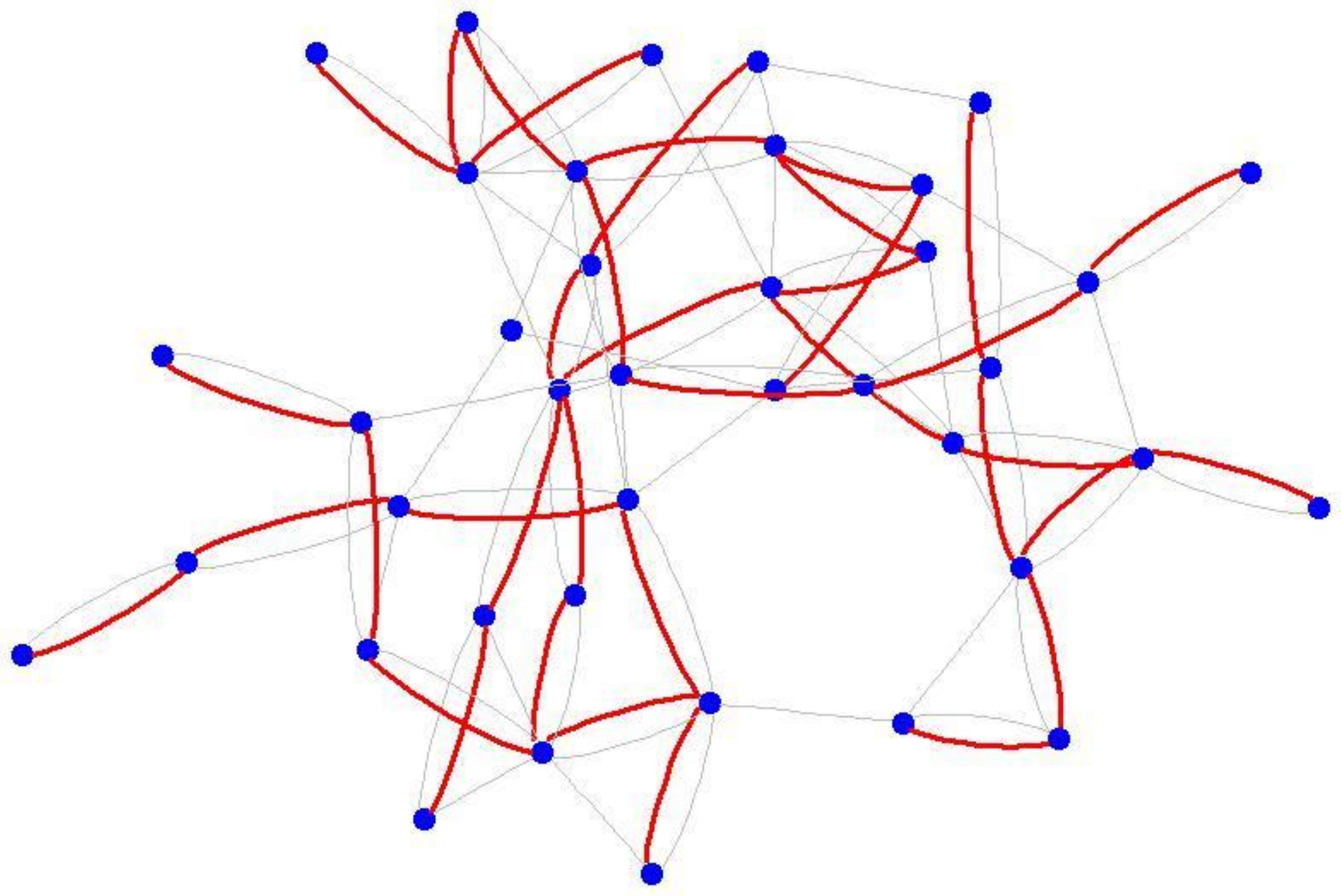}
 } \\
\subfigure[70 deletions]{
\includegraphics[scale=0.2]{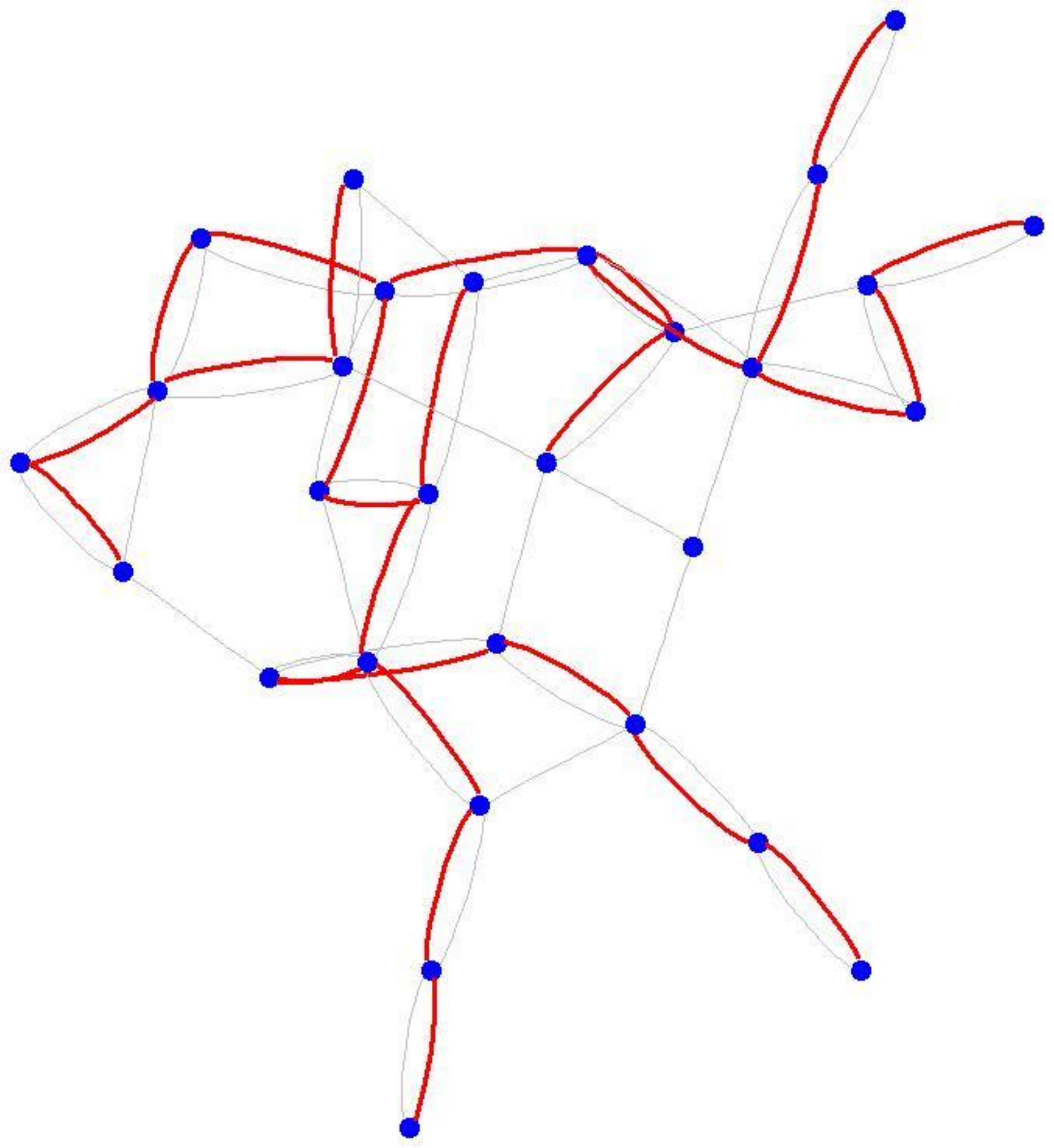}
%\label{fig:subfig2}
 }
 \subfigure[80 deletions]{
\includegraphics[scale=0.19]{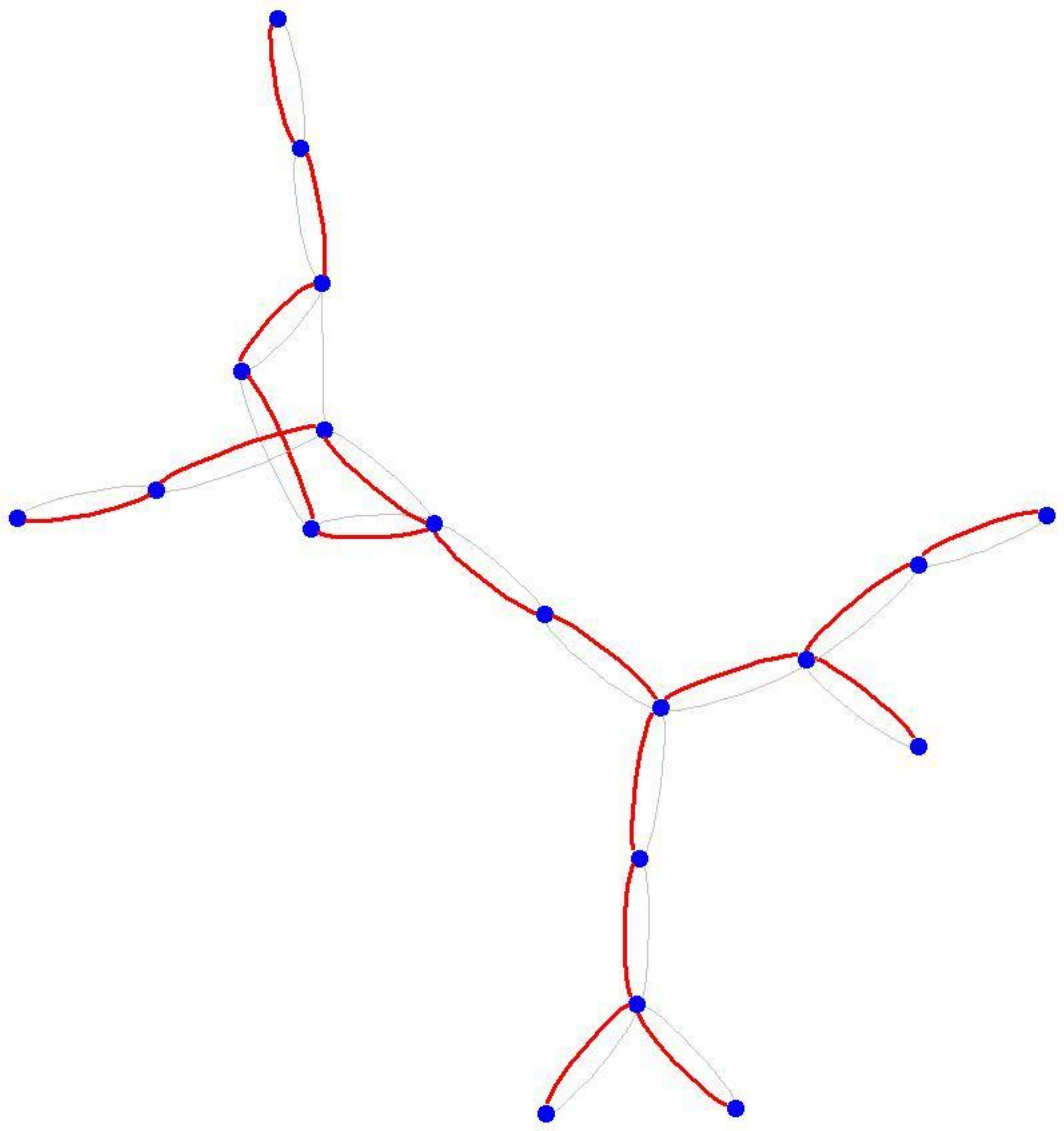}
%\label{fig:subfig2}
 }
\subfigure[90 deletions]{
\includegraphics[scale=0.2]{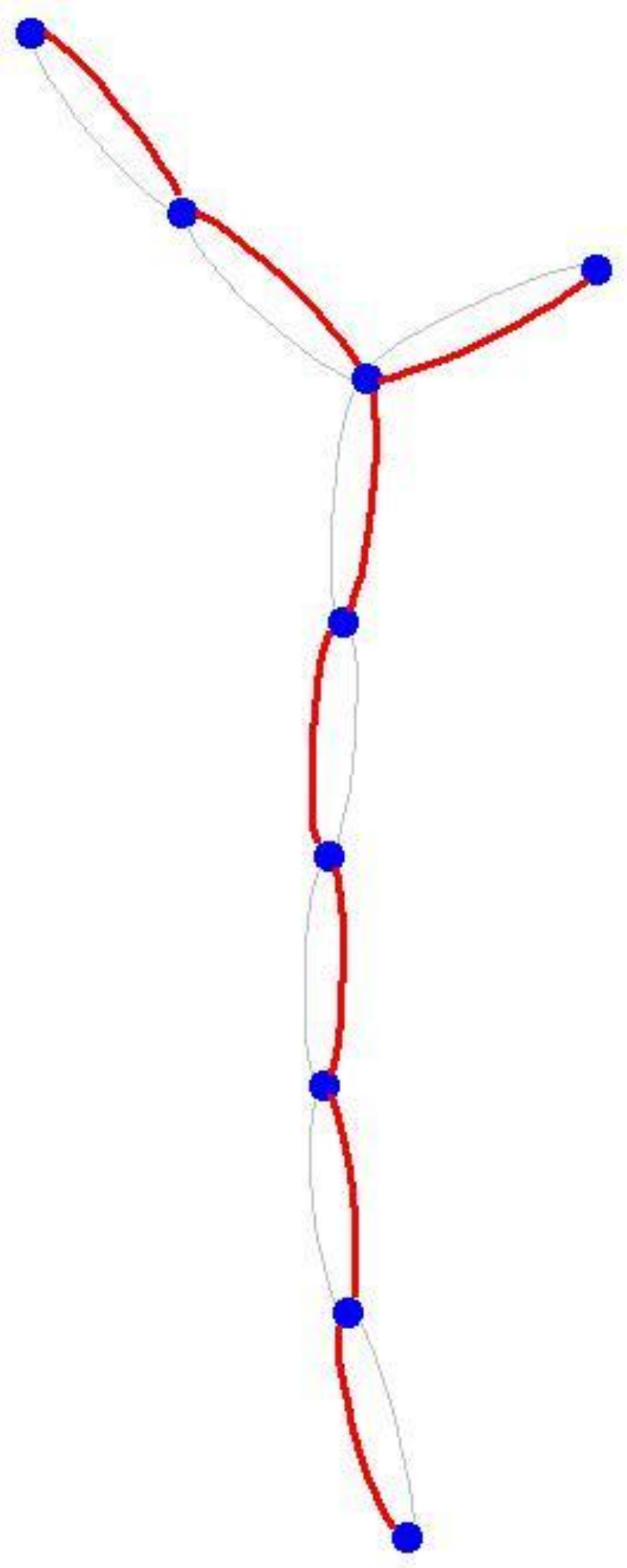}
%\label{fig:subfig3} 
}
\caption{A timeline of deletions and self healing in a network with 100 nodes. The gray edges are the original edges and
the red edges are the new edges added by our self-healing algorithm.}
\label{fig: DASHhealtimeline}
\end{figure}

\subsection{Degree increase}
\label{subsec: empirical-degree}

 The  $Neighbor of Max Strategy$ consistently resulted in higher degree increase, hence, we report results for only this attack strategy. Our experimental results clearly show that $\DASH$ and $\SDASH$  are good healing strategies. It performed well against both adversary strategies. Figure \ref{fig: Degplot} shows that $\DASH$ and $\SDASH$ have much lower degree increase than the other more naive strategies. Also, this degree increase was less than $\log n$, which is consistent with our theoretical results. $\SDASH$ has the additional nice property that it keeps path lengths small over multiple adversarial deletions.

\begin{figure}[h!]
\centering
\includegraphics{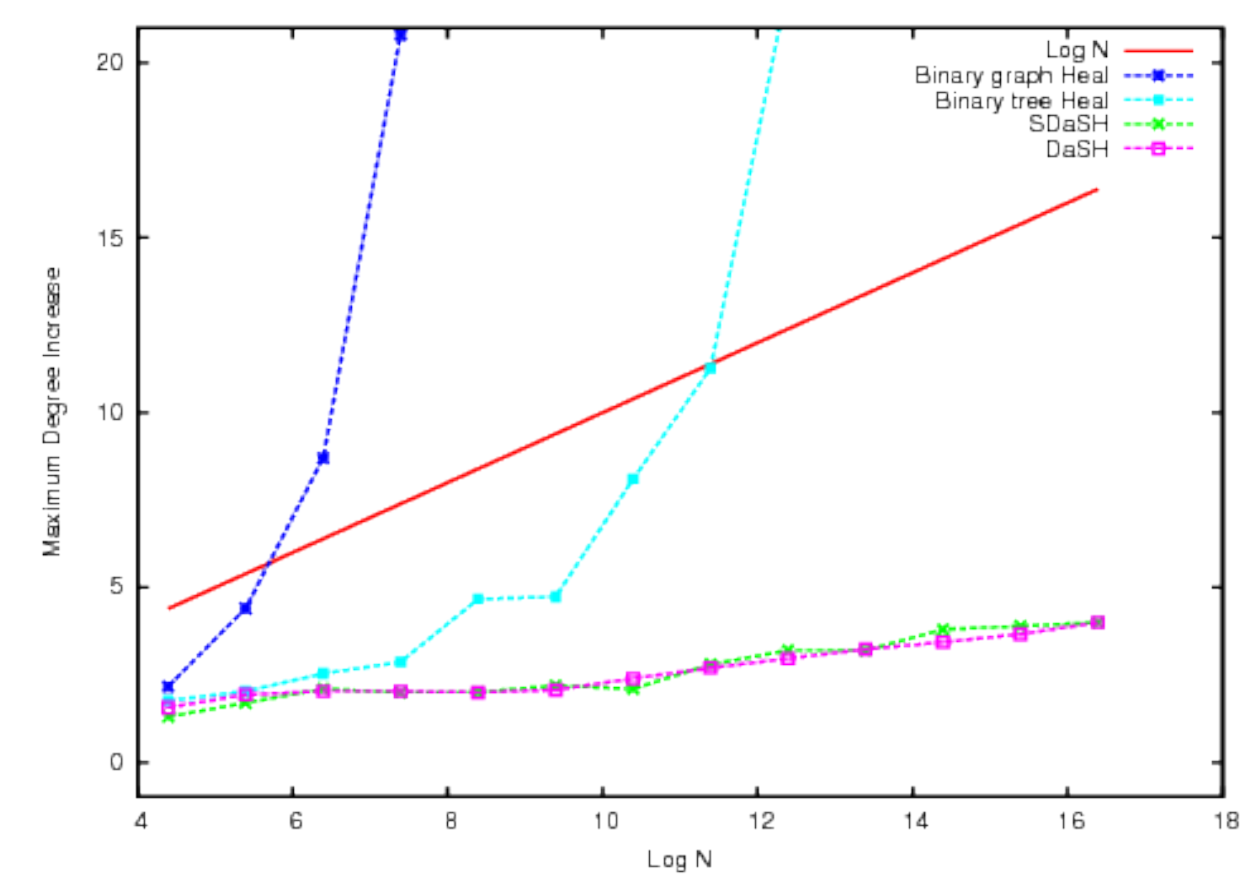}
\caption{Maximum Degree increase: DASH vs other algorithms}
\label{fig: Degplot}
\end{figure}

\subsection{Messages}
\label{subsec: empirical-messages}

Figure \ref{fig: nbrID} shows that the number of time a nodes $ID$ changes is less than $\log n$, as expected, for all healing strategies. Figure \ref{fig: msgs} shows the maximum number of messages a node sent out for the different strategies. Note that the number of messages a node sends out has to be less than or equal to the number of times a node changes ID times the degree of the node. Thus, algorithms with higher degree increase perform poorly.

\begin{figure}[h!]
   \centering
\includegraphics{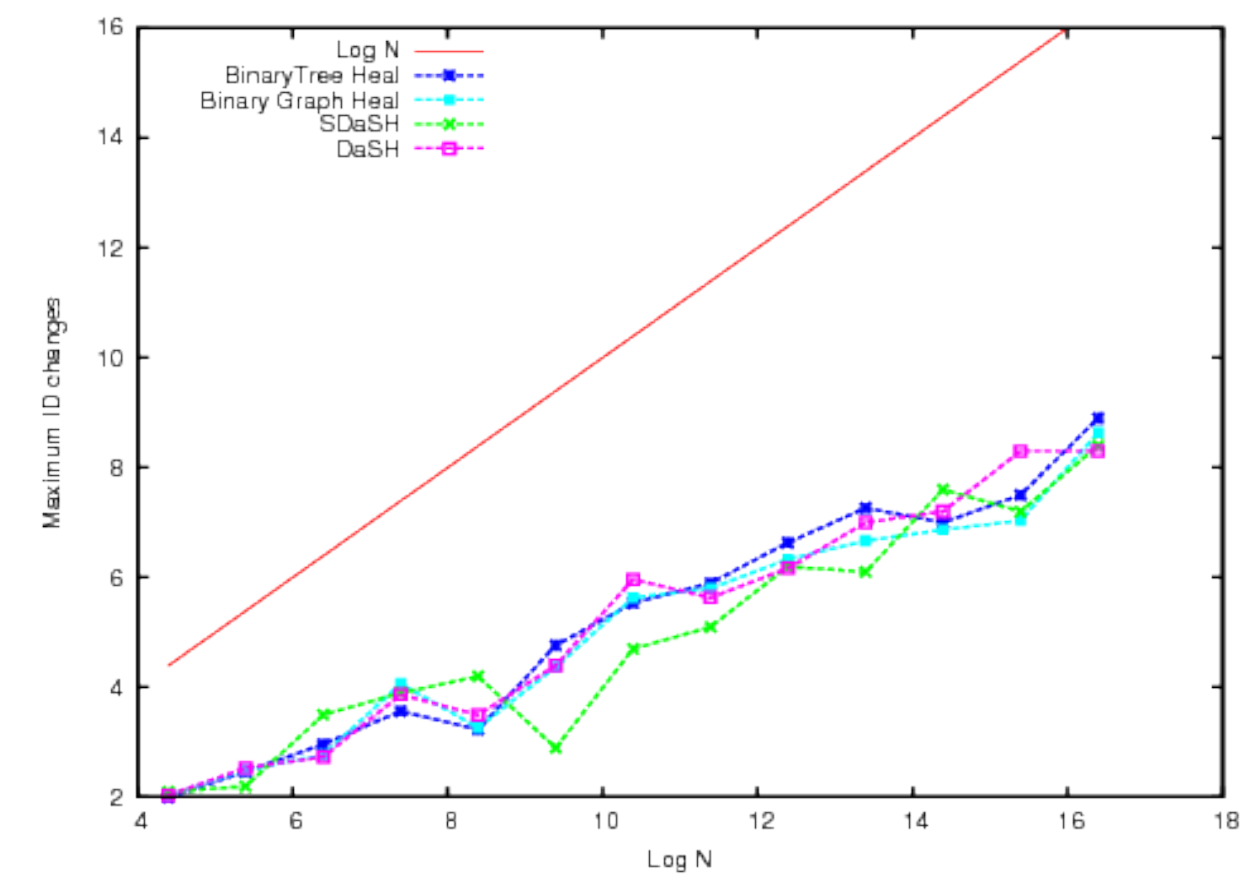}
     \caption{ID changes for nodes}
       \label{fig: nbrID}           %% label
  \end{figure}

\begin{figure}[h!]
\includegraphics{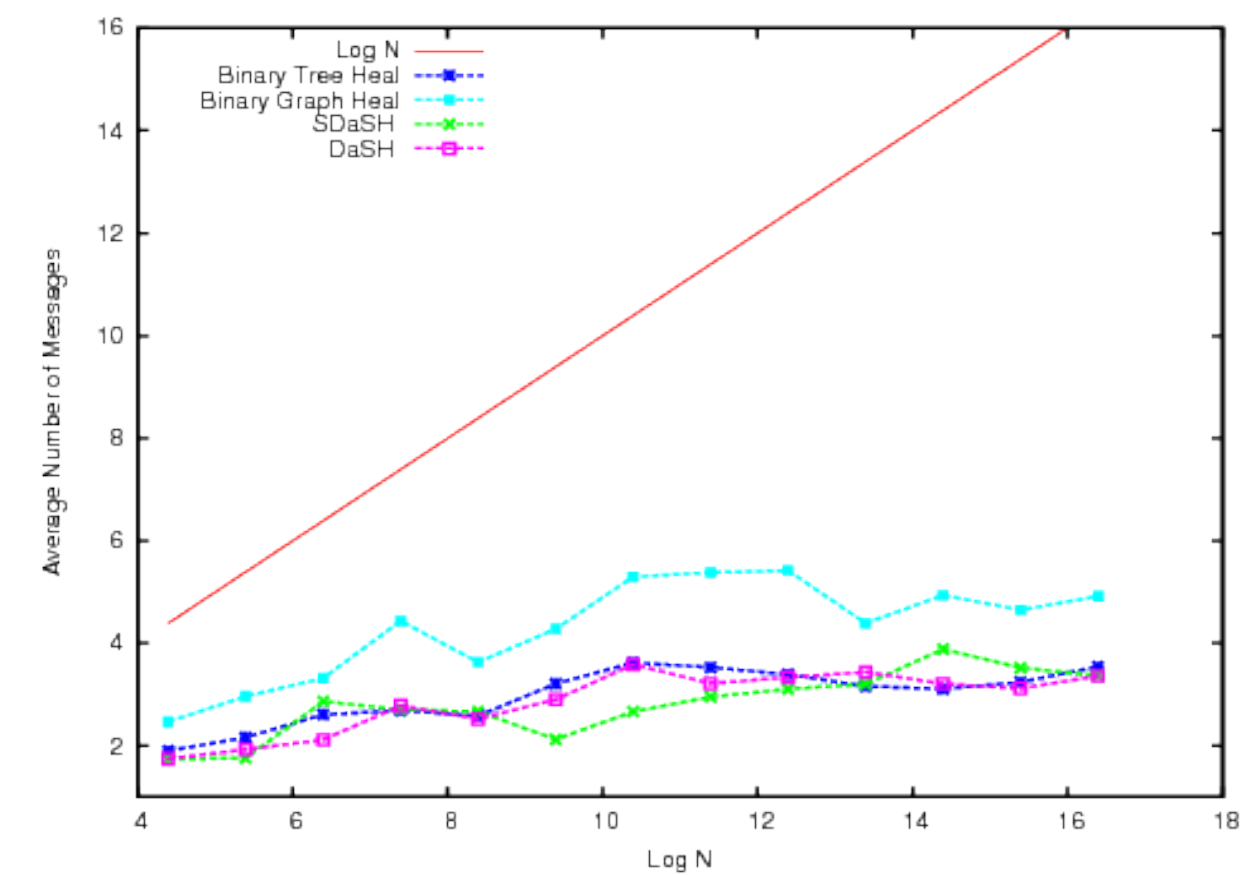}
    \caption{Number of messages exchanged for Component(ID) information  maintenance}
        \label{fig: msgs}           %% label for second figure
% \label{fig: IDmsgs}                  %% label for entire figure
% \caption{ID changes and Messages exchanged per node}
\end{figure}

%\begin{figure}[h!]
%   \centering
%   \subfigure[ID changes for nodes]{
%        \label{fig: nbrID}           %% label for first subfigure
%\includegraphics[width=2.5in]{DASH/plots/AvgMaxRbreakBinDashPlot.pdf}}
%        %\includegraphics[width=1.0in]{graphic.eps}}
%   \hspace{0.1in}
%   \subfigure[Number of messages exchanged for Component(ID) information maintenance]{
%        \label{fig: msgs}           %% label for second subfigure
%\includegraphics[width= 2.5in]{DASH/plots/AvgTotMsgAllPlot.pdf}}
%%\includegraphics[width= 2.2in]{plots/AvgMaxMsgnbrnodeDashPlot.jpg}}
%       % \includegraphics[width=1.5in]{graphic.eps}}
% \caption{ID changes and Messages exchanged per node}
%   \label{fig: IDmsgs}                  %% label for entire figure
%\end{figure}

\subsection{Heuristics and experiments involving Stretch}
\label{app: stretch}

%\subsubsection{Stretch}
Stretch is  an important property we would also like our self-healing algorithms to minimize. The stretch for
any two nodes is the ratio between their distance in the new healed network and their distance in the original network. Stretch for the network is the maximum stretch over all pairs of nodes. Stretch is also closely related to the diameter of the network. In some sense, maintaining low degree increase and low stretch are contradictory aims since a high-degree node will lead to shorter paths and possibly lower stretch in the network.

\subsubsection{$\SDASH$: a strategy with good empirical results}
\label{subsec: SDaSH}
 $\SDASH$ is an algorithm we have devised which empirically has both low degree increase and low stretch. During
self-healing, we say a node \emph{surrogates} if it replaces its deleted neighbor in the network. i.e. it takes all the
connections of the deleted neighbor to itself. Surrogation never increases stretch since the paths never increase in
length. In certain situations, it turns out that surrogation can be done without degree increase. In such situations,
$\SDASH$ does surrogation else it simply applies $\DASH$. $\SDASH$ is described in Algorithm \ref{algo: sdash}.

\begin{algorithm}[h!]
\caption{\textbf{SDASH:} Surrogate Degree-Based Self-Healing}
\label{algo: sdash}
\begin{algorithmic}[1]
\STATE \emph{Init:} for given network $G(V,E)$, Initialize each vertex with a random number $ID$ between [0,1] selected
uniformly at random. 
\WHILE {true}
\STATE \emph{If a vertex $v$ is deleted, do}
\STATE Let $m \in UN(v,G) \cup N(v,G_h)$ be the node with Maximum degree increase ($\delta$) of all nodes in $UN(v,G) \cup N(v,G_h)$.   
\IF {$ w \in UN(v,G) \cup N(v,G_h)$ and $\delta(w) + | UN(v,G) \cup N(v,G_h) | -1 \le \delta(m)$ }
\STATE connect all nodes in $ UN(v,G) \cup N(v,G_h)$ to $w$. 
\ELSE 
\STATE Nodes in $UN(v,G) \cup N(v,G_h)$ are reconnected into a \emph{complete binary tree}. To connect the tree, go left
to right, top down, mapping nodes to the \emph{complete binary tree} in increasing order of $\delta$ value.
%that is they reconnect into a binary tree such that the Breadth First traversal gives a list sorted in ascending order
%on the Degree of the nodes. The following procedure 
\ENDIF
\STATE Let $MINID$ be the minimum $ID$ of any node in $UN(v,G) \cup N(v,G_h)$.
 Propagate $MINID$ to all the nodes in the tree of $UN(v,G) \cup N(v,G_h)$ in $G_h$. All these nodes now set their $ID$ to
$MINID$.
%\ENDFOR
\ENDWHILE
\end{algorithmic}
\end{algorithm}

 As can be seen in the figures that follow, $\SDASH$ seems to allow a degree increase  up to $O(\log n)$ and
stretch up to $O(\log n)$. We are working on proving theoretical properties of this algorithm.

\subsubsection{Stretch: empirical results}
\label{subsec: empirical-stretch}

 Figure \ref{fig: stretchplot} shows the performance of some of our algorithms for stretch.  We determined that the $Max Node strategy$ is most effective for the adversary when trying to maximize stretch and so our results in Figure \ref{fig: stretchplot} are against that adversarial strategy.   The more naive degree-control healing strategies do a good job of minimizing stretch.  However, it is important to keep in mind that these more naive algorithms increase the node degrees to a point where they are unlikely to be useful for many applications.  In contrast, our experiments show that $\SDASH$  does a good job of minimizing both stretch and degree increase.
  
\begin{figure}[h!]
\centering
\includegraphics{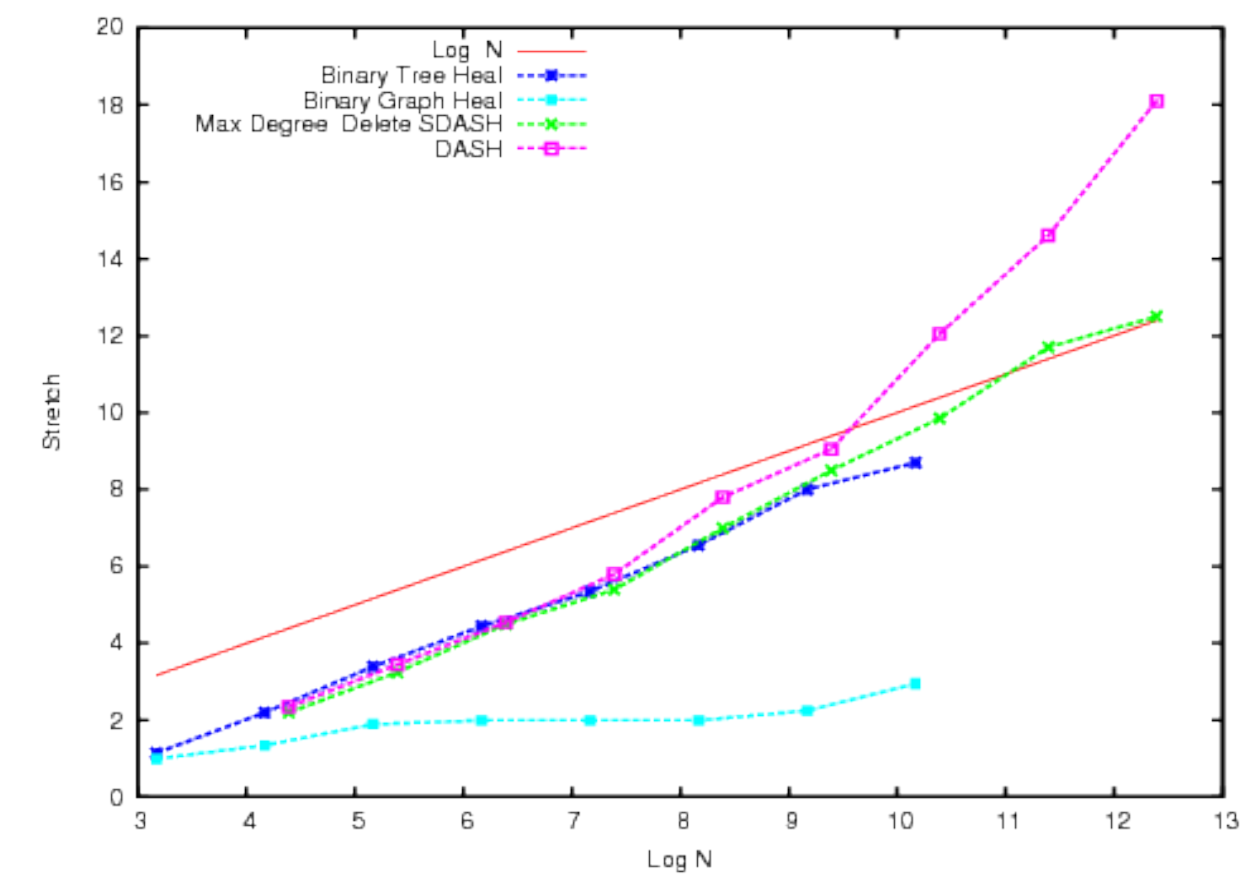}
\caption{Stretch for various algorithms}
\label{fig: stretchplot}
\end{figure}

\section{Conclusions and future work}
\label{sec: conclusions}

In this chapter, we have studied the problem of self-healing in networks that are reconfigurable in the sense that new edges can be added to the network.  We have described $\DASH$, a simple, efficient and localized algorithm for self-healing, that provably maintains network connectivity, even while increasing the degree of any node by no more than $O(\log n)$.  We have shown that $\DASH$ is asymptotically optimal in terms of minimizing the degree increase of any node.  Further, we have presented empirical results on power-law networks showing that $\DASH$ significantly outperforms the naive algorithms for this problem.

Several interesting problems remain open including the following:  Can we not only maintain connectivity, but also provably ensure that lengths of shortest paths in the graph do not increase by too much?  Can we remove the need for propagating IDs in order to maintain connected component information, or is such information strictly necessary to keep the degree increase small?  Can we use the self-healing idea to protect invariants for combinatorial objects besides graphs?  For example, can we provide algorithms to rewire a circuit so that it maintains essential functionality even when multiple gates fail?

%\pagebreak

%\bibliography{selfhealltx8} 
%\bibliographystyle{plain}
%\end{document}
%%%%%%%%%%%%%%%%%%%%%%%%%%%%%%%%%%%%%%%%%%%%%%%%%%%%%%%%
%% Chapter 3: Forgiving Tree
\chapter{Forgiving Tree}
\label{chapter: FT}

\begin{epigraphs}
\qitem{\epitext{My roots are strong\\ 
My branches free,
But only because I'm a forgiving tree.}%
 }%
{%
\episource{The Forgiving Tree}\\ %
\epiauthor{Cheryl Merriweather}%
}

\end{epigraphs}

In this chapter, we present the algorithm $\FTree$ which first appeared in \emph{Principles of Distributed Computing 2008}~\cite{HayesPODC08}. We consider the problem of self-healing in peer-to-peer networks that
are under repeated attack by an omniscient adversary.  We assume that the following process continues for up to $n$ rounds where $n$ is the total number of nodes initially in the network: the adversary deletes an arbitrary node from the network, then the network responds by quickly adding a small number of new edges.

We present a distributed data structure that ensures two key
properties.  First, the diameter of the network is never more than
$O(\log \Delta)$ times its original diameter, where $\Delta$ is the
maximum degree of the network initially.  We note that for many
peer-to-peer systems, $\Delta$ is polylogarithmic, so the diameter
increase would be a $O(\log \log n)$ multiplicative factor.  Second,
the degree of any node never increases by more than $3$ over its
original degree.  Our data structure is fully distributed, has $O(1)$
latency per round and requires each node to send and receive $O(1)$
messages per round.  The data structure requires an initial setup
phase that has latency equal to the diameter of the original network,
and requires, with high probability, each node $v$ to send $O(\log n)$
messages along every edge incident to $v$.  Our approach is orthogonal
and complementary to traditional topology-based approaches to
defending against attack.

\section{Introduction}
\label{sec: FT-Intro}

 In Chapter~\ref{chapter: Intro}, we have made a case highlighting the need of using \emph{responsive} approaches for maintaining robustness and self-healing in networks.
 
 % For completeness, we repeat the motivation in the next paragraphs here. 
 
%Many modern networks are \emph{reconfigurable}, in the sense that the topology of the network can be changed by the nodes in the network.  For example, peer-to-peer, wireless and mobile networks are reconfigurable.  More generally, many social networks, such as a company's organizational chart; infrastructure networks, such as an airline's transportation network; and biological networks, such as the human brain, are also reconfigurable.  Unfortunately, our mathematical and algorithmic tools have not developed to the point that we are able to fully understand and exploit the flexibility of reconfigurable networks.  For example, on August 15, 2007 the Skype network crashed for about $48$ hours, disrupting service to approximately $200$ million users due to what the company described as failures in their ``self-healing mechanisms''~\cite{garvey, fisher,malik, moore, ray, stone}. We believe that this outage is indicative of a much broader problem.

In this chapter, we focus on a new, \emph{responsive} approach for
maintaining robust reconfigurable networks.  Our approach is responsive in the sense
that it responds to an attack (or component failure) by changing the
topology of the network.  Our approach works irrespective of the
initial state of the network, and is thus orthogonal and complementary
to traditional non-responsive techniques.  There are many desirable
invariants to maintain in the face of an attack.  Here we focus only on
 the simplest and most fundamental invariants: ensuring the diameter of the network 
and the degrees of all nodes do not increase by much.

\medskip
\noindent {\bf Our Model:} We now describe our model of attack and
network response.  We assume that the network is initially a connected
graph over $n$ nodes.  An adversary repeatedly attacks the
network.  This adversary knows the network topology and our
algorithms, and it has the ability to delete arbitrary nodes from the
network.  However, we assume the adversary is constrained in that in
any time step it can only delete a single node from the network.  We
further assume that after the adversary deletes some node $x$ from the
network, that the neighbors of $x$ become aware of this deletion and
that the network has a small amount of time to react by adding and deleting some
edges.  This adversarial model captures what can happen when a
worm or software error propagates through the population of nodes.
Such an attack may occur too quickly for human intervention or for the
network to recover via new nodes joining.  Instead the nodes that
remain in the network must somehow reconnect to ensure that the
network remains functional.

We assume that the edges that are added can be added anywhere in the
network.  We assume that there is very limited time to react to
deletion of $x$ before the adversary deletes another node.  Thus, the
algorithm for deciding which edges to add between the neighbors of $x$
must be fast. The detailed model used in $\FTree$ and its relation to the general model we described in  Section~\ref{sec: Intro-self-healingModel} is given in Section~\ref{sec: FT-Model}. 

\medskip
\noindent {\bf Our Results:}  A naive approach to this problem is
simply to 'surrogate' one neighbor of the deleted node to take on the
role of the deleted node, reconnecting the other neighbors to this
surrogate. However, an intelligent adversary can always cause this approach to 
increase the degree of some node by $\theta(n)$. On the other hand, we may try to keep the degree
increase low by connecting neighbors of the deleted node as a straight
line, or by connecting the neighbors of the deleted node in a binary tree.  However, for both of these techniques the diameter can increase by $\theta(n)$ over multiple deletions by an intelligent adversary~\cite{BomanSAS06, SaiaTrehanIPDPS08}.

In this chapter, we describe a new, light-weight distributed data structure that ensures that: 1) the diameter of the network never increases by more than $\log \Delta$ times its original diameter, where $\Delta$ is the
maximum degree of a node in the original network; and 2) the degree of
any node never increases by more than $3$ over over its original
degree.  Our algorithm is fully distributed, has $O(1)$ latency per
round and requires each node to send and receive $O(1)$ messages per
round.   The formal statement and proof of these results is in Section~\ref{subsec: upperbounds}.  Moreover, we show (in Section~\ref{subsec:
lowerbounds}) that
in a sense our algorithm is asymptotically optimal, since any algorithm that increases node degrees by no more than a
constant must, in some cases, cause the diameter of a graph to increase by a $\log \Delta$ factor.

The algorithm requires a one-time setup phase  to do the
following two tasks.  First, we must find a breadth first spanning
tree of the original network rooted at an arbitrary node.  In the synchnronous communication model, this can be
done with latency equal to the diameter of the original network, and,
with high probability, each node $v$ sending $O(\log n)$ messages
along every edge incident to $v$, as in the algorithm due to
Cohen~\cite{Cohen}. The second task required is to set up a simple
data structure for each node that we refer to as a will.  This will,
which we will describe in detail in the Section~\ref{sec:algorithm}, gives
instructions for each node $v$ on how the children of $v$ should
reestablish connectivity if $v$ is deleted.  Creating the will
requires $O(1)$ messages to be sent along the parent and children
edges of the global breadth-first search tree created in the first
task.
 
% Our main results can be summarized in the following two theorems.

\medskip
\noindent {\bf Related Work:}

In this chapter, we build on earlier work in~\cite{BomanSAS06, SaiaTrehanIPDPS08}.

There have been numerous papers on dealing with adversarial atttacks in networks.  Kuhn et al \cite{Kuhn2006Blueprint, Kuhn2005Self-Repairing} describe efficient algorithms that provably ensure that node degree and network diameter stay small even in the case where an adversary can either add or delete up to a fixed number of nodes in any time step. 
They describe algorithms for the hypercube~\cite{Kuhn2005Self-Repairing} and pancake topology~\cite{Kuhn2006Blueprint} and suggest how their approach can apply to any recursively defined peer-to-peer topology.  In contrast, our algorithm does not handle adversarial insertions, but it is immediately applicable to any arbitrary reconfigurable network, even those that are not recursively defined.   

%\pagebreak

\section{Delete and Repair Model}
\label{sec: FT-Model}

We now describe the details of our delete and repair model.  Let $G = G_0$ be an arbitrary graph on $n$ nodes,
which represent processors in a distributed network.  One by one, the Adversary deletes nodes until none are left.  After each deletion, the Player gets to add some new edges to the graph, as well as deleting old ones.   The Player's goal is to maintain connectivity in the network, keeping the diameter of the graph small.  At the same time, the Player wants to minimize the resources spent on this task,
in the form of extra edges added to the graph, and also in 
terms of the number of connections maintained by each node
at any one time (the degree increase).  We seek an algorithm which gives performance guarantees 
under these metrics for each of the $n!$ possible deletion orders.

Unfortunately, the above model still does not capture the
behaviour we want, since it allows for a centralized Player
who ignores the structure of the original graph, and simply
installs and maintains a complete binary tree, using a leaf 
node to substitute for each deleted node. 
%\tom{Probably want either more or less detail here. Or really, this is probably not what we want to say.}
%(Later on, we will view these nodes as representing
%agents or processors in a distributed system, and will
%consider computation and communication costs.)

To avoid this sort of solution, we require a distributed
algorithm which can be run by a processor at each node.
Initially, each processor only knows its neighbors in $G_0$,
and is unaware of the structure of the rest of the $G_0$.
After each deletion (forming $H_t$), only the neighbors of the deleted
vertex are informed that the deletion has occurred.
After this, processors are allowed to communicate by
sending a limited number of messages to their direct 
neighbors.  We assume that these messages are always
sent and received successfully.  The processors may 
also request new edges be added to the graph to form $G_t$.
The only synchronicity assumption we make is that the
next vertex is not deleted until the end of this round
of computation and communication has concluded.
To make this assumption more reasonable, the per-node
communication should be $O(\log n)$ bits, and should 
moreover be parallelizable so that the entire protocol
can be completed in $O(1)$ time if we assume synchronous
communication.

We also allow a certain amount of pre-processing to be
done before the first deletion occurs.  This may, for
instance, be used by the processors to gather some
topological information about $G_0$, or perhaps to 
coordinate a strategy.  Another success metric is
the amount of computation and communication needed
during this preprocessing round.  Our full model is described as Model~\ref{algo: model-2}.

This model can be seen as a special case of our general model (Section~\ref{sec: Intro-self-healingModel}). We  assume we begin with a connected graph of $n$ vertices and do not explicitly discuss node insertions in $\FTree$.  Since only deletions happen, $n$ can only decrease.  For this reason, for our bounds, we need only compare our graph properties in the present graph at timestep $t$ ($G_t$), to the initial graph $G_0$ which has $n$ vertices.

\floatname{algorithm}{Model}

\begin{algorithm}[h!]
\caption{The Delete and Repair Model -- Distributed View.}
\label{algo: model-2}
\begin{algorithmic}
\STATE Each node of $G_0$ is a processor.  
\STATE Each processor starts with a list of its neighbors in $G_0$.
\STATE Pre-processing: Processors may send messages to and from
their neighbors.
\FOR {$t := 1$ to $n$}
\STATE Adversary deletes a node $v_t$ from $G_{t-1}$, forming $H_t$.
\STATE All neighbors of $v_t$ are informed of the deletion.
\STATE {\bf Recovery phase:}
\STATE Nodes of $H_t$ may communicate (in parallel) 
with their immediate neighbors.  These messages are never lost or
corrupted, and may contain the names of other vertices.
\STATE During this phase, each node may insert edges
joining it to any other nodes as desired. 
Nodes may also drop edges from previous rounds if no longer required.
\STATE At the end of this phase, we call the graph $G_t$.
\ENDFOR
\STATE {\bf Success metrics:} Minimize the following ``complexity'' measures:
 \begin{enumerate}
\item{\bf Degree increase.} $\max_{t<n} \max_{v} \Degree(v,G_t) - \Degree(v,G_0)$
\item {\bf Diameter stretch.} $\max_{t<n} \diam(G_t) / \diam(G_0)$
\item{\bf Communication per node.} The maximum number of bits sent by a single node in a single recovery round.
% \tom{Want to modify this or omit?}
\item{\bf Recovery time.} The maximum total time for a recover round,
assuming it takes $1$ bit no more than $1$ time unit to traverse any edge and unlimited local computational power at each node.
\end{enumerate}
\end{algorithmic}
\end{algorithm}

\section{The Forgiving Tree algorithm}
\label{sec:algorithm}

At a high level, our algorithm works as follows. We begin with a
rooted spanning tree $T$, which without loss of generality may as well
be the entire network.

Each time a non-leaf node $v$ is deleted, we think of it as being
replaced by a balanced binary tree of ``virtual nodes,'' 
with the leaves of the virtual tree taking $v$'s
place as the parents of $v$'s children. Depending on certain
conditions explained later, the root of this ``virtual tree'' or
another virtual node (known as $v$'s \emph{heir}---this will be
discussed later) takes $v$'s place as the child of $v$'s
parent. This is illustrated in figure \ref{fig: RT}.
Note that each of the virtual nodes which was added is of
degree $3$, except the heir, if present.

\begin{figure}[h!]
\centering
\includegraphics[scale=0.5]{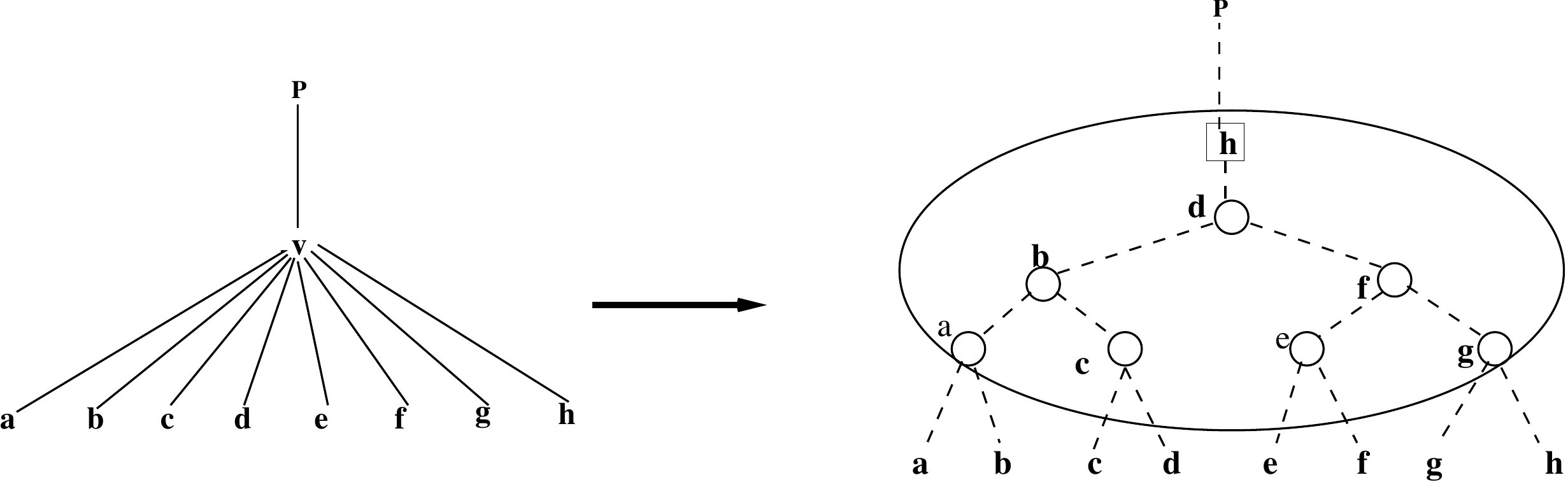}
\caption{Deleted node $v$ replaced by its Reconstruction Tree. The nodes in the oval are helper nodes. Regular helper nodes are depicted by circles and the heir helper node by a rectangle.}
 \label{fig: RT}
\end{figure}

When a leaf node is deleted,
we do not replace it.  However, if the parent of the deleted 
leaf node was a virtual node, its degree has now reduced
from $3$ to $2$, at which point we consider it redundant
and ``short-circuit'' it, removing it from the graph,
and connecting its surviving child directly to its parent.
This helps to ensure that, except for heirs, every virtual
node is of degree exactly $3$.

After a long sequence of such deletions, we are left with a
tree which is a patchwork mix of virtual nodes and original nodes.
We note that the degrees of the original nodes never increase
during the above procedure.
Also, because the virtual trees are balanced binary trees, the deletion of
a node $v$ can, at worst, cause the distances between its
neighbors to increase from $2$ to $2 \lceil \log d  \rceil$, where 
$d$ is the degree of $v$.  This ensures that,
even after an arbitrary sequence of deletions, the distance
between any pair of surviving actual nodes has not increased
by more than a $\lceil \log \Delta \rceil$ factor, 
where $\Delta$ is the maximum degree of the original tree.

Since our algorithm is only allowed to add edges and not nodes, 
we cannot really add these virtual nodes to the network.
We get around this by assigning each virtual node to an actual
node, and adding new edges between actual nodes in order to 
allow ``simulation'' of each virtual node.  More precisely,
our actual graph is the homomorphic image of the tree
described above, under a graph homomorphism which fixes 
the actual nodes in the tree and maps each virtual node
to a distinct actual node which is ``simulating'' it.
The existence of such a mapping is a consequence of the 
fact that all the virtual nodes have degree $3$, except heirs,
which have degree $2$ (and there are not too many of these), 
and will be proved later.
Note that, because each actual node ever simulates at most one
virtual node at a time, and virtual nodes have degree at most $3$,
this ensures that the maximum degree increase of our algorithm
is at most $3$.

The heart of our algorithm is a very efficient distributed algorithm 
for keeping track of which actual node is assigned to simulate
each virtual node, so that the replacement of each deleted node by
its virtual tree can be done in $O(1)$ time.  
We accomplish this using a system of ``wills,'' 
in which each vertex $v$ instructs each of its children (or their ``heirs'') 
in the event of $v$'s deletion, how to simulate the
virtual tree replacing $v$, and also the virtual node $v$ was
simulating (if any).

This will is prepared in advance, before $v$'s deletion, and entrusted
to $v$'s children or their surviving heirs.  An example of this is
shown in figure \ref{fig: RTbreakup}.  Certain events, such as the
deletion of one of $v$'s children, or a change in which virtual node
$v$ is simulating, may cause $v$ to revise its will, informing the
affected children or their surviving heirs.  As shall be seen, the
total number of messages and node IDs which must be sent is $O(1)$ per
deleted vertex; the number of bits sent is thus $O(\log n)$. In addition, there is a startup cost for communicating
the initial wills: this is $O(1)$ latency; and $O(1)$ messages and $O(\log n)$ bits per edge in the original network.

%\subsection{Detailed description} 
\subsection{Distributed implementation} 

To begin with, in Table~\ref{tab: nodedata} we list the data kept by
each real node $v$ required for the ForgivingTree algorithm. We have
four main classes of fields, according to the way they are used by the
node.  `Current fields' give a node's present configuration and status
in the tree. `Reconstruction fields' hold the data needed for a node
to reconstruct connections when one of its neighbors gets
deleted. `Helper fields' hold information with regard to the helper
node being simulated by this node. Each node also stores some special
flags with regard to its helper or heir status. In the description
that follows, we shall refer directly to these fields.

\begin{table}[h!]
\centering
%\begin{tabular}{|l|p{1.8in}|}
\begin{tabular}{|l|p{0.65\textwidth}|}
%\begin{tabular}{|l|l|}
\hline 
\textbf{Current fields}& Fields having  information about a node's current neighbors.\\ \cline{2-2}
%\hline
  \texttt{parent(v)}& Parent of $v$.\\
  \texttt{children(v)}& Children of $v$.\\
  \texttt{SubRT(v)}& Stores the Reconstruction Tree ($\RT$) of $v$ minus a possible helper node
simulated by $\heir(v)$.  This tree of helper and real nodes shall replace $v$ if $v$ is deleted.\\
 \texttt{heir(v)} & The heir of $v$.\\
\hline
\textbf{Helper fields} & Fields specifying a node's role as a helper node.\\ \cline{2-2}
%\hline
 \texttt{hparent(v)}& Parent of the helper node  $v$ may be simulating. \\
 \texttt{hchildren(v)}& Children of the helper node $v$ may be simulating.\\
\hline
\textbf{Reconstruction fields}& Fields used by a node to reconstruct its connections when its neighbor is deleted.\\ \cline{2-2}
%\hline
 \texttt{nextparent(v)}& The node which will be the next $\parent$ of $v$. \\
 \texttt{nexthparent(v)}& The node which will be the next $\hparent$ of $v$. \\
 \texttt{nexthchildren(v)}& The node(s) which will be the next $\hchildren$ of $v$.\\
\hline
\textbf{Flags}& Specifying a node's helper or heir status.\\ \cline{2-2}
%\hline
\texttt{ishelper(v)}& (boolean field). True if $v$ is simulating a \emph{helper} node, false otherwise.\\
\texttt{isreadyheir(v)}&  (boolean field). True if $v$ is simulating an heir in ready state, false otherwise (wait or deployed state). \\
\hline
\end{tabular}
\caption{The fields maintained by a node $v$}
\label{tab: nodedata}
\end{table}

%\item $\heir(v)$: The heir of $v$ (a specific helper node in $\RT(v)$).
%
%\item $\ishelper(v)$: (boolean field). Tells whether $v$ is 
%also simulating a \emph{helper} node.
%\item $\isreadyheir(v)$:  (boolean field). Tells whether $v$ is an 
%'unemployed' heir. 
%\item $\helper(v)$: the helper node which $v$ is simulating.
%
%\item $\hparent(v)$: Parent of $\helper(v)$ (a superneighbor). 
%\item $\hchildren(v)$: Children of $\helper(v)$ (list of superneighbors).

At the top level, our algorithm is specified as Algorithm \ref{algo: forgiving}~: \textsc{Forgiving tree}. Algorithm \ref{algo: forgiving}~ uses Algorithms 2 to 9, which will be described at the appropriate places. As referred to earlier, \textsc{Forgiving tree} works on a tree which may be obtained from the original graph during a preprocessing phase. The next stage is an initialisation phase in which the appropriate data structures are setup. Once these are setup, the network is ready to face the adversarial attacks as and when they happen.

\begin{figure}[h!]
\centering
\includegraphics[scale=0.39]{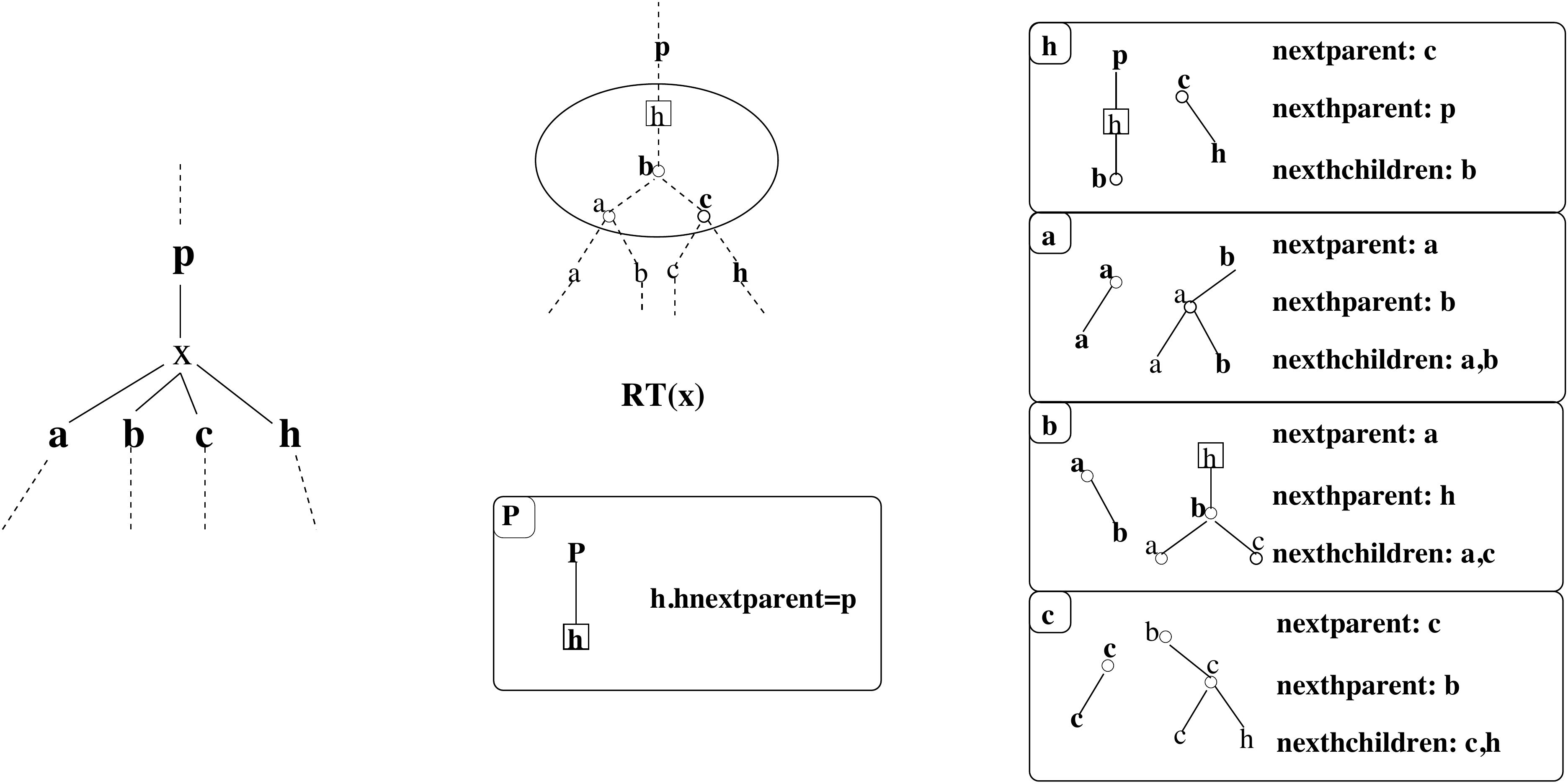}
\caption{The leftmost column shows a small segment of the network.  The RT(x) corresponding to this figure is shown. 
Every neighbor of node $x$ stores the portion of $\RT(x)$ relevant to it. Each rectangular box is labelled with a neighbor and shows the portions and the value of the corresponding fields . }
\label{fig: RTbreakup}
\end{figure}

%commented out  
% \begin{comment} 
% \begin{figure}[h!]
%\centering
%%\subfigure[$\helper(v)$ is ancestor of $v$.]{\label{sfig: ldc1} \includegraphics[scale=0.2]{PODC_LDcase1.pdf}}\hspace{0.1in}
%\subfigure[Turn1: Adversary deletes $v$. Vertices $a$ through $h$ take over helper nodes in $\RT(v)$. $h$ is $v$'s
%heir and connects to both $p$ and $d$. Note that the real graph now contains a cycle ($b,c,d$). ]{\label{sfig:
%storyc1}
%\includegraphics[scale=0.3]{for_story_stage1_10jun}}\hspace{0.3in}
%\subfigure[]{\label{sfig: storyc2} \includegraphics[scale=0.3]{for_story_stage2_10jun}}\\
%\subfigure[]{\label{sfig: storyc3} \includegraphics[scale=0.3]{for_story_stage3_10jun}}\hspace{0.3in}
%\subfigure[]{\label{sfig: storyc4} \includegraphics[scale=0.3]{for_story_stage4_10jun}}\hspace{0.3in}
%\subfigure[]{\label{sfig: storyc5} \includegraphics[scale=0.3]{for_story_stage5_10jun}}
%%\subfigure[ $T$ after round3]{\label{sfig: t4lr3b} \includegraphics[scale=0.5]{ternary4l3rb.jpg}}
%\caption{An illustrative sequence of deletions and healings. In two-column figure, the left column shows the real
%network and the right column the Forgiving Tree.}
%\label{fig: leafdels}
%\end{figure} 
% \end{comment} 
 %commented out
  
 \begin{figure}[h!]
\centering
\includegraphics[width = 0.7\textwidth, height= 0.9\textheight]{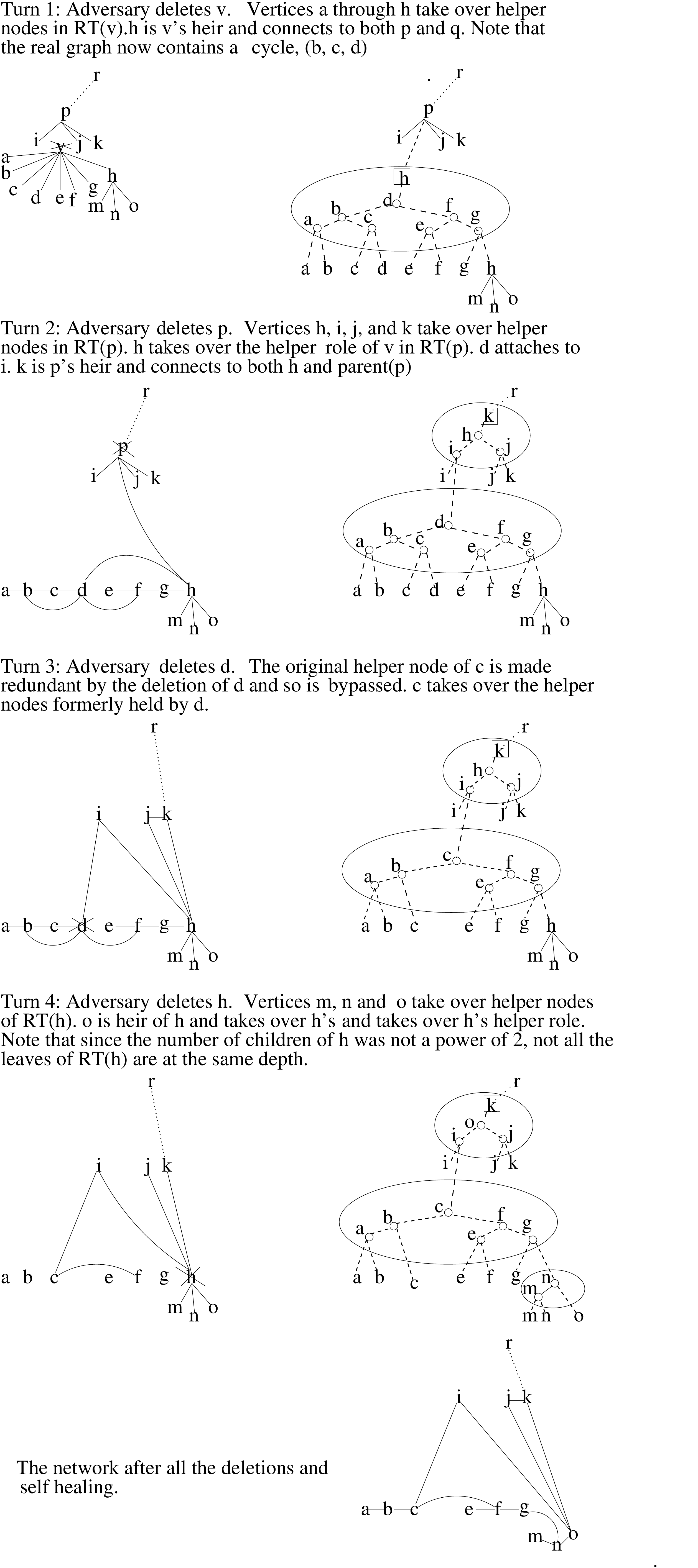}

\caption{An illustrative sequence of deletions and healings.}
\label{fig: story}
\end{figure}

\subsubsection{The Initialization phase}
This phase is specified in Algorithm \ref{algo: init}~: \textsc{Init()}. We assume each node $v$ has a unique
identification number  which we call $ID(v)$. Every node in the tree initializes the fields we have listed in Table~\ref{tab: nodedata}.   In our descriptions if no data is available or appropriate for a field, we set it to $\Empty$.
Since no deletion has happened yet and there are no helper nodes in the system, the helper fields are set to $\Empty$.
The current fields $\parent(v)$ and $\children(v)$ are assigned pointers to the parent and children of $v$. Of course,
if $v$ is a leaf node $\children(v)$ is $\Empty$ and if $v$ is the root of the tree $\parent(v)$ is $\Empty$. 

 As stated earlier, the heart of our algorithm is the system of
$\will$s created by nodes and distributed among its neighbors. 
The will of a node, $v$, has two parts: firstly, a Reconstruction Tree
($\SubRT(v)$), which will replace $v$ when 
it is deleted by the Adversary, 
and secondly, the delegation of $v$'s helper responsibilities 
(if any) to a child node, $\heir(v)$.
For concreteness, we initially designate the child of $v$ with
the highest $\ID$ as $\heir(v)$.  In the event that $\heir(v)$
is deleted, its role will be taken over by its heir, if any.
If $\heir(v)$ is a leaf when it is deleted, then $v$ will 
designate its new heir to be the surviving child whose helper
node has just decreased in degree from $3$ to $2$.

Algorithm \ref{algo: genRT}: \textsc{GenerateSubRT} computes
$\SubRT(v)$. If the node $v$ has no helper responsibilities, as is
during this phase, $\RT(v)$ is simply $\SubRT(v)$ with a helper node
simulated by $\heir(v)$ appended on as the parent of the root of
$\SubRT(v)$. Figure \ref{fig: RT} and Turn 1 in Fig~\ref{fig: story}
depict such Reconstruction Trees. If the node $v$ has helper
responsibilities $\RT(v)$ is the same as $\SubRT(v)$. Node $v$ uses
Algorithm \ref{algo: genRT} to compute $\SubRT(v)$ as follows: All the
children of $v$ are arranged as a single layer in sorted (say,
ascending) order of their $\ID$s. Then a set of helper nodes - one
node for each of the children of $v$ except the heir are arranged
above this layer so as to construct a balanced binary search tree
ordered on their $ID$s.

 The last step of the initialization process is to finalize the will and transmit it to the children. Each child is
given only the portion of the will relevant to it. Thus, only this portion needs to be updated whenever a will changes.
The division of $\RT$ into these portions is shown in figure \ref{fig: RTbreakup}.  There are fundamentally two
different kinds of wills : one prepared by leaf nodes who have helper responsibilities and the other by non-leaf nodes.
Obviously, during the initialization phase, only the second kind of will is needed. This is finalized and distributed as
shown in Algorithm \ref{algo: makewill}: \textsc{MakeWill}. 
The children of $v$ initialize their reconstruction fields
with the values from $\SubRT(v)$. If later $v$ gets deleted these values will be copied to present and helper fields
such that $\RT(v)$ is instantiated. Notice that the role the heir will assume is decided according to whether $v$ is a
helper node or not. Since $v$ cannot be a helper node in this phase, the heir node simply sets its reconstruction
fields so as to be between the root of $\SubRT(v)$ and $\parent(v)$. In this case when $\RT(v)$ will be instantiated,
the helper node simulated by $\heir(v)$ shall have only one child: we will say that $\heir(v)$ is in the ready phase
(explained later) and set the flag $\isreadyheir(v)$ to true. In the initialization phase both the $\isreadyheir$ and 
$\ishelper $ flags will be set to false. 

This completes the setup and initialization of the data structure. Now our network is ready to handle adversarial
attacks. In the context of our algorithm, there are two main events that can happen repeatedly and  need to be handled
differently:
\subsubsection{Deletion of an internal node}
\label{subsec: internaldel}
 The healing that happens on deletion of a non-leaf node is specified in Algorithm \ref{algo: fixnode}:
\textsc{FixNodeDeletion}. In our model, we assume that the failure of a node is only detected by its neighbors in the
tree, and it is these nodes which will carry out the healing process and update the changes wherever required. If the 
node $v$ was deleted, the first step in the reconstruction process is to put $\RT$ into place according to Algorithm
\ref{algo: makeRT}: \textsc{makeRT}. Note that all children of $v$ have lost their parent. Let us discuss the
reconstruction performed by non-heir nodes first. They make an edge to their new parent (pointer to which was
available as nextparent()) and set their current fields. Then they take the role of the helper nodes as specified in
$\RT(v)$ and Algorithm \ref{algo: makehelper}: \textsc{MakeHelper} and make the required edges and field changes to
instantiate $\RT(v)$.

To understand what the heir node does in this case,  it will be useful here to have a small discussion on the states of
a regular/heir node:\\

\begin{figure}[h!]
\centering
\includegraphics {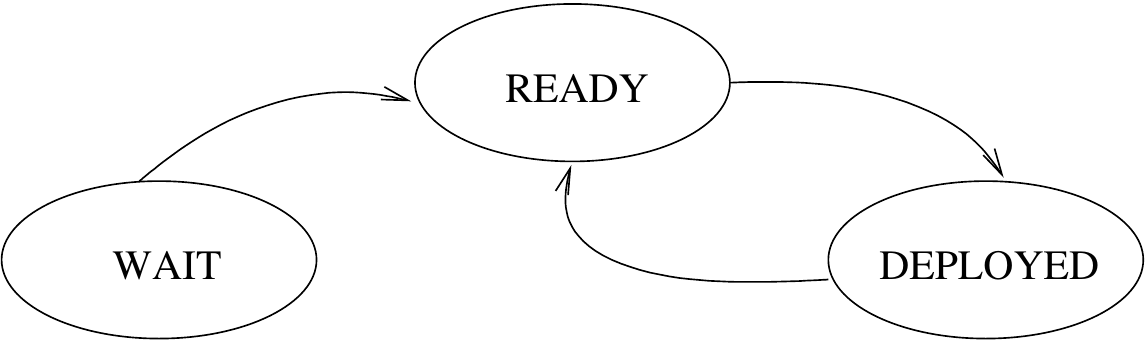}
\caption{The states of a node with respect to helper duties: Waiting, Ready and Deployed}
\label{fig: nodestates}
\end{figure}
  
% \pagebreak
  
\noindent {\bf States of a heir/regular node:} Consider a node $v$ and its heir $h$. From the point of view of $h$, we
can imagine $h$  to be in one of three states which we  call \emph{wait}, \emph{ready} and \emph{deployed}. These states
are illustrated in figure ~\ref{fig: nodestates}. For ease of discussion, let us call the helper node that a node is simulating
$\helper(node)$.
In brief, a node is considered to be in the wait state when it has no helper
responsibilities, in the ready state when  $\helper(node)$ with one child, and in the deployed state when  $\helper(node)$
 has  two children (which is the maximum possible).  Notice that the node can  be in the wait state only when  $v$ has
not been deleted and thus, $h$ has assumed no helper responsibilities. It only has the will of $v$ and is in limbo
with regard to helper duties. Now consider the case when $v$ gets deleted.  Following are the possibilities:

\begin{itemize}
\item \emph{node $v$ had no helper responsibilities}: This happens when $v$'s original parent was not deleted. Thus,
 $v$ could be a regular child or a heir in the wait state. On $v$'s deletion $h$ moves to the ready state and sets its
flag $\isreadyheir$ to True. This is the state in which  $\helper(h)$ has only one child i.e.
the root of $\SubRT(v)$. This happens when $h$ executes its portion of the will of $v$ using Algorithm \ref{algo: makeRT}:
\textsc{makeRT}. Note that this may not be the final state for the helper node of $h$, and is thus called the ready
state.

\pagebreak

\item \emph{node $v$ had helper responsibilities}: There are two further possibilities:
\begin{itemize}
\item \emph{$\helper(v)$ had one child}: This can only happen when $v$ was a heir node  in the ready state. Thus,
$v$'s flags $\ishelper$ and $\isreadyheir$ were both set to True. Node $h$
will take over the helper responsibilities of $v$ and thus, in turn, $h$ will now have one child i.e. will be in the
ready state and will set its flags $\ishelper$ and $\isreadyheir$ to True. Notice that if $v$ was an heir, $h$ will
now also take over those responsibilities, and on future deletions of $v$'s ancestors could move further up the tree
either as an heir in ready state or in a deployed state become a full helper node.
\item \emph{$\helper(v)$ had two children}: Node $v$ could be a regular child or heir. $h$ will fully take over the
helper responsibilities of $v$, and thus $\helper(h)$ shall acquire two children and move on to the deployed state.
Notice that previously $h$ could have been in either  wait or  ready state. Since  it is now not in the ready
state, it will set its $\isreadyheir$ flag to False and $\ishelper$ flag to True.
\end{itemize}
\end{itemize}
 It is easy to see that a regular i.e. non-heir node can be in either wait or deployed state.\\
 
 Here we also define the following operation, which is used in Algorithms  \ref{algo: fixleaf} and \ref{algo: makeRT}:
\begin{description}
\item[bypass(x):] \emph{Precondition}: $|\hchildren(x)| = 1$ i.e the helper node has a single child. \emph{Operation}:
\emph{Delete} $\helper(x)$ i.e. $\hparent(x)$ and $\hchildren(x)$  remove their edges with $x$ and make a
new edge between themselves. \\
 $\hparent(x) \leftarrow \Empty$; $\hchildren(x) \leftarrow \Empty$.
\end{description}

  We can now easily see how the heir of $v$, $h$ takes part in the reconstruction according to Algorithm
\ref{algo: makeRT}: \textsc{makeRT}. Node $h$ can be either in wait state or ready state. If it is in the wait
state it simply takes its helper responsibilities according to Algorithm \ref{algo: makehelper}: \textsc{MakeHelper},
as in turn 1 of figure~\ref{fig: story} .
Note that here $h$ checks  if it has moved to the ready state and sets its $\isreadyheir$ flag accordingly.
 If $h$ was already in ready state, it relinquishes its present helper role and moves on to the new helper
role. To relinquish its present role, node $h$ intimates  $\hparent(h)$ and $\hchildren(h)$, and together they accomplish
this as specified by the operation bypass($h$). Turn 2 in Fig~\ref{fig: story} illustrates this. 

Once $\RT(v)$ is in a place, there may be a need for the parent of $v$ to recompute its will. This happens only when
$v$ did not already have a helper role or equivalently when $\heir(v)$ moves to a ready state. Lines \ref{algline:
ND2} to \ref{algline: ND6} of Algorithm \ref{algo: fixnode} deals with this situation. Node $\parent(v)$ simply
replaces $v$ by $\heir(v)$ in its will and retransmits it.  At the end of this healing process, the children of the
deleted nodes check if they need to leave the second kind of will, which we call a \emph{LeafWill}. This will is
required only for those nodes which are leaves in our tree and have virtual responsibilities. Since they have no
children to take over their helper responsibilities they leave this responsibility to their parent. We will discuss
this in greater detail in the next section.

\subsubsection{Deletion of a leaf node}
\label{subsec: leafdel}
\begin{figure}[h!]
\begin{centering}
%\subfigure[$\helper(v)$ is ancestor of $v$.]{\label{sfig: ldc1} \includegraphics[scale=0.2]{PODC_LDcase1.pdf}}\hspace{0.1in}
%\subfigure[$\helper(v)$ is ancestor of $v$.]{\label{sfig: ldc1} \includegraphics[scale=0.22]{LDcase1-crop}}\hspace{0.3in}
%\subfigure[$w$ and $\helper(w)$ share a neighbor.]{\label{sfig: ldc2} \includegraphics[scale=0.22]{LDcase2-crop}}\\
%\subfigure[ $z$ and $\helper(z)$ do not share neighbors.]{\label{sfig: ldc3} \includegraphics[scale=0.22]{LDcase3-crop}}\hspace{0.3in}
%\subfigure[$z$ is an heir in Ready state.]{\label{sfig: lhdc0} \includegraphics[scale=0.22]{LHDcase0-crop.pdf}}
\subfigure[$\helper(v)$ is ancestor of $v$.]{\label{sfig: ldc1} \includegraphics[scale=0.6]{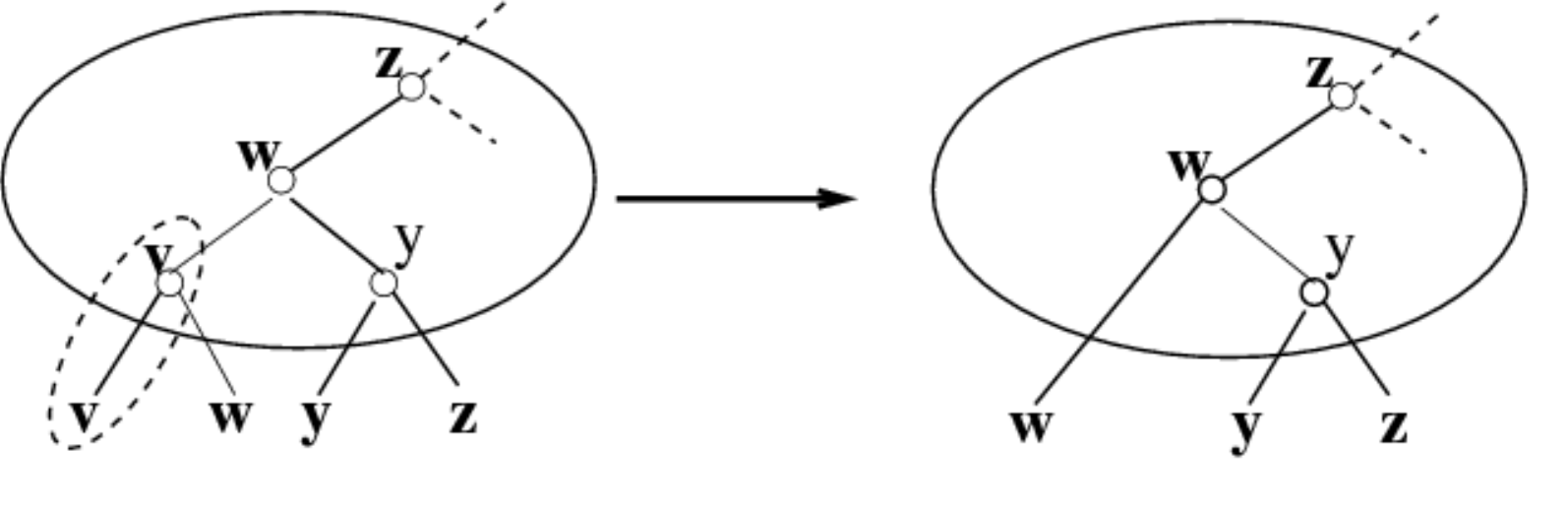}} \\ 
%\hspace{0.3in} 
\subfigure[$w$ and $\helper(w)$ share a neighbor.]{\label{sfig: ldc2} \includegraphics[scale=0.6]{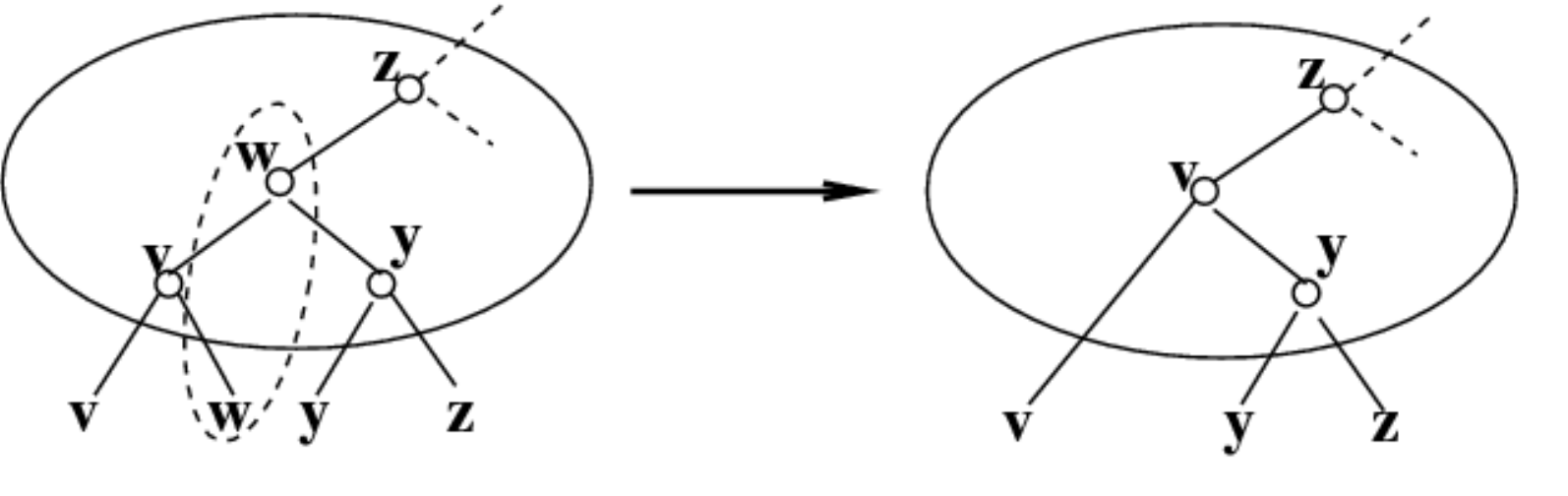}} \\
\subfigure[ $z$ and $\helper(z)$ do not share neighbors.]{\label{sfig: ldc3} \includegraphics[scale=0.6]{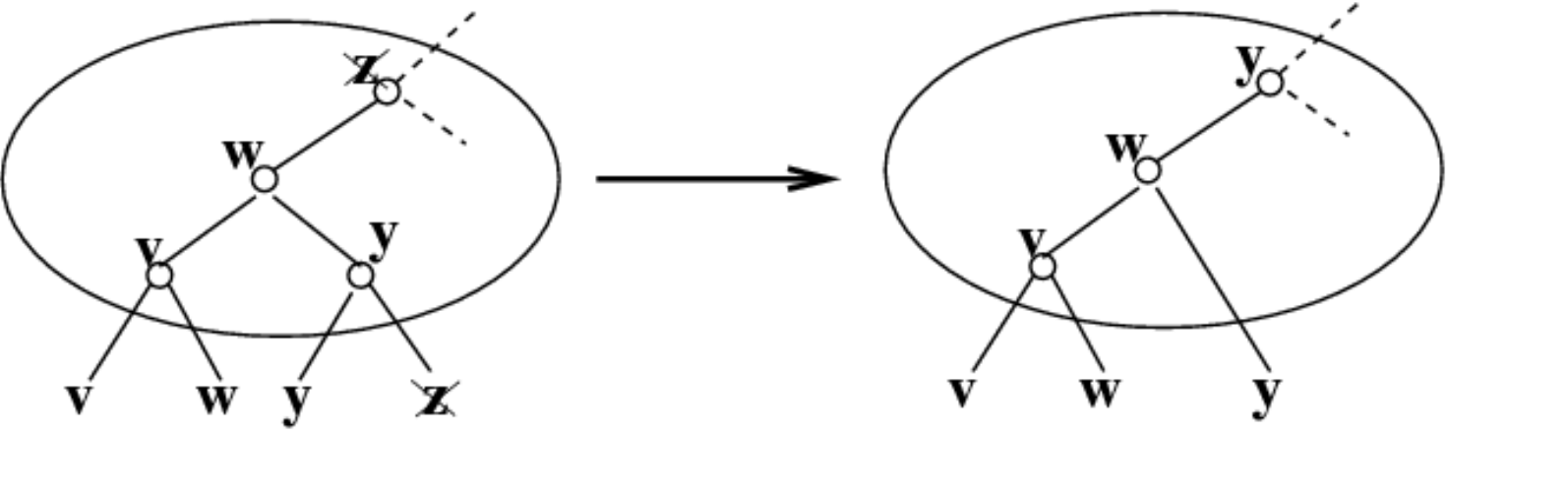}} \\ 
%\hspace{0.3in}
\subfigure[$z$ is an heir in Ready state.]{\label{sfig: lhdc0} \includegraphics[scale=0.6]{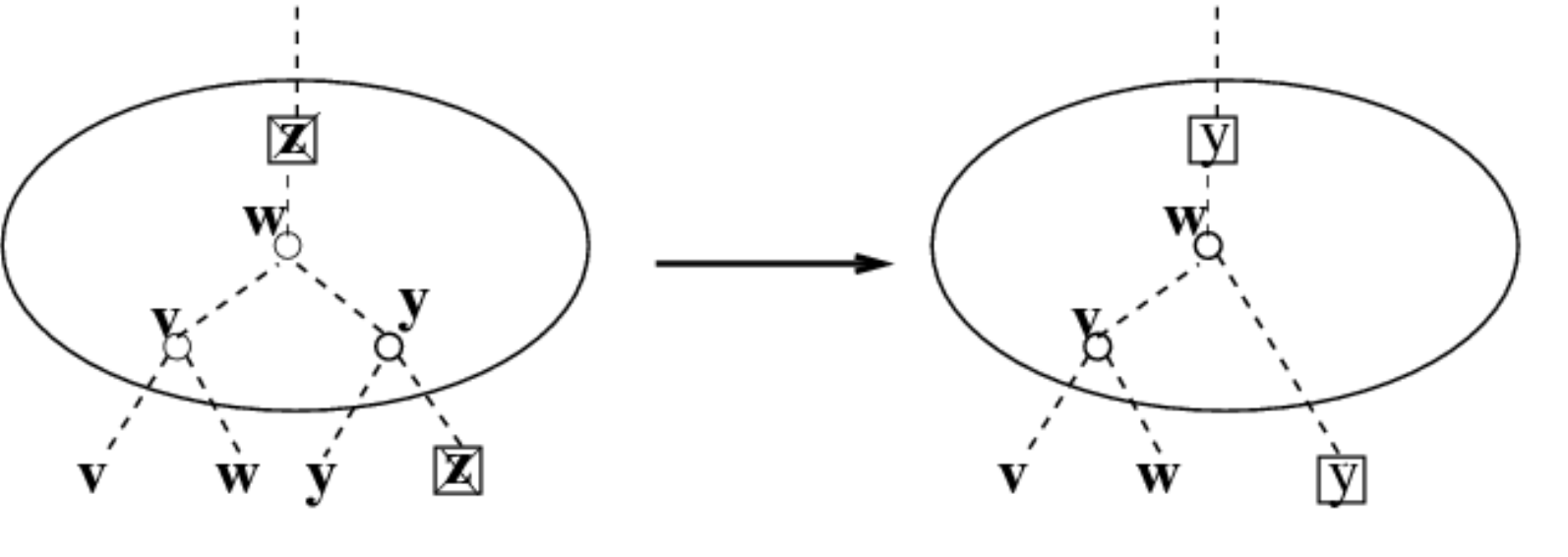}}
%\subfigure[ $T$ after round3]{\label{sfig: t4lr3b} \includegraphics[scale=0.5]{ternary4l3rb.jpg}}
\caption{Various cases of Leaf deletions}
\label{fig: leafdels}
\end{centering}
\end{figure} 

 If the adversary removes a leaf node from the system, the healing is accomplished by its neighbors as specified in
Algorithm \ref{algo: fixleaf}: \textsc{FixLeafDeletion}. Let $v$ be the deleted leaf node and $p$ be its parent. Let us
consider the simple case first.  This is when the deleted node had no helper responsibility. This also implies its
original parent did not suffer a deletion. Node $p$ simply removes $v$ from the list of its children and then
recomputes and redistributes its will.\\
  Now, consider the situation where the deleted node had helper responsibilities.  In this case, the one node whose
workload has been reduced by this deletion is $p$. Using Algorithm \ref{algo: makeleafwill}:
\textsc{MakeLeafWill} $v$ hands over the list of  its helper responsibilities to $p$.  Here, a special
case may arise when $v$ is simulating a helper node which has $v$ itself as one of its $\hchildren$. Recall that $\parent(v)$ is $v$'s
ancestor closest to $v$ in the tree. This implies that $\parent(v) = \hparent(v) = p$. The only thing that $p$ needs
to do if $v$ is deleted is to remove $v$ from its $\hchildren$ and add itself (for consistency). This is the will
conveyed by $v$ to $p$. When $v$ is deleted, $p$ simply updates its helper fields.
 For other cases, $v$ simply sends its helper fields to $p$ to be copied to $p$'s helper reconstruction fields. In this
situation, when $v$ is actually deleted the following happens: the helper node that $p$ is simulating is deleted and
bypassed by the bypass operation defined earlier.  Node $p$ now simulates a new helper node that has the same helper
responsibilities previously fulfilled by $v$. In case the deleted leaf node was itself an heir in ready state, $p$
detects this and sets its flags accordingly. Again at the end of the reconstruction, the leaf nodes reconstruct their
wills. An example of such a leaf deletion is the deletion of  node $d$  at Turn 3 as shown in  figure \ref{fig: story}.

\par\noindent{\bf Important Note:}\quad 
When implementing the pseudocode for 
Algorithm~\ref{algo: makewill} \\ (\textsc{MakeWill}), 
it is important to bear in mind that when $\RT(v)$ is being
updated due to a node deletion, most of $\SubRT(v)$ will be
unchanged.  In fact, only $O(1)$ nodes will need to have their
fields updated.  These can be found and updated more efficiently
by a more detailed algorithm based on case analysis.

%, which we  defer to the full version of this paper.

%Let $C(v,T)$ be the children of the vertex $v$, and $P(v,T)$ be the parent of the vertex $v$ in the Tree $T$ that

\floatname{algorithm}{Algorithm}

\begin{algorithm}[ph!]
\begin{algorithmic}[1]
\STATE Given a tree $T(V,E)$
\STATE \textsc{Init(T)}.
\WHILE {true}
\IF{a vertex $x$ is deleted}
\IF { $\children(x)$ is $\Empty$}
\STATE \textsc{FixLeafDeletion(x)}
\ELSE 
\STATE \textsc{FixNodeDeletion(x)}
\ENDIF
\ENDIF
\ENDWHILE
\end{algorithmic}
\caption{\textsc{Forgiving tree}: The main function.}
\label{algo: forgiving}
\end{algorithm}
 
 \begin{algorithm}[h!]
\begin{algorithmic}[1]
\REQUIRE{each node of T has a unique ID}
\FOR{each node  $v\in T$} 
\STATE $\children(v)\leftarrow $ children of $v$.
\STATE $\parent(v)\leftarrow $ if $v$ is root of T then $\Empty$ else parent of $v$.
\STATE $\isreadyheir(v)\leftarrow false$. 
\STATE $\ishelper(v)\leftarrow false$.
\STATE $\hparent(v)\leftarrow \Empty$. 
\STATE $\hchildren(v)\leftarrow \Empty$. 
\STATE $\heir(v)\leftarrow$ if v is a leaf node then $\Empty$ else child of $v$ with highest ID.
\STATE $\SubRT(v)\leftarrow \textsc{generateSubRT}(v)$.
\STATE \textsc{MakeWill}($v, \SubRT(v)$).
\ENDFOR
\end{algorithmic}
\caption{\textsc{Init(T)}: initialization of the Tree T} 
\label{algo: init}
\end{algorithm}

\begin{algorithm}[h!]
\caption{\textsc{FixNodeDeletion($v$)}: Self-healing on deletion of internal node }
\label{algo: fixnode}
\begin{algorithmic}[1]
\STATE \textsc{MakeRT}($\children(v),\parent(v)$).
\STATE \label{algline: ND2} let $h = \heir(v)$.Let $p = \parent(v)$
\IF {$\isreadyheir(h) =true$} 
\STATE  $\hparent(h)$ replaces $v$ by $h$ in $\SubRT(\hparent(h))$.
\STATE \textsc{MakeWill}($\hparent(h),\SubRT(\hparent(h))$).
\ENDIF \label{algline: ND6}
\FOR {each node $y\in \children(v)$}
 \IF {$\children(y)$ is $\Empty$}
\STATE \textsc{MakeLeafWill}($y$).
\ENDIF
\ENDFOR
\end{algorithmic}
\end{algorithm}

 \begin{algorithm}[h!]
\caption{\textsc{FixLeafDeletion($v$)}: Self-healing on deletion of leaf node}
\label{algo: fixleaf}
\begin{algorithmic}[1]
\STATE let $p=\parent(v)$
\IF {$\ishelper(p)=false$}
\STATE $p$ removes $v$ from $\children(p)$
\STATE $SubRT(p)\leftarrow \textsc{GenerateSubRT}(p)$
\STATE \textsc{MakeWill}($p$)
\ELSE
%\STATE Let z=$\hparent(v)$
\STATE Let $z=\parent(v)$.
\IF {$z \neq \hparent(v)$}
\STATE bypass($z$).
\ENDIF
\STATE $z $ makes edges with $\nexthparent(z)$,$\nexthchildren(z)$.
\STATE $\hparent(z)\leftarrow \nexthparent(z)$.
\STATE $\hchildren(z)\leftarrow \nexthchildren(z)$.
\IF{$\left|\hchildren(z)\right|$=1}
\STATE $\isreadyheir(z)=true $ 
\ENDIF
\ENDIF
\FOR {each node $y\in \children(v)$}
 \IF {$\children(y)$ is $\Empty$}
\STATE \textsc{MakeLeafWill}($y$).
\ENDIF
\ENDFOR
\end{algorithmic}
\end{algorithm}

  \begin{algorithm}[h!]
\caption{\textsc{GenerateSubRT($v$):} Computes  the Reconstruction Tree ($\RT$) of $v$ minus a possible helper node
simulated by $\heir(v)$. }
\label{algo: genRT}
\begin{algorithmic}[1]
\STATE Let $Lset$ be a set of vertices representing all members of $\children(v)$, and $Iset$ be another set of vertices.
representing  all members of $\children(v)$ except the one with the highest $ID$.
\STATE Arrange $Lset$ in ascending order of their $ID$s.
\STATE Using the arranged $Lset$ as leaves and $Iset$ as the internal nodes construct a  Balanced Binary Search Tree
$\SubRT$  ordered on the nodes $ID$.
\RETURN $\SubRT$
\end{algorithmic}
\end{algorithm}

  \begin{algorithm}[h!]
\caption{\textsc{MakeWill($v,\SubRT(v))$:} Makes and distributes the will of v }
\label{algo: makewill}
\begin{algorithmic}[1]
\STATE Let p = $\parent(v)$. Let  $rv$ be root of $\SubRT(v)$.
\FOR {each node $y\in \children(v)$}
    \STATE let $ly$ be the leaf vertex representing $y$ in $\SubRT(v)$. Let $hy$ be the internal node in $\SubRT$ representing $y$. 
    \STATE \label{algline: parentdefined} If $hy$ is $ly$'s parent in $\SubRT(v)$ then $\nextparent(y)\leftarrow$ parent
of $hy$ in $\SubRT$ else $\nextparent(y)\leftarrow$ parent of $ly$ in $\SubRT$.
     \IF{ $y\neq heir(v)$}
     \STATE $\nexthchildren(y)\leftarrow$ children of $hy$ in $\SubRT$.   
     \STATE $\nexthparent(y)\leftarrow$ parent of $hy$ in $\SubRT$.  
       \ELSE
        \IF{$\ishelper(v) = true$}
          \STATE $\nexthchildren(y)\leftarrow \hchildren(v)$.
          \STATE $\nexthparent(y)\leftarrow \hparent(v)$.
        %  \STATE $\isreadyheir(y)\leftarrow \isreadyheir(v)$.
          \STATE $\nexthparent(rv)\leftarrow p$.     
         \ELSE
                 \STATE $\nexthchildren(y)\leftarrow rv$.
                  \STATE $\nexthparent(y)\leftarrow p$.
                  \STATE $\nexthparent(rv)\leftarrow y$.
         \ENDIF
     \ENDIF
\ENDFOR
\end{algorithmic}
\end{algorithm}

\begin{algorithm}[h!]
\caption{\textsc{MakeLeafWill($v$)}:  Leaf node leaves a will for its parent.}
\label{algo: makeleafwill}
\begin{algorithmic}[1]
\STATE let $z=\parent(v)$.
\IF{$z = \hparent(v)$}
%[This is a special case.]
\STATE $\nexthparent(z)\leftarrow \hparent(v)$.
 \STATE $\nexthchildren(z)\leftarrow \hchildren(z)/\{v\} \cup \{z\}$. \COMMENT{$z$ will take on itself as a child of
its helper node.}
 \ELSE
\STATE $\nexthparent(z)\leftarrow \hparent(v)$.
\STATE $\nexthchildren(z)\leftarrow \hchildren(v)$.
\ENDIF
\end{algorithmic}
\end{algorithm}

\begin{algorithm}[h!]
\caption{\textsc{makeRT(children(v),parent(v))}: Replace the deleted node by its  $\RT$}
\label{algo: makeRT}
\begin{algorithmic}[1]
\FOR{each node $x\in$ $\children$(v)}
  \IF{$\isreadyheir(x)=true$}
  \STATE bypass($x$).     \COMMENT{ $\hparent(x)$ and $\hchildren(x)$ bypass $x$ and connect themselves. }
 % \STATE $\hchildren(x)$ makes edge with $\nextparent(x)$ \COMMENT{$\nextparent(x)$ points to $v$'s temporary parent.}
   %\STATE The edge $\hchildren(x)$,$x$ is  deleted. 
   %\STATE $\hparent(\hchildren(x)) = \nextparent(x)$
  % \IF{$\parent(\hchildren(x)) = x$}
     %   \STATE $\parent(\hchildren(x))=\nextparent(x)$
    %\ENDIF    
       \STATE \textsc{MakeHelper}($x$).
  \ELSE     
\STATE $x$ makes edge between itself and $\nextparent(x)$.
\STATE $\parent(x)\leftarrow \nextparent(x)$.
\STATE \textsc{MakeHelper}($x$).
\ENDIF
\ENDFOR 
\end{algorithmic}
\end{algorithm}

\begin{algorithm}[h!]
\caption{\textsc{MakeHelper($v$)}: $v$ takes over helper node responsibilities}
\label{algo: makehelper}
\begin{algorithmic}[1]
\STATE $v$ makes edges between itself and $\nexthchildren(v)$, and $\nexthparent(v)$.
\STATE $\hparent(v)\leftarrow \nexthparent(v)$.
\STATE $\hchildren(v)\leftarrow \nexthchildren(v)$.
\STATE $\ishelper(v) = true$.
\IF{$|\hchildren(v)| = 1$}
 \STATE $\isreadyheir(v) = true$.  \COMMENT{Only an 'unemployed' $\heir$ has a single child.}
 \ENDIF
\end{algorithmic}
\end{algorithm}

\pagebreak
%\newpage

\section{Results}
\label{subsec: Results}

% ** TODO: Add text on homomorphism at some place in this section **

\subsection{Upper Bounds}
\label{subsec: upperbounds}
Before considering the main theorem, we shall prove a couple of lemmas.
\begin{lemma}
\label{lemma: onehelper}
In the Forgiving Tree, a real node can simulate at most one helper node at a time.
\end{lemma}
\begin{proof}
A node simulates a new helper node if and only if its parent is deleted (Section~\ref{subsec: internaldel}) or
a sibling that is a leaf node in the Forgiving Tree is deleted (Section~\ref{subsec: leafdel}). We will show that whenever either
of the above events happens and the node has to simulate a new helper node, it no longer needs to simulate the
helper node it was simulating prior to these events occuring and thus it always simulates at most one helper node at a
time. Let us consider the cases in more detail. Consider a node $v$ and it's parent node $p$.
\begin{itemize}
\item \emph{Parent node $p$ is deleted:} There are three possibilities:
\begin{itemize}
\item \emph{Node $v$ is in Wait state (i.e. no previous helper role):} Node $v$ will now take over the role of exactly
one helper node as specified for $\RT(p)$ (Figure~\ref{fig: RT}).
\item \emph{Node $v$ is in  Ready heir:} Node $v$ will remove its previous helper node using operation $\bypass(v)$, and 
be redeployed in the  Ready state or as a Deployed node (Figure~\ref{fig: nodestates}). An example is node $h$ at turn 2 in
Figure~\ref{fig: story}.
\item \emph{Node $v$ is in Deployed state:} By construction and by definition of parent in the Forgiving Tree
(Line~\ref{algline: parentdefined}, Algorithm~\ref{algo: makewill}:\textsc{MakeWill}), $p$ is the parent of $v$ through
$\helper(p)$. This implies that $v$'s parent in the Forgiving Tree is $\helper(p)$ (not $p$ itself). Thus $v$ does not
feature in the will of $p$ and will not simulate a new helper node on deletion of $p$. 
\end{itemize}
\item \emph{In the Forgiving Tree, a leaf node sibling of $v$ is deleted:} Refer to Figure~\ref{fig: leafdels}, cases b,c and d, and node $c$ in turn 3,
Figure~\ref{fig: story}.
On deletion of a leaf node, exactly one helper node becomes redundant, and this can be removed.
If $v$ takes on the role of a new helper node, its old helper node is removed using $\bypass(v)$.
%In all cases, either the helper node of the deleted node is removed alongwith the deleted node, or the deleted
%node's sibling (i.e. our generic node $v$) deletes its helper node using $\bypass(v)$ and takes over the helper node of
%the deleted node.
\end{itemize}

\end{proof}

Let $\FT_{i}$ be the Forgiving Tree which has undergone $i$ rounds of deletions and healings. A time step is a
single deletion followed by healing.

\begin{lemma}
\label{lemma: ancestry}
 If an original node $x$ is an ancestor of another original node $v$ in $\FT_{i}$ for some time step $i$, then node $x$
must also have been an ancestor of node $v$ in $\FT_{0}$.
\end{lemma}
\begin{proof}
We will prove this by induction on  time step $i$.

\emph{Base case:} $i=0$: This is trivially true. 

\emph{Inductive step:} Let $v_{x}$ be the node deleted at time step $i$, and let $v$ be an arbitrary node in $\FT_{i}$.
By the inductive hypothesis, we need only show that the deletion of $v_{x}$ will not violate the invariant.
Note that if $v_{x}$ has a helper node then when $\helper(v_{x})$ is deleted, no new original node become an ancestor of
$v$ in $\FT_{i}$, since either a new helper node takes the place of $\helper(v_{x})$ or $\helper(v_{x})$ is bypassed.

We also note that when $v_{x}$ is deleted, no new original node can become the ancestor of $v$ in $\FT_{i}$. To see
this, note that when the deletion of $v_{x}$ creates an $\RT$ no \emph{original} node that was a child of $v_{x}$ can
becomes a new ancestor of $v$ in $\FT_{i}$.

\begin{comment}
Forgiving Tree $\FT_{0}$ consists only of original nodes. $\FT_{i}$ will consist of a patchwork of original and helper
nodes. Consider arbitrary nodes $v$ and  $x$, an ancestor of $v$. Consider another node $v_{j}$ on the path between
nodes $v$ and $x$. Let  $v_{j}$ be deleted at timesetp $j$. By our algorithm, node $v_{j}$ will be replaced by $\RT(v_{j})$
(Figure \ref{fig: RT}). Note that in $\RT(v_{j})$, all the surviving original  children of $v_{j}$ will be only on the
leaf layer. Thus, no original node will promote itself to be a parent of another original node in the Forgiving Tree.
However, helper nodes simulated by siblings may be promoted in the hierarchy. Thus, after timestep $j$, node $x$ will
still be an ancestor of $v$  except that node $v_{j}$ on their path will be replaced by a string of helper nodes. 
 Thus, if original node $x$ is an ancestor of $v$ in $\FT_{i}$, it was an ancestor of $x$ in $\FT_{0}$ too.
 \end{comment}
 
\end{proof}

\pagebreak

 Let $\Delta$ be its maximum degree of a node in $\FT_{0}$.
 
\begin{lemma}
\label{lemma: deldepth}
 Let $\delanc_{i}(v)$ be the number of ancestors of $v$ in $\FT_{0}$ that have been deleted by time step $i$.
 
  For all nodes $v$, 
 \[
 \depth_{i}(v) \le \depth_{0}(v) + \log \Delta \times \delanc_{i}(v)
 \]
 
\end{lemma}
\begin{proof}
We shall prove this by induction on time $i$.

\emph{Base case:} $i$ = 0: This is trivial since there have been no deletions so far.

\emph{Inductive step:} Let $v_{i}$ be the node deleted at the $i^{th}$ deletion. Consider an arbitrary original node
$v$ in $\FT_{i}$.  First, observe that the removal of $\helper(v_{i})$, if it exists  ($\helper(v_{i}$) will be removed
on
deletion of node $v_{i}$), never increases the depth of any node. This is because the helper node is either  replaced by
another helper node or it is removed in the $\bypass$ operation, which will never increase the depth of any node. We now
consider the deletion of the original node $v_{i}$. There are two cases for node $v$:
\begin{itemize}
\item \emph{Node $v$ is not in the subtree rooted at $v_{i}$:}
 Here,\\
  $\depth_{i}(v) \le \depth_{i-1}(v)$ and thus the induction holds. 
\item \emph{Node $v$ is in the subtree rooted at $v_{i}$:} 
By lemma~\ref{lemma: ancestry}, node $v_{i}$ must have been an ancestor of $v$ in $\FT_{0}$. Since Algorithm~\ref{algo:
fixnode} replaces $v_{i}$ with $\RT(v_{i})$, which is a balanced binary tree, we know that,
\[
\depth_{i}(v) \le \depth_{i-1}(v) + \log \Delta
\]
Also, by the Inductive hypothesis,
\[
\depth_{i-1}(v) \le \depth_{0}(v) + \log \Delta \times \delanc_{i-1}(v)
\]
These two equations imply that
\begin{eqnarray*}
\depth_{i}(v) & \le & \depth_{0}(v) + \log \Delta \times \delanc_{i-1}(v) + \log \Delta\\
    & \le & \depth_{0}(v) + \log \Delta \times \delanc_{i}(v) 
\end{eqnarray*}
\end{itemize}

\end{proof}

Now, we prove our main theorem. Let $\FT_{0}$ be  the original tree, and let $D$ be its diameter.
\begin{theorem}
The Forgiving Tree has the following properties:
\label{theorem: forgiving}
\begin{enumerate}
\item\label{th: degree} 
  The Forgiving Tree increases the degree of any vertex by at most $3$.
\item\label{th: diameter} 
  The Forgiving Tree always has diameter $O(D \log \Delta)$.
\item \label{th:cost} 
  The latency per deletion and number of messages sent per node per 
  deletion is $O(1)$; each message contains $O(1)$ node IDs and thus $O(\log n)$ bits.
\end{enumerate}
\end{theorem}
\begin{proof}
 Parts~\ref{th: degree} and~\ref{th:cost} follow directly by
construction of our algorithm.  For part~\ref{th: degree}, we note
that for a node $v$, any degree increase for $v$ is imposed by its
edges to $\hparent$($v$) and $\hchildren(v)$. By lemma~\ref{lemma: onehelper}, node
$v$ can play the role of at most one helper node at any time and the
number of $\hchildren$ is never more than $2$, because the
reconstruction trees are binary trees.  Thus the total degree increase
is at most $3$.  Part~\ref{th:cost} also follows directly by the
construction of our algorithm, noting that, because the virtual nodes all have degree at most $3$, healing one deletion 
results in at most $O(1)$ changes to the edges in each 
affected reconstruction tree.  In fact, the changes to 
$\RT(w)$ for an affected node $w$ do not require new information,
which allows these messages to be computed and distributed in parallel.
%We leave the details to the full version due to space constraints.
%\end{proof}

\hspace{-1em}We next show Part~\ref{th: diameter}, that the diameter of the Forgiving Tree  is always $O(D \log \Delta)$.
Consider the Forgiving Trees $\FT_{0}$ and $\FT_{i}$.  Let their respective heights be $h$ and $\h_{i}$. Consider a
node $x$ in $\FT_{i}$ which has the maximum depth, equal to $\h_{i}$. By lemma~\ref{lemma: deldepth},
\begin{eqnarray*}
\h_{i} & \le &\depth_{0}(x) + \log \Delta \times \delanc_{i}(x)\\
  & \le & \h + \log \Delta \times \delanc_{i}(x)
\end{eqnarray*}
 Since, node $x$ can have at most $\h$ ancestors,
 \begin{eqnarray*}
 \h_{i} & \le & \h + \log \Delta \times h\\
  & \le & \log \Delta \times (h + 1)
 \end{eqnarray*}

 Since the diameter of a tree can at most be twice the height of the tree, the diameter of $\FT_{i}$ is at most $2 (\h +
1)\log \Delta$, or $O(D \log \Delta)$.

\end{proof}

\subsection{Lower Bounds}
\label{subsec: lowerbounds}

\begin{theorem}
Consider any self-healing algorithm that ensures that: 1) each node
increases its degree by at most $\alpha$, for some $\alpha \geq 3$;
and 2) the diameter of the graph increases by a multiplicative factor
of at most $\beta$.  Then for any positive $\Delta$, for some initial
graph with maximum degree $\Delta$, it must be the case that
$\beta  \geq \frac{1}{2} [\log_{\alpha + 1} \Delta - 1]$.
\end{theorem}

\begin{proof}

\begin{figure}[h!]
\centering
\includegraphics[scale=0.8]{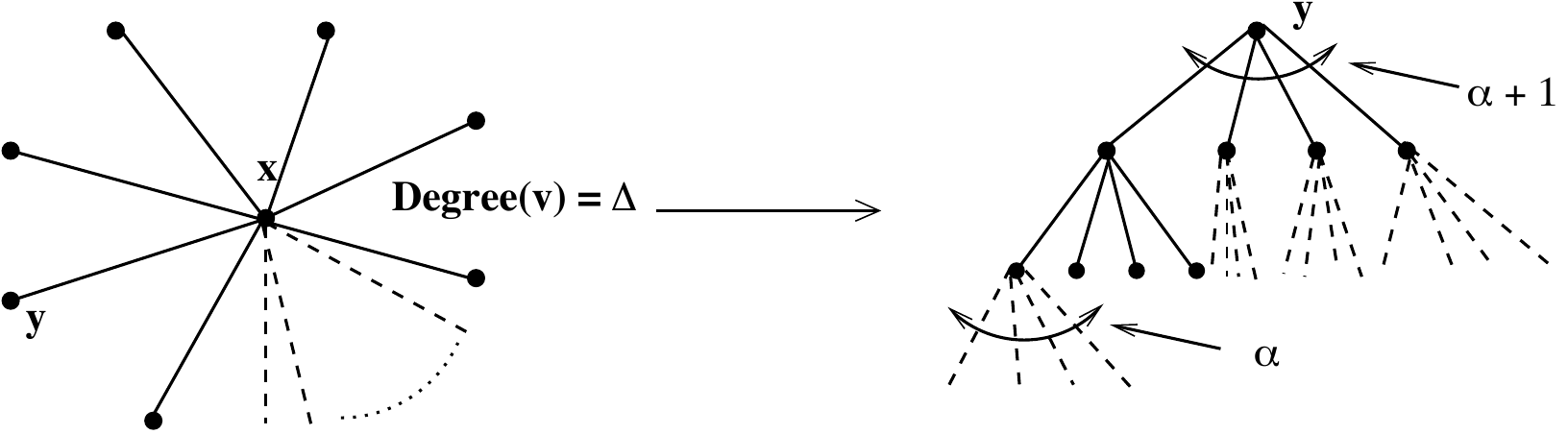}
\caption{Deletion of the central node $v$ of a star leads to an increase in the diameter. Here, the healing algorithm increases the degree of any node by at most $\alpha$.}
\label{fig: lowerboundFT}
\end{figure}

Let $G$ be a star on $\Delta + 1$ vertices, where $x$ is the root node, and $x$ has $\Delta$ edges with each of the
other nodes in the graph.  Let $G'$ be the graph created after the adversary deletes the node $x$.  Consider a breadth
first search tree, $T$, rooted at some arbitrary node $y$ in $G'$.  We know that the self-healing algorithm can increase
the degree of each node by at most $\alpha$, thus the root node in $T$ can have at most $\alpha+1$ children, and other nodes can have at most $\alpha$ children.  Let $h$ be
the height of $T$.  Then we know that $1 + (\alpha + 1) \sum_{i=0}^{h-1}
\alpha^{i} \geq \Delta$.  This implies that $(\alpha+1)^{h+1} \geq \Delta$ for $\alpha \geq 3$, or ${h+1} \geq \log_{\alpha + 1} 
\Delta$.  Since the diameter of $G$ is $2$, we know that $\beta \geq h/2$, and thus $2\beta + 1 \geq \log_{\alpha + 1}
\Delta$. 
% Raising $\alpha$ to the power of both sides, we get that $\alpha^{2\beta + 1} \geq \Delta$. Taking log of both sides
Rearranging, we get $\beta  \geq \frac{1}{2} [\log_{\alpha + 1} \Delta - 1]$. This is illustrated in figure~\ref{fig: lowerboundFT}.
\end{proof}

\medskip
\noindent
We note that this lower-bound compares favorable with the general
result achieved with our data structure.  The Forgiving Tree can be
modified so that it ensures that 1) the degree of any node increases
by no more than $\alpha$ for any $\alpha \geq 3$; and that the
diameter increases by no more than a multiplicative factor of $\beta
\leq 2 \log_{\alpha} \Delta + 2$.

\section{Conclusion}

In this chapter, we have presented a distributed data structure that withstands
repeated adversarial node deletions by adding a small number of new
edges after each deletion.  Our data structure ensures two key
properties, even when up to all nodes in the network have been
deleted.  First, the diameter of the network never increases by more
than $O(\log \Delta)$ times its original diameter, where $\Delta$ is
the maximum original degree of any node.  For many peer-to-peer
systems, $\Delta$ is at most polylogarithmic, and so the diameter
would increase by no more than a $O(\log \log n)$ multiplicative
factor.  Second, no node ever increases its degree by more than $3$
over its original degree.

Several open problems remain.    For example, how do we extend our model and algorithm to  handle insertions of nodes
and multiple deletions? Can we protect other invariants? Can we extend our distributed data structure to ensure that
the stretch
between any pair of nodes increases by no more than a certain amount? Can we design our algorithms so they can work
directly on graphs instead of spanning trees of those graphs?  We have some preliminary positive results 
answering the above questions that build on this work.
We can also consider extending self-healing beyond our present model. For example, Can we design algorithms for less
flexible networks such as sensor networks? 
 Can we extend the concept of self-healing to other objects besides graphs? For example, can we design algorithms to rewire a circuit so that it
maintains its functionality even when multiple gates fail? Can our approach be used  to better understand self-healing
in biological systems such as the human brain?

%% End of chapter 3
%%%%%%%%%%%%%%%%%%%%%%%%%%%%%%%%%%%%%%%%%%%%%%%%%%%%%%%%
%% Chapter  4: Forgiving Graph
%\chapter{Future Work}
%\label{chapter: FW}

\chapter{Forgiving Graph}
\label{chapter: FG}

\begin{epigraphs}
\qitem{\epitext{The weak can never forgive. Forgiveness is the attribute of the strong.}%
 }%
{%
\epiauthor{Mahatma Gandhi.}%
}
%\qitem{ \epitext{Curioser and curioser.}%
%}%
%{%
%\episource{Alice in Wonderland}
%}
\end{epigraphs}

%\begin{abstract}
In this chapter, we present the final of our algorithms discussed in this Dissertation. To recap,  we consider the problem of self-healing in peer-to-peer networks that are under repeated attack by an omniscient adversary. Here, we will assume that, over a sequence of rounds, an adversary either inserts a node with arbitrary connections or
deletes an arbitrary node from the network. The network responds to each such change by quick ``repairs," which consist of adding or deleting a small number of edges.

These repairs essentially preserve closeness of nodes after adversarial deletions, without increasing node degrees by too much, in the following sense.   At any point in the algorithm, nodes $v$ and $w$ whose distance would have been $\ell$ in the graph formed by considering only the adversarial insertions (not the adversarial deletions), will be at distance at most $\ell \log n$ in the actual graph, where $n$ is the total number of vertices seen so far. Similarly, at any point, a node $v$ whose degree would have been $d$ in the graph with adversarial insertions only, will have degree at most $3d$ in the actual graph.  Our distributed data structure, which we call the Forgiving Graph, has low latency and bandwidth requirements.

The Forgiving Graph improves on the Forgiving Tree distributed data structure from Chapter~\ref{chapter: FT}, ~\cite{HayesPODC08},  in the following ways: 1) it ensures low stretch over all pairs of nodes, while the Forgiving Tree only ensures low diameter increase;
2) it handles both node insertions and deletions, while the Forgiving Tree only handles deletions; 3) it does not
require an initialization phase, while the Forgiving Tree initially requires construction of a spanning tree of the
network.

%We present a distributed algorithm that ensures two key properties.  First, the distance between any two nodes in the 
%the network is never more than $O(\log n)$ times their original distance, where $n$ is the number of nodes in the network at the present instance.   Second,
%the degree of any node never increases by more than $3$ times over its original degree.    Our algorithm is simple to setup with only certain additional fields for each
%node, and should have low latency and bandwidth requirement. The algorithm presented here is incomplete. Note that the
%Forgiving Graph will be an improvement over our previous algorithms since it will be able to bind both stretch and
%degree increase. It does not subsume those algorithms however since in some scenarios those algorithms may be more
%efficient.

%\end{abstract}

\section{Introduction}

In Chapter~\ref{chapter: Intro}, we have made case for the need of using \emph{responsive} approaches in reconfigurable networks for maintaining robustness and self-healing in networks.   In this chapter, we describe a distributed data structure for maintaining invariants in a reconfigurable network.  We note that our approach is responsive in the sense that it responds to an attack by changing the network topology.  Thus, it is orthogonal and complementary to traditional non-responsive techniques for ensuring network robustness.

%In this paper, we describe \emph{responsive} algorithms for maintaining desirable properties in reconfigurable networks.  Our approach is responsive in the sense that it responds to an attack (or node failure) by changing the topology of the network.  One benefit of this approach is that it works for any type of initial network topology and is thus orthogonal and complementary to traditional non-responsive techniques.  There are many desirable properties that one might want to maintain on a network in the face of attack.  In this paper, we focus on two fundamental invariants: 1) maintaining low stretch i.e. keeping the distance between any pair of nodes $v$ and $w$ close to their distance in the network without deletions; 2) keeping the degree increase of all nodes small.

This work builds significantly on results achieved in~\cite{HayesPODC08} (Presented in Chapter~\ref{chapter: FT}), which presented a responsive, distributed data structure called the \emph{Forgiving Tree} for maintaining a reconfigurable network in the face of attack. 
Over a complete run of Forgiving Tree: 1) The diameter of the network can never exceed its original diameter by more than a multiplicative factor of $O(\log \Delta)$ where $\Delta$ is the maximum degree in the graph; and 2) the total  increase in the  degree of any node can never be more than $3$.
The Forgiving Tree ensured two invariants: 1) the diameter of the network never increased by more than a multiplicative
factor of $O(\log \Delta)$ where $\Delta$ is the maximum degree in the graph; and 2) the degree of a node never
increased by more than an additive factor of $3$.  

In the following pages, we present a new, improved distributed data structure called the \emph{Forgiving Graph}.  The improvements of the Forgiving Graph over the Forgiving Tree are threefold.  First, the Forgiving Graph maintains  low stretch i.e. it ensures that the distance between any pair of nodes $v$ and $w$ is close to what their distance would be even if there were no node deletions.  It ensures this property even while keeping the degree increase of all nodes no more than a multiplicative factor of $3$.  Moreover, we show that this tradeoff between stretch and degree increase is asymptotically optimal.  Second, the Forgiving Graph handles both adversarial insertions and deletions, while the Forgiving Tree could only handle adversarial deletions (and no type of insertion).  Finally, the Forgiving Graph does not require an initialization phase, while the Forgiving Tree required an initialization phase which involved sending $O(n \log n)$ messages, where $n$ was the number of nodes initially in the network, and had a latency equal to the initial diameter of the network.  Additionally, the Forgiving Graph is divergent technically from the Forgiving Tree, it makes significant use of a novel distributed data structure that we call a Half-full Tree or ``haft''. $\haft$s are discussed in Section~\ref{sec: hafts}. Our main algorithm is described in Section~\ref{sec: FGalgorithm} and Section~\ref{sec: FGdetail}.

\medskip
\noindent {\bf Our Model:} We remind the reader about the model we have been using in this work. We assume that the network is initially a connected graph over $n$ nodes.  An adversary repeatedly attacks the network. This adversary knows the network topology and our algorithm, and it has the ability to delete arbitrary nodes from the  network or insert a new node in the system which it can connect to any subset of the nodes currently in the system.   However, we assume the adversary is constrained in that in any time step it can only delete or insert a single node.
The detailed model is described in Section~\ref{sec: FGmodel}.

\medskip
\noindent {\bf Our Results:}  For a peer-to-peer network that has both insertions and deletions, let $G'$ be the graph consisting of the original nodes and inserted nodes without any changes due to deletions. Let $n$ be the number of
nodes in $G'$. The Forgiving Graph ensures that: 1) the distance between any two nodes of the actual network never increases by more than $\log n$ times their distance in $G'$; and 2) the degree of any node in the actual network never increases by more than $3$ times its degree in $G'$.  Our algorithm is completely distributed and resource efficient.  Specifically, after deletion, repair takes
$O(\log d\log n)$ time and requires sending $O(d\log n)$ messages, each of size $O(\log n)$ where $d$ is the degree
of the node that was deleted. The formal statement and proof of these results is in Section~\ref{subsec: upperbounds}. 

\medskip
\noindent {\bf Related Work:} 
Our work significantly builds on work in~\cite{HayesPODC08} as described above.   Our model of attack and repair builds on earlier work in~\cite{BomanSAS06, SaiaTrehanIPDPS08} (The later is presented in Chapter~\ref{chapter: DASH}).

%which proposed a simple line algorithm for self-healing to maintain
%network connectivity.

%Using negative hspace below to preserve right margin
\section{\hspace*{-10pt}Node Insert, Delete and Network Repair Model}
\label{sec: FGmodel}
 
 We now describe the details of our node insert, delete and network repair model.  Let $G = G_0$ be an arbitrary graph on $n$ nodes, which represent processors in a distributed network.  In each step, the adversary either deletes or adds a node.  After each deletion, the algorithm gets to add some new edges to the graph, as well as deleting old ones.  At each insertion, the processors follow a protocol to update their information.
The algorithm's goal is to maintain connectivity in the network, keeping the distance between the nodes small.  At the same time, the algorithm wants to minimize the resources spent on this task, especially keeping node degree small.  

Initially, each processor only knows its neighbors in $G_0$, and is unaware of the structure of the rest of $G_0$.
After each deletion or insertion, only the neighbors of the deleted or inserted vertex are informed that
the deletion or insertion has occurred. After this, processors are allowed to communicate by sending a limited number
of messages to their direct  neighbors.  We assume that these messages are always sent and received successfully.  The
processors may also request new edges be added to the graph. The only synchronicity assumption we make is that no
other  vertex is deleted or inserted until the end of this round of computation and communication has concluded.
To make this assumption more reasonable, the per-node communication cost should be very small in $n$ (e.g. at most logarithmic).

We also allow a certain amount of pre-processing to be done before the first attack occurs.  This may, for instance,
be used by the processors to gather some topological information about $G_0$, or perhaps to 
coordinate a strategy.  Another success metric is the amount of computation and communication needed during this
preprocessing round.  Our full model is described in Figure~\ref{algo:model-2}.

For our success metrics, at any time $T$, we compare the actual graph $G_T$ to the graph $G'_{T}$ which is the graph with only the original nodes (those at $G_0$) and insertions without regard to deletions and healing. This is the graph which would have been present if the adversary was not doing any deletions and (thus) no self-healing algorithm was active. This is the natural graph for comparing results. Notice if there were no insertions happening in our model, we could have compared $G_T$ to $G_0$ but since insertions are happening, $G_T$ may not even have the same nodes as $G_0$ rendering a node-based comparison impossible. Figure~\ref{fig: ComparisonGraphs} shows an example of $G'_T$ and a corresponding $G_T$. The figure also shows, in  $G'_T$,  the nodes and edges inserted  and deleted,  and in $G_T$, the edges inserted by the healing algorithm, in different colors, as the network evolved over time. Figure~\ref{fig: GraphCompDegrees} shows how the two graphs compare with regards to degree of a particular node $v$, and figure~\ref{fig: GraphCompStretch} shows how the healing algorithm effects the distance between two nodes, $u$ and $v$. Our algorithm gaurantees our invariants on the 'complexity' measures at every time step that the algorithms is in execution.

%% Manually position the float below by having the previous paragraph before or after it%%

%\floatname{algorithm}{Model}
%\begin{algorithm}[h!]
\begin{figure}[t]
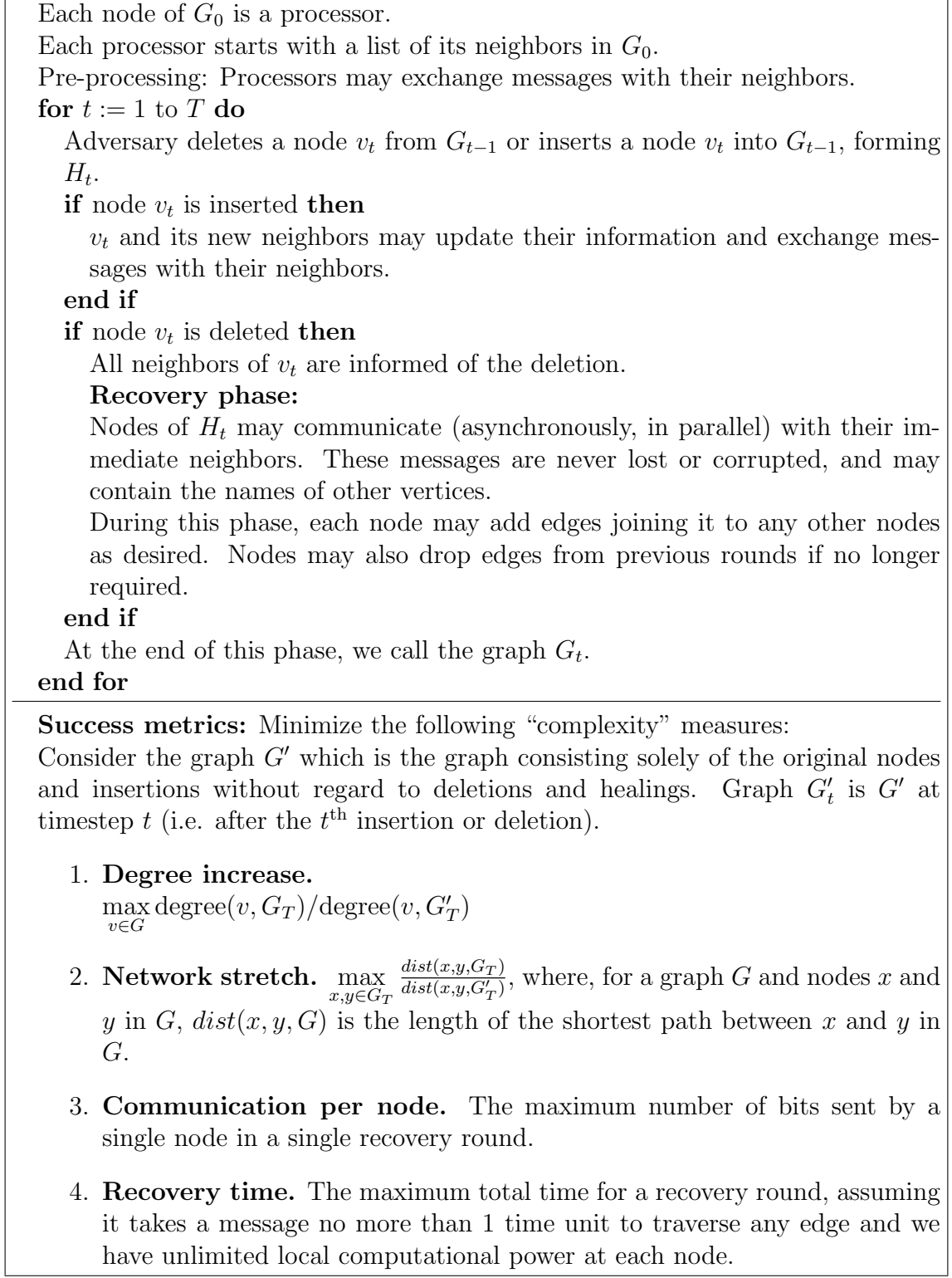

\caption{The Node Insert, Delete and Network Repair Model -- Distributed View.}
\label{algo:model-2}
\begin{boxedminipage}{\textwidth}
\begin{algorithmic}
\STATE Each node of $G_0$ is a processor.  
\STATE Each processor starts with a list of its neighbors in $G_0$.
\STATE Pre-processing: Processors may exchange messages with their neighbors.
%send messages to and from their neighbors.
\FOR {$t := 1$ to $T$}
\STATE Adversary deletes a node $v_t$ from $G_{t-1}$ or inserts a node $v_t$ into $G_{t-1}$, forming $H_t$.
\IF{node $v_t$ is inserted} 
\STATE  $v_t$ and its new neighbors may update their information and exchange messages with their neighbors.
\ENDIF
\IF{node $v_t$ is deleted} 
\STATE All neighbors of $v_t$ are informed of the deletion.
\STATE {\bf Recovery phase:}
\STATE Nodes of $H_t$ may communicate (asynchronously, in parallel) 
with their immediate neighbors.  These messages are never lost or
corrupted, and may contain the names of other vertices.
\STATE During this phase, each node may add edges
joining it to any other nodes as desired. 
Nodes may also drop edges from previous rounds if no longer required.
\ENDIF
\STATE At the end of this phase, we call the graph $G_t$.
\ENDFOR
\vspace{5pt}
\hrule
\STATE
\STATE {\bf Success metrics:} Minimize the following ``complexity'' measures:\\
Consider the graph  $G'$ which is the graph consisting solely of the original nodes and insertions without regard to
deletions and healings. Graph $G'_{t}$ is $G'$ at timestep $t$ (i.e. after the $t^{\mathrm{th}}$ insertion or deletion).
%Graph $G'_{t}$ is $G'$ at timestep $t$ which is equivalent to $G'_{t'}$ where the $t' \le t$ is
%the timestep at which the latest insertion on or before $t$ occurred. 
 \begin{enumerate}
\item{\bf Degree increase.} \\
  $\max_{v \in G} \Degree(v,G_T) / \Degree(v,G'_T)$
%\item {\bf Network stretch.} $\max \left( (x,y) \in G_{t}, G_{t'}; t, t' <T, \distance_{t'}(x,y) / \distance_{t}(x,y)
%\right)$
\item {\bf Network stretch.} $\max_{x, y \in G_{T}} \frac{dist(x,y,G_{T})}{dist(x,y,G'_{T})}$, where, for a graph $G$ and nodes $x$ and $y$ in $G$, $dist(x,y,G)$ is the
length of the shortest path between $x$ and $y$ in $G$.
%For any pair of nodes $x$ and $y$, $ \distance(x,y,G_{T}) / \distance(x,y,G'_t)$
\item{\bf Communication per node.} The maximum number of bits sent by a single node in a single recovery round.
% \tom{Want to modify this or omit?}
\item{\bf Recovery time.} The maximum total time for a recovery round,
assuming it takes a message no more than $1$ time unit to traverse any edge and we have unlimited local computational power at each node.
\end{enumerate}
\end{algorithmic}
\end{boxedminipage}
\end{figure}
%\end{algorithm}

\clearpage

\begin{figure}[h!]
\centering
\subfigure[$G'_T$: Nodes in red (dark gray in grayscale) deleted, and nodes in green (patterned) inserted, by the adversary.]{ \label{sfig: GraphOrigInserts}
\includegraphics[scale=0.6]{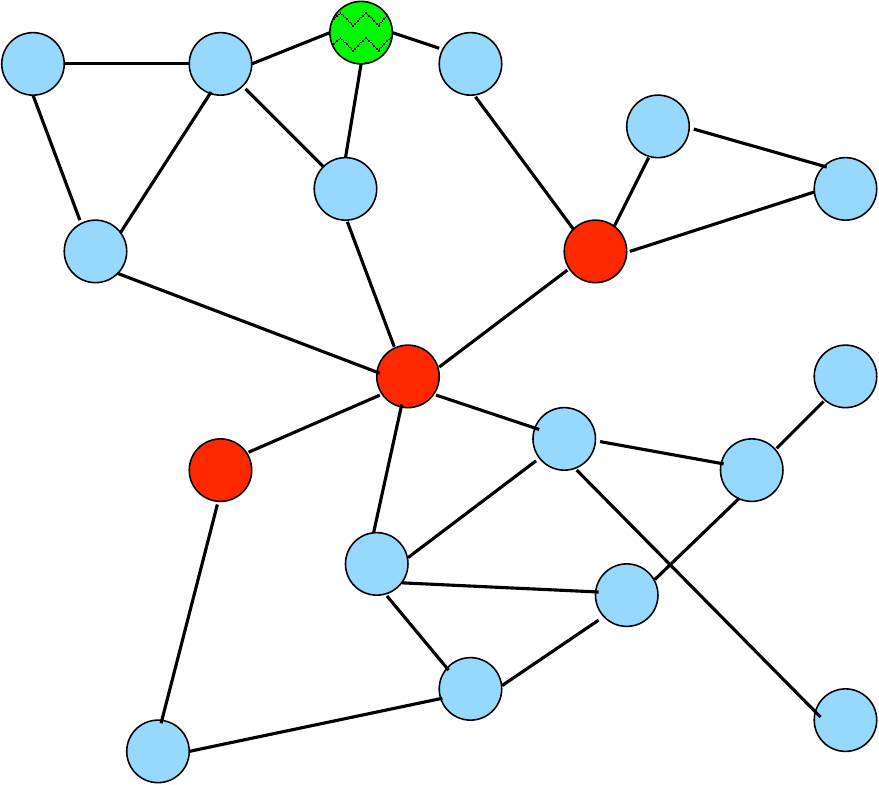} }
\hspace{30pt}
\subfigure[$G_T$: The actual graph. Edges added by the healing algorithm shown in gold (light shaded in grayscale) color.]
%and Every haft is a union of complete binary trees. In our notation, $T_a$ is a complete binary
%tree and $|T_{a}|$ is the number of leaf nodes in $T_{a}$. ]
%{ \label{sfig: haft-as-join} \includegraphics[scale=0.9]{images/FG/Haft} }
{ \label{sfig: GraphHealed} \includegraphics[scale=0.6]{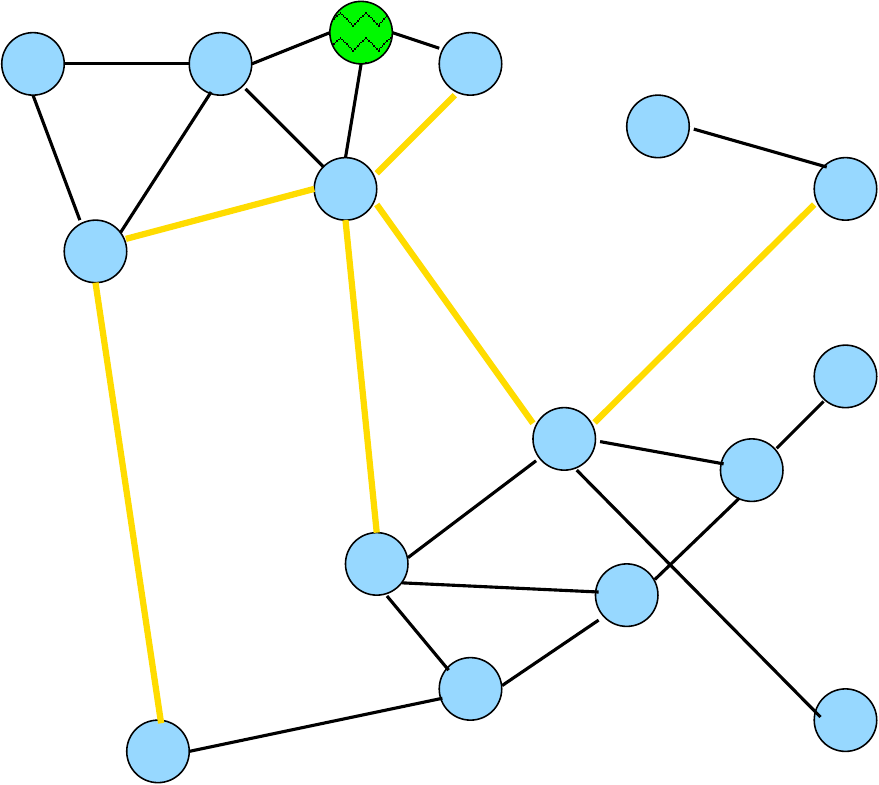} }
\caption{ \emph{Graphs at time T}. $G'_T$: The graph of initial nodes and insertions over time, $G_T$: The actual healed graph.}
\label{fig: ComparisonGraphs}
\end{figure}

\begin{figure}[h!]
\centering
%\subfigure[$\FGraph$ ensures degree of no node increases by more than a factor of 3.]{ \label{sfig: GraphCompDegrees}
%\includegraphics[scale = 0.5]{images/FG/Haft7examples} }
\includegraphics[scale=0.6]{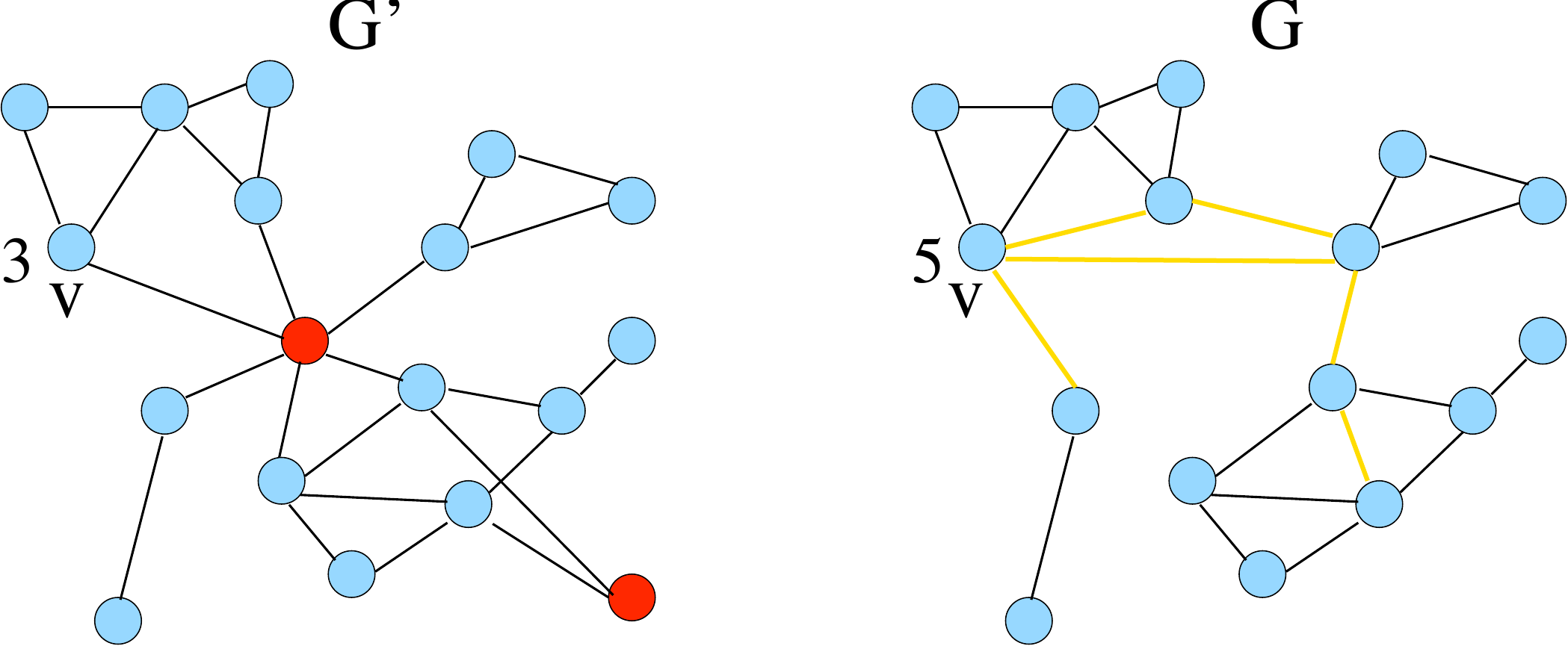} 
%\hspace{4pt}
%\label{fig: ResultsDegreeStretch}
\caption{Comparing degrees: In the figure the degree of node $v$ in graph of only original and inserted nodes is 3, and in the actual healed network it is 5. The nodes in red (dark gray in grayscale) were deleted by the adversary and the golden (light shaded) edges were the ones added by the healing algorithm.}
\label{fig: GraphCompDegrees}
\end{figure}

\begin{figure}[h!]
\centering
\includegraphics[scale=0.6]{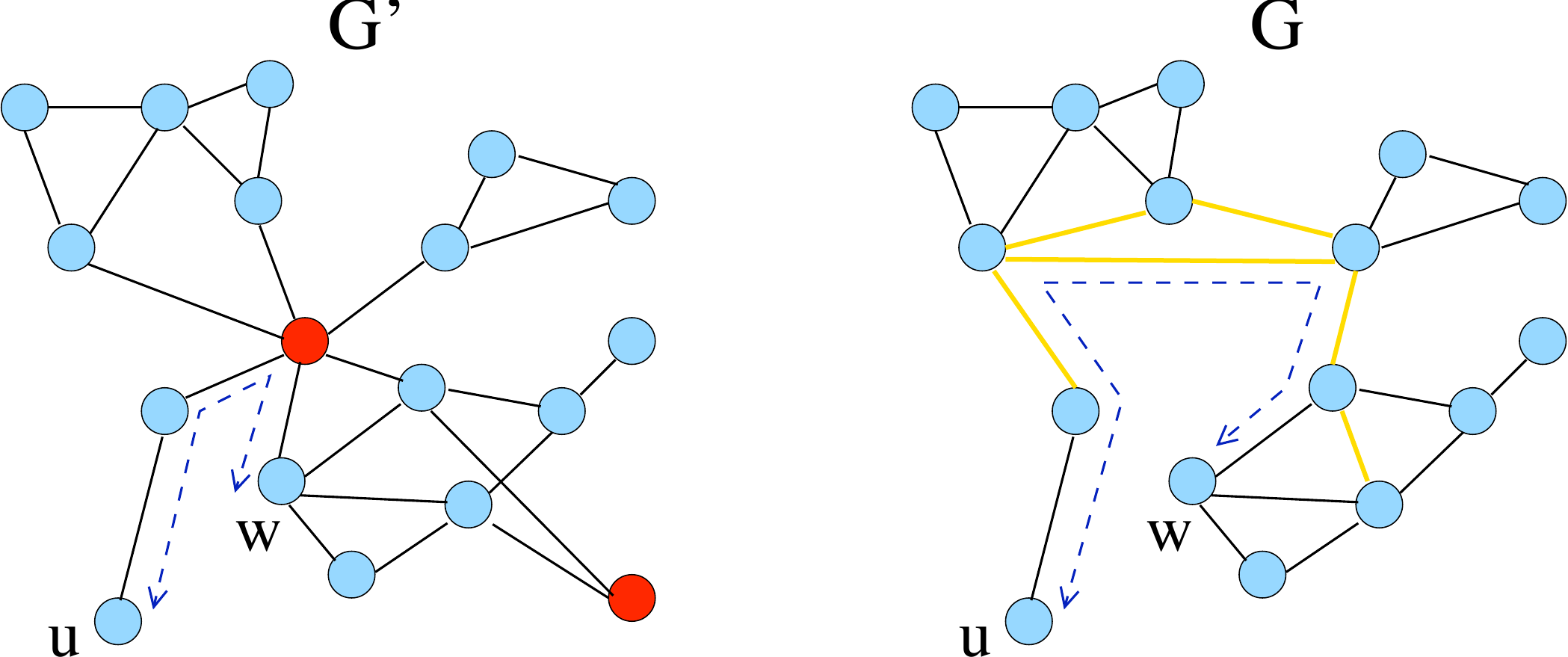} 
%\hspace{4pt}
%\label{fig: ResultsDegreeStretch}
\caption{Comparing distances: In the figure nodes $u$ and $w$ have their distance increased to 5 in the actual healed network compared to their distance of $3$ in the graph of only original and inserted nodes. The nodes in red (darker in grayscale) were deleted by the adversary and the golden edges (lighter shade) are the ones added by the healing algorithm}
\label{fig: GraphCompStretch}
\end{figure}

%\pagebreak

\section{The Forgiving Graph algorithm}
\label{sec: FGalgorithm}

\begin{figure}[h!]
\centering
\includegraphics[scale=0.55]{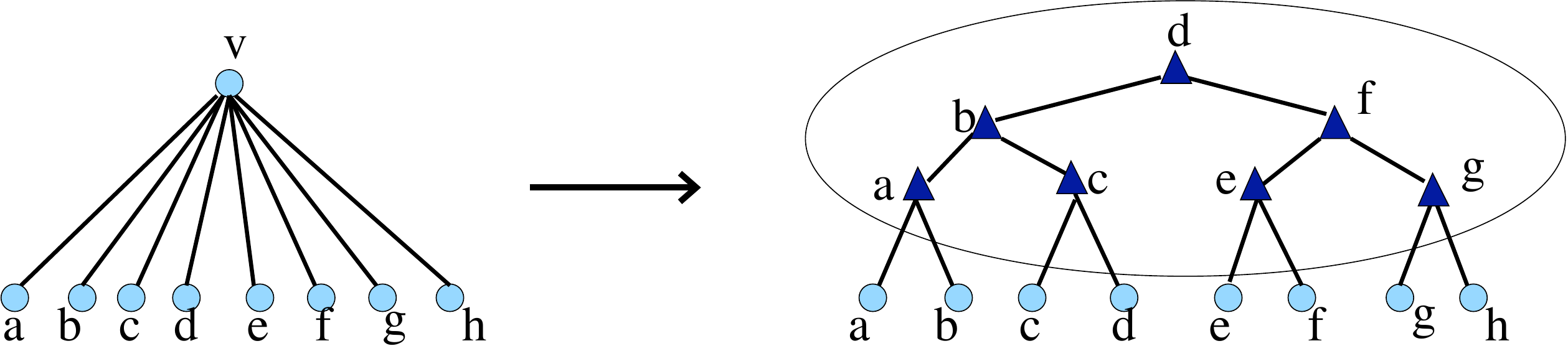}
\caption{Deleted node $v$ replaced by its Reconstruction Tree. The  triangle shaped nodes are 'virtual' helper nodes simulated by the 'real' nodes which are in the leaf layer.}
 \label{fig: RT}
\end{figure}

Here, we give a high level description of our algorithm.
An adversary can effect the network in one of two ways: inserting a new node in the network or deleting an
existing node from the network. Node insertion is straightforward and is dependent on the specific policies of the
network. When an insertion happens, our incoming node and its neighbors update the data structures that are
used by our algorithm. We will also assume that nodes maintain some neighbor-of-neighbor information. There are many
ways to maintain neighbor of neighbor information~\cite{MNW, Naor04knowthy}.  Maintaining neighbor of neighbor information requires regular updates, and may be used for 
other purposes such as routing, thus, we do not explicitly include this maintenance cost in our analysis.

 %This information is regularly updated for
%other purposes such as routing information, thus, we do not explicitly include this cost in our analysis.

Each time a node $v$ is deleted, we can think of it as being replaced by a Reconstruction Tree ($\RT(v)$, for short) which is a  \textit{haft} (defined in  Section~\ref{sec: hafts}) having ``virtual'' nodes as internal nodes and
neighbors of $v$ (which we call real nodes) as the leaf nodes. Note that each virtual node has a degree of at most  $3$. A single real node itself is a trivial $\RT$ with one node.
 $\RT(v)$ is formed by merging all the neighboring $\RT$s of
$v$ using the strip and merge operations from Section~\ref{sec: hafts}. Thus, following a deletion, we may have a graph with both real and virtual nodes.  After a long sequence of such insertions and deletions, this  graph is a patchwork mix of virtual nodes and real nodes. Let us call this graph $\FG$ (short for $\FGraph$). As for the other graphs, $FG_T$ is the graph $\FG$ at time $T$.

Also, because the virtual trees (hafts) are balanced binary trees, the deletion of a node $v$ can, at worst, cause the
distances between its neighbors to increase from $2$ to $2 \lceil \log d  \rceil$ by traveling through its $\RT$, where 
$d$ is the degree of $v$ in $\G'$ (the graph consisting solely of the original nodes and insertions without regard to
deletions and healings). 
However, since this deletion may cause many $\RT$s to merge and the new $\RT$ formed may
involve all the nodes in the graph, the distances between any pair of actual surviving nodes may increase by no
 more than a $\lceil \log n \rceil$ factor.
 
Since our algorithm is only allowed to add edges and not nodes, 
we cannot really add these virtual nodes to the network.
We get around this by assigning each virtual node to an actual
node, and adding new edges between actual nodes in order to 
allow ``simulation'' of each virtual node.  More precisely,
our actual graph is the homomorphic image of the graph
described above, under a graph homomorphism which fixes 
the actual nodes in the graph and maps each virtual node
to a distinct actual node which is ``simulating'' it. Figure~\ref{fig: homomorphism} shows this homomorphism where the graph $\FG$  is mapped to the graph $\G$.  We discuss this homomorphism and its relationship to our results in more detail in Section~\ref{sec: Results} .

Note that, because each actual node simulates at most one
virtual node for each of its deleted neighbors, and virtual nodes have degree at most $3$,
this ensures that the maximum degree increase of our algorithm
is at most $3$ times the node's degree in $\G'$. 

%\pagebreak

\section{Half-full Trees (``HAFTS'')}
\label{sec: hafts}

\epigraph{\epitext{Is the glass half full, or half empty? It depends on whether you're pouring, or drinking.} }%
{\epiauthor{Bill Cosby}}

\begin{figure}[h!]
\centering
\subfigure[The first seven hafts. The nodes marked by a circle are the primary roots, and those in boxes are
the spine nodes.]{ \label{sfig: haftexample}
\includegraphics[scale=0.9]{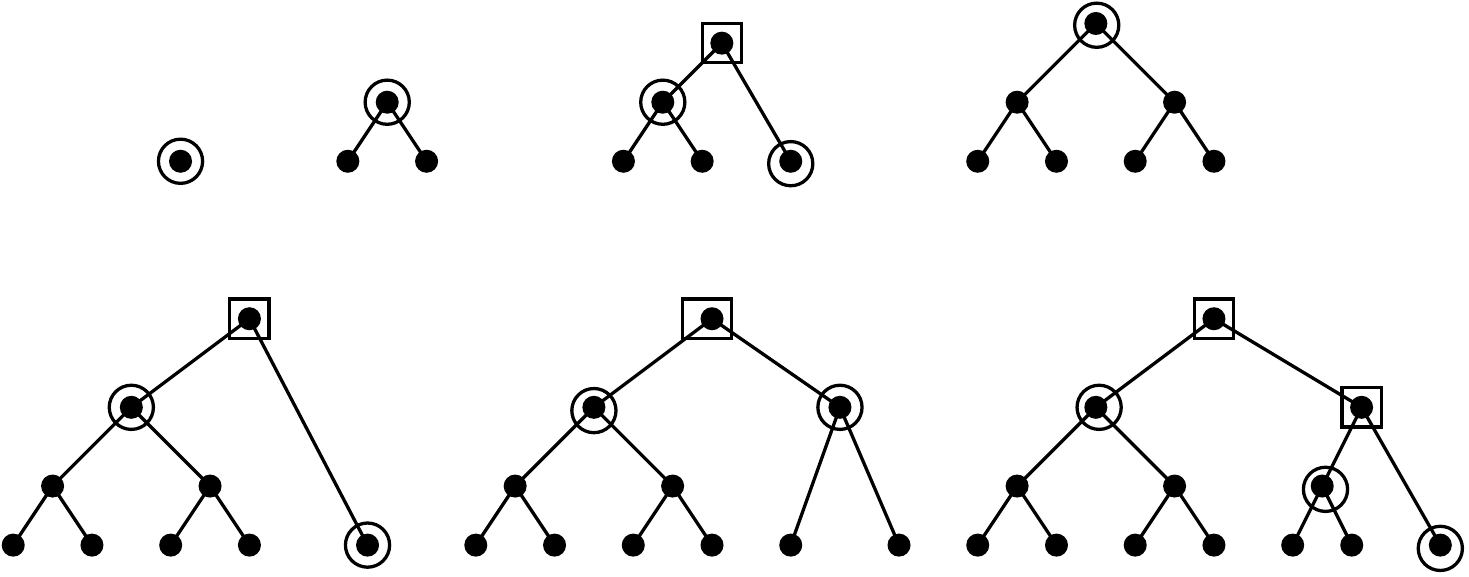} }
\hspace{0.4in}
\subfigure[
Structure of a haft. Each $T_i$ is a complete binary tree, 
with 
$|T_1| > |T_{2}| > \dots > |T_{k}|$. The spine nodes are the nodes in red (darker in grayscale). The left child of each spine node, and  the right child of the rightmost spine node are the primary roots, shown in green (lighter in grayscale).]
%and Every haft is a union of complete binary trees. In our notation, $T_a$ is a complete binary
%tree and $|T_{a}|$ is the number of leaf nodes in $T_{a}$. ]
%{ \label{sfig: haft-as-join} \includegraphics[scale=0.9]{images/FG/Haft} }
%{ \label{sfig: haft-as-join} \includegraphics{images/FG/Haft} }
{ \makebox[0.9\textwidth][c]{\includegraphics{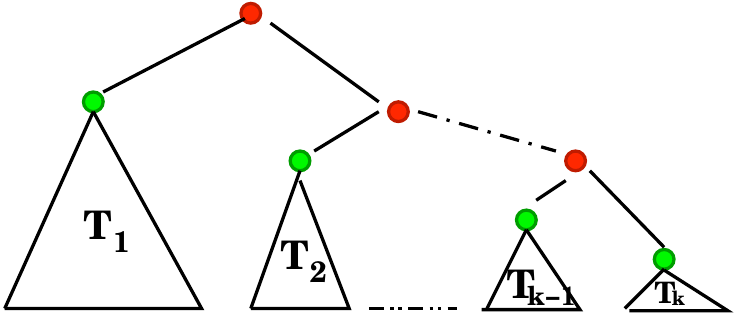}  \label{sfig: haft-as-join} }
}
\caption{haft (half-full tree)}
\end{figure}

 In this section, we define half-full trees (or hafts, for short), 
and describe their most important properties for our present application.
This type of tree has been studied before, by Vaucher~\cite{Vaucher-Essays04}, who called them ``staircase trees.''  However, our presentation will be
self-contained.
% lemmas:

%- Any half-full binary tree can be decomposed into full subtrees
% such that for any x a power of 2, there is at most one subtree of x leaf nodes in this decomposition.

%- In merging any set of \ell half-full binary trees, the number of new nodes that most be added to perform this merge
% is at most \ell - 1.
%- Q1: Can we get O(1) latency without directed links and keeping the will sizes polylog?
%(idea: use other types of data structures to achieve this)

\begin{description}
\item[Half-full tree:] A \emph{half-full tree}, or \emph{haft}, is a rooted binary tree in which every non-leaf node $v$
has the following properties:
\begin{itemize}
 \item $v$ has exactly two children.
 \item The left child of $v$ is the root of a complete binary subtree
that contains at least half of $v$'s  descendants.
\end{itemize}
%\end{description}
%\begin{description}
\pagebreak
\item[Primary root:] A \emph{primary root} is a node in a $\haft$ such that:
\begin{itemize}
\item It is the root of a complete subtree.
\item Its parent, if it has one, is not the root of a complete subtree.
\end{itemize}
%\end{description}
%\begin{description}
\item[Spine:] A \emph{spine node} is the parent of a primary root. 
Equivalently, it is a node in a $\haft$ which is not the root
of a complete subtree.
The \emph{spine} of a $\haft$ is the set of
all spine nodes.  We observe that, if non-empty, 
the spine consists of the vertices of a path,
with the root of the $\haft$ as one endpoint.
\end{description}

\noindent Figure \ref{sfig: haftexample} shows several examples of hafts. 
%For any positive $\ell$, there is a unique 
%haft having $\ell$ leaves (see lemma~\ref{lemma: hftproperties}), which% we refer to as as $\haft(\ell)$
We now give a simple structural lemma which completely characterizes 
any haft as a function of the number of its leaves.  This will
be useful later when we wish to perform merging operations on the
hafts used by our algorithm.

\begin{lemma}[Binary representation of Hafts]
\label{lemma: hftproperties}
Let $\ell$ be a positive integer. Then there is a unique haft $T$
having $\ell$ leaves.
Moreover, let $h$ be the number of ones in the binary representation of $\ell$, and suppose
$x_1 > x_2 > \dots > x_h$ are the indices of these ones, so that
 \[
\ell = \sum 2^{x_{i}} .
 \]
Then either
\begin{itemize}
\item $h=1$, and $T$ is a complete tree of depth $x_{1}$, or
\item $h \ge 2$, and $T$ consists of $h-1$ spine nodes $s_{1}, \dots s_{h-1}$,
together with $h$ complete binary trees $T_{1},\dots,T_{h}$, where
\begin{itemize}
\item $s_{1}$ is the root of $T$,
\item each $T_{i}$ has depth $x_{i}$,
\item each $s_{i}$ has the root of $T_{i}$ as its left child
\item for $1 \le i \le h-2$, $s_{i}$ has $s_{i+1}$ as its right child
\item $s_{h-1}$ has the root of $T_{h}$ as its right child
\end{itemize}
\end{itemize}
\end{lemma}

\begin{corollary}
Let $T$ be a haft having $\ell$ leaves.  Then the depth of $T$ equals
$\lceil \log \ell \rceil$.
\end{corollary}

\begin{proof}[Proof of Lemma~\ref{lemma: hftproperties}]
We will prove the detailed structure of $T$, from which the 
uniqueness is apparent.

First, consider the case $h=1$ (\emph{i.e.,} $\ell$ is a power of $2$).
If $\ell = 1$, there is nothing to prove.  Assume $\ell > 1$.
Now the left subtree of $T$ is complete, and hence has number of
leaves equal to a power of two.  Since at least half of the leaves
are on the left subtree, this power of two is at least $\ell/2$.
Since the root of $T$ has two children, not all of the leaves
are on the left subtree, and hence there are exactly $\ell/2$
leaves on the left subtree, and thus also $\ell/2$ leaves on the
right subtree.  Since it is immediate from the definition
that any subtree of a haft is also a haft, it follows by induction
on $\ell$ (being a power of two) 
that the right subtree is also a complete subtree.
Thus, $T$ is complete.

Now, suppose $h \ge 2$.  Let us denote the root of $T$ by $s_1$.
Because $\ell$ is not a power of two, $s_1$ must be
a spine node.
Since the left subtree, $T_1$, is complete and contains between
$\ell/2$ and $\ell$ leaves, it must have depth $x_1$.
Since the right subtree is a haft having number of leaves equal to
\[
\ell - 2^{x_1} = \sum_{i=2}^h 2^{x_i}
\]
it follows by induction on $\ell$ (being any positive integer) 
that it has the claimed structure.  Thus, $T$ is also as claimed.
\end{proof}

\subsection{Operations on Hafts}
We Define the following operations on hafts:
\begin{enumerate}
%\item \emph{Strip}: $\HFT  \rightarrow \{\FT\ |\  \FT\ \textrm{is a full tree}  \}$
\item \emph{Strip}: Suppose $T$ is a haft with $h$ ones in its binary representation. The Strip operation removes 
$h-1$ nodes from $T$ returning  a forest of $h$ complete trees. 

%\item \emph{Merge}: $\HFT_{1} + \HFT_{2} + \textrm {new node}  \rightarrow \HFT_{3}  $\\
%$\HFT_{1} + \HFT_{2} + \cdots + \HFT_{l} + (l-1)\textrm{new nodes} \rightarrow \HFT_{l+1}  $
\item \emph{Merge}: The Merge operation joins  hafts together using additional isolated single nodes,
to create a single new haft.
%\item \emph{Delete}: $\HFT - \textrm{leaf node} \rightarrow \HFT$
%\item \emph{Delete}:  The Delete operation deletes a leaf node from a haft, and does further operations so as to
%return another haft.
\end{enumerate}

\pagebreak

\noindent We now describe these operations in more detail:

\subsubsection{Strip}
\label{subsec: haftstrip}
% The signature of operation $\Strip$ can be rewritten as  $\HFT  \rightarrow  F$ where $F$ is a forest of full
%trees. Thus 
By Lemma~\ref{lemma: hftproperties}, if we remove the spine from a
haft, $T$, we are left with a forest of $h$ complete binary trees,
where $h$ is the number of ones in the binary representation of
the number of leaves of $T$.
The operation $\Strip(T)$ returns this forest.

%\tom{Shorten this description?}
The $\Strip$ operation works as follows: 
 If  $T$ is a complete tree, then return $T$ itself. Note that the root of the $T$ is the only primary root in this
case. If $T$ is not a complete tree, then $F$ is obtained as follows. Starting from the root of $T$, traverse the 
direct path towards the rightmost leaf of $T$. Remove a node if it is not a primary root. Stop when a primary root or a
leaf node (which is a primary root too) is discovered. 
In figure \ref{sfig: haft-as-join} the $\Strip$ operation removes the nodes indicated by the square boxes. \\
%% Write something different below.
We now give intuition as to why the Strip operation works.
\begin{lemma}
The Strip operation returns the subtrees rooted at all primary roots in the input $\haft$.
\end{lemma}
\begin{proof}
By the definitions of  $\haft$ and primary root, if a vertex is not the root of a complete subtree, its left
child is guaranteed to be a primary root. Thus, either the root of the $\haft$ is a primary root or its left child is. If
the left child is a primary root, there can be no other primary root in the left subtree, so we we return the tree rooted at that
child.  Recursively applying the same test to the right child, we get all the primary roots. 
\end{proof}

\begin{figure}[h!]
\centering
\includegraphics[scale=0.8]{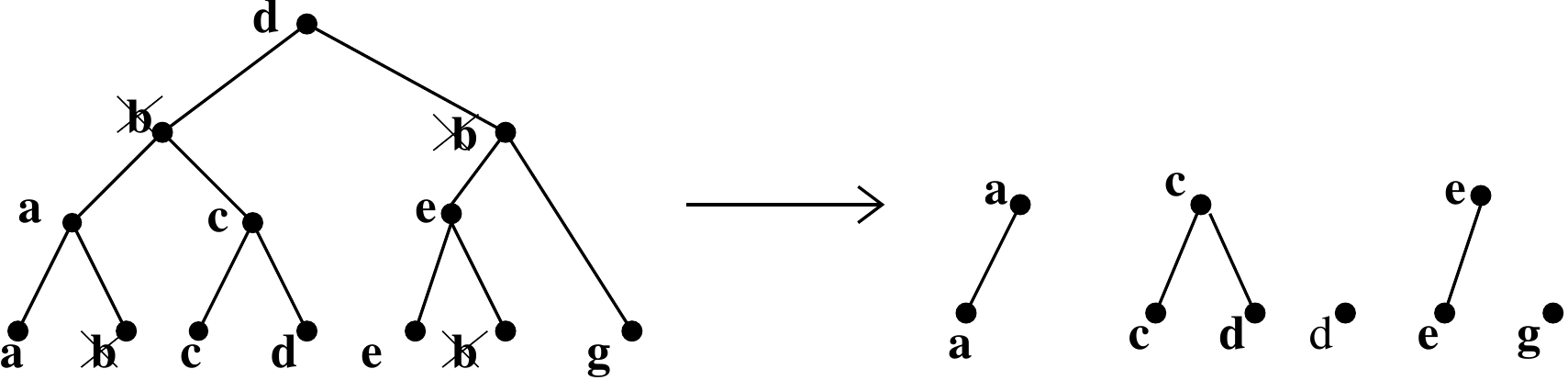}
\caption{Deletion of a node and its helper nodes lead to breakup of RT into components. The Strip operation or a simple variant (for non-hafts) returns a set of complete trees, which can then be merged.} 
\label{fig: HaftStripMerge}
\end{figure}

\subsubsection{Merge}
\label{subsec: haftmerge}
 By Lemma~\ref{lemma: hftproperties}, every $\haft$ is completely characterized by its number of leaves.
Merging $\haft$s is analogous to binary addition of these numbers. 
The new binary number obtained is the number of leaves in the
haft produced by the Merge operation.
This is illustrated in figure~\ref{fig: haftmergebinary}.\\
 The first step of the $\Merge$ operation is to apply the $\Strip$  operation on the input trees. This gives a forest of
complete trees.  These complete trees can be recombined with the help of extra nodes to obtain a new $\haft$. 
%Later in the section in lemma \ref{lemma: hftmerge} we give the number of such extra nodes needed. 
Let $Size(X)$ be the number of nodes in a tree $X$. Consider two complete trees
$T_1$ and $T_2$ (Size($T_1) > Size(T_2$)), with roots $r_{1}$ and $r_{2}$ respectively, and an extra node $v$. To merge
these trees, make $r_{1}$ the left child and $r_{2}$ the right child of $v$ by adding edges between them. The merged
tree is always a $\haft$.
 Thus, the merge operation $\Merge(\haft_1,\haft_2, \ldots)$ is as follows:\\

\begin{enumerate}
\item Apply $\Strip$ to all the hafts to get a forest of  complete trees.
%\item  Till there is a single tree $\HFT_{3}$:
 %\begin{enumerate}
  % \item Take 2 trees $T_{i}$ and $T_{j}$, where $|T_{i} - T_{j}|$ = minimum ($|T_{i} - T_{j}|$) for all $i,j$. Let $v$
%be an extra node. If $T_{i}$ and $T_{j}$ are both full trees, join $T_{j}$ to $T_{i}$ by making $v$ the parent of $r_{i}$ and $r_{j}$ and making the corresponding
%edges returning $T_{i}$, else $T_{i} \gets \Merge (T_{i}, T_{j}, v)$.
   %\end{enumerate}
  \item Let $T_{1} , T_{2}, \ldots, T_{k}$ be the $k$ complete trees sorted in ascending order of their
size. Traverse the list from the left, let $T_{i}$ and $T_{i+1}$ be the first two adjacent trees of the same size and $v$ be a
single isolated vertex, join $T_{i}$ and $T_{i+1}$ by making $v$ the parent of the root of $T_{i}$ and the root of
$T_{i+1}$, to give a new tree. Reinsert this tree in the correct place in the sorted list. Continue traversal of the
list from the position of the last merge, joining pairs of trees of equal sizes. At the end of this traversal, we are
left with a sorted list of complete trees, all of different sizes.
\item Let $T_1, T_2, \ldots,  T_l$   be the sorted list of complete trees obtained after the previous step. Traverse
the list from left to right, joining adjacent trees using single isolated vertices.  Let $w$ be a single isolated
vertex. Join $T_1$ and $T_2$ by making the root of $T_2$ the left child and the root of $T_{1}$ the right child of
$w$, respectively. This gives a new haft. Join this haft and $T_3$  by using another available isolated
vertex, making the larger tree ($T_{3}$)  its left child. Continue this process till there is a single haft.
\end{enumerate}

\begin{figure}[h!]
\centering
\includegraphics[scale=0.7]{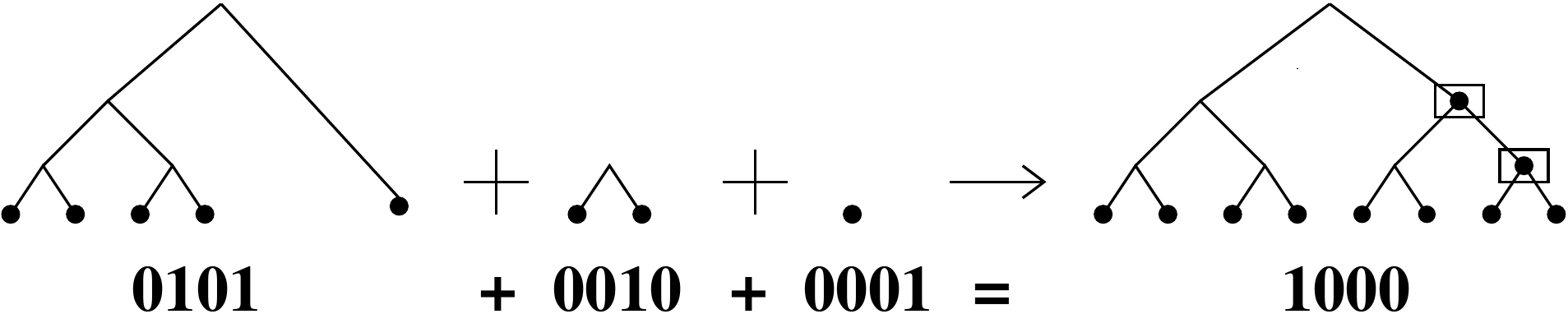}
\caption{Merging three hafts. The vertices in the square boxes are the new isolated vertices used to join the complete
The square shaped vertices are the isolated vertices used to join the complete trees. Merging is analogous to binary
number addition, where the number of leaves are represented as binary numbers.}
\label{fig: haftmergebinary}
\end{figure}

%\section{FG: Detailed description} 
\section{FG: Distributed implementation} 

\label{sec: FGdetail}

 \begin{figure}[h!]
\centering
\subfigure[The original graph. Node $v$ attacked.]{
\makebox[0.45\textwidth][c]{\includegraphics[scale=0.5]{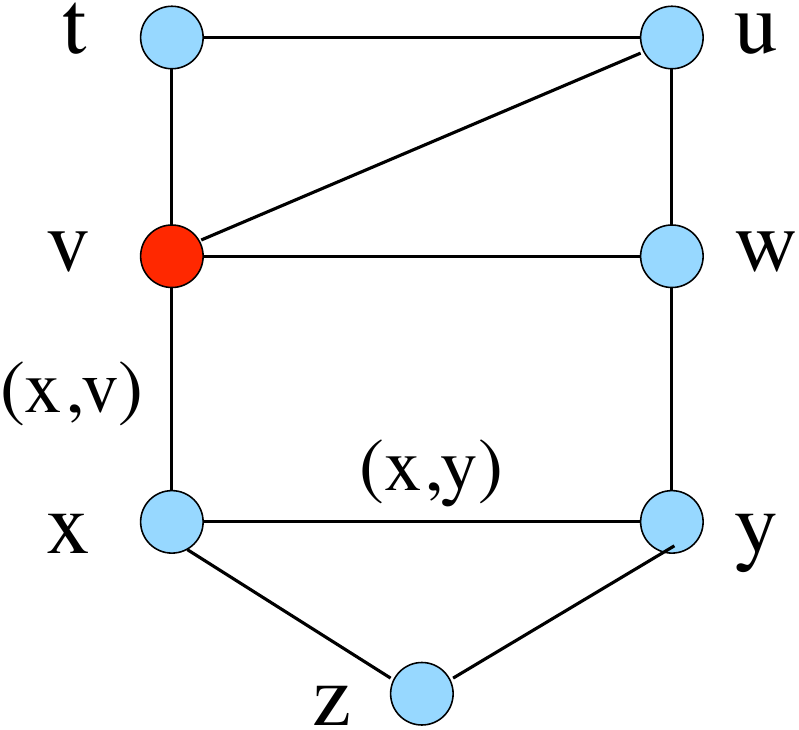}}
 } 
\subfigure[Healed graph. The new nodes inside ellipse are helper nodes.]{
\makebox[0.45\textwidth][c]{\includegraphics[scale=0.5]{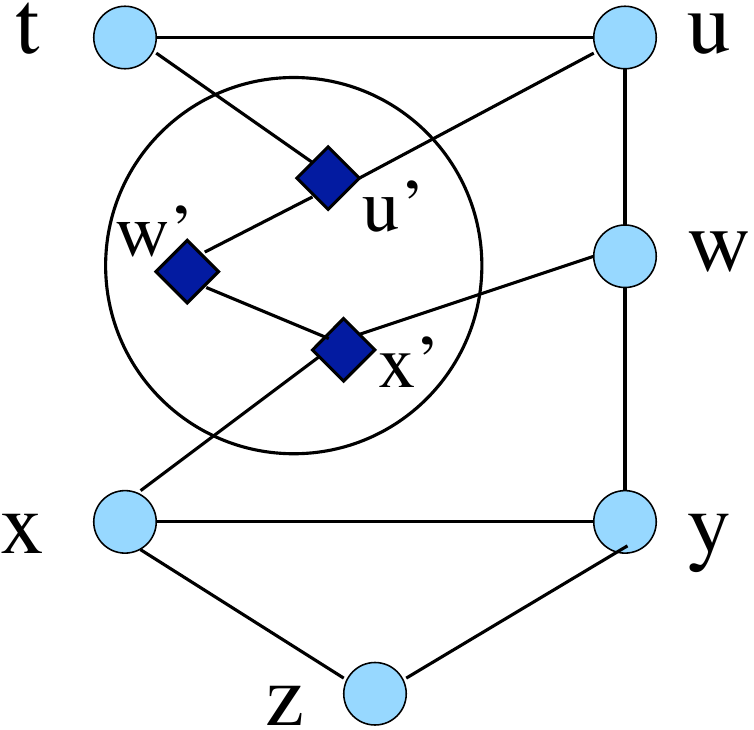}}
%\label{fig:subfig2}
 }\\
\subfigure[Node $y$ attacked.]{
\makebox[0.45\textwidth][c]{\includegraphics[scale=0.5]{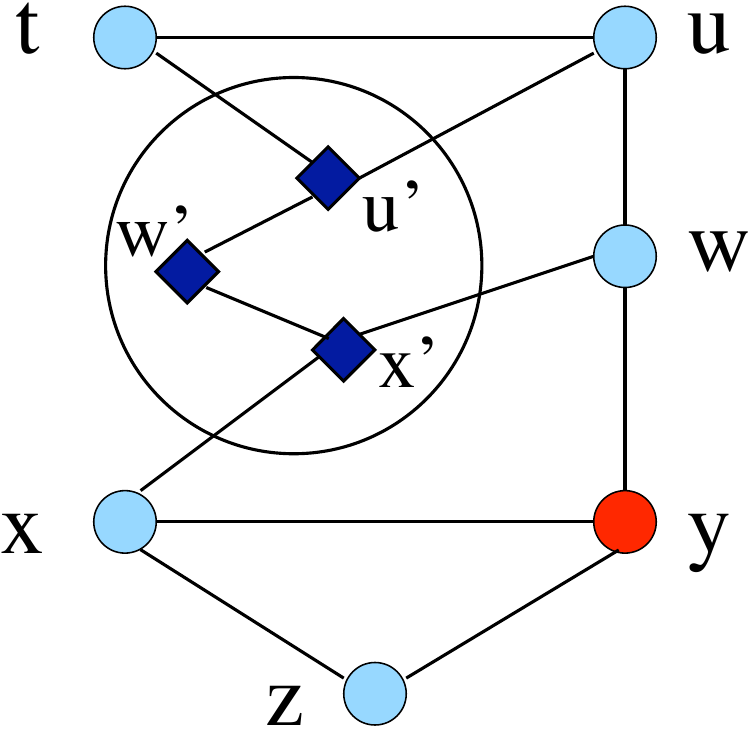}}
%\label{fig:subfig3} 
}
\subfigure[Healed Graph. Notice two $\RT$s with common leaf nodes.]{
\makebox[0.45\textwidth][c]{\includegraphics[scale=0.5]{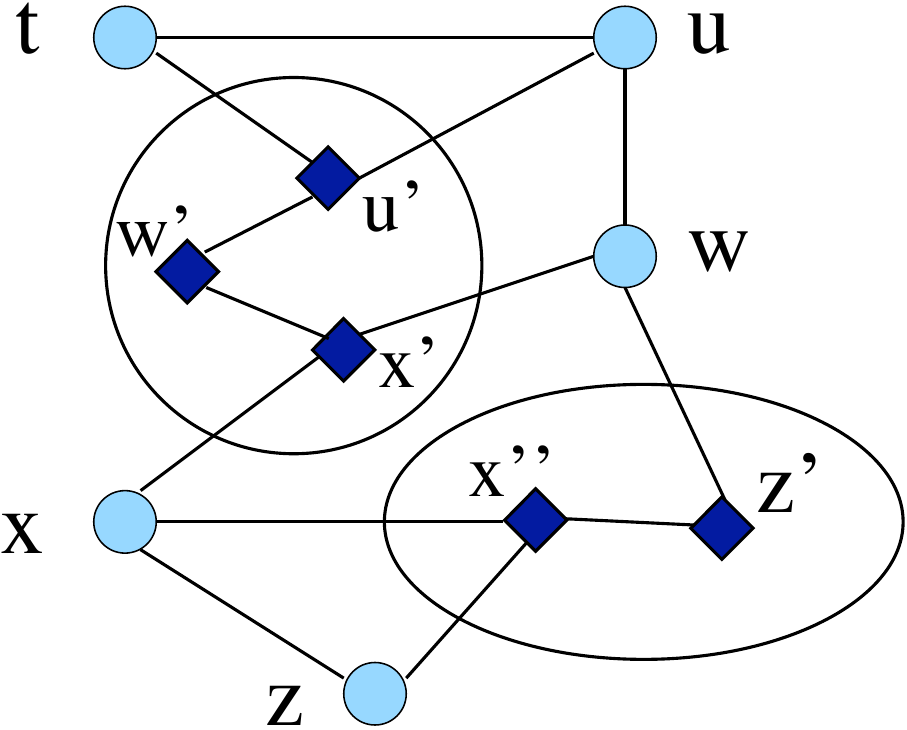}}
 }\\
\subfigure[Node $w$ attacked: notice $w$ is a common leaf of both $\RT$s]{
\makebox[0.45\textwidth][c]{\includegraphics[scale=0.5]{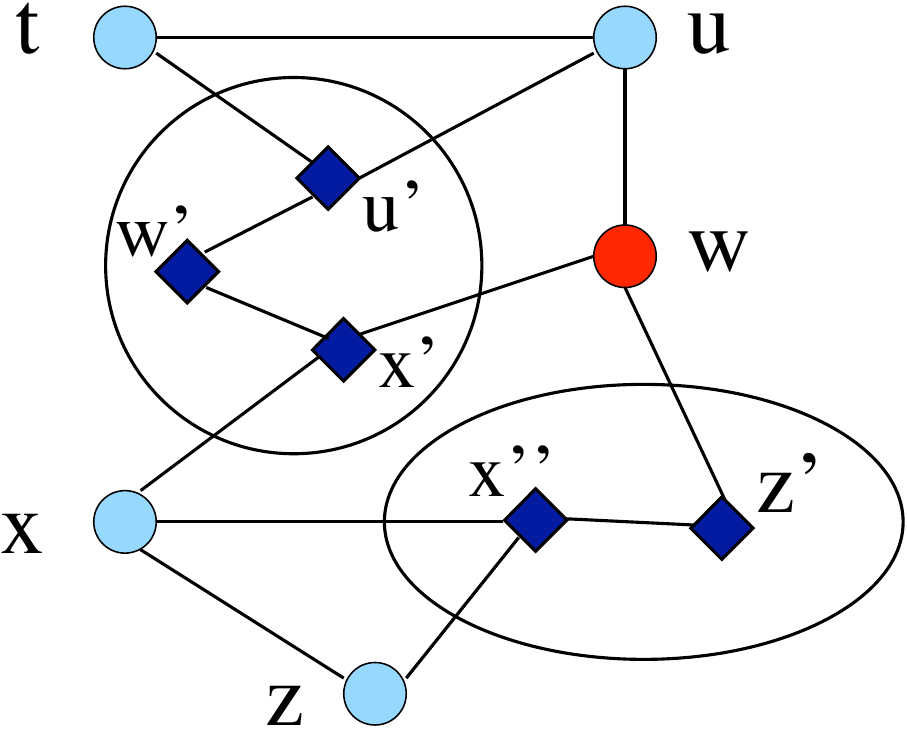}}
 }  
 \subfigure[Healed Graph. The $\RT$s have merged. Some of the leaf nodes ($x$'s, $u$'s) are
identical (so the picture no longer shows the  $\RT$ resembling a $\haft$. However, refer figure~\ref{fig: RTrepresentations}). ]{
\makebox[0.45\textwidth][c]{\includegraphics[scale=0.5]{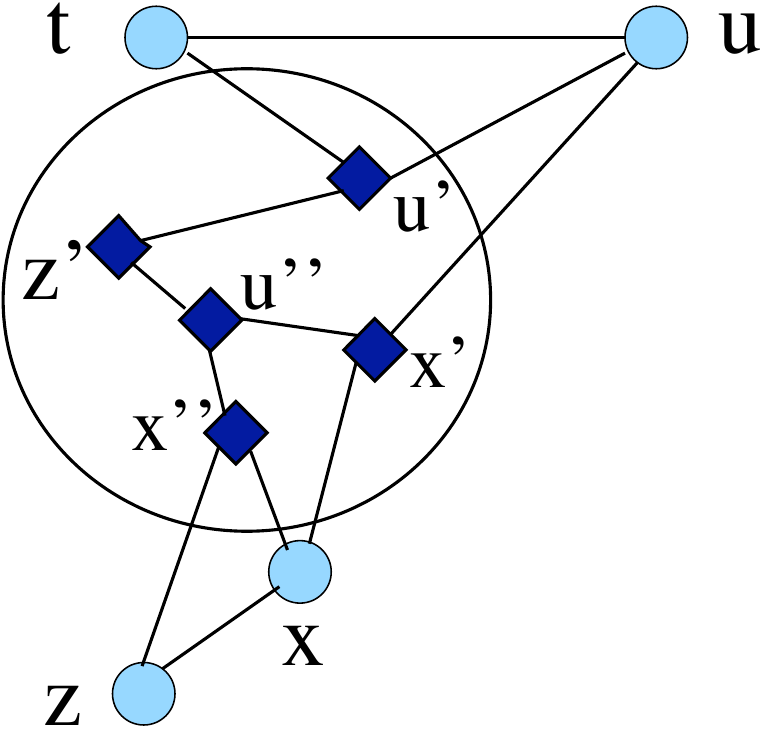}}
%\mbox{\includegraphics[scale=0.5]{images/FG/FGAction6}}
\label{sfig: fgseries6} 
 } 
\caption{Effect of 3 deletions on a graph.  The $\RT$ for each deleted node consists of the helper nodes, plus
the neighbors of the deleted node which form the leaves of the tree. In this example, the deleted nodes form an
independent set, so the structure of the $\RT$s does not depend on the deletion order. }
\label{fig: fgseries}
\end{figure}

\begin{figure}[h!]
\centering
 \subfigure[From figure~\ref{sfig: fgseries6}. Nodes $x$ and $u$ have two edges each going into the haft corresponding to two of their deleted neighbors. ]{
\makebox[0.4\textwidth][c]{\includegraphics[scale=0.5]{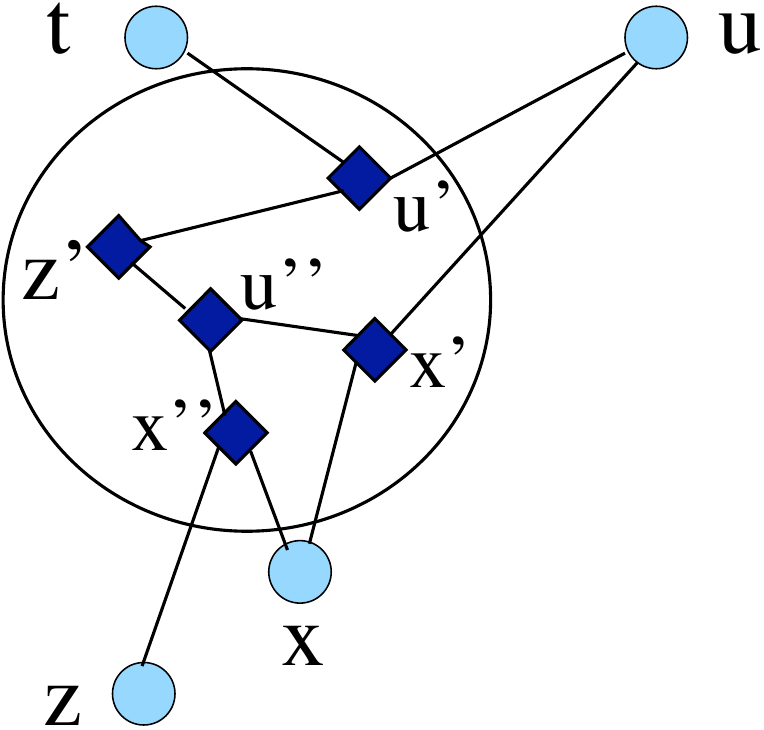}}
%\mbox{\includegraphics[scale=0.5]{images/FG/FGAction6}}
 } \hspace{0.1\textwidth}
 \subfigure[Nodes $x$ and $u$ repeated as leaf nodes of $\RT$s with edges corresponding to their deleted neighbors. This shows the $\haft$ structure of the $\RT$. ]{
 \makebox[0.4\textwidth][c]{\includegraphics[scale=0.5]{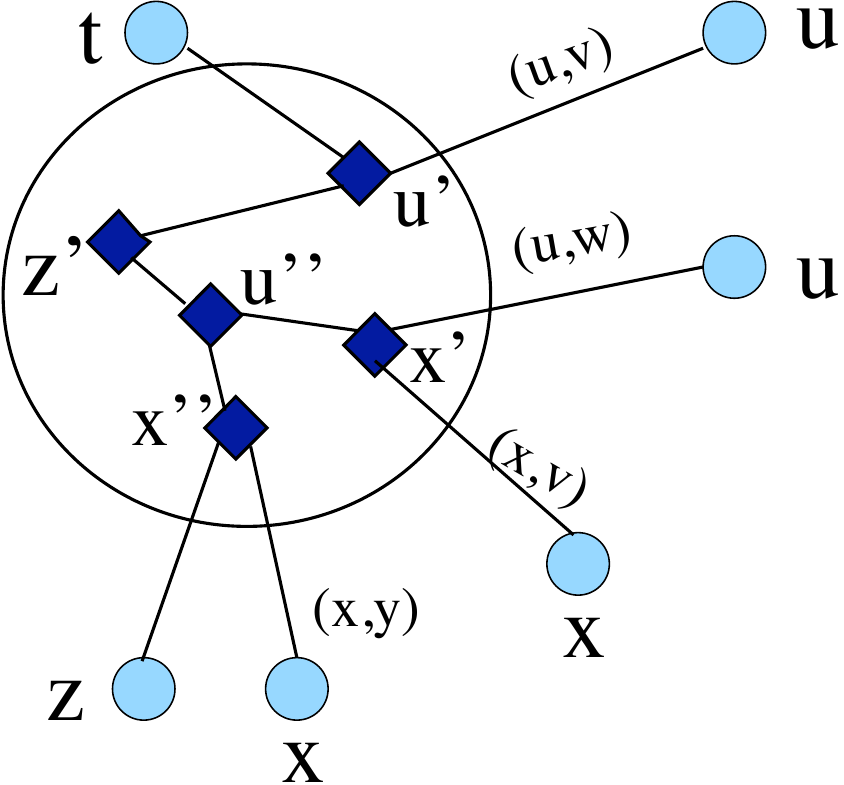}}
%\mbox{\includegraphics[scale=0.5]{images/FG/FGAction6}}
 } 
\caption{Equivalent Representations of a $\RT$.}
\label{fig: RTrepresentations}
\end{figure}

%\begin{figure}[h!]
%\centering
%\includegraphics[scale = 0.6]{images/FG/FGstory}
%\caption{Effect of 3 deletions on a graph. a) The original graph. b) The graph after 2 nodes have been deleted. The
%new nodes inside the ellipses are the helper nodes. The $\RT$ for each deleted node consists of the helper nodes, plus
%the neighbors of the deleted node which form the leaves of the tree. In this example, the deleted nodes form an
%independent set, so the structure of the $\RT$s does not depend on the deletion order. c) After the deletion of the
%third node, adjacent to the first two, the 3 $\RT$s are merged into one. Note that now some of the leaf nodes are
%identical (so now the picture no longer shows the $\RT$).}
%%\label{fig: fgstory}
%\label{fig: fgseries}
%\end{figure}

\begin{table}[h!]
%\begin{tabular}{|l|p{1.7in}|}
%\begin{tabular}{|l|l|}
\begin{tabular}{|l|p{3.7in}|}
\hline 
\textbf{Processor v: Edge(v,x)}&\\ \cline{2-2}
\hline 
\textbf{Real node fields}&\\ \cline{2-2}
\hline
 %  \texttt{parent(v)}& Parent of $v$.\\
 % \texttt{children(v)}& Children of $v$.\\
 % \texttt{neighbors(v)}& Neighbors of $v$.\\
  \texttt{Endpoint}& The node that represents the other end of the edge. For edge(v,x) this will be node $x$ if $x$ is
alive or $\RTparent$ if $x$ is not. \\
 % \texttt{RTID} &  $\ID$ of the $\RT$ of which edge is presently part of.\\
  \texttt{hashelper}& (boolean field). True if there is a helper node simulated by $v$ corresponding to this edge.\\
  \texttt{RTparent} & Parent of $v$ in $\RT$. Non NULL only if $x$ has been deleted.\\
   \texttt{Representative}& This is $v$ itself. Field used during merging of $\RT$s. \\
 % \texttt{SubRT(v)}& Stores the Reconstruction Tree ($\RT$) of $v$.  This tree of helper and real nodes that shall
%replace $v$ if $v$ is deleted.\\
 %\texttt{heir(v)} & The heir of $v$.\\
\hline
\textbf{Helper node fields} & Fields for helper node corresponding to the edge. Non NULL only if the helper node exists.
Sometimes, we will refer to a helper field as \emph{edge.helper.field}\\
\cline{2-2}
\hline
 \texttt{hparent}& Parent of helper node. \\
 \texttt{hrightchild}& Right Child of helper node. \\
 \texttt{hleftchild}& Left Child of helper node. \\
 \texttt{height} &  Height of the helper node.\\
 \texttt{descendantcount} & The number of descendants of the helper node.\\
 \texttt{Representative}& The unique leaf node of the subtree of $(v,x).\helper$ in  $(v,x).\helper$'s $\RT$ that does
not have a helper node in that subtree. This node is used during merging of $\RT$s.\\
%\hline
%\textbf{Reconstruction fields}& Fields used by a node to reconstruct its connections when its neighbor is deleted.\\ \cline{2-2}
%\hline
%\texttt{Representative(v,[RT])}& The heir of node $v$ in $\RT$, used for new helper roles during merging of $\RT$s.\\
% \texttt{nextparent(v)}& The node which will be the next $\parent$ of $v$. \\
 %\texttt{nexthparent(v,[nbrs])}& The node which will be the next $\hparent$ of $v$ in the $\RT$ of the particular
%neighbor(s).
%\\
 %\texttt{nexthchildren(v,[nbrs])}& The node(s) which will be the next $\hchildren$ of $v$ in the $\RT$ of the
%particular neighbor(s).\\
%\hline
%\textbf{Flags}& Specifying a node's status.\\ \cline{2-2}
 % \hline
%\texttt{ishelper(v,[RT])}& (boolean field). True if $v$ is simulating a \emph{helper} node in a particular $\RT$.\\
%\texttt{isactiveheir(v,[nbrs])}&  (boolean field). True if $v$ is simulating an activated heir  for the particular
%neighbor(s), false otherwise. \\
\hline
\end{tabular}
\caption{The fields maintained by a processor $v$ for edge$(v,x)$, which is an edge in $G'$, the graph of only original nodes and  insertions. Here $\RT$ refers to the reconstruction tree of which  $v: edge(v,x)$ is a part.}
\label{tab: nodedata}
\end{table}

% Problem with direct pointers: Each time the right most node changes, it's identity will need to be propagated to
%% every node in the RT (O(n)). Solution: O(log n ) healing. Send the message of deletion to the rightmost node.%%
%Propagate by sending up/down the tree, send only to your rightmost child, will work esp if rightmost node knows it is%%
%the one. 
% Problem 2: Primary roots may change too!
% proposed solution: Use placeholders for making 'wills' - give this info to your parent in each RT, and send messages%%
% to find primary roots for each RT, and store with this information node in each RT
As mentioned earlier, deletion of  a node $v$ leads to it being replaced by  a Reconstruction Tree ($\RT(v)$, for short) in $G$ (Refer to Table~\ref{algo:model-2} for definitions). The $\RT$ is a  $\haft$ (discussed in  Section~\ref{sec: hafts}) having ``virtual'' nodes as internal nodes  and real neighbors of $v$ as the leaf nodes. 
%The role of the virtual nodes are simulated by the leaf nodes. After a long sequence of such
%insertions and deletions, we are left with a graph which is a patchwork mix of virtual nodes and original nodes. 
%The real network is a homomorphic image of  this virtual graph. The nodes in the virtual graph refer to the
%corresponding processor in the network, as shown in Figure~\ref{fig: processornodes}.  The nodes in $G$ corresponding to
%an edge of $v$ in $G'$ and forming the leaf nodes in any $\RT$ are called real nodes, and those internal to a $\RT$
%and simulated by the real nodes (more precisely, by the processor) are called helper nodes.
The virtual nodes are called helper nodes. Recall that the graph $G'$ is the graph consisting of solely the original
nodes and insertions (Table~\ref{algo:model-2}). 
% After a series of deletions and merges, it is possible that a node may occur as a leaf node multiple times in the
%same $\RT$ (refer figure~\ref{fig: fgseries}). The edge information is used to differentiate between these occurrences.

Figure~\ref{fig: fgseries} shows a small series of deletions and repairs by the $\FGraph$ algorithm. Notice that after healing on the third deletion some nodes are occuring as leaf nodes multiple times (figure~\ref{sfig: fgseries6}). Here, edge information is useful for differentiating between these nodes.  A node takes part in a $\RT$ only if one of its neighbors got deleted. It can only have two edges into a $\RT$ if two of its neighbors have already been deleted. Each edge from a real node into a $\RT$ corresponds to a deleted neighbor. We can imagine this edge never got deleted and just that its other endpoint got replaced by a helper node. Thus, if there was an edge between nodes $x$ and $y$, and node $y$ got deleted, we can keep this edge labelled as $(x,y)$. Alternatively, the edge is labelled with it's name in $G'$, which will always be $(x,y)$ since $\G'$ has no deletions.  For convenience, when a  node occurs as a leaf node multiple times in a $\RT$, we will often consider each occurance as a seperate node and depict it as such. Figure~\ref{fig: RTrepresentations} shows this alternate representation. Notice that it is easy to see the haft structure in this representation and we stay in the realm of trees.  Thus, when we refer to a leaf node of a $\RT$, we will mean a real node augmented with the edge information. Thus, when we state that there is at most one helper node corresponding to a leaf node of a $\RT$, this is equivalent to saying that there is at most one helper node in a $\RT$ corresponding to an edge  in the graph $\G'$ .

 The actual processor or entity in the network in which we are executing the algorithm is the one which has to keep track of its real nodes, edges and helper nodes.  In Table~\ref{tab: nodedata},  we list the information each processor $v$ requires for each of its edges in $\G'$ in order to execute the ForgivingGraph algorithm.  For node $v$, the end point of the edge is stored in the field  $v.\Endpoint$. For an edge $(v,x)$, if $x$ is a real
node (i.e. not a helper node) then the field $v.\Endpoint$ is simply the node $x$. When one of the nodes of the edge gets deleted, in $\FG$, a helper node from the new $\RT$ may take place of the previous  node. We will still refer to this edge as $(v,x)$ i.e. by its name in $G'$ but update the fields $\Endpoint$ and $\RTparent$. Moreover, the processor may now simulate a helper node corresponding to this edge. Since each edge is uniquely identified, the real nodes and helper nodes corresponding to that edge can also be uniquely identified. This identification is used by the processors to pass messages along the correct paths. The \textsc{Forgiving graph} algorithm is given in pseudocode form in Algorithm~\ref{algo: forgiving} along with the
required subroutines. 
%For ease of description, the real and helper nodes belonging to the same processor may not be
%explicitly distinguished in the pseudocode.

\begin{figure}[h!]
\centering
\includegraphics[scale=0.35]{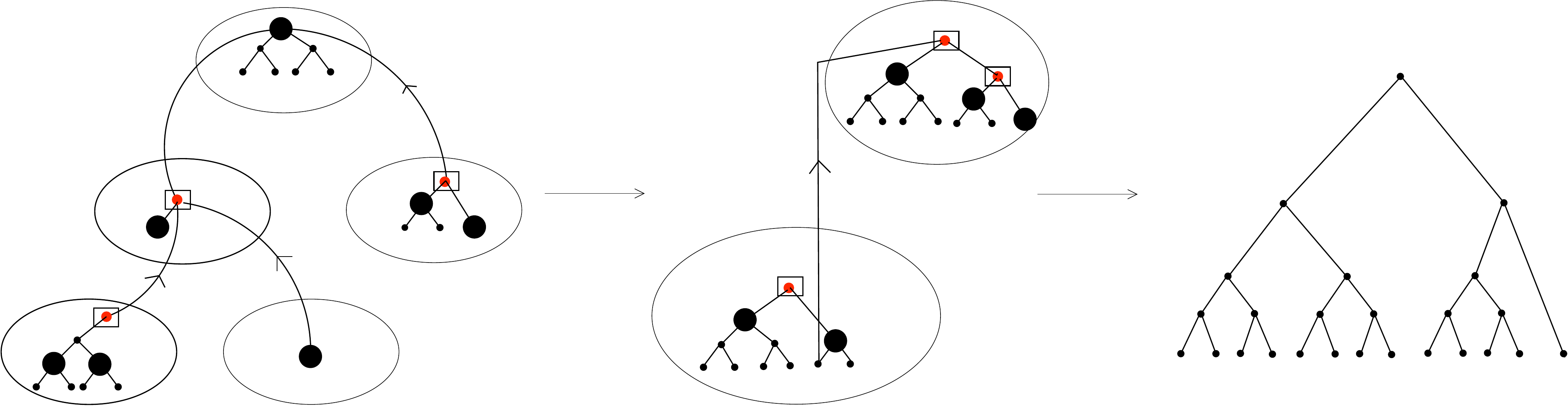}
\caption{On deletion of a node $v$, The $\RTfragment$s to be merged are connected by a binary tree $BT_{v}$. 
The leaf $\RTfragment$s merge with their parents till a single $\RT$ is left. The solid circles are the primary roots.
The (red color) nodes in the square boxes are spine nodes removed at each step.}
\label{fig: Anchormerge}
\end{figure}

\begin{figure}[h!]
\centering
\includegraphics[scale=0.75]{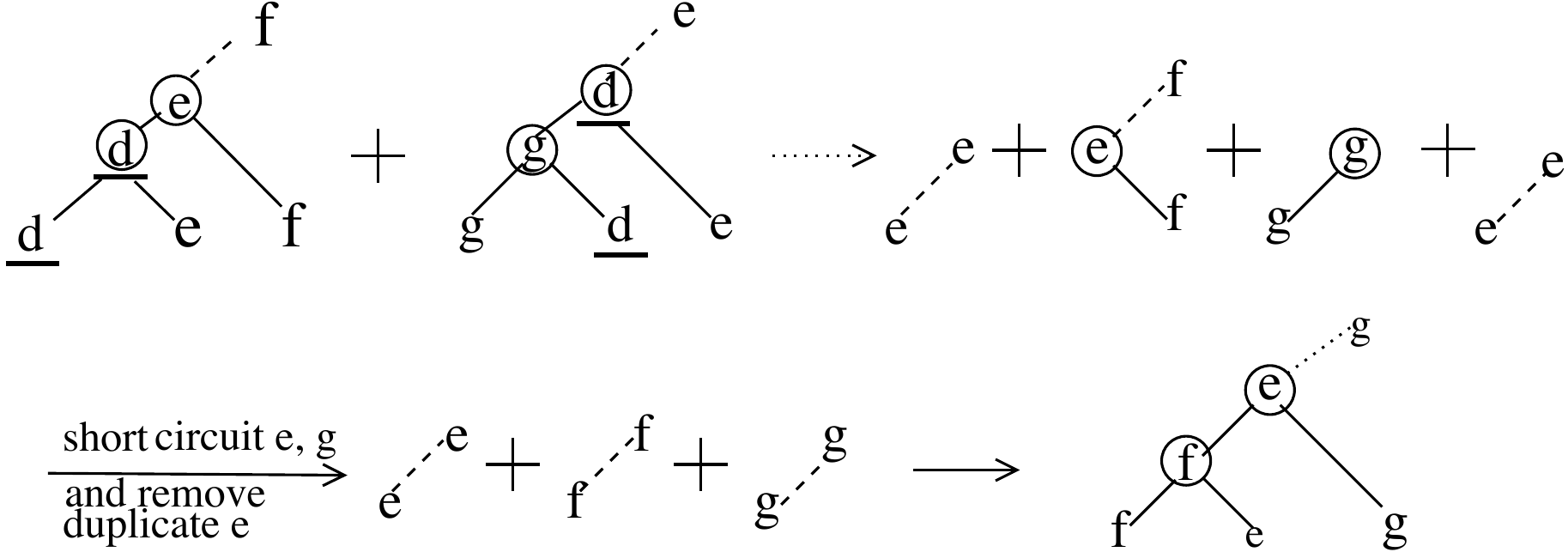}
\caption{The underlined node $d$ and corresponding helpers are deleted. This leads to the graph breaking into
components which are then merged using $BT_{d}$ (the binary tree of anchors) and the primary roots in the components. The dashed edges
show the representative for that node.}
\label{fig: RTmerge}
\end{figure}

%We have four main classes of fields, according to the way they are used by the
%node.  `Current fields' give a nodes present configuration and status
%in the tree. `Reconstruction fields' hold the data needed for a node
%to reconstruct connections when one of its neighbors gets
%deleted. `Helper fields' hold information with regard to the helper
%node being simulated by this node. Each node also stores some special
%flags with regard to its helper or heir status. In the description
%that follows, we shall refer directly to these fields.

%\pagebreak
On deletion of a node, the repair proceeds in two phases. The first phase is a quick $O(1)$ phase in which the neighbors
of the deleted node connect themselves in the form of a binary tree (Algorithm~\ref{algo: fixnode},Figure~\ref{fig: Anchormerge}). Consider the
effect of the deletion of $v$ on one of the $\RT$s of which $v$ is a leaf. Removal of this leaf and of the helper node
corresponding to that leaf (if any) splits this $\RT$ into connected components.
We select particular nodes which were neighbors of the deleted nodes from  each of these components. Let $Nset$ be the
collection of all these nodes together with any undeleted neighbors of $v$ in $\FG$. We shall call a component
taking part in the merge process (irrespective of whether it is a $\haft$ or not) as  a  \emph{RTfragment}, to
distinguish it from the final $\RT$ formed at the end of the merge process.
 In phase 2,  the $\RTfragment$s  are merged (Figure~\ref{fig: Anchormerge}). 
 %We are only interested in the complete trees in these since we can discard all other helper nodes.
 Before we can reconnect these $\RTfragment$s into a single $\haft$, we need to further break them up into $\haft$s (we actually break them into complete trees) so that we can merge them. We now go into details of the communication
protocol that achieves this merge.  Let $v$ be the processor deleted. Then, the nodes in $Nset$ connect in the form 
of a binary tree we call $BT_{v}$. We call the nodes forming $BT_{v}$ as \emph{anchors}. Formally, we define an anchor as follows:
\begin{description}
\item[Anchor]: An anchor is a designated node in a $\RTfragment$ that takes part in the binary
tree $BT_{v}$.
\end{description}
  The anchors send probe messages to discover the primary roots which head these complete trees
(Algorithm~\ref{algoline: findprroots}). This is similar to the Strip operation described in Section~{\ref{subsec:
haftstrip}}. The nodes maintain information about their height and number of their children in their $\RT$ or
$\RTfragment$. Thus, they are able to identify themselves as primary roots. At the same time, the nodes outside the
complete trees are identified and marked for removal. It is possible that a $\RTfragment$ may appear more than once in a $\BT_{v}$ through multiple nodes acting as anchors. However, we want one complete tree to take part only once in the merge. This is accomplished as follows: Every anchor sends probe messages to discover the primary roots in its $\RTfragment$. Nodes further pass on these probes till they reach a primary node. However, if an anchor receives a probe message originating from another anchor, it will reject the message and return it to the sender, which will send it back towards the source anchor.   This ensures that a primary root (thus, a complete tree) will be discovered by only one anchor.
The complete trees are then merged pairwise in a bottomup fashion till only a
single haft remains. This is illustrated in figure \ref{fig: Anchormerge}. At each round, every leaf $\RT$ in $BT_{v}$
will merge with its parent $\RT$. This can be done in parallel, so that the  number of rounds of merges will be
equivalent to the height of the tree. For two trees to merge, as shown in the Merge operation (Section~\ref{subsec:
haftmerge}), an additional node is needed that will become the parent of these two trees. This node must be simulated by a real node that is not already simulating a helper node in the trees. Since the number of internal  nodes in a tree is one less than the leaf nodes, there is exactly one such leaf node for each tree. The roots of these two trees have the identity of this node for their tree. This node is called a \emph{Representative} (of the root node).  For merging, we use an algorithm that we call the  \emph{representative mechanism}. The formal definition of a representative and details of the representative mechanism are given in Section~\ref{sec: representative}. Each node keeps the identity of its representative stored in the field \emph{Representative} (Table~\ref{tab: nodedata}).

Now, we briefly describe merging using  representatives. When two trees (Note that a tree may even be a single node) are merged (Algorithm~\ref{algo: makeRT} and Algorithm~\ref{algo: ComputeHaft}), the representative of the root of the bigger tree (or of one of the trees, if they have the same size) instantiates a new helper node, and makes the two roots its children. To  make the new structure a $\haft$, the root of the bigger tree shall become the left child of the new helper node. The new helper node will now inherit as its representative the representative of the
root of its right subtree, since this is the node in the merged tree that does not have a helper node. An
example of merging using this algorithm is shown in Figure~\ref{fig: RTmerge}.

 At the end of each round, we have a  new set of leaf $\RT$s. Each new leaf is now a
merged haft of the previous leaves and their parent. We need a new anchor for this haft. We can continue having the anchor of the parent $\RTfragment$ as the anchor. However, this
 node may be one of the extra nodes marked for removal. In this case, the anchor designates one of the nodes that was a primary root in its $\RTfragment$ as the new anchor, passes on its links and removes itself.
 The newly formed leaf $\haft$s may have primary roots which are different from those of the previous ones. The new anchor will send probe messages and gather the relevant information and inform the new primary roots of their role. This process will continue till we are left with a single $\RT$. This is shown in Figure~\ref{fig: Anchormerge}.

\floatname{algorithm}{Algorithm}

\begin{algorithm}[ph!]
\begin{algorithmic}[1]
\STATE given a Graph $G(V,E)$
\REQUIRE{each node of G has a unique ID}
\FOR{each node  $v\in G$} 
\STATE \textsc{Init($v$)}
\ENDFOR
\WHILE {true}
\IF{a vertex $v$ is inserted}
\STATE vertex $v$ and new neighbors add appropriate edges
\STATE \textsc{Init($v$)}
\ELSIF{a vertex $v$ is deleted}
\STATE \textsc{DeleteFix($v$)}
\ENDIF
\ENDWHILE
\end{algorithmic}
\caption{\textsc{Forgiving graph}: The main function.}
\label{algo: forgiving}
\end{algorithm}

 \begin{algorithm}[ph!]
 \caption{\textsc{Init($v$)}: initialization of the node $v$} 
\label{algo: init}
\begin{algorithmic}[1]
\FOR{each $\edge (v,x)$}
\STATE $(v,x).\representative = v$
\STATE set other fields to NULL
\ENDFOR
\end{algorithmic}
\end{algorithm}

\begin{algorithm}[ph!]
\caption{\textsc{DeleteFix($v$)}: Self-healing on deletion of a node }
\label{algo: fixnode}
\begin{algorithmic}[1]
\STATE $\Nset = \{ \}$
\FOR{each $\edge (v,x)$}
\IF{$(v,x).hashelper = \True $}
\STATE $\Nset = \Nset \cup (v,x).\hparent \cup (v,x).\hrightchild$
\ENDIF
\STATE $\Nset = \Nset \cup  (v,x).\Endpoint$
\ENDFOR
\STATE \label{algoline: BTquickfix} Nodes in $\Nset$ make new edges to make a balanced binary tree $\BT_{v} (\Nset,E_{v})$
\STATE \textsc{BottomupRTMerge($\BT_{v},v$)}
\STATE delete the edges $E_{v}$
\end{algorithmic}
\end{algorithm}

\begin{algorithm}[ph!]
\caption{\textsc{BottomupRTMerge($\BT_{v}, v$)}: The nodes of $\BT_{v}$ merge  their $\RT$s starting from the leaves going up forming a new $\BT_{v}$. }
\label{algo: bottomupmerge}
\begin{algorithmic}[1]
\IF{$\BT_{v}$ has only one node}
\STATE return
\ENDIF
\FOR{$y \in \BT_{v}$}\label{algoline: findprroots}
 \IF{$y$ is a real node}
 \STATE let $y.\PrRoots \gets y$
 \ELSIF{$y = (v,x).\Endpoint$}
 \STATE $y.\PrRoots \gets$ \textsc{FindPrRoots($y, 1,(v,x), \True, y$ )}
 \ELSIF {$y.\helper.hparent = v$ OR $y.\helper.hleftchild = v$ OR $y.\helper.hrightchild = v$}
  \STATE let $y.\PrRoots \gets$ \textsc{FindPrRoots($y, $ $ v.\descendantcount,v.\helper,
\True, y$)}
  \ELSE
  \STATE let $y.\PrRoots \gets$ \textsc{FindPrRoots($y, $ $v.\descendantcount$$,v.\helper, \False, y$)}
  \ENDIF
\ENDFOR
\FOR{all nodes $y$ s.t. node $y$ is a parent of a leaf in $\BT_{v}$}
\IF{$y$ has two children in $BT_{v}$}
%\STATE \textsc{Haft_Merge($y$, $y$ 's left child, $y$ 's right child)}
\STATE \textsc{Haft\_Merge}($y$, $y$ 's left child in $BT_{v}$, $y$ 's right child in $BT_{v}$)
\ELSE
\STATE \textsc{Haft\_Merge}($y$, $y$'s left child, NULL)
%%\STATE \textsc{TwoRTMerge($y, y$'s parent in $\BT_{v}$)}
\ENDIF
\ENDFOR
\STATE \textsc{BottomupRTMerge($\BT_{v}$)} \COMMENT{New leaf nodes merge again till only one is left.}
\end{algorithmic}
\end{algorithm}

%\begin{algorithm}[ph!]
%\caption{\textsc{TwoRTMerge($y,z$)}: The $\RT$s of $y$ and $z$ merge  forming a new $\RT$, with $z$ being the new anchor. }
%\label{algo: twortmerge}
%\begin{algorithmic}[1]
%\STATE Node $y$ conveys to node $z$ its list of primary roots (and their representatives) and vice versa.
%\STATE Node $y$ and node $z$ propagate the lists to the primary roots in their $\RT$s.
%\STATE Every primary root computes the merged $\RT$ according to a standard convention and makes the corresponding edges.
%\STATE Node $z$ is the new anchor for the $\RT$.
%\STATE The new primary roots are informed and $z$ updates its list. \COMMENT{By conjecture, any  primary root of the
%merged tree is at most $\lfloor{\log (\# \RT s \mathrm{merged})}\rfloor$ away from some primary root of some original
%tree.}
%\end{algorithmic}
%\end{algorithm}

\begin{algorithm}[ph!]
\caption{\textsc{FindPrRoots}($y$, numchild, sender, Breakflag, origin): Find primary roots in the $\RTfragment$
(Section~\ref{sec: FGdetail} of main text) containing node $y$. If Breakflag is set, the tree is  a $\RTfragment$ formed due to the deletion of the node prior to any merges and the nodes need to adjust their descendant count. }
\label{algo: findprroots}
\begin{algorithmic}[1]
\IF{$y$ is an Anchor node AND $y \neq origin$}
\STATE return NULL \COMMENT{Anchors reject probe messages from other anchors.}
\ENDIF
\IF{Breakflag = $\True$ AND (sender = $y.\hrightchild$ OR sender = $y.\hleftchild$ )}
  \STATE $y.\descendantcount = y.\descendantcount$ - numchild
\ENDIF
%\IF{\textsc{TestFullRoot($y, \RTID$)} = $\True$}
\IF{$y.\descendantcount = 2^{y.\height}$} %[Test for node heading a full subtree]  
  \IF {\textsc{TestPrimaryRoot($y$)} = $\True$} 
  \STATE return \{$y$,\textsc{FindPrRoots}($y.\hparent, 0, y$, Breakflag, origin) \}
%\textsc{FindPrimaryRoots($\hrightchild(x)$)}  }
  \ELSE
  \STATE return \{\textsc{FindPrRoots}($y.\hparent, 0, y$, Breakflag, origin) \} \COMMENT{Node itself not a primary
root but parent maybe.}
  \ENDIF
\ELSE
\STATE mark node red
  \IF{exists($y.\hleftchild$) AND sender $\neq$ $y.\hleftchild$}
   \STATE \textsc{FindPrRoots}($y.\hleftchild, y.\descendantcount, y$, Breakflag, origin)
   \ELSIF{exists($y.\hrightchild$) AND sender $\neq$ $y.\hrightchild$}
   \STATE \textsc{FindPrRoots}($y.\hrightchild, y.\descendantcount, y$, Breakflag, origin)
   \ELSIF{exists($y.\hparent$) AND sender $\neq$ $y.\hparent$}
   \STATE \textsc{FindPrRoots}($y.\hparent, y.\descendantcount, y$, Breakflag, origin)
%\STATE return \{ \textsc{FindPrimaryRoots($y.\hparent$)}, \textsc{FindPrimaryRoots($y.\hleftchild$)}, \textsc{FindPrimaryRoots($y.\hrightchild$)} \}
%\textsc{FindPrimaryRoots($\hrightchild(y)$)}  }
  \ENDIF
\ENDIF
\end{algorithmic}
\end{algorithm}

%\begin{algorithm}[ph!]
%\caption{\textsc{RightMostProbe($y, \RTID$)}: Probe for and return the rightmost node in $\RT$ }
%\label{algo: rightmostnode}
%\begin{algorithmic}[1]
%\IF{y is a real node} 
%\STATE return y
%\ELSIF{event: $y$ receives probe message from it's $\hparent$}
%\STATE \textsc{RightMostProbe($y.\hrightchild, \RTID$)}
%\ELSIF{$y$ is the root of $\RT$}
% \STATE \textsc{RightMostProbe($y.\hrightchild, \RTID$)}
% \ELSE
% \STATE \textsc{RightMostProbe($y.\hparent, \RTID$)}
% \ENDIF
%\end{algorithmic}
%\end{algorithm}

\begin{algorithm}[ph!]
\caption{\textsc{TestPrimaryRoot($y$)}: Tell if helper node $y$ is a primary root in $\RT$ }
\label{algo: testprimaryroot}
\begin{algorithmic}[1]
\IF{$y.\descendantcount = 2^{y.\height}$} %[Test for node heading a full subtree] 
 \IF{$y.\hparent = NULL$}
\STATE return $\True$
\ELSIF{$y.\hparent.\descendantcount \neq  2^{y.\hparent.\height}$}
\STATE return $\True$
 \ENDIF
 \ENDIF
\STATE return $\False$ 
\end{algorithmic}
\end{algorithm}

\begin{algorithm}[ph!]
\caption{\textsc{Haft\_Merge($p,\ell,r$)}: Merge the hafts mediated by anchors $p,\ell$ and $r$}
\label{algo: haftmerge}
\begin{algorithmic}[1]
\STATE Nodes $p,\ell\ \mathrm{and}\ r$ exchange $p.\PrRoots,$ $ \ell.\PrRoots(l),$ $ \PrRoots(r)$
\STATE let $\RT \gets $ \textsc{MakeRT($\PrRoots(p),\ell.\PrRoots, r.\PrRoots$)}
%\STATE \textsc{PropagateID($\RT, max_{\height}(p,l,r).\RTID, \Root(\RT)$)}
  \IF{$p$ is marked red} %\COMMENT{$p$ needs to be removed}
  \STATE $p$ transfers its edges in $BT_{v}$ to one of $p.\PrRoots$ \COMMENT{$p$ needs to be removed, $BT_{v}$ needs
to be maintained}
  \ENDIF
 \STATE remove all helper nodes marked red \COMMENT{Some helper nodes marked red may have been reused and unmarked by
\textsc{MakeRT}}
\end{algorithmic}
\end{algorithm}

\begin{algorithm}[ph!]
\caption{\textsc{MakeRT}(PRoots1, PRoots2, PRoots3): The sets of Primary roots make a new RT }
\label{algo: makeRT}
\begin{algorithmic}[1]
\FOR{all $y \in (\mathrm{PRoots1} \cup \mathrm{PRoots2} \cup \mathrm{PRoots3}) $ }
  \STATE let $T \gets$ \textsc{ComputeHaft}(PRoots1, PRoots2, PRoots3)
  \STATE make helper nodes and set fields and make edges according to  $T$
 \ENDFOR
\end{algorithmic}
\end{algorithm}

\begin{algorithm}[ph!]
\caption{\textsc{ComputeHaft($\mathrm{PRoots1, PRoots2, PRoots3})$}: (Implementation of Haft Merge) The primary roots
compute the new haft}
\label{algo: ComputeHaft}
\begin{algorithmic}[1]
\STATE let $R =  \mathrm{PRoots1} \cup \mathrm{PRoots2} \cup \mathrm{PRoots3} $
\STATE let $L =  R$ sorted in ascending order of number of children, NodeID
\STATE suppose $L$ is $(r_{1}, r_{2}, \ldots, r_{k})$ where the $r_{i}$ are the $k$ ordered primary roots.
\STATE set $ctr = 1, count = k$
\WHILE{$ctr < count$}
  \IF{$r_{ctr}.\numchildren = r_{ctr+1}.\numchildren$ }
  \STATE \label{algocode: mergeeqrep} Make helper node $\helper(r_{ctr}.\representative)$. Initialize fields
to NULL.
  \STATE make  $\helper(r_{ctr}.\representative)$ the parent of $r_{ctr}$ and $r_{ctr + 1}$
  \IF{$r_{ctr}$ is a real node}
   \STATE  set $\helper(r_{ctr}.\representative).\height = 1$ 
   \ELSE
    \STATE set $\helper(r_{ctr}.\representative).\height = 2 r_{ctr}.height$ 
  \ENDIF
  \STATE set $\helper(r_{ctr}.\representative).\representative = r_{ctr+1}.\representative$ 
  \STATE remove $r_{ctr}, r_{ctr + 1}$, insert  $\helper(r_{ctr}.\representative)$ in correct place in $L$.
  \STATE set $ctr \gets ctr - 1$, $count \gets count - 1$
  \ENDIF
  \STATE set  $ctr \gets ctr + 1$,
\ENDWHILE
\STATE set $ctr = 1$
\WHILE{$ctr < count$}
\STATE \label{algocode: mergeneqrep} make helper node $\helper(r_{ctr+1}.\representative)$. Initialise its fields to NULL
\STATE set  $\helper(r_{ctr + 1}.\representative).\hleftchild = r_{ctr+1}$
\STATE set  $\helper(r_{ctr + 1}.\representative).\hrightchild = r_{ctr}$
\STATE set $\helper(r_{ctr + 1}.\representative).\height =  r_{ctr+1}.\height + 1$ 
\STATE set $\helper(r_{ctr + 1}.\representative).\representative =  r_{ctr}.\representative$ 
\STATE In $L$, replace $r_{ctr + 1}$ by  $\helper(r_{ctr + 1}.\representative)$ 
\ENDWHILE
\end{algorithmic}
\end{algorithm}

%\begin{conjecture}
%If a number of half-full trees are merged, Any  primary root of the final merged tree is at most $\lfloor{\log (\# \RT
%s\ \mathrm{ merged})}\rfloor$ away from some primary root of some original tree.
%\end{conjecture}
%\begin{proof}
%
%\end{proof}

\subsection{Representative mechanism}
\label{sec: representative}

In this section, we discuss representatives and their use in merging in more detail.
 Formally, we define a representative as follows:
\begin{description}
\item[Representative:]
In the Forgiving Graph $\FG$, given a node $y$, the representative of $y$ is a real node, decided as follows:
\begin{itemize}
 \item If $y$ is a real node, then $y$ itself.
 \item If $y$ is a helper node, then the unique leaf node that is a descendant of $y$ and does not have a helper node in
the subtree headed by $y$.
% unique leaf node of $y$'s subtree in $y$'s $\RT$ that does not have a helper
%node in that subtree. 
\end{itemize}
\end{description}

Recollect that one of our objectives is to maintain an invariant that a real node simulate at most one helper node. Moreover, this has to happen in the dynamic environment of nodes getting deleted, inserted, $\RT$s breaking and merging.  The representative mechanism allows us to do this in an efficient manner, as we shall show. Intuitively, a representative is a real node who we know is not simulating a helper node yet and so is available for providing a helper node. Each node in the Forgiving Graph $\FG$ has a representative. Formally, for a node $y$ in $\FG$, if $y$ is a real node, $y$ is its own representative.  This makes sense since 
%if we are looking to merge the  real node $y$ with another $\RT$, 
$y$ is the root of a $\RT$ (a single node $\RT$) and  not simulating a real node. If node $y$ is a helper node its representative is the unique leaf node that is $y$'s descendant in $y$'s subtree that is not simulating a helper node. Notice that there is exactly one such leaf node in any subtree since the number of internal nodes are one less than the number of the leaf nodes, and as a consequence of our invariant, all other leaf nodes are simulating exactly one helper node each in that subtree. Due to the way our merge operations operate, each helper node gets assigned a representative when the helper node is created and moreover it never changes its representative during its lifetime. This is a very useful property as we shall see later.

 \begin{figure}[h!]
\centering
\includegraphics[scale=0.8]{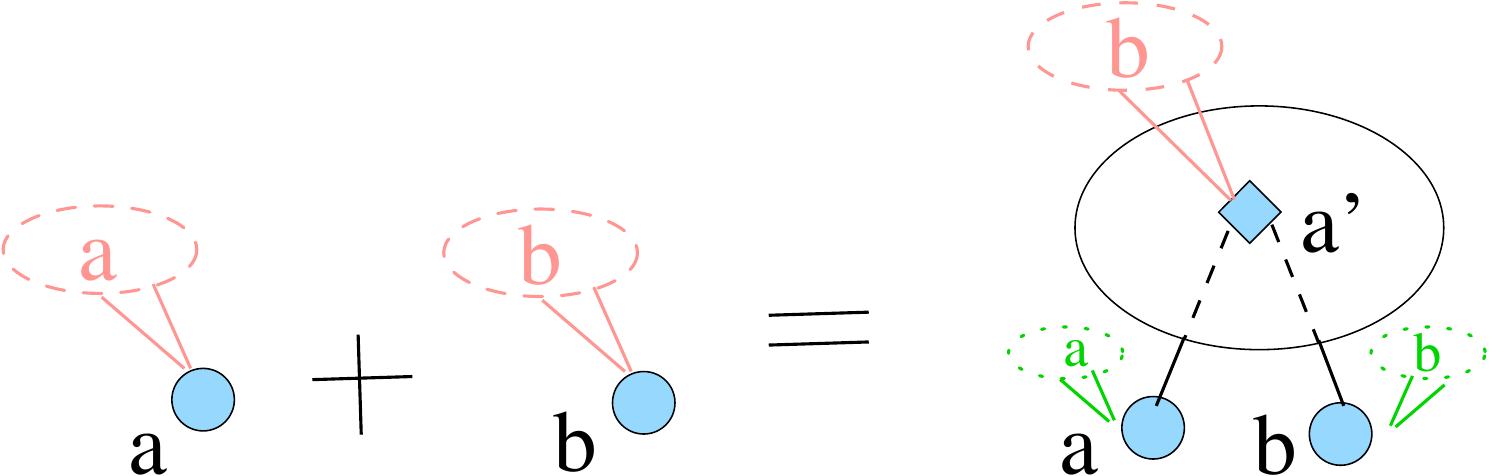}
\caption{Merging with representatives: Two singleton hafts of real nodes $a$ and $b$ merge.  Here $a$ creates the parent helper node, and this helper node inherits the representative of its right child ($b$) as its representative. Notice $b$ is the unique real node in $a.\helper$'s subtree that is not simulating a helper node. With regard to merging, the  root nodes representatives are 'active'  (shown in pink, dashed outline), while others are 'dormant' (shown in green, dotted outline).   }
\label{fig: RepMerge}
\end{figure}

First, let's see how  representatives are used to merge hafts. The simplest example is shown in figure~\ref{fig: RepMerge}: two real nodes (a real node is a singleton haft) merge using their representatives. 
To recollect, when two hafts merge, a new helper node is needed to become the parent of both. We choose this node to be simulated by the representative of the root of the bigger haft. If the hafts are of the same size, either can be selected. 
 The chosen representative is informed: it instantiates a new helper node and makes the two roots its children. To  make the new structure a $\haft$, the root of the bigger tree shall become the left child of this new helper node. The new helper node now needs a representative of its own. The obvious choice is the representative of its right child, since that leaf node still has not supplied a helper node. This is consistent with the definition of a representative (this can be verified for the small example of figure~\ref{fig: RepMerge}). This is the conceptual picture. In the distributed implementation, as described earlier, this communication takes place through the anchors which exchange information among the merging anchors. This information consists of the identity of the primary roots, their height and representative information. Each anchor is then able to run the merge algorithm in its memory, and it directly contacts the nodes with which it has to make edges. If this is a new node it is also provided with the identity of its representative.

\begin{figure}[h!]
\centering
\includegraphics[scale=0.6]{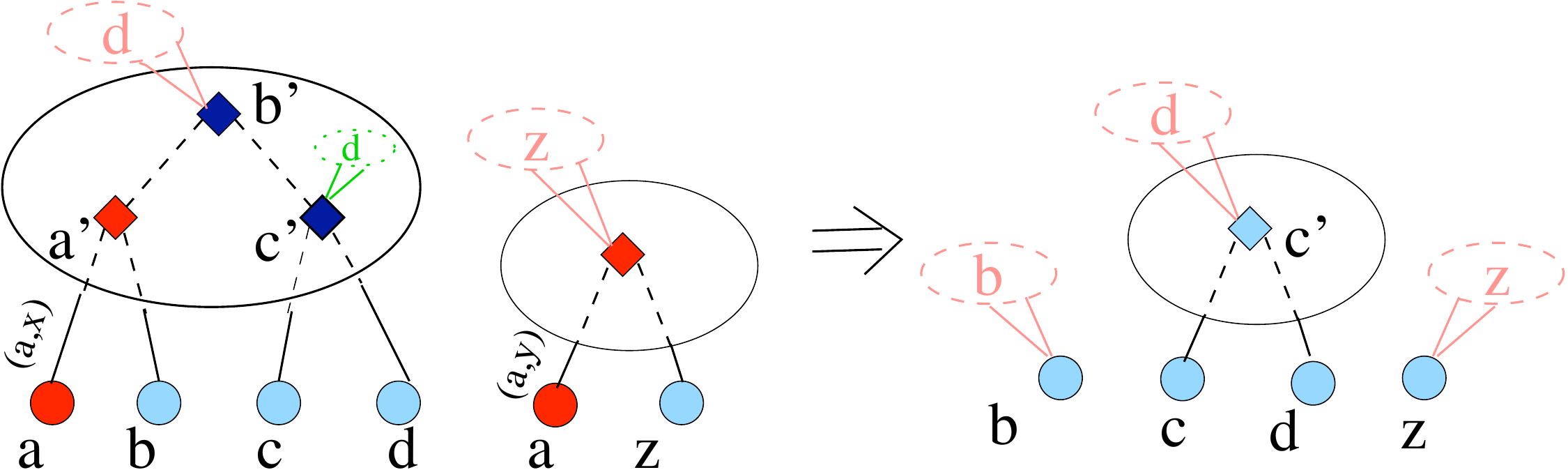}
\caption{Reusing representative information: $\RT$s split into complete trees on deletion of node $a$. A node always has a  representative assigned to it at birth and it never changes its representative. In the figure, node $c'$ has $d$ as its representative:- 'dormant' before the split (green, dotted outline),  'active' afterwards (pink, dashed outline).}
\label{fig: RepSplit}
\end{figure}

 What happens when a deletion happens and a $\RT$ splits into smaller complete trees? To merge back, we need to find the representatives of the roots of these trees. Should we traverse the subtree of these roots to find the representative? Obviously, this is expensive. Fortunately, the representative mechanism renders this unnecessary. To recall, merging happens using primary roots, which are the roots of complete trees. After a split, we are only left with complete trees. Obviously, complete trees have not had a deletion in their subtree, thus, none of the nodes in these trees need to change their representatives. Since only the nodes of the complete trees will be merging (via their roots) we need only worry about  their representative information. This implies that no node need ever change its representative. This is shown in figure~\ref{fig: RepSplit}. As shown in the picture, we can imagine that the representatives of the primary roots are in an 'active' state i.e. they will be used for the upcoming merge, whereas representatives of all internal nodes are in a 'dormant' state meaning though they are not required at the present stage, they may be utilized in the future.

\section{Real graph from the Forgiving Graph}
\label{sec: homomorphism}

\begin{figure}[h!]
\centering
\includegraphics[scale=0.7]{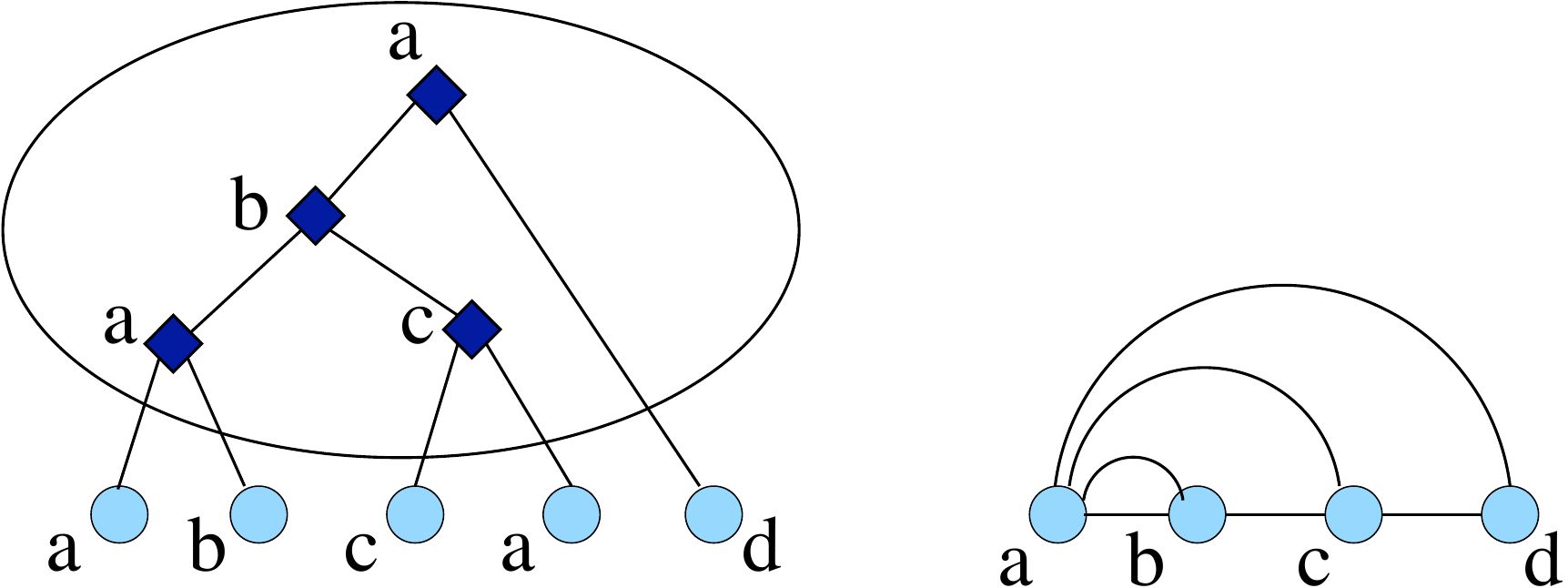}
\caption{The actual graph $\G$ (on the right) is a homomorphic image of the Forgiving Graph $\FG$ (left) where the helper nodes are mapped to the nodes simulating them. Note both the node degrees and distances between nodes in the real graph cannot be more than those in the Forgiving Graph.}
\label{fig: homomorphism}
\end{figure}

It is easy to see that the Forgiving Graph $\FG$ maps to the real graph $G$ in a straightforward way: map all the helper nodes to the real nodes simulating them.  Figure~\ref{fig: homomorphism} shows an example. 
More formally, $\G$ is a homomorphic image of $\FG$. Consider two graphs $G_1 = (V_1, E_1)$, and $G_2 = (V_2, E_2)$.
In this context, a homomorphism may be defined as follows: 
A homomorphism is a function $f: V_1 \rightarrow V_2$ such that if undirected edge $\{v,w\}$  is in $E_1$ (the edge set of $G_1$) this implies that the edge $\{f(v), f(w) \}$ is  in $E_2$. Moreover, we say that $G_2$ is the homomorphic image of $G_1$ under $f$ if the edges of $G_2$ are exactly the images of the edges of $G_1$ under the homomorphism.
We know that, in $\FG$,  there can be multiple real and helper nodes corresponding to a processor in the network  that performs all the functions required of those nodes. Each node is identified by its processor and some additional information. For node $v$ in $\FG$, let $Processor(v)$ be  the name of that processor. Also, in the graph $\G$, there is only one node per processor and consider this node to be labelled with the name of that processor. Then,  our homomorphism $H:V(\FG) \rightarrow V(\G)$ is simply $H(v) = Processor(v)$. 

Let us make the following observations about homomorphisms which will be useful to us in proving our results (Section~\ref{sec: Results}).

\begin{observation}
\label{obs: homodegree}
For any graph homomorphism $F: G_1 \rightarrow G_2$, for all nodes $u, v$ in $V$, $dist_{G_2}(F(u), F(v)) \le dist_{G_1}(u,v)$ where $dist_{G}(x,y)$ is the distance between two nodes $x$ and $y$ in a graph $G$. 
\end{observation}

\begin{observation}
\label{obs: homostretch}
If the graph $G_2$ is the homomorphic image of  graph $G_1$ under a graph homomorphism $F: G_1 \rightarrow G_2$,  then for all nodes $v'$ in $G_2$, $deg_{G_2}(v') \le \sum_{v \in F^{-1}(v')} deg_{G_1}(v)$, where $deg_{G}(x)$ is the degree of the node $x$ in a graph $G$.
\end{observation}

\section{Results}
\label{sec: Results}

\subsection{Upper Bounds}
\label{subsec: upperbounds}

%%% Fix the lemmas. Use FG instead of G when meaning the virtual graph %%%

As earlier, let $\G$ be the graph of the network, $\FG$ the Forgiving Graph, and  $\G'$  the graph consisting solely of the original nodes and insertions without regard to deletions and
healings. Let $\G_{T}$, $\FG_{T}$ and $\G'_{T}$ be these graphs at time $T$. 

\begin{lemma}
\label{lemma: hnodecounts}
Given the edge $(v,x)$ in $\G'_{T}$,
\begin{enumerate}
\item \label{lmpart: onehelper} There can be at most one helper node in $\FG_T$ corresponding to $(v,x)$.
\item  \label{lmpart: twohelper} During the Repair phase, there can be at most two helper nodes corresponding to the edge $(v,x)$. Moreover, one of these could also be an anchor in $BT_{v}$.
\end{enumerate}
\end{lemma}
\begin{proof}
 There is only one `real' node in $\FG_T$ corresponding to an edge in $\G'_T$ (Figure~\ref{fig: fgseries}). Let us refer to this node  as simply $v$. Moreover, $v$ can only be a leaf node of a $\RT$, and a helper node can only be an internal node. 
%We now prove part~\ref{lmpart: onehelper} by contradiction. 

%Part 1: At most one helper node for a real node.
We prove  part~\ref{lmpart: onehelper} by contradiction. Suppose there are two helper nodes in $\FG_T$ corresponding to the real node $v$. Let us call these nodes $v'$ and $v''$. The following cases arise:

	%Case 1: nodes in different RTs
	\begin{enumerate}[i.]
	
	\item \emph{$v'$ and $v''$ belong to different $\RT$s:}
	\label{pfcase: sameRT}
	
	\begin{figure}[h!]
	\centering
	\includegraphics[scale=0.85]{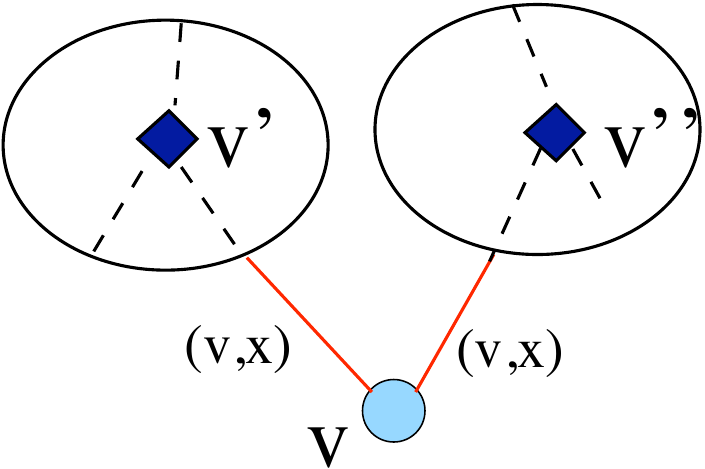}
	\caption{Proof by contradiction: Case 1. Two helper nodes in different $\RT$s.}
	\label{fig: OneHelperCase1}
	\end{figure}
	
	This case is depicted in figure~\ref{fig: OneHelperCase1}. We assume that both $v'$ and $v''$ exist but that they are in different $\RT$s.
	By the representative mechanism, a helper node is created only if the real node that simulates it is the representative of a node (e.g. in line~\ref{algocode: mergeeqrep} in Algorithm~\ref{algo: ComputeHaft}). By definition, the representative of a node is a unique leaf node in the subtree headed by that node in its $\RT$. If both $v'$ and $v''$ exist and  belong to different $\RT$s, this implies that node $v$ exists as a leaf node in two different $\RT$s. This is a contradiction.
	 
	%Case 2: Same RT 
	\item \emph{$v'$ and $v''$ belong to the same $\RT$:}
	
	 Without loss of generality, assume that the  $v''. height \ge v'.height$. The following cases arise:
	 
		 \begin{enumerate}
		 % Case 2a: Same RT. nodes in different subtrees.
		  \item \emph{$v'$ is a node not in the subtree headed by $v''$:}
		
		 \begin{figure}[h!]
		\centering
		\includegraphics[scale=0.85]{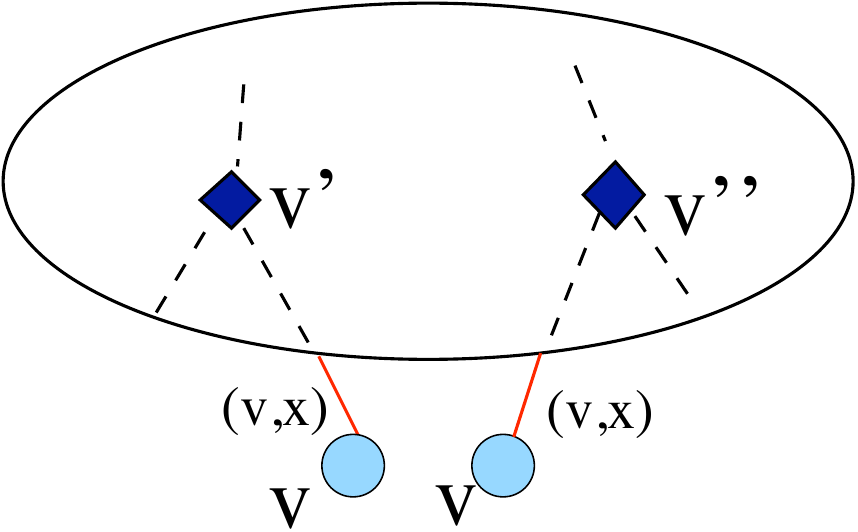}
		\caption{Proof by contradiction: Case 2(a). Two helper nodes in same  $\RT$, but in different subtrees.}
		\label{fig: OneHelperCase2}
		\end{figure}

		  This case is  shown in figure~\ref{fig: OneHelperCase2}. We assume  that both $v$ and $v''$ exist, and that they are in the same $\RT$ but in different subtrees i.e. $v''$ is not an ancestor of $v'$. The proof  is similar to that of case~\ref{pfcase: sameRT}.   The representative mechanism and definition of a representative implies that node $v$ was a
		representative in two non-intersecting subtrees in the same $\RT$. This implies that node $v$ occurs as a leaf
		twice in that $\RT$. This is not possible. 
		 
		 %Case 2b: Same RT. Same subtree.
		  \item \emph{$v'$ is a node in the subtree headed by $v''$:}
		  
		  \begin{figure}[h!]
		\centering
		\includegraphics[scale=0.85]{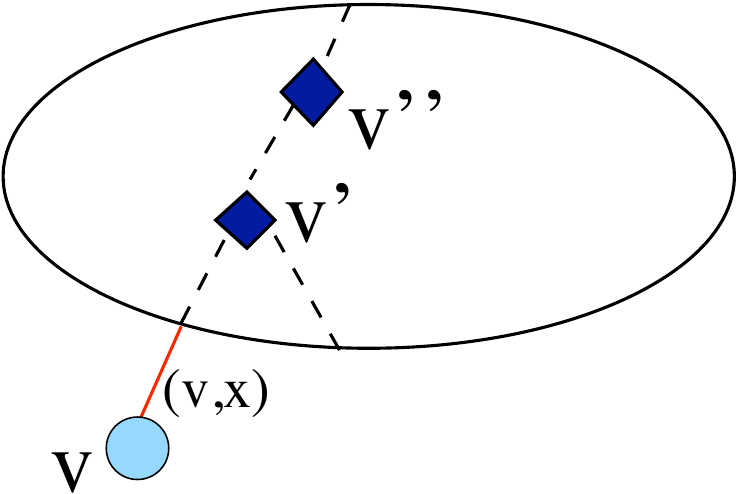}
		\caption{Proof by contradiction: Case 2(b). Two helper nodes in the same subtree.}
		\label{fig: OneHelperCase3}
		\end{figure}
		
		     This case is shown in figure~\ref{fig: OneHelperCase3}. We assume  that both $v$ and $v''$ exist, and that they are in the same $\RT$ and moreover  $v''$ is not an ancestor of $v'$. Note that by the representative mechanism, when two nodes are to be joined, the representative of one of them provides the single node that will be their parent. This new node inherits the other (unused) representative as its representative. The tree gets built up bottom up with available representatives propagating upwards. Thus, node $v'$ will be created before node $v''$.
		 By definition of a representative, neither $v'$ nor any of its ancestors can now have $v$ as a representative since $v$ is now already simulating a helper node. Thus, $v''$ was created without any of its children having $v$ as a representative. However, this is not possible.   
		 % once a helper node is created, its identity cannot change. Thus, the only possibility is $v''$ is created after $v'$.
%		Let node $y$ be the child of node $v''$ that had $y.\representative = v''$ when $v''$ was created. However, 
%		   $y.\representative$ could not have been $v''$, since by definition, $y.\representative$ has to be the unique leaf 
%		node not  simulating a helper node in the subtree of $y$, but $v$ is already simulating $v'$ in $v''$'s subtree.
		 \end{enumerate}
	\end{enumerate}
%End of proof of single helper node per real node

 Now, we prove part~ \ref{lmpart: twohelper}. As stated earlier, at each stage of the merge procedure,
$\RTfragment$s in $\BT_{v}$ will merge with their parent. Suppose that $v'$ is a helper node simulated by real node $v$,
and $v'$ is not part of any complete subtree in such a $\RTfragment$. This means that $v'$ will be marked red
and removed when this stage of merge is completed (Refer Figure~\ref{fig: Anchormerge}). Let node $y$ be the root of the
 complete subtree (i.e. a primary root in that $\RTfragment$) that has $v$ as a leaf node. 
Node $v'$ is an ancestor of node $y$ since $v'$ cannot be $y$'s descendant. By definition, $y.\representative = v$, since
$v$ will be the unique leaf node in $y$'s subtree not simulating a helper node in that subtree.
 When the trees are being merged, $v$ may be asked to create another helper node. Thus, $v$ may have two helper nodes.
 Also, each $\RTfragment$ has exactly one anchor node. This anchor may be $v'$ or another node. Thus, in the
repair phase, a real node may simulate at most two helper nodes, and one of these helper nodes may be an anchor.
 However, node $v'$ will be removed as soon as this stage is completed, and if $v'$ was an anchor, a new anchor is
chosen from the existing nodes. Since at the end of the merge, $\BT_{v}$  collapses to leave one $\RT$, the extra helper
nodes and the edges from the anchor nodes are not present in $\FG_T$, thus, not contradicting part~\ref{lmpart: onehelper}.
\end{proof}

\begin{lemma}
\label{lemma: cost}
After each deletion, the repair phase requires the sending of at most $O(d \log n)$ messages, each of
length $O(log^2 n)$.  Moreover, this can be done in parallel by the neighbors of the deleted node, in time $\mathrm{polylog}(d,n)$.
%After each deletion, the repair can take
%% $O(\log d\log n)$ 
%$\mathrm{polylog}(d,n)$ time to exchange $O(d \log n)$ messages of size $O(log^{2}
%n)$, where $d$ is the degree of the deleted node.
\end{lemma}

\begin{proof} There are mainly two types of messages exchanged by the algorithm. They are the probe messages sent by the
\textsc{FindPrRoots()} (Algorithm~\ref{algo: findprroots}) within a $\RT$ and the messages containing the information
about the primary roots exchanged by the anchors in $\BT_{v}$ and among the primary roots themselves (Algorithm~\ref{algo:
haftmerge}: \textsc{ComputeHaft()}). Let $\size(\BT_{v})$ be the number of $\RT$s of $\BT_{v}$. Since a helper
node can split a $\RT$ into maximum 3 parts, and there can be at most $d$ helper nodes, where $d$ is the degree of the deleted node $v$, $\size(\BT_{v}) \le 3d$.
Now, let us calculate the number of messages:
\begin{itemize}
%Probe messages
% anchor of a RTportion. say, anchor is leftmost node. It meets 2 log numnodes on the way, each has a a primary root except the root of the tree. thus, search msgs = 2. 2 log numnodes, anchor conf. msgs ~ for 2 log numnodes, but these are piggybacked on the way back, as each node waits for its msgs to come back. If not a primary root, even then node returns a confirmation msg. thus, each msg has a reply. total msgs = 2.2.2\log numnodes = 8\log numnodes.  
\item \emph{Probe messages (Algorithm~\ref{algo: findprroots})}: A probe message is generated by an anchor of a $\RT$. This is similar to the \emph{Strip} operation
(Section~\ref{subsec: haftstrip}). The path that the probe message follows is the direct path from the originating
node to the rightmost node of the $\RT$. At most 2 messages can be generated for every node on the way.
Each node waits for a reply to its message. If it had a neighbor as a primary root, it will hear back from it with the root's identity. If it had an anchor as a neighbor, it will get an 'end of path' message. This node will then reply back to the message it had received from the its neighbor on the path from the requesting anchor.  Thus, each message generated by the request from the anchor will get a reply back with identities of one or multiple primary roots or end of path messages. By the property of hafts, each node on this path will have a primary root as a neighbor, thus, the longest path a message can take is equal to the diameter of the tree, which is the longest path in the tree.
% Further, there can be one confirmatory message transmitted from a primary root back to the anchor. 
Let $numnodes$ be the number of nodes and $numprobes$ be number of probe messages sent in a single $\RT$. The length of  the longest path is $2\log\ numnodes$.
Thus, 
\begin{eqnarray*}
 numprobes & \le & 2.2.2 \log numnodes  \\
   & \le & 8 \log n  
\end{eqnarray*}

%List messages
\item \emph{Exchange of primary roots lists (Algorithm~\ref{algo: haftmerge})}:  At each step of Algorithm~\ref{algo:
bottomupmerge} (\textsc{BottomupRTMerge()}), leaves in $BT_{v}$ merge with their parents. Let $rtlistmsgs$
be the number of messages exchanged for every such merge.  The anchors of the leaves of $BT_{v}$ send their primary
roots lists to the parent, which in turn can send both it's list and the sibling's list to the child.  Thus,
$rtlistmsgs = 4$. In addition, every anchor will send this list to the primary roots in its $\RT$, generating at most
another $\log n$ messages (Let us call this $AtoRmsgs$).
\end{itemize}
As stated earlier, in the $BT_{v}$, leaves merge with their parents. The number of such merges
before we are left with a single $\RT$ is  $\lceil\size(BT_{v})/2 - 1\rceil$. Also, at most 3 $\RT$s are involved in
each merge. Let $totmessages$ be the total number of messages exchanged. Hence,

\begin{eqnarray*}
totmessages & = & \lceil \size(BT_{v})/2 - 1 \rceil\\
& &( 3 ( numprobes  + AtoRmsgs) + rtlistmsgs )\\
 & \le & \lceil 3d/2 - 1 \rceil ( 27 \log n + 4 )\\
 & \in & O(d \log n)
\end{eqnarray*}
 In $BT_{v}$, leaves and their parents merge. This can be done in parallel such that each time the level of $BT_{v}$
reduces by one. Within each $\RT$, the time taken for message passing is still bounded by $O(\log n)$
assuming constant time to pass a message along an edge. Since there are at most $\lceil log d \rceil$ levels, the time
taken for passing the messages is $O(\log d \log n)$ i.e $\mathrm{polylog}(d,n)$. The biggest message exchanged may have information
about the primary roots of upto two $\RT$s. This may be the message sent by a parent $\RT$ in $BT_{v}$ to its children
$\RT$. Since there can be at most $O(\log n)$ primary roots, the size of messages containing their ID is $O(\log^{2} n)$.
\end{proof}

We now state our main result.  Recall that $\G_T$ is the graph produced after $T$ steps of our algorithm, while $\G'_T$ is the graph resulting from the insertions only, with no deletions or repairs.
 \begin{theorem}
 The Algorithm $\FGraph$ has the following properties:
\label{theorem: forgiving}
\begin{enumerate}
\item\label{th: degree} 
 \emph{Degree increase:} 
 For any node $v$ in $V(G_T)$, after any number of time steps, $T$, the degree of $v$ in $G_T$ is at most 3 times the degree of $v$ in $G'_T$.
% For any node $v$, $d(v,G_{T}) \le 3d(v,G'_{T})$, where $d$ is the degree of the node $v$. 
 % The Forgiving Graph increases the degree of any vertex by at most $3$ times over its degree in $G'$.
\item\label{th: stretch} \emph{Stretch:} 
For any nodes $x,y$ in $V(G_T)$, after any number of time steps, $T$, the distance between $x$ and $y$ in $G_{T}$ is at most
$\log(n)$ times the distance in $G'_{T}$.
%For any pair of nodes $x$ and $y$,\\
% $distance(x,y, G_{T}) \le (\log n) \times distance(x,y,G'_{T})$.
\item \label{th: cost} 
\emph{Cost:} After each deletion, the repair phase requires the sending of at most $O(d \log n)$ messages, each of
length $O(log^2 n)$.  Moreover, this can be done in parallel by the neighbors of the deleted node, in time $\mathrm{polylog}(d,n)$.
% \emph{Cost:} After each deletion, the repair takes at most $O(\log d\log n)$ time with $O(d
%\log n)$ messages of size at most $O(log^{2} n)$, where $d$ is the degree of the deleted node.
%The Forgiving Graph increases the distance between any two nodes of the network by no more than $\log n$
%times their original distance in $G'$.
%\item \label{th: cost}   The latency per deletion and number of messages sent per node per 
 % deletion is $O(1)$; each message contains $O(1)$ bits and node IDs. 
\end{enumerate}
\end{theorem}

\begin{proof}
 Part~\ref{th: degree} follow directly by construction of our algorithm.   Note that for a real
node $v$ in $\FG_T$, any degree increase for $v$ is imposed by the  edges of its helper node to $\hparent$($v$) and $\hchildren(v)$. From lemma~\ref{lemma: hnodecounts} part~\ref{lmpart:
onehelper}, we know that, in $\FG_T$, node $v$ can play the role of at most one helper node for any of its neighbors in $\G'_T$ at  any time (i.e. equal to the degree of $v$ in $\G'_{T}$ ). The number of $\hchildren$ of a helper node are never more than $2$, because the reconstruction trees are binary trees. 
Thus the total degree of $v$ in $\FG_{T}$ is at most $3$ times its degree in $\G'_T$. From observation~\ref{obs: homodegree} and noting that $\G_T$ is a homomorphic image of $\FG_T$, we can see that  the degree of $v$ in $\G_{T}$ is at most $3$ times its degree in $G'_T$ .

%** Part~\ref{th:cost} also follows directly by the
%construction of our algorithm, noting that, because the virtual nodes all have degree at most $3$, healing one deletion 
%results in at most $O(1)$ changes to the edges in each 
%affected reconstruction tree.  In fact, the changes to 
%$\RT(w)$ for an affected node $w$ do not require new information,
%which allows these messages to be computed and distributed in parallel.
%We leave the details to the full version due to space constraints. **
%\end{proof}

We next show Part~\ref{th: stretch}.  We show  that the stretch of the Forgiving Graph $\FG_T$ is  $O(D \log n)$, where $n$ is the number of nodes in $\G_{T}$. The distance between any two nodes $x$ and $y$ cannot increase by more than the factor of the longest path in the largest $\RT$ on the path between $x$ and $y$. Since the number of nodes in $\FG_T$ is $O(n)$, This factor is $\log n$ at the maximum. Since there is a homomorphism from the graph $\FG_T$ to $\G_T$, the result  follows directly from observation~\ref{obs: homostretch}.

The proof of Part~\ref{th: cost} 
%is left to the journal version due to space constraints.
 follows from Lemma~\ref{lemma: cost}. Note that besides the communication of the messages
 discussed, the other operations can be done in constant time in our algorithm.
 \end{proof}

\pagebreak

\subsection{Lower Bounds}

\label{subsec: lowerbounds}

\begin{theorem}
Let $n$ be a positive integer, $\alpha \ge 3$ and  $\beta = \frac{1}{2} (\log_{\alpha } (n-1) - 1)$.  Then there exists a
graph on $n$ vertices and a vertex deletion  such that any way of repairing this deletion
% any self-healing algorithm in our model 
 under our model must either increase the degree of some node by more than a factor of $\alpha$, or it must increase the distance between some pair of nodes by at least a factor of $\beta$.

%Consider any self-healing algorithm that ensures that: 1) each node increases its degree by a multiplicative factor of 
%at most $\alpha$, where $\alpha \geq 3$; and 2) the stretch of the graph increases by a multiplicative factor of at most
%$\beta$. Then, for some initial graph with $n$ nodes, it must be the case that
% $\beta \geq \frac{1}{2} \log_{\alpha-1}( n - 1)$.
%Then for any positive $\Delta$, for some initial
%graph with maximum degree $\Delta$, it must be the case that
%$\beta  \geq \frac{1}{2} [\log_{\alpha} n - 1]$.
\end{theorem}

\begin{proof}

\begin{figure}[h!]
\centering
\includegraphics[scale=0.8]{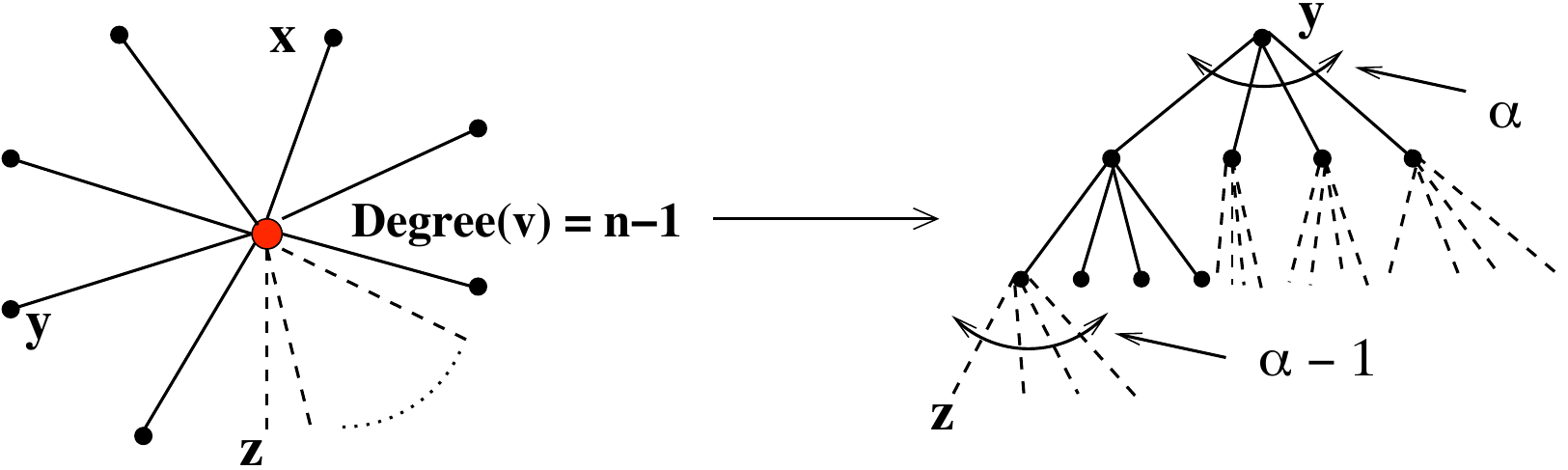}
\caption{Deletion of the central node $v$ of a star leads to an increase in the stretch. Here, the healing algorithm can increase the degree of any node by at most a factor of $\alpha$.}
\label{fig: lowerboundFG}
\end{figure}

Let $G$ be a star on $n$ vertices, where $x$ is the root node, and $x$ has an edge with each of the
other nodes in the graph. The other nodes (besides $x$) have a degree of only 1. Let $G'$ be the graph created after the
adversary deletes the node $x$.  Consider a breadth
first search tree, $T$, rooted at some arbitrary node $y$ in $G'$.  We know that the self-healing algorithm can increase
the degree of each node by at most a factor of $\alpha$, thus every node in $T$ besides $y$ can have at most $\alpha-1$ children. 
Let $h$ be the height of $T$.  Then we know that $1 + \alpha \sum_{i=0}^{h-1}
(\alpha-1)^{i} \geq n-1$.  This implies that  $(\alpha)^{h+1} \geq n-1$ for $\alpha \geq 3$, or ${h + 1} \geq
\log_{\alpha}(n-1)$.  Let $z$ be a leaf node in $T$ of largest depth. Then, the distance between $y$ and $z$ in $G'$ is $h$ and the distance 
between $y$ and $z$ in $G$ is 2. Thus, $\beta \ge h/2$, and  $2\beta 
\geq \log_{\alpha} (n-1) - 1$, or $\beta \geq \frac{1}{2}(\log_{\alpha}( n-1) - 1)$. This is illustrated in figure~\ref{fig: lowerboundFG}.
% The distance between
%$u$ and $v$ in $G$ is $2$, and the diameter of $G'$ is at least $h$. Thus, $\beta = h/2$,
% Raising $\alpha$ to the power of both sides, we get that $\alpha^{2\beta + 1} \geq
%\Delta$. Taking log of both sides and rearranging, we get $\beta  \geq \frac{1}{2} [\log_{\alpha} \Delta - 1]$.
\end{proof}
%\medskip
\noindent
Note that the upper bound on the degree increase and stretch of our algorithm is within a constant factor of matching
this lower bound.
%We note that this lower-bound compares favorably with the general
%result achieved with our data structure.  
%The Forgiving Tree can be
%modified so that it ensures that 1) the degree of any node increases
%by no more than $\alpha$ for any $\alpha \geq 3$; and that the
%diameter increases by no more than a multiplicative factor of $\beta
%\leq 2 \log_{\alpha} \Delta + 2$. 

\section{Conclusion}
In this chapter, we have presented a distributed data structure that withstands
repeated adversarial node deletions by adding a small number of new
edges after each deletion.  Our data structure is efficient and ensures two key
properties, even in the face of both adversarial deletions and adversarial insertions.
First, the distance between any pair of nodes never increases by more than a $\log n$ multiplicative factor than what the distance would be without the adversarial deletions.  Second, the degree of any node never increases by more than a $3$ multiplicative factor.

Several open problems remain including the following. Can we design algorithms
for less flexible networks such as sensor networks?  For example, what
if the only edges we can add are those that span a small distance in
the original network?  Can we extend the concept of
self-healing to other objects besides graphs?  For example, can we
design algorithms to rewire a circuit so that it maintains its
functionality even when multiple gates fail?

%End of chapter

%%%%%%%%%%%%%%%%%%%%%%%%%%%%%%%%%%%%%%%%%%%%%%%%%%%%%%%%

% Chapter: future directions

\chapter{Future Directions}
\label{chapter: FD}

\begin{epigraphs}
%\qitem{ \epitext{ I've done the calculation and your chances of winning the lottery are identical whether you play or not.}%
%}%
% { \epiauthor{Fran Lebowitz}
% }
% \qitem {\epitext{There is no such thing as a failed experiment, only experiments with unexpected outcomes}}%
%{\epiauthor{Richard Buckminster Fuller}}
%
%\qitem{ \epitext{Psychiatry enables us to correct our faults by confessing our parents' shortcomings.} }%
%  {\epiauthor{Laurence J. Peter}%
%  }
%  \qitem{ \epitext{I've gone into hundreds of [fortune-teller's parlors], and have been told thousands of things, but nobody ever told me I was a policewoman getting ready to arrest her.}%
%  }%
%  { \epiauthor{New York City detective}%
%  }
\qitem {\epitext{There is no such thing as a failed experiment, only experiments with unexpected outcomes}}%
{\epiauthor{Richard Buckminster Fuller}}
\end{epigraphs}

\noindent In this chapter, we point out some related open problems and discuss the future directions in which this research can be extended. 
%  We are already working on some of these ideas.
%Section \ref{sec: Sensor} presents a current idea we are discussing to extend self-healing to the specific
%area of sensor networks. Section \ref{sec: Social} discusses questions relating our research to Social networks.

%\pagebreak

\section{Empirical study of self-healing algorithms beyond assumptions}

How would our algorithms perform beyond the assumptions of the model we have used? There are certain assumptions our algorithms make and we would like to know how well our algorithms perform even when those assumptions don't hold e.g. our model assumes single failure before each recovery. How would our algorithms perform if there are multiple failures in close physical or temporal proximity? Akin to an ecological disaster, we would see how the algorithms perform if a set of nodes (a clump of species) are simultaneously deleted or there are cascading failures. We will also place restrictions on the topology of the network and add additional rules to the algorithm to simulate different networks found in nature.  In particular, we have already begun work on simulating $\FGraph$ for these purposes.  

\section{Routing in Self-healing structures}
\label{sec: FD-Routing}
Can we implement efficient updates to routing tables? Small changes to the network, e.g., deletion of an edge or a node can lead to major changes in the tables. Can our algorithms keep track of these? Self-healing routing is an important research question especially given the dynamic nature of modern networks~\cite{Branch07SelfHealingRouting, Gui03short:self-healing, Meyer08Self-healingMultipathRouting}. We will like to propose solutions in our framework which incorporate routing in addition to the invariants we already maintain. This could involve proposing efficient routing schemes to go with our self-healing structures or developing new structures that help routing.

\section{Load balanced Self-healing}
\label{sec: FD-LoadBalancing}
Trees are not the best structure for effective load balancing e.g. in a balanced tree, half of the paths will go through the root.  Can we improve the load balancing using a different self-healing data structure? Good load balancing would be ensured if upon healing there are not likely to be bottlenecks for communication traffic. There has been some previous work on load balancing in structured P2P systems~\cite{Karthik03loadbalancing}. This may also be related to the earlier question (Section~\ref{sec: FD-Routing}) on self-healing Routing. There are many interesting ideas we are looking at out there which may potentially contribute to a solution that we are looking at e.g. The Chord P2P structure~\cite{Stoica03Chord}, Skip graphs~\cite{Stoica03Chord}, and Small-world network models~\cite{Kleinberg00SmallWorld}.
%We have already begun work on this, trying  an algorithm based on skip graphs~\cite{AspnesS2007}.

\section{Self-healing in Sensor Networks}
\label{sec: FD-Sensor}
 Directional antennas are increasingly becoming important in sensor networks e.g. ~\cite{Huang02topologycontrol}. They also allow us to use our concept of self-healing where we have an edge in the underlying graph for two nodes in communication with each other.  In a wireless ad-hoc network multi-hop connectivity can be easily lost when a transceiver
goes silent and does not relay messages any longer. Successful
pairwise communication occurs in a wireless network only in the
absence of interference which is usually achieved by
frequency or time or code division multiplexing, i.e., by assigning
non-interfering channels (colors) to the pairwise links (edges) that
are necessary for global connectivity. Therefore, to restore
connectivity after a node failure, it is important to also restore an
interference-free channel assignment for the pairwise links in the
repaired network.

In the disk graph model of a wireless ad-hoc network, there are $n$ transceivers and the transmission and reception range  of a transceiver $u$ is a disk $D(u)$ centered at $u$ with radius $r(u)$. The transceivers are vertices of a directed graph $G=(V,E)$. The directed edge $u \rightarrow v$ belongs to $E$ if and
only if $v$ is in $D(u)$. Two transceivers $u$ and $v$ can communicate
directly (without intermediate hops) if and only if both directed edges $u  \rightarrow v$ and $v \rightarrow u$ are present.  We say that a pair of transceivers $(u,v)$ are connected if there exists a path between $u$ and $v$
consisting of bidirectional edges. Two edges $(u,v)$ and $(p,q)$ may exhibit \emph{primary interference} if either $u \in \{p,q\}$ or $v \in \{p,q\}$, i.e., if they have a common vertex. They may exhibit \emph{secondary interference}  if they share a common edge i.e. there is an edge whose one end-point is either $u$ or $v$ and the other is either $p$ or $q$.
%there exists at least one of the eight directed edges between one of ${u,v}$ and one of ${p,q}: u \rightarrow p,
%u \rightarrow q, v \rightarrow p, v \rightarrow q, p \rightarrow u, p \rightarrow v, q \rightarrow u, q \rightarrow v$.
%\end{enumerate}

%\subsection*{Model of interference}
%Bidirectional edges are the only edges that can used for
%communication, but two such edges may interfere with each other, which
%means that they must be assigned different channels for simultaneous
%interference-free communication. All edges are responsible for
%contributing to interference, which is of two types:
%\begin{enumerate}
%\item \textbf{Primary interference:} Two bidirectional edges $(u,v)$ and $(p,q)$ are
%in primary interference if either $u \in {p,q}$ or $v \in {p,q}$, i.e., if
%the two edges are incident on a common vertex.
%\item \textbf{Secondary interference:} Two bidirectional edges $(u,v)$ and $(p,q)$
%are in secondary interference if there exists at least one of the
%eight directed edges between one of ${u,v}$ and one of ${p,q}: u \rightarrow p,
%u \rightarrow q, v \rightarrow p, v \rightarrow q, p \rightarrow u, p \rightarrow v, q \rightarrow u, q \rightarrow v$.
%\end{enumerate}

The question of maintaining this interference-free communication graph can then be reduced to maintenance of strong edge coloring where each edge is assigned a color such that no interfering edge shares a color, in the presence of an adversary. We can again assume that the adversary removes one node at a time and the neighbors are alerted of this. One possible approach is to use  the self-healing idea such as the notion of "wills" (as in Forgiving Tree) to compute an efficient, local way to repair the network by re-connecting (a subset of) the nodes in the neighborhood of the deleted vertex. We can then use a distributed, randomized algorithm (i.e., a protocol), a la Luby \cite{Luby-STOC85}, to implement the repair. In \cite{thite-percom06} Barett at al adapted Luby's algorithm to the case of distance-2 coloring and showed that it was sufficient for each node to know the so-called ``active degree" to determine its wake-up probability. Much of the ideas in this section were from discussion with Shripad Thite (Google).

\section{Self-healing/ Behavioral robustness in Social Networks}
\label{sec: Social}

% This Section is from discussions with Willemien Kets (Omidyar Fellow, Santa Fe Institute). 
There are some interesting avenues to explore in the context of robustness in social networks. Some of the questions in this context may involve achieving behavioral robustness as opposed to topological invariance like we have used so far in our self-healing work.
One of  the question we want to explore is the following:  Does a phenomena like the minority game which normally achieves equilibrium achieves homeostasis even in the presence of an adversary. This has the flavor of behavioral invariance.
%\subsection{Homeostasis in an adversarial minority game}
%We talked briefly about a learning process that can converge to an
%"anti-coordination" equilibrium in the minority game if the agents'
%memory length is equal to 2.
In \cite{Kets-Minority2007}, Willemien et al show that a
learning process in which players best-reply to a history of limited
length and in which they have a preference for more recent best
replies ("recency bias") eventually settles down in one of the pure
Nash equilibria (optimal anti-coordination) if the memory length of
players is at least 2. The proof uses the fact  that
you can construct a path from any initial history to a state where
players play according to a pure Nash equilibrium in each following
time period, and that such a path will occur with probability 1 in the
long run. This gives  an algorithm for reaching  an
anti-coordination equilibrium if players best-reply to beliefs based
on a limited history of play and they have a recency bias.

Our idea is to study what will happen to the equilibrium if an adversary (in the sense of external perturbations) is
introduced into the mix. Will we still achieve the anti-coordination equilibrium?, and what are the implications?

% Note that
%the recency bias is crucial: without the recency bias players would
%keep choosing the two actions with positive probability (this result
%follows from Hurkens (1995)). I'm not sure whether this is of direct
%interest to you, but perhaps it is a good starting point for further
%work.

\section{Self-* problems}
\label{sec: FD-Self*}
Can we go beyond Self-healing (demand even stronger guarantees)? There is strong interest in the so-called self-* algorithms. We have earlier discussed these properties in Section~\ref{sec: Intro-Self-*}.
One such objective is self-stabilization. A distributed system that is self-stabilizing will end up in a correct state no matter what state it is initialized with. This is a highly desirable property for distributed systems, and worth investigating. Another direction would be too look at the network layers themselves. Our work is based on overlay networks. However, it may be beneficial to consider what happens below that layer, at the physical layer itself, to come up with practical, efficient and robust network designs.

\section{Evolution of social and computer networks\\ and study of group formation}
\label{sec: FD-dynamicnets}
 
  It is important to study the mechanisms behind the  formation and evolution of networks, particularly, networks like the Internet, and social networks, in particular with regards to their stability and self-* properties.
 Techniques from various areas like  game theory can often be profitably applied here.  There are many models seeking to explain network formation e.g. ~\cite{LSSTOC09}; Some models seek to explain the formation of networks in a game theoretic manner by having nodes as players making connections (edges) with other players to maximize their utility function~\cite{FLMPSPODC03}.
 
There is  interest in discovering mechanisms for formation of groups.  Our attempts at simple toy models suggest this is a difficult problem. However, there has been interesting research in this area incorporating both theoretical and experimental (including field observations) work.
 Dan Rubenstein, an ecologist from  Princeton, collected data on the social structure graphs of the thriving plains zebra and the endangered Grevy's zebras, from the plains of Africa.  He and Tanya Berger-Wolff, from University of Illinois Chicago, then modeled the Zebra's group behavior by looking at the network of their social behavior to find interesting patterns~\cite{TanyaZebraScienceNews, TantiSIGKDD07}.  Jared Saia and Tanya Berger-Wolff have also proposed mathematical and computational framework that enables analysis of dynamic social networks and that explicitly makes use of information about when social interactions occur~\cite{TanyaJaredKDD06}.
 There are also many interesting data sets  e.g. on mobile phone usage patterns~\cite{MITRealityDataSet}, which can help investigate such questions as  routing of messages, group formation and social motivation. 
 
 There are many interesting questions: How do groups self-heal i.e are groups sensitive to perturbation, leaving and joinings of agents? What are the mechanisms that explain formation and dissolution of groups in real networks? Can we propose such game theoretic cost functions? When agents cooperate to form a group, how does that influence formation of other groups? 

\section{Byzantine agreement: Distributed computing in  presence of byzantine faults}
\label{sec: FD:-Byzantine}
This section owes itself to  discussions with Professor Valerie King. The failure models we have considered so far ignore byzantine faults (e.g. by adversarial code corruption), but it is important for the network to be able to function/self-heal in presence of these faults. 
A fundamental problem in distributed computing is that of coordinating behavior by processors in the presence of an adversary who controls a constant fraction of processors. At its most basic, it is
formulated as the  Byzantine Agreement Problem.  Each of $n$ processors are given an input bit; they execute a protocol, at the end of which all output the same bit equal to one of their input bits. 

This problem, in the asynchronous model (an adversary controls the order in which messages are delivered), is known to be impossible to solve deterministically in the full information model, i.e., if the adversary has access to all messages sent and there are no cryptographic assumptions made.
A randomized protocol exists in which each processor has private random bits but it requires an exponential number of messages. Both of these results were shown in the 1980's. Last year, Kapron, Kempe, King, Saia and Sanwalani~\cite{KKKSSSODA08} showed a polylogarithmic time protocol which succeeds with high probability for this problem if the choice of corrupt processors is made independently of the random bits, and the adversary is ``non-adaptive". In addition, King and Saia showed that $\tilde{O}(n^{3/2})$ total  bits of communication suffice~\cite{KingSDISC09}. Without the assumption, all known Byzantine Agreement protocols in the synchronous model (where messages are delivered in rounds) and the asynchronous model use $\Omega(n^2)$ messages, even with private channels and cryptographic assumptions.

Several intriguing problems remain open, in decreasing order of difficulty:

\begin{enumerate}
\item
Can we close the gap between the lower bound of $n^2$ and the upper bound of exponential time for asynchronous Byzantine agreement (with an adaptive adversary) in the full information model?

\item
Can we do Byzantine agreement with cryptography or private channels with an adaptive adversary in $o(n^2)$ bits per processor? Is it possible to prove a nontrivial lower bound here? Is there a practical protocol for this? (Recently, King and Saia have published an algorithm which solves this problem in the synchronous model~\cite{KingSaia10BreakingBarrierArxiv}. Their algorithm assumes private channels and takes only  $\tilde{O}(\sqrt{n})$ bits per processor  and has polylog latency).

\item
Can we load balance the Byzantine agreement problem so that no processor uses more than $o(n)$ bits with the assumption, for the synchronous model? For the asynchronous model?

\item
Can we enhance the protocols designed by King and Saia, so that they are robust to an adversary who can also
remove and insert new nodes, some of which are corrupt, still in the full information model? 

\end{enumerate}

Techniques for proving lower bounds for randomized distributed problems like this are scarce and may involve techniques from communication complexity. There is a well known method for deterministic lower bounds in distributed computing using algebraic topology, but there is no known extension to randomized algorithms. This would be interesting to explore. Any answer to the first question will be a major breakthrough in a widely studied problem area that has been open for over 25 years.

\end{document}